\documentclass[11pt]{article}
\usepackage[ruled,vlined,linesnumbered,hangingcomment]{algorithm2e}
\usepackage[
  colorlinks=true,
  linkcolor=blue!92,
  citecolor=teal!85!black,
  urlcolor=magenta!85
]{hyperref}
\usepackage{amsfonts,amssymb}
\usepackage{comment}
\usepackage{titlesec}
\usepackage{authblk}
\usepackage{tikz}
\usepackage{bbding}
\usepackage{pifont}
\usepackage{titletoc}
\usepackage{hhline}
\usepackage{amsmath}
\usepackage{listings}
\usepackage{verbatim}
\usepackage{makecell}

\setcellgapes{3pt}
\makegapedcells

\usepackage{tabu}
\usepackage{graphicx}
\usepackage{enumitem}
\usepackage{mathrsfs}
\usepackage{setspace}
\usepackage{tcolorbox}
\usepackage{cite}
\usepackage{bbm}
\usepackage{mathtools}
\usepackage{titletoc}
\usepackage{geometry} 
\usepackage{url} % not crucial - just used below for the URL 
\usepackage{listings}
\usepackage{amssymb,amsmath,amsfonts,latexsym}
\usepackage{amsmath,graphicx,bm,xcolor,url}
\usepackage[caption=false]{subfig} 
\usepackage{array}%array and tabular environments
\usepackage{verbatim}
\usepackage{bm}
\usepackage{verbatim}
\usepackage{textcomp}
\usepackage{mathrsfs}
\usepackage{relsize}
\usepackage{subfig}
 \usepackage{amsthm}

\catcode`~=11 \def\UrlSpecials{\do\~{\kern -.15em\lower .7ex\hbox{~}\kern .04em}} \catcode`~=13 

\allowdisplaybreaks[3]

\newcommand{\nn}{\nonumber}

\newcommand{\calC}{\mathcal{C}}
\newcommand{\calD}{\mathcal{D}}

\newcommand{\calH}{\mathcal{H}}

\newcommand{\calK}{\mathcal{K}}
\newcommand{\calL}{\mathcal{L}}
\newcommand{\calM}{\mathcal{M}}
\newcommand{\calN}{\mathcal{N}}
\newcommand{\calO}{\mathcal{O}}
\newcommand{\calP}{\mathcal{P}}

\newcommand{\calS}{\mathcal{S}}
\newcommand{\calT}{\mathcal{T}}
\newcommand{\calU}{\mathcal{U}}
\newcommand{\calV}{\mathcal{V}}
\newcommand{\calW}{\mathcal{W}}
\newcommand{\calX}{\mathcal{X}}

\newcommand{\ba}{\mathbf{a}}
\newcommand{\bA}{\mathbf{A}}

\newcommand{\bD}{\mathbf{D}}
\newcommand{\be}{\mathbf{e}}

\newcommand{\bF}{\mathbf{F}}
\newcommand{\bg}{\mathbf{g}}
\newcommand{\bG}{\mathbf{G}}
\newcommand{\bh}{\mathbf{h}}

\newcommand{\bI}{\mathbf{I}}

\newcommand{\bL}{\mathbf{L}}

\newcommand{\bM}{\mathbf{M}}

\newcommand{\bp}{\mathbf{p}}
\newcommand{\bP}{\mathbf{P}}
\newcommand{\bq}{\mathbf{q}}

\newcommand{\bR}{\mathbf{R}}

\newcommand{\bS}{\mathbf{S}}

\newcommand{\bu}{\mathbf{u}}

\newcommand{\bv}{\mathbf{v}}

\newcommand{\bw}{\mathbf{w}}

\newcommand{\bx}{\mathbf{x}}
\newcommand{\bX}{\mathbf{X}}
\newcommand{\by}{\mathbf{y}}

\newcommand{\bz}{\mathbf{z}}
\newcommand{\bZ}{\mathbf{Z}}

\newcommand{\bbeta}{\bm{\beta}}

\newcommand{\bdelta}{\bm{\delta}}

\newcommand{\btau}{\bm{\tau}}

\newcommand{\bTheta}{\bm{\Theta}}

\DeclareMathOperator{\diag}{diag}

\DeclareMathOperator{\supp}{supp}

\newtheorem{theorem}{Theorem}

\newtheorem{definition}{Definition} 

\newcommand{\qednew}{\nobreak \ifvmode \relax \else
      \ifdim\lastskip<1.5em \hskip-\lastskip
      \hskip1.5em plus0em minus0.5em \fi \nobreak
      \vrule height0.75em width0.5em depth0.25em\fi}

\newcommand{\scrH}{\mathscr{H}}

\newcommand{\scrN}{\mathscr{N}}

\newcommand{\scrP}{\mathscr{P}}

\usepackage[T1]{fontenc}
\usepackage[utf8]{inputenc}
\numberwithin{theorem}{section}
\usepackage{amsthm}
\usepackage{xcolor, tikz}

\numberwithin{definition}{section}
\numberwithin{pro}{section}

\newtheorem{coro}{Corollary}
\newtheorem{lem}{Lemma}
\newtheorem{rem}{Remark}

\numberwithin{claim}{section}
\numberwithin{rem}{section}
\numberwithin{lem}{section}
\numberwithin{coro}{section}
\usepackage{bm}

\usepackage{etoolbox}

\apptocmd{\thebibliography}{%
  \setlength{\itemsep}{0pt}%
  \setlength{\parskip}{0pt}%
  \setlength{\parsep}{0pt}%
}{}{}

\DeclareMathOperator{\sign}{sign}

\DeclareMathOperator{\rad}{rad}

\DeclareMathOperator{\sfP}{\mathsf{P}}

\DeclareMathOperator{\sfT}{\mathsf{T}}
\DeclareMathOperator{\sfI}{\mathsf{I}}
\DeclareMathOperator{\sfJ}{\mathsf{J}}
\DeclareMathOperator{\sfE}{\mathsf{E}}

\DeclareMathOperator{\dist}{dist}

\usepackage[font=small]{caption} 
\title{One-Bit Phase Retrieval: Optimal Rates and Efficient Algorithms}

\date{(\today)}

\author[$\ast$]{Junren Chen}
\author[$\dag$]{Ming Yuan}
\affil[$\ast$]{Department of Mathematics, University of Maryland}
\affil[$\dag$]{Department of Statistics, Columbia University}
\begin{document}
\maketitle
 
%\footnotetext[1]{This research was supported in part by NSF Grants DMS-2015285 and DMS-2052955. The work was done when J.~Chen was a visiting Ph.D. Student at the Department of Statistics, Columbia University. Email: {\hypersetup{urlcolor=black}{chenjr58@connect.hku.hk}; {ming.yuan@columbia.edu}}.}
    
\long\def\symbolfootnote[#1]#2{\begingroup\def\thefootnote{\fnsymbol{footnote}}\footnote[#1]{#2}\endgroup}

\symbolfootnote[0]{This research was supported by NSF Grants DMS-2015285 and DMS-2052955. The work was carried out while J.C. was a graduate student at the Department of Mathematics, University of Hong Kong and a visiting student at the Department of Statistics, Columbia University. (email: {\href{mailto:jchen58@umd.edu}{\texttt{jchen58@umd.edu}}; \href{mailto:ming.yuan@columbia.edu}{\texttt{ming.yuan@columbia.edu}}}.)}

\begin{abstract}
In this paper, we study the sample complexity and develop efficient algorithms for 1-bit phase retrieval, i.e., the recovery of a signal $\bx\in\mathbb{R}^n$ from $m$ phaseless bits $\{\sign(|\ba_i^\top\bx|-\tau)\}_{i=1}^m$ with standard Gaussian $\ba_i$. By investigating a phaseless version of random hyperplane tessellation, we show that (constrained) hamming distance minimization uniformly  recovers all unstructured signals with Euclidean norm bounded away from zero and infinity to the error $\calO(\frac{n}{m}\log(\frac{m}{n}))$, and $\calO(\frac{k}{m}\log(\frac{mn}{k^2}))$ when restricting to $k$-sparse signals. Both error rates are  information-theoretically optimal up to a logarithmic factor. Intriguingly, the optimal rate for sparse recovery matches that of 1-bit compressed sensing, suggesting that the phase information is non-essential for 1-bit compressed sensing. 
We also develop efficient and near-optimal algorithms for 1-bit (sparse) phase retrieval. Specifically, we prove that (thresholded) gradient descent with respect to the one-sided $\ell_1$-loss, when initialized via spectral methods, 
converges linearly and attains the near-optimal reconstruction error, with sample complexity $\calO(n)$ for unstructured signals and $\calO(k^2\log(n)\log^2(\frac{m}{k}))$ for $k$-sparse signals. 
Our proof is based upon the observation that a certain local approximate invertibility condition is respected by Gaussian measurements. Our results establish the major findings of (memoryless) 1-bit compressed sensing in a phaseless setting. 
\end{abstract}
%{\color{blue}Checklist:}
%\begin{itemize}
%	[leftmargin=5ex,topsep=0.25ex]
%		\setlength\itemsep{-0.1em}
%    \item {\color{blue}[absolute]} 
    %\item {\color{blue}[avoid saying IT limits]}
    %\item {\color{blue}[absolute , universal constants {\color{red}[numerical cons?]}, remove warm initialization]}
%\end{itemize}

 %\newpage
%% removing this when separating
%% ****************************
%\setcounter{tocdepth}{2} 
%% ****************************

%\pagenumbering{roman}
{% \hypersetup{linkcolor=black}\small 
%	\thispagestyle{empty}
%	\tableofcontents
%	\thispagestyle{empty}
}
 %\newpage

\section{Introduction}\label{sec:intro}
The problem of phase retrieval arises naturally in applications such as X-ray crystallography \cite{millane1990phase}, quantum mechanics \cite{reichenbach1998philosophic}, and astronomical imaging \cite{shechtman2015phase}, where capturing the phase information is prohibitively expensive or infeasible; see, e.g., \cite{candes2013phaselift,candes2015phase,zhang2017nonconvex,chen2017solving,wang2017solving,netrapalli2013phase,lecue2015minimax,eldar2014phase,huang2024performance}. To further increase the sampling rate and reduce energy consumption in phase retrieval, quantized measurements have been considered; see, e.g., \cite{kishore2020wirtinger,zhu2019phase,mukherjee2018phase,domel2022phase, eamaz2022one,mroueh2013quantization}. In spite of these recent progresses and the practical successes, still very little is known about either the information-theoretic aspect  or efficient algorithms for quantized phase retrieval, that is, phase retrieval from quantized magnitude-only measurements. The primary objective of this work is to provide a clear understanding on the optimal error rates of 1-bit phase retrieval, which concerns the recovery of a signal $\bx$ in ${\mathbb R}^n$ from the binary phaseless measurements\footnote{While phase retrieval is already challenging due to the loss of phase information, 1-bit phase retrieval is even more difficult, as it additionally discards almost all magnitude information, retaining only a single bit.}
\begin{align}\label{eq:1bpr_mea}
    y_i=\sign(|\ba_i^\top\bx|-\tau),\quad i=1,\ldots, m,
\end{align} 
where $\bA=[\ba_1,\ldots, \ba_m]^\top\in \mathbb{R}^{m\times n}$ is a sensing matrix and $\tau>0$ is a fixed threshold, and develop efficient and optimal algorithms for the recovery of unstructured or sparse signals.

In particular, we derive sample complexity bounds for 1-bit phase retrieval when $\bA$ is a Gaussian matrix, i.e., its entries are independent standard Gaussian random variables. Ignoring logarithmic factor, the optimal error rate for reconstructing all unstructured signals bounded away from zero and infinity is of the order $\tilde{\calO}(\frac{n}{m})$,  and of order $\tilde{\calO}(\frac{k}{m})$ if we restrict to $k$-sparse signals. Both rates can be achieved by (constrained) hamming distance minimization. The optimal rate for recovering sparse signals is identical to that of 1-bit compressed sensing (e.g., \cite{boufounos20081,jacques2013robust,plan2012robust,plan2013one,matsumoto2024binary,awasthi2016learning,domel2023iteratively}) where one observes 
\begin{align}\label{eq:1bcs_mea}
    y_i=\sign(\ba_i^\top\bx),\quad i=1,\ldots, m.
\end{align}
 This suggests an intriguing message that {\it phase information is  inessential for 1-bit compressed sensing}. There are, however, also important distinctions between the two problems \eqref{eq:1bpr_mea} and \eqref{eq:1bcs_mea}. 
Specifically, in 1-bit compressed sensing, the signal norm information is completely lost in  quantization, hence it is conventional to assume $\bx$ lives in the standard Euclidean sphere. In fact, the literature of 1-bit compressed sensing advocates the use of dithering to enable full signal reconstruction with norm (see, e.g.,  \cite{dirksen2021non,knudson2016one,dirksen2020one,baraniuk2017exponential,dirksen2023robust}). In contrast, a shift of the quantization threshold away from $0$, which is obviously necessary for quantizing the magnitude-only measurements, also overcomes the limitation and enables norm reconstruction. Moreover, establishing the attainability of the optimal error rate requires a new technical tool on the phaseless hyperplane tessellation of (a subset of) an annulus, which may be of independent interests.

While hamming distance minimization achieves the optimal error rates, it is not computationally tractable. To overcome this challenge, we also develop efficient algorithms for 1-bit phase retrieval. We show that under a near-optimal number of Gaussian measurements, gradient descent with respect to the one-sided $\ell_1$-loss converges linearly and attains the near-optimal error rate for recovering unstructured signals, provided that a good initialization is given. A similar near-optimal local convergence guarantee is established for recovering $k$-sparse signals via thresholded gradient descent, which incorporates an additional hard thresholding operation at each gradient descent step. To further complement these  results, we prove that the desired initialization can be obtained by spectral methods under appropriate sample sizes. The high-level architecture of our proof is inspired by the recent work \cite{matsumoto2024binary} who  first theoretically analyzed the optimality of a similar iterative thresholding algorithm for 1-bit compressed sensing. The main ideas are to derive tight concentration bounds for the directional gradients by conditioning on the index set of hyperplane separations and treat separately the ``large-distance regime'' and ``small-distance regime''. Yet our analysis for 1-bit phase retrieval is more delicate and poses significant new challenges. First note that the one-sided $\ell_1$-loss for 1-bit compressed sensing is only non-differentiable at $\bigcup_i\{\bz\in\mathbb{R}^n:\ba_i^\top\bz=0\}$. For 1-bit phase retrieval, the gradient and landscape of our loss function appear far  more intricate. To establish tight concentration bounds for the directional gradients, we need to take a close inspection on the two index sets associated with (phaseless) hyperplane separations, and decompose the gradient into a main term and a higher-order term. To deal with the main term in the large-distance regime, since the iterates are no longer restricted to the standard sphere, a more general orthogonal decomposition and substantially more technicalities are in order to cover all possible directional gradients. In controlling the additional higher-order term, we make use of a key observation that  ``double separation'' probability---the probability of two points being separated by both the hyperplane and the phaseless hyperplane---exponentially decays with the distance. {Another challenge occurs in the small-distance regime  where we seek to bound the directional gradients at points that are close  to the underlying signal. Our approach is to leverage the sharp bounds in phaseless hyperplane tessellation to show that the non-zero contributors to the   gradient cannot be too numerous, and then invoke a powerful concentration bound to uniformly control the directional gradient within the order of the optimal error. This line of reasoning substantially  simplifies the small-distance regime analysis in \cite{matsumoto2024binary}; see a more detailed comparison in Section \ref{app:matter}.} We believe the arguments and technical tools developed in the present paper will prove useful in many other related problems; indeed, some of them are already used to provide a unified treatment to quantized compressed sensing in our follow-up work \cite{chen2024optimal}.

The problem of 1-bit phase retrieval is closely related to 1-bit compressed sensing and phase retrieval, both of which have been extensively researched in recent years. To put our contribution in perspective, we shall now briefly review some of  the most relevant theoretical developments in these areas. 

%\subsection{1-Bit Compressed Sensing}
The optimal rate of 1-bit compressed sensing is due to   Jacques et al. \cite{jacques2013robust} who showed that the optimal error rate is at the order $\calO(\frac{k}{m})$, up to a logarithmic factor. On the computational side, \cite{jacques2013robust} also introduced the (normalized) binary iterative hard thresholding (NBIHT) that performs thresholded  gradient descent with respect the one-sided $\ell_1$-loss and empirically observed its optimal performance. There have been some attempts to provide theoretical justification for the optimality of NBIHT since the pioneering work of \cite{jacques2013robust}. For example,  Liu et at. \cite{liu2019one} showed that the later iterates of NBIHT  remain bounded with the same level of estimation error, and Friedlander et al. \cite{friedlander2021nbiht} proved the first non-asymptotic error rate for NBIHT that decays with $m$ almost optimally but exhibits highly sub-optimal dependence on the sparsity $k$. More recently, a near-optimal guarantee was presented by Matsumoto and Mazumdar \cite{matsumoto2024binary}, showing that NBIHT with a particular step size  achieves estimation error $\calO(\frac{k}{m}\log(\frac{mn}{k^2}))$. We also note in passing a few efficient   solvers with sub-optimal theoretical guarantees. The linear programming approach in \cite{plan2013one} is perhaps the first efficient and provable 1-bit compressed sensing solver, achieving an error rate of $\tilde{\calO}((\frac{k}{m})^{1/5})$. There are also efficient solvers robust to noise or for more general non-linear models, see, e.g., \cite{plan2012robust,plan2016generalized,plan2017high,awasthi2016learning}. These algorithms, however, achieve error rates no faster than $\calO(\sqrt{k/m})$. Readers can consult the surveys \cite{boufounos2015quantization,dirksen2019quantized} for an overview of quantized compressed sensing. More broadly, our work is related to the rapidly growing literature on the quantized versions of various estimation problems  and 
signal reconstruction from general non-linear observations. See, e.g., \cite{genzel2023unified,jung2021quantized,dirksen2022covariance,chen2022high,davenport20141,Cai2013AMC,xu2020quantized,thrampoulidis2020generalized,jacques2021importance}. {Table~\ref{table11} compares our results on 1-bit (sparse) phase retrieval with existing results on 1-bit compressed sensing, demonstrating that our work essentially extends the main theoretical findings of 1-bit compressed sensing to the phase retrieval setting.}

\begin{table}[ht!]
    \centering
    \begin{tabular}{|c|c|c|c|}
        \hline 
        ~ & \makecell{1-Bit Compressed \\ Sensing} &  \makecell{1-Bit Phase \\ Retrieval} & \makecell{1-Bit Sparse   Phase \\ Retrieval}
        \\
         \hline 
       Problem Setup & \makecell{$\sign(\bA\bx)$ \\ $\to \bx\in \Sigma^{n}_k\cap\mathbb{S}^{n-1}$} & \makecell{$\sign(|\bA\bx|-\tau) $ \\ $\to\bx\in \mathbbm{A}_{\alpha,\beta}$} & \makecell{$\sign(|\bA\bx|-\tau)$ \\ $\to \bx\in \mathbbm{A}_{\alpha,\beta}\cap \Sigma^n_k$} \\
         \hline 
        Optimal Rate & \makecell{$\|\hat{\bx}-\bx\|_2=\tilde{\calO}(k/m)$\\\cite[Thms 1--2]{jacques2013robust}} & \makecell{$\dist(\hat{\bx},\bx)=\tilde{\calO}(n/m)$\\ Our Thms \ref{thm:hdm_gua}--\ref{thm:lower}} &  \makecell{$\dist(\hat{\bx},\bx)=\tilde{\calO}(k/m)$\\ Our Thms \ref{thm:hdm_gua}--\ref{thm:lower}}
        \\
          \hline 
        \makecell{ Optimal and \\ 
        Efficient Algorithm} & \makecell{NBIHT with \\$m=\tilde{\calO}(k)$ \cite{matsumoto2024binary}}  & \makecell{Alg \ref{alg:SI_low} $+$ Alg \ref{alg:pgd} \\ with $m=\tilde{\calO}(n)$\\Our Thms \ref{thm:SI_low}, \ref{thm:pgd_low}}  & \makecell{Alg \ref{alg:SI_high} $+$ Alg \ref{alg:pgd_high} \\ with $m=\tilde{\calO}(k^2)$\\Our Thms \ref{thm:SI_high}, \ref{thm:pgd_high}} 
        \\
        \hline 
         \makecell{Other Algorithms\\ with guarantee} & \makecell{suboptimal and no \\ faster than $\tilde{\calO}(\sqrt{k/m})$} & \makecell{N/A \\see also Rem. \ref{rem:ca1bpr}} & N/A
        \\ 
        \hline 
    \end{tabular} 
    \caption{\small A comparison between 1-bit compressed sensing (Column 2)  and our results on 1-bit phase retrieval of unstructured signals (Column 3) and sparse signals (Column 4). Readers may consult \cite[Table 1]{chen2024optimal} for a review of other (suboptimal) 1-bit compressed sensing algorithms.\label{table11}} 
\end{table}

The problem of phase retrieval, mathematically formulated as the reconstruction of $\bx$ from $\{y_i=|\ba_i^\top\bx|\}_{i=1}^m$, is  known to be NP-hard in general (e.g., \cite{sahinoglou1991phase,candes2013phaselift}). Some of the earliest heuristic methods include the Gerchberg-Saxton and Fienup algorithms (see, e.g., \cite{fienup1982phase}). In the past decade, a number of algorithms with theoretical guarantees have been developed. Notable examples inlcude the convex Phaselift algorithm \cite{candes2013phaselift,candes2015phase1}, the non-convex Wirtinger flow and several variants \cite{candes2015phase,chen2017solving,wang2017solving,zhang2017nonconvex}, Phasemax \cite{goldstein2018phasemax,bahmani2019solving}, among others \cite{waldspurger2015phase,netrapalli2013phase,waldspurger2018phase}.
Specifically, note that most non-convex solvers rely on spectral methods to find a good initialization, and then execute gradient descent for further improvement. Some of these algorithms have also been adapted to solve sparse phase retrieval where the signal is $k$-sparse; see, e.g.,  \cite{li2013sparse,cai2016optimal,wang2017sparse,jaganathan2017sparse,netrapalli2013phase}. However, all these algorithms require $\tilde{\calO}(k^2)$ phaseless measurements  that are higher than the information-theoretic optimal sample complexity $\calO(k\log(\frac{en}{k}))$ \cite{eldar2014phase,lecue2015minimax}. It remains an open problem if there exists an efficient algorithm that can recover any $k$-sparse signal from a near-optimal number of measurements. See  \cite{shechtman2015phase,jaganathan2016phase} for more in-depth reviews of phase retrieval.

%Notable examples include dithered model, multi-bit model, non-Gaussian measurements and adaptive quantization (e.g., \cite{knudson2016one,xu2020quantized,thrampoulidis2020generalized,genzel2023unified,jung2021quantized,baraniuk2017exponential,dirksen2021non,huynh2020fast,ai2014one}); see the surveys \cite{boufounos2015quantization,dirksen2019quantized} for an overview of this area.
\paragraph{Notation:} We shall first introduce some notation that are used throughout the paper. We use boldface letters to denote vectors and matrices, regular letters to denote scalars.  We measure the difference of $\bu,\bv\in \mathbb{R}^n$  by the phaseless $\ell_2$-norm $\dist(\bu,\bv)=\min\{\|\bu-\bv\|_2,\|\bu+\bv\|_2\}$, the directional error $\dist_{\rm d}(\bu,\bv)=\dist\big(\frac{\bu}{\|\bu\|_2},\frac{\bv}{\|\bv\|_2}\big)$ and error in norm $\dist_{\rm n}(\bu,\bv) = |\|\bu\|_2-\|\bv\|_2|$. For a matrix $\bM$, we work with its Frobenius norm $\|\bM\|_{\rm F}$ and operator norm $\|\bM\|_{\rm op}$. 
Let the $L^p$-norm of a random variable $X$ be $\|X\|_{L^p}=(\mathbbm{E}|X|^p)^{1/p}$, then we define its sub-Gaussian norm by $\|X\|_{\psi_2}:= \sup_{p\ge 1}\frac{\|X\|_{L^p}}{\sqrt{p}}$, sub-exponential norm by $\|X\|_{\psi_1}:=\sup_{p\ge 1}\frac{\|X\|_{L^p}}{p}$. The sub-Gaussian norm extends to random vector $\bX\in \mathbb{R}^n$ by taking  supremum over all 1-dimensional margins  $\|\bX\|_{\psi_2} = \sup_{\bv\in\mathbb{S}^{n-1}}\|\bv^\top\bX\|_{\psi_2}$.  
We denote the cardinality of the set $\calS$ by $|\calS|$.
For $\calK\subset \mathbb{R}^n$, we let its radius   be $\rad(\calK)= \sup_{\bu\in\calK}\|\bu\|_2$ and the secant set by $\calK_-=\calK-\calK$; we also let $\calP_{\calK}$ be the projection onto $\calK$ under $\ell_2$-norm. We let $\mathbb{B}_2^n=\{\bu\in\mathbb{R}^n:\|\bu\|_2\le 1\}$ denote the standard Euclidean ball, and more generally write $\mathbb{B}_2^n(\bv;r)=\{\bu\in\mathbb{R}^n:\|\bu-\bv\|_2 \le r\}$. 
Throughout the work, $\bx,\bA,\tau,\by$ are reserved for the underlying signal, the sensing matrix, the fixed quantization threshold in 1-bit phase retrieval, and the 1-bit observations, respectively. We also use $C,C_1,C_2,c,c_1,c_2,...$ to denote constants whose values may differ in each appearance. It is important to note that {\it these constants may only depend on $(\alpha,\beta,\tau)$}, otherwise stated.  We write $T_1=\calO(T_2)$ or $T_1\lesssim T_2$ if $T_1\le C_1T_2$ for some  constant $C_1$, $T_1 =\Omega(T_2)$ or $T_1\gtrsim T_2$ if $T_1\ge c_2T_2$ for some   constant $c_2>0$, and $T_1\asymp T_2$ if $T_1=\calO(T_2)$ and $T_1 = \Omega(T_2)$ simultaneously hold. Note that $\tilde{\calO}(\cdot)$ is the less precise version of $\calO(\cdot)$ which ignores logarithmic factors. 
%More notation will be introduced in subsequent development when needed, and tables of recurring notation  are provided at the beginning of Appendix to improve the readability. 

\paragraph{Outline:} The rest of the paper is organized as follows. We first present the main theoretical findings in the next section.  Sections \ref{sec:ITproof} and \ref{sec:pgdproof} provide the general frameworks and main arguments for our results. Due to the highly technical nature of our proofs, further details are relegated to the Appendix.  Our theories are complemented by  experimental results in Section \ref{sec:numerics}. Section \ref{sec:conclusion} closes the paper with a few concluding remarks. The complete proofs  of our results appear in the appendices. 

\section{Main Results}\label{sec:main_res}
We shall begin with a roadmap of our main results. Information-theoretically,  Theorem~\ref{thm:local_embed} establishes a geometric result on phaseless random hyperplane tessellations, Theorem~\ref{thm:hdm_gua} provides performance guarantees for (constrained) Hamming distance minimization, and Theorem~\ref{thm:lower} gives matching lower bounds for the reconstruction of unstructured and sparse signals. Computationally,  Theorems~\ref{thm:pgd_low} and~\ref{thm:pgd_high} show that Algorithms~\ref{alg:pgd} and~\ref{alg:pgd_high} achieve near-optimal local recovery for unstructured and sparse signals, respectively.
    Finally, Theorems \ref{thm:SI_low} and \ref{thm:SI_high} provide guarantees on spectral initializations that complement the local convergence guarantees.

\subsection{Information-Theoretic Bounds}\label{sec:it}
The first question that we address in this paper is the information-theoretical aspect of 1-bit phase retrieval:
\begin{tcolorbox}[colback=white,
  colframe=black!80,
  boxrule=0.4pt,
  width=0.7\textwidth,
  center] \centering\textit{How well can we recover $\bx$ from $\by= \sign(|\bA\bx|-\tau)$?} \end{tcolorbox}   
  \noindent Note that the sign of $\bx$ is not identifiable, we shall therefore consider measuring the difference between  $\bu,\bv\in \mathbb{R}^n$ by
$
	\dist(\bu,\bv) = \min\{\|\bu-\bv\|_2,\|\bu+\bv\|_2\}.$
We assume throughout the paper that the underlying signals reside in a given annulus 
\begin{align}
	\label{eq:signal_range}
	\mathbbm{A}_{\alpha,\beta} := \{\bu\in \mathbb{R}^n:\alpha\le \|\bu\|_2\le\beta\},
\end{align}
for some constants $0<\alpha\le \beta<\infty$. In other words, we assume that the signals are bounded away from zero and infinity in Euclidean norm.

%\subsubsection{Upper Bound}
%Suppose that there exist constants $0<\alpha\le \beta<\infty$ so that the signals of interest live in the annulus 
We first derive an  upper bound for 1-bit phase retrieval of generally structured signals. More specifically, we assume that the signals reside in $\calK$, an arbitrary subset of $\mathbbm{A}_{\alpha,\beta}$ that encodes potential signal structure. 
 Of particular interest to us later on are the 1-bit phase retrieval of unstructured signals $\bx\in\calK:= \mathbbm{A}_{\alpha,\beta}$, and 1-bit sparse phase retrieval where  $\bx\in\calK:=\Sigma^n_k\cap\mathbbm{A}_{\alpha,\beta}$. Here, $\Sigma^n_k$ denotes the set of  $k$-sparse vectors in $\mathbb{R}^n$, and we also write $\Sigma^{n,*}_k = \Sigma^n_k\cap\mathbb{S}^{n-1}$. Moreover, we allow for a $\zeta$-fraction of adversarial corruption, under which we are only able to observe the corrupted bits $\hat{\by}$ obeying $
	d_H(\hat{\by},\by)\le \zeta m,$ 
where $d_H(\bu,\bv)=\sum_{i=1}^m \mathbbm{1}(u_i\ne v_i)$ denotes the hamming distance between $\bu=(u_i),\bv=(v_i)\in \{-1,1\}^m$. Our  upper bound is attained by constrained hamming distance minimization
\begin{align}\label{eq:hdm}
	\hat{\bx}_{\rm hdm}\in  \text{arg}\min~d_H\big(\sign(|\bA\bu|-\tau),\hat{\by}\big),~\text{subject to }\bu \in \calK,
\end{align}
and hence is of information-theoretic nature due to the computational infeasibility of (\ref{eq:hdm}). %due to the discrete hamming distance loss and typically non-convex constraint. The computational infeasibility renders the performance bound of (\ref{eq:hdm}) of information-theoretic nature.   

We pause to introduce some geometric notions. Recall that a hyperplane with a normal vector $\ba \in \mathbb{R}^n$ and a shift parameter $\tau$ is defined as $
	\calH_{\ba,\tau}=\{\bu \in \mathbb{R}^n:\ba^\top\bu=\tau\}.$ For some $\tau>0$, we will work with  a pair of symmetric hyperplanes 
	\begin{align}\label{eq:pl_hyperdef}
	\calH_{|\ba|}  = \{\bu \in \mathbb{R}^n:|\ba^\top\bu|=\tau\}   = \calH_{\ba,\tau}\cup \calH_{\ba,-\tau},  
\end{align}
referred to as a phaseless hyperplane and denoted by $\calH_{|\ba|}$ hereafter for convenience (we drop the dependence on the given $\tau$ for short). Note that $\calH_{|\ba|}$ separates $\mathbb{R}^n$ into two subsets (sides):
	\begin{align}\label{eq:twosides}
		\calH_{|\ba|}^+:=\{\bu \in \mathbb{R}^n:|\ba^\top \bu|\ge \tau\}~\text{ and }~\calH_{|\ba|}^-:=\{\bu \in \mathbb{R}^n:|\ba^\top\bu|<\tau\}.
	\end{align}
We refer to this as a phaseless separation in that $\bu$ and $-\bu$ always live in the same side of $\calH_{|\ba|}$. An arrangement of $m$ phaseless hyperplanes, say $\{\calH_{|\ba_1|},...,\calH_{|\ba_m|}\}$, separates the set $\calK$ into different cells, each taking the form $\calK \cap \calH_{|\ba_1|}^{*}\cap\calH_{|\ba_2|}^{*} \cap ...\cap \calH_{|\ba_m|}^{*}$ where $\calH_{|\ba_i|}^{*} = \calH_{|\ba_i|}^{+}$ or $\calH_{|\ba_i|}^{-}$. We shall refer to this as a phaseless hyperplane tessellation of $\calK$, see Figure \ref{fig:comp_embed} for an intuitive graphical illustration and  the comparison with the usual hyperplane tessellation.

The analysis  of (\ref{eq:hdm}) is closely related to phaseless hyperplane tessellation of the signal set $\calK$. The connection can be best illustrated in the absence of corruption: in this case, $\hat{\bx}_{\rm hdm}$ can be geometrically interpreted as any signal in $\calK$ that lives in the same cell with $\bx$ under the tessellation posed by $\{\calH_{|\ba_i|}\}_{i=1}^m$; as a result, achieving 
uniform recovery up to some accuracy $\rho$ via (\ref{eq:hdm}),
namely
$\dist(\hat{\bx}_{\rm hdm},\bx)\le \rho$ uniformly over all $\bx\in \calK$, is equivalent to all the cells in the tessellation of $\calK$ by $\{\calH_{|\ba_i|}\}_{i=1}^m$   having diameter (with respect to $\dist(\cdot,\cdot)$) no larger than $\rho$.
%\begin{figure}
%   \begin{centering}
	%      \includegraphics[width=0.80\columnwidth]{figs/1bcspr_Mar14.eps} 
	%     \par
	% \end{centering}

%\caption{\label{fig:existing_noiseless}}
%\end{figure}

\begin{figure}[ht!]
	\begin{centering}
		\includegraphics[width=0.35\columnwidth]{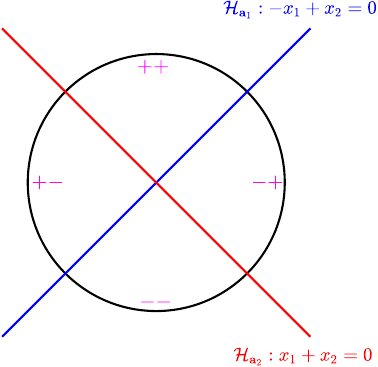} \quad \includegraphics[width=0.40\columnwidth]{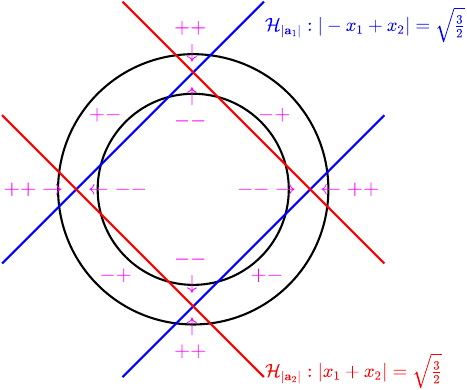}
		\par
	\end{centering}
	
	\caption{\label{fig:comp_embed} \small Graphical illustrations for the  classical hyperplane tessellation (Left) and our proposed phaseless hyperplane tessellation (Right). In the left figure, the standard sphere $\{x_1^2+x_2^2=1\}$ in $\mathbb{R}^2$ is separated into 4 cells by two hyperplanes that pass through $(0,0)^\top$, providing geometric understanding of  the standard 1-bit compressed sensing problem (e.g., \cite{plan2013one,plan2014dimension,oymak2015near}). In the right figure, the annulus $\mathbbm{A}_{1,\sqrt{2}}=\{1\le x_1^2+x_2^2\le 2\}$ in $\mathbb{R}^2$ is separated into $4$ cells by two phaseless hyperplanes with a common shift $\tau = \sqrt{3/2}$, which reflects the geometric aspect of   1-bit phase retrieval.}
\end{figure}

%The above connection directs our attention towards studying     phaseless hyperplane tessellation to understand (\ref{eq:hdm}). 

To ensure that the cells are uniformly small, our first main result shows that the required number of Gaussian phaseless hyperplanes is proportional to the ``complexity'' of $\calK$, which is jointly captured by   Gaussian width $\omega(\calK): = \mathbbm{E} _{\bg\sim \calN(0,\bI_n)}\sup _{\bu\in\calK}|\langle \bg, \bu\rangle|$ and metric entropy  $	\scrH(\calK,r): = \log \scrN(\calK,r)$; here, $\scrN(\calK,r)$ denotes  the minimal number of radius-$r$ $\ell_2$-balls $\mathbb{B}^n_2(r):=\{\bu\in \mathbb{R}^n:\|\bu\|_2\le r\}$ for covering $\calK$ (i.e., covering number). %Hence the metric entropy is simply a  notion equivalent to covering number. 
We will also work with the localization of $\calK$ given by $\calK_{(r)}:=(\calK-\calK)\cap \mathbb{B}^n_2(r).$
%will also be recurring. 

\begin{theorem}[Local Binary Phaseless Embedding]\label{thm:local_embed}
	Suppose $\ba_1,...,\ba_m$ are i.i.d. $\calN(0,\bI_n)$ vectors.
	Given a symmetric set $\calK\subset \mathbbm{A}_{\alpha,\beta}$ and some sufficiently small $r$ satisfying $rm\ge C_0$ for some  constant $C_0\ge1$, we let $r' = \frac{c_1r}{\log^{1/2}(r^{-1})}$ for some small enough $c_1$. For some  constants $(C_2,C_3,c_4,C_5,c_6)$, if 
	\begin{align}\label{eq:local_emb_sample}
		m\ge C_2  \Big(\frac{\omega^2(\calK_{(3r'/2)})}{r^3}+\frac{\scrH(\calK,r')}{r}\Big),
	\end{align}
	then  with probability at least $1-C_3\exp(-c_4rm)$ the following two events $E_{s}$ and $E_{l}$   hold:
	\begin{itemize}
		[leftmargin=2ex,topsep=0.25ex]
		\setlength\itemsep{-0.1em}
		\item Event $E_s$: For any $\bu,\bv\in\calK$ with $\dist(\bu,\bv)\le \frac{r'}{2}$ we have $$d_H\big(\sign(|\bA\bu|-\tau),\sign(|\bA\bv|-\tau)\big)\le C_5rm;$$ 
		\item Event $E_l$: For any $\bu,\bv\in\calK$ with $\dist(\bu,\bv)\ge 2r$ we have 
		$$d_H\big(\sign(|\bA\bu|-\tau),\sign(|\bA\bv|-\tau)\big)\ge c_6m\dist(\bu,\bv).$$ 
	\end{itemize}
\end{theorem}

While there are parallel results for random hyperplane tessellation  in the 1-bit compressed sensing literature \cite{oymak2015near,dirksen2021non}, we will need a novel probabilistic observation to prove Theorem \ref{thm:local_embed}. See Lemma \ref{lem:Puv}.   %Note that Theorem \ref{thm:local_embed} is a local binary embedding result, for instance, the  event $E_s$ only  controls  $d_H\big(\sign(|\bA\bu|-\tau),\sign(|\bA\bv|-\tau)\big)$ for $\bu,\bv$ being sufficiently close. As with \cite{oymak2015near,dirksen2022sharp}, a stronger global binary embedding can be achieved under a sample complexity exhibiting worse dependence on $r$. 
The performance bound for (constrained) hamming distance minimization readily follows from Theorem \ref{thm:local_embed}. 

\begin{theorem}
	[Recovery via Hamming Distance Minimization] \label{thm:hdm_gua}Suppose that the entries of $\bA$ are i.i.d. drawn from $\calN(0,1)$, $\calK$ is a given symmetric set contained in $\mathbbm{A}_{\alpha,\beta}$. Suppose $d_H(\hat{\by},\sign(|\bA\bx|-\tau))\le \zeta m$ and $\hat{\bx}_{\rm hdm}$ is defined by (\ref{eq:hdm}). Given some small $\epsilon$ such that $\epsilon m\ge C_0$ for some  constant $C_0$, we let $\epsilon' =  \frac{c_1\epsilon}{\log^{1/2}(\epsilon^{-1})}$ for some small enough $c_1$. Then for some  constants $(C_1,C_2,C_3,c_4)$, provided that \begin{align}\label{eq:gua_sample}
		m \ge C_1 \Big(\frac{\omega^2(\calK_{(3\epsilon'/2)})}{\epsilon^3}+\frac{\scrH(\calK,\epsilon')}{\epsilon}\Big),
	\end{align}
	the event $
		\dist(\hat{\bx}_{\rm hdm},\bx)\le 2\epsilon + C_2\zeta~(\forall\bx\in\calK)$ holds with probability at least $1-C_3\exp(-c_4 \epsilon m)$. 
\end{theorem}

%Remarkably, Theorem \ref{thm:hdm_gua} applies to arbitrary $\calK\subset \mathbbm{A}_{1,2}$.
%Next, we specialize Theorem \ref{thm:hdm_gua} to 1-bit phase retrieval of unstructured $\bx\in\mathbbm{A}_{\alpha,\beta}$, and 1-bit sparse phase retrieval where $\bx\in\Sigma^n_k\cap\mathbbm{A}_{\alpha,\beta}$. 
%Let $\|\bu\|_0$ be the number of non-zero entries in $\bu$, then we write $\Sigma^n_k = \{\bu\in \mathbb{R}^n:\|\bu\|_0\le k\}.$

\begin{comment}
\begin{coro}
	[Hamming Distance Minimization with $\calK=\mathbbm{A}_{\alpha,\beta}$] \label{cor:low_upper}In the setting  of Theorem \ref{thm:hdm_gua} with $\calK = \mathbbm{A}_{\alpha,\beta}$, for some  constants $(C_1,C_2,c_3,C_4)$, if
 \begin{align*}
		m\ge \frac{C_1n}{\epsilon}\log\Big(\frac{1}{\epsilon}\Big), 
	\end{align*} then with probability at least $1-C_2\exp(-c_3 \epsilon m)$, we have $\dist(\hat{\bx}_{\rm hdm},\bx)\le \epsilon+C_4\zeta$
	holding for all $\bx\in \mathbbm{A}_{\alpha,\beta}$.  
\end{coro}
\begin{coro}
	[Hamming Distance Minimization with $\calK=\Sigma^n_k\cap \mathbbm{A}_{\alpha,\beta}$] \label{cor:high_upper}In the setting of Theorem \ref{thm:hdm_gua} with $\calK =  \Sigma^n_k\cap\mathbbm{A}_{\alpha,\beta}$, for some  constants $(C_1,C_2,c_3,C_4)$, if \begin{align*}
		m\ge \frac{C_1k}{\epsilon}\log\Big(\frac{n}{\epsilon k}\Big), 
	\end{align*} then with probability at least $1-C_2\exp(-c_3\epsilon m)$, we have $\dist(\hat{\bx}_{\rm hdm},\bx)\le \epsilon+C_4\zeta$ holding for all $\bx\in \Sigma^n_k\cap\mathbbm{A}_{\alpha,\beta}$. 
\end{coro}
\end{comment}
\begin{rem}
Note that in the corruption-free setting where $\hat{\by}=\by=\sign(|\bA\bx|-\tau)$, Theorem \ref{thm:hdm_gua} implies the attainability $$\sup_{\bx\in \mathbbm{A}_{\alpha,\beta}}\dist(\hat{\bx}_{\rm hdm},\bx) = \calO\Big(\frac{n}{m}\log\frac{m}{n}\Big)$$
for recovering unstructured signals, and the attainability 
	$$
		\sup_{\bx\in\Sigma^n_k\cap \mathbbm{A}_{\alpha,\beta}}\dist(\hat{\bx}_{\rm hdm},\bx) =\calO\Big(\frac{k}{m}\log \frac{mn}{k^2}\Big) $$
for recovering sparse signals. \label{rem:attain}
\end{rem}

In the following, both rates in Remark \ref{rem:attain} are  shown to be optimal up to a logarithmic factor.

%\subsubsection{Lower Bound}
%The lower bound is more restricted than Theorem \ref{thm:hdm_gua} in that it is useful only when $\calK$ contains $\calV\cap \mathbbm{A}_{\alpha,\beta}$ for some linear subspace $\calV$, but more general in that the bound is valid for {\it arbitrary} design matrix $\bA$ (rather than the Gaussian one only). 
%%Without loss of generality we consider $(\alpha,\beta)=(1,2)$. 

\begin{theorem}
	[Information-Theoretic Lower Bound]\label{thm:lower} Let $\calV$ be a $\nu$-dimensional subspace in $\mathbb{R}^n$. Suppose   we seek to recover $\bx \in  \calV\cap\mathbbm{A}_{\alpha,\beta}$ from $\by = \sign(|\bA\bx|-\tau)$ with given $\beta>\alpha$ and arbitrary measurement matrix $\bA\in \mathbb{R}^{m\times n}$. Suppose $m\ge \nu$, then for any decoder which, takes $\by$ as input and returns $\hat{\bx}$ as an estimate of $\bx$, we have 
	$$
		\sup_{\bx\in \calV\cap \mathbbm{A}_{\alpha,\beta}}\dist(\hat{\bx},\bx)\ge \frac{c\nu}{m}$$ 
	for some constant $c$ depending on $(\alpha,\beta,\tau)$.  
\end{theorem}

\begin{rem}
Taking $\calV = \mathbb{R}^n$ in Theorem \ref{thm:lower} yields a lower bound 
$$
		\sup_{\bx\in \mathbbm{A}_{\alpha,\beta}}\dist(\hat{\bx},\bx)=\Omega\Big(\frac{n}{m}\Big)$$ for recovering unstructured signals. Similarly, by taking $\calV = \mathbb{R}^k\times \{\bm{0}^{n-k}\}$ that is a subset of $\Sigma^n_k$  we derive the lower bound  
$$
	  \sup_{\bx\in \Sigma^n_k\cap \mathbbm{A}_{\alpha,\beta}}\dist(\hat{\bx},\bx)= \Omega\Big(\frac{k}{m}\Big)%\Longrightarrow     \sup_{\bx\in \Sigma^n_k\cap \mathbbm{A}_{\alpha,\beta}}\dist(\hat{\bx},\bx)=\Omega\Big(\frac{k}{m}\Big)
	$$
  for recovering sparse signals. These bounds match the attainability bounds in Remark \ref{rem:attain}, up to a logarithmic factor. 
\end{rem}

%It is also worth noting that the lower bound given by Theorem \ref{thm:lower} applies to any sensing matrix $\bA$, thus indicating that Gaussian sensing matrix is near-optimal  since no other sensing matrix could attain essentially faster error rate.\footnote{We restrict our discussion to memoryless quantization where the sensing matrices is fixed before observing any 1-bit measurement. See, e.g.,  \cite{baraniuk2017exponential,boufounos2015quantization} for adaptive quantization.}

  Theorems \ref{thm:hdm_gua}--\ref{thm:lower} provide clear understanding on the information-theoretic aspect of 1-bit (sparse) phase retrieval. Since the error rate $\tilde{O}(\frac{k}{m})$ was also shown to be optimal for 1-bit compressed sensing in \cite{jacques2013robust}, we arrive at an intriguing takeaway message that {\it phases are non-essential for 1-bit compressed sensing}. More precisely, when losing all phase information, we can simply shift the quantization threshold to fully preserve the ability to recover sparse  signals from the 1-bit measurements; in fact, this  additionally enables norm reconstruction. While this message is concluded from an information-theoretic perspective, shortly it will be complemented by experimental results.

\subsection{Efficient Near-Optimal Algorithms}
%While the constrained hamming distance minimization is near-optimal, it should be noted that it is computationally intractable in terms of the discrete objective and non-convex constraint, rendering the upper bounds in Corollaries \ref{cor:low_upper}--\ref{cor:high_upper} of information-theoretic nature. 
Note that $\hat{\bx}_{\rm hdm}$ is infeasible to compute, hence the following computational question remains unaddressed: 
\begin{tcolorbox}[colback=white,
  colframe=black!80,
  boxrule=0.4pt,
  width=0.7\textwidth,
  center] \centering\textit{Can we achieve the optimal error rates for recovering}\\
    \textit{unstructured or sparse signals via efficient algorithms?} \end{tcolorbox}
%\begin{align*}
%    &\textit{Can we achieve the optimal error rates for recovering}\\
%    &\textit{unstructured or sparse signals via efficient algorithms?}
%\end{align*}  
\noindent We shall show that (thresholded) gradient descent with respect to the one-sided $\ell_1$-loss can achieve the near-optimal error rates. 
 
%answering the question in affirmative.
%,\footnote{In our algorithmic developments, when no confusion could arise, we use ``1-bit phase retrieval'' to refer to the specific ``low-dimensional case'' where $\bx\in\mathbbm{A}_{\alpha,\beta}$ is unstructured, and ``1-bit sparse phase retrieval'' for the ``high-dimensional case'' where $\bx\in\Sigma^n_k\cap\mathbbm{A}_{\alpha,\beta}$.} answering the following question in affirmative: 
%\begin{question}
%	\label{question2}
%	Can we use efficient decoders to achieve the near-optimal error rates $\tilde{\calO}(\frac{n}{m})$ for 1-bit phase retrieval of unstructured signals, and $\tilde{\calO}(\frac{k}{m})$ for 1-bit sparse phase retrieval?
%\end{question}

\subsubsection{Recovering Unstructured Signals}\label{sec:alg_unstructure}
%Our algorithmic works are presented  in a noiseless case (where $\hat{\by}=\by$) for succinctness.
In general, the computational difficulty of \eqref{eq:hdm} can be attributed to the two sources of nonconvexity: the hamming distance loss function  and the feasible set $\calK$. Interestingly, for unstructured signals, it suffices to convexify the loss function. A common way to do so is to replace the hamming distance loss with the one-sided $\ell_1$-loss
	\begin{align}\label{eq:l1loss}
		\calL(\bu) = \frac{1}{m}\sum_{i=1}^m \max\big\{0,-y_i(|\ba_i^\top\bu|-\tau)\big\} = \frac{1}{2m}\sum_{i=1}^m \Big[||\ba_i^\top\bu|-\tau| - y_i(|\ba_i^\top\bu|-\tau)\Big].
	\end{align}
The (sub-)gradient of $\calL$ at $\bu$ can be computed explicitly:
\begin{subequations}\label{eq:subgra}
	\begin{align} 
		&\partial \calL(\bu) = \frac{1}{2m}\sum_{i=1}^m \big(\sign(|\ba_i^\top\bu|-\tau)-y_i\big)\sign(\ba_i^\top\bu)\ba_i\\
		&= \frac{1}{2m}\sum_{i=1}^m \big(\sign(|\ba_i^\top\bu|-\tau)-\sign(|\ba_i^\top\bx|-\tau)\big)\sign(\ba_i^\top\bu)\ba_i.
	\end{align}
\end{subequations}
As an attempt to minimize  $\calL(\bu)$, we propose to perform gradient descent with step size $\eta$ in the $t$-th iteration, as formalized in Algorithm \ref{alg:pgd}.

\begin{algorithm}
	\caption{Gradient Descent for 1-bit Phase Retrieval ({GD-1bPR})	\label{alg:pgd}}
	\textbf{Input}: model data $(\bA,\by,\tau)$, initialization $\bx^{(0)}$,  step size $\eta$ ($\eta=\sqrt{\frac{\pi e}{2}}\tau$ by default)

	\textbf{For}
	$t = 1, 2, 3,\cdots$  \textbf{do}
	\begin{gather}
		\bx^{(t)} = \bx^{(t-1)} - \eta\cdot \partial\calL(\bx^{(t-1)})  
		.\label{eq:gd}
	\end{gather}
\end{algorithm}
\begin{rem}
While we assume the existence of some $\beta\ge\alpha>0$ for constraining signal norms, it is worth noting that Algorithm \ref{alg:pgd}, in fact all algorithms presented in this section, does not require knowing $\alpha$ and $\beta$ a priori.
\end{rem}
The following result shows that Algorithm \ref{alg:pgd} converges linearly to a solution with near-optimal estimation error, provided that $\bx^{(0)}$ is close enough to $\pm\bx$. We note that our  Theorem \ref{thm:pgd_low}  and the forthcoming Theorem \ref{thm:pgd_high} are uniform in nature. 
%that is, $\bx^{(0)}$ enters a contraction domain surrounding $\pm\bx$.  

\begin{theorem}
	[{GD-1bPR} is Near-Optimal]\label{thm:pgd_low} We consider the iteration sequence $\{\bx^{(t)}\}_{t=0}^\infty$ produced by Algorithm \ref{alg:pgd}.
	Suppose $m\ge C_1n$ holds with a sufficiently large $C_1$, $\bx\in \mathbbm{A}_{\alpha,\beta}$ is the desired signal, and  the initialization $\bx^{(0)}$  satisfies $\dist(\bx^{(0)},\bx)\le \frac{c_2}{\log^{1/2}(\frac{m}{n})}$ for some small enough $c_2$. 
	If {GD-1bPR} is executed with  $\eta=\sqrt{\frac{\pi e}{2}}\tau$, 
	%\begin{align*}
	%	\eta=c_\eta\cdot\sqrt{\frac{\pi e}{2}}\tau
	%\end{align*}
	%or the adaptive step sizes $\eta=c_\eta\cdot \sqrt{\frac{\pi e}{2}}\|\bx^{(t-1)}\|_2$, 
	%\begin{align*}
	%	\eta=c_\eta\cdot \sqrt{\frac{\pi e}{2}}\|\bx^{(t-1)}\|_2,
	%\end{align*}
	then the following statement holds with probability at least $1-C_3\exp(-c_4n\log(\frac{m}{n}))$: universally for any $\bx\in\mathbbm{A}_{\alpha,\beta}$, we have $$\dist(\bx^{(t)},\bx) \le \frac{C_5n}{m}\log^2\big(\frac{m}{n}\big),\qquad\forall\,t\ge C_6\log\big(\frac{m}{n}\big)$$
	for some constants $C_5,C_6$.  
\end{theorem}
\begin{rem}
    We focus on the noiseless case to keep the presentation concise. We will note in Section \ref{sec:conclusion} how to further establish the robustness of Algorithms \ref{alg:pgd}--\ref{alg:pgd_high} to adversarial bit flips. 
\end{rem}
\begin{rem}
Note that in similar non-convex algorithms for phase retrieval (e.g., \cite{candes2015phase,chen2017solving,wang2017solving,zhang2017nonconvex}), the initial guess $\bx^{(0)}$ only needs to satisfy $\dist(\bx^{(0)},\bx)\le c$ for some small enough $c$. In contrast, our Theorem \ref{thm:pgd_low} requires the slightly more stringent $\dist(\bx^{(0)},\bx)\lesssim  (\log(\frac{m}{n}))^{-1/2}$. Such logarithmic  scaling helps ensure that the ``double separation probability'' is exponentially smaller than the ``separation probability'', which will be useful in Lemma \ref{lem:final_h2}. 
\end{rem}

%In signal reconstruction problems with random sensing ensemble, uniform recovery---meaning that a single draw of the random ensemble allows for accurately recovering all signals of interest---is a highly sought-after notion. See, e.g., \cite{genzel2023unified,chen2023unified,dirksen2021non,xu2020quantized}. 

\subsubsection{Recovering Sparse Signals}\label{sec:alg_sparse}
To accurately recover sparse signals from $m$ phaseless bits  that might be much fewer than the ambient dimension $n$, it is imperative to enforce the sparsity constraint. To this end, we shall consider the following iterative hard thresholding procedure. 

\begin{algorithm}
	\caption{Binary Iterative Hard Thresholding for 1-bit Sparse Phase Retrieval ({BIHT-1bSPR})\label{alg:pgd_high}}
	\textbf{Input}:   model data $(\bA,\by, \tau)$, sparsity $k$, initialization $\bx^{(0)}\in \Sigma^n_k$, step size $\eta$ ($\eta=\sqrt{\frac{\pi e}{2}}\tau$ by default)

	\textbf{For}
	$t = 1, 2, 3,\cdots$  \textbf{do}
	\begin{subequations}
		\begin{gather}
			\tilde{\bx}^{(t)} = \bx^{(t-1)} - \eta\cdot \partial\calL(\bx^{(t-1)})  
			;\label{eq:pgd_high_gd}\\\label{eq:thresholding_biht}
			\bx^{(t)} = \calT_{(k)}(\tilde{\bx}^{(t)}),
		\end{gather}
	where $\calT_{(k)}(\bu)$ zeros out all entries of $\bu$ except for the top $k$-largest ones in absolute value.
	\end{subequations}
\end{algorithm}

By interpreting $\calT_{(k)}$ as the projection onto $\Sigma^n_k$, the above algorithm can also be viewed as a procedure of projected gradient descent. Similar to Theorem \ref{thm:pgd_low}, the following theorem shows that {BIHT-1bSPR} is near-optimal for 1-bit sparse phase retrieval.% of sparse signals.

\begin{theorem}
	[{BIHT-1bSPR} is Near-Optimal]\label{thm:pgd_high}   We consider the iteration sequence $\{\bx^{(t)}\}_{t=0}^\infty$ produced by    Algorithm \ref{alg:pgd_high} with $\eta=\sqrt{\frac{\pi e}{2}}\tau$. Suppose $m\ge C_1 k\log(\frac{mn}{k^2})\log(\frac{m}{k})$ holds with sufficiently large $C_1$, $\bx\in \Sigma^n_k\cap\mathbbm{A}_{\alpha,\beta}$ is the desired signal, and the initialization $\bx^{(0)}$ satisfying $\dist(\bx^{(0)},\bx)\le \frac{c_2}{\log^{1/2}(\frac{m}{k})}$  for some small enough $c_2$ is provided. If 
		\begin{align}
		\max\Big\{\max_{w\in [\frac{\tau}{1.01\beta},\frac{\tau}{0.99\alpha}]}\Big|1-w^3\exp\Big(\frac{1-w^2}{2}\Big)\Big|,\max_{w\in[\frac{\tau}{1.01\beta},\frac{\tau}{0.99\alpha}]}\Big|1-w\exp\Big(\frac{1-w^2}{2}\Big)\Big|\Big\}	\le 0.49,
			\label{eq:fixed_ratio}
		\end{align}
	 then the following statement holds  with probability at least $1-C_3\exp(-c_4k\log(\frac{mn}{k^2}))$: Universally for any $\bx\in \Sigma^n_k\cap\mathbbm{A}_{\alpha,\beta}$, we have $$\dist(\bx^{(t)},\bx) \le \frac{C_5k}{m}\log\big(\frac{mn}{k^2}\big)\log\big(\frac{m}{k}\big),\qquad \forall \,t\ge C_6\log(\frac{m}{k})$$
			for some  constants $C_5,C_6$. %The  constants in this theorem may depend on $c_\eta$.   
		\end{theorem}

%		\begin{rem}[The Selection of $\eta$] \label{rem:ratio_cons}
\begin{rem}
We elucidate on the condition (\ref{eq:fixed_ratio}) which translates into certain ratio conditions (i.e., $\frac{\tau}{\alpha}\le \xi_1$ and $\frac{\beta}{\tau}\le \xi_2$ for some positive thresholds $\xi_1,\xi_2>1$), and is not seen in the unstructured case. Specifically, for some function $F(\eta, a,b)$ we  show
$$
\|\bx^{(t)}-\bx\|_2 \le \sup_{a,b} F(\eta,a,b)\cdot\|\bx^{(t-1)}-\bx\|_2+\textrm{other terms}$$
 for {GD-1bPR},
but only have $$
\|\bx^{(t)}-\bx\|_2 \le 2\sup_{a,b}F(\eta,a,b)\cdot \|\bx^{(t-1)}-\bx\|_2+\text{other terms}$$ for {BIHT-1bSPR},  
 with extra leading factor $2$ arising as the result of hard thresholding. 
 This means that $\sup_{a,b}F(\eta,a,b)<1$ suffices to ensure the contraction of {GD-1bPR}, while     we will need $\sup_{a,b}F(\eta,a,b)<\frac{1}{2}$ for {BIHT-1bSPR}, which gives rise to (\ref{eq:fixed_ratio}). 
 This creates a gap between the practical algorithm BIHT-1bSPR and the intractable decoder (\ref{eq:hdm}) since the latter does not require ratio conditions of this type. It is possible to relax (\ref{eq:fixed_ratio}) by using adaptive step size or tighter hard thresholding bound, which we do not attempt here. {It would be of greater interest to study whether (\ref{eq:fixed_ratio}) is some fundamental algorithmic barrier or simply a proof artifact that can be completely removed, which we leave as an open question.}  
\end{rem}
\begin{rem} \label{rem:tau}
    {Since the ratio conditions from (\ref{eq:fixed_ratio}) require a properly chosen $\tau$, an important practical issue is the tuning of $\tau$. We shall provide some brief discussion here (although a fuller account is beyond the scope of this theoretical work). For the (non-uniform) recovery of a specific $\bx$, we can set $\alpha=\beta = \|\bx\|_2$, and to fulfill (\ref{eq:fixed_ratio}) we shall choose $\tau$ as an accurate estimator of $\|\bx\|_2$; suppose that it is feasible to collect some unquantized phaseless measurements $\{y_i=|\ba_i^\top\bx|\}_{i=1}^{n_0}$ before the quantization of $\{y_i=|\ba_i^\top\bx|\}_{i=n_0+1}^n$, then one can set $\tau = \sqrt{\frac{\pi}{2}}\frac{1}{n_0}\sum_{i=1}^{n_0}y_i$, and standard concentration argument shows that $n_0=O(\log n)$ suffices to ensure (with high probability) that $(1-c')\|\bx\|_2\le\tau\le(1+c')\|\bx\|_2$ for some small $c'>0$.}  
\end{rem}

		\subsubsection{Spectral Initialization}\label{sec:spectralini}
We proceed to discuss how a good initialization can be obtained via spectral methods. First observe that $\frac{\bx}{\|\bx\|_2}$ is the leading eigenvector of the expectation of
		$
			\hat{\bS}_{\bx} = \frac{1}{m}\sum_{i=1}^m y_i \ba_i\ba_i^\top.$ 
We can then initialize the ``direction'' of $\bx$ by the leading eigenvector of $\hat{\bS}_{\bx}$. Additionally, note that the expectation of
	$	\hat{\lambda}_{\bx} = \frac{1}{m}\sum_{i=1}^m \mathbbm{1}(y_i = 1)$
equals to $\mathbbm{P}\big(|\ba_i^\top\bx|\ge\tau\big)=2\Phi(-\frac{\tau}{\|\bx\|_2})$, where $\Phi$ is the  cumulative distribution function of standard Gaussian variable. We can therefore estimate $\|\bx\|_2$ by
$ \hat{\lambda}_{\rm SI}:=  -\tau/\Phi^{-1}(\frac{\hat{\lambda}_{\bx}}{2}).
$
We thus propose the following spectral initialization:		
		\begin{algorithm}
			\caption{Spectral Initialization for  1-bit Phase Retrieval ({SI-1bPR})\label{alg:SI_low}}
			\textbf{Input}: Model data $(\bA,\by,\tau)$

			\textbf{Estimate $\frac{\bx}{\|\bx\|_2}$:} Compute $\hat{\bv}_{\rm SI}$ as the normalized leading eigenvector of $\hat{\bS}_{\bx}= \frac{1}{m}\sum_{i=1}^m y_i \ba_i\ba_i^\top$

			\textbf{Estimate $\|\bx\|_2$:} Compute $\hat{\lambda}_{\rm SI}=\frac{-\tau}{\Phi^{-1}(\frac{\hat{\lambda}_{\bx}}{2})}$  
%			\begin{align}\label{eq:norm_est}
%				\hat{\lambda}_{\rm SI}:= \frac{-\tau}{\Phi^{-1}(\frac{\hat{\lambda}_{\bx}}{2})}
%			\end{align} where $\Phi^{-1}(\cdot)$ is the inverse function of $\Phi(\cdot)$ and $\hat{\lambda}_{\bx} = \frac{1}{m}\sum_{i=1}^m \mathbbm{1}(y_i = 1)$
						
			\textbf{Output:} Spectral initialization $\hat{\bx}_{\rm SI} = \hat{\lambda}_{\rm SI}\hat{\bv}_{\rm SI}$
		\end{algorithm}

While drawing the inspiration from the population level, we obtain the non-asymptotic error rate of Algorithm \ref{alg:SI_low} as follows. %Note that we separately present the non-uniform and uniform guarantees in the theorems for spectral initialization. 

		\begin{theorem}
	[Guarantees for {SI-1bPR}]\label{thm:SI_low} Under standard Gaussian $\bA$, we let $\hat{\bx}_{\rm SI}$ be the spectral estimator in Algorithm \ref{alg:SI_low}. 
	Suppose that $m\ge C_1n$ for sufficiently large $C_1$. Then, for a fixed $\bx\in\mathbbm{A}_{\alpha,\beta}$  we have $$	\mathbbm{P}\left(\dist(\hat{\bx}_{\rm SI},\bx) \le C_2\sqrt{\frac{n}{m}}\right) \le 1-\exp(-c_3n).$$
	We also have a uniform bound $$\mathbbm{P}\left(\sup_{\bx\in\mathbbm{A}_{\alpha,\beta}}\dist(\hat{\bx}_{\rm SI},\bx) \le C_4\sqrt{\frac{n\log(\frac{m}{n})}{m}}\right)\ge 1-\exp\big(-c_5n\log(\frac{m}{n})\big).$$ 
	% More specifically, given some $\delta\in(0,1)$, $m\ge C_6n$ for large enough $C_6$ depending on $(\alpha,\beta,\gamma,\delta)$ suffices to ensure $\dist(\hat{\bx}_{\rm SI},\bx)\le \delta$ holding for all $\bx\in \mathbbm{A}_{\alpha,\beta}$ with probability at least $1-\exp(-c_7n\log\frac{m}{n})$.  
\end{theorem}
\begin{rem}
Combining Theorem \ref{thm:pgd_low} and Theorem \ref{thm:SI_low}, it immediately follows that Algorithm \ref{alg:pgd} initialized by Algorithm \ref{alg:SI_low} provides an efficient near-optimal algorithm for 1-bit phase retrieval, which uniformly recovers all unstructured signals in $\mathbbm{A}_{\alpha,\beta}$ to the accuracy $\calO(\frac{n}{m}\log^2(\frac{m}{n}))$ under the sample complexity $m=\calO(n)$.  
\end{rem}

%We pause to provide a novel perspective for 1-bit phase retrieval. Restricting our attention to $\bx\in\mathbb{S}^{n-1}$, by lifting \cite{candes2013phaselift,candes2015phase1,li2013sparse} we can linearize the magnitude-only measurements and recast the phaseless bits as 
%\begin{align*}
  %  y_i = \sign(|\ba_i^\top\bx|-\tau)=\sign(|\ba_i^\top\bx|^2-\tau^2)=\sign(\langle \ba_i\ba_i^\top-\tau^2\bI_n,\bx\bx^\top\rangle):=\sign(\langle\bA_i,\bX\rangle).
%\end{align*}
%This interprets the 1-bit phase retrieval problem as 1-bit compressed sensing of the rank-1 positive semi-definite matrix $\bX:=\bx\bx^\top$ under the unusual design $\bA_i:=\ba_i\ba_i^\top-\tau^2\bI_n$. Along this perspective, the direction estimator $\hat{\bv}_{\rm SI}$ in Algorithm \ref{alg:SI_low} can be evidently identified with 
 %\begin{align*}
  %       \hat{\bV}_{\rm SI}:=\text{arg}\max ~ \Big\langle \frac{1}{m}\sum_{i=1}^my_i\bA_i,\bV\Big\rangle,~\text{subject to }\bV = \bv\bv^\top,~\bv\in \mathbb{S}^{n-1}.
%\end{align*}
% Interestingly, the above objective resembles the one for the 1-bit compressed sensing solver in \cite{plan2012robust}, which seeks a signal maximally correlated with $\frac{1}{m}\sum_{i=1}^m \sign(\ba_i^\top\bx)\cdot\ba_i$. 

\begin{rem} \label{rem:ca1bpr}
There exist related results from \cite{domel2022phase} who derived non-asymptotic guarantees for    a different 1-bit phase retrieval model with   ``heavier'' sensing ensemble 
\begin{align*}
    y_i = \sign\Big(\|\bP_i \bx\|_2^2-\frac{1}{2}\Big),~i=1,2,\cdots,m,
\end{align*}
where  $\bP_i$ denotes the projection onto the i.i.d. uniformly distributed $(\frac{n}{2})$-dimensional linear subspace $\calV_i$, and the signal  $\bx$ resides in $\mathbb{S}^{n-1}$. Their estimator essentially returns the   leading eigenvector of $\frac{1}{m}\sum_{i=1}^m y_i \bP_i$ and was proved to achieve a comparable non-uniform error rate $\calO((\frac{n\log (n)}{m})^{1/2})$ but an essentially worse uniform error rate $\calO(n(\frac{\log(m)}{m})^{1/2})$. We also note in passing that the fundamental limit of Algorithm \ref{alg:SI_low} can be derived from the theories in \cite{lu2020phase,mondelli2018fundamental}.
\end{rem}

%that attains the information theoretical lower bound, up to a logarithmic factor.

%For sparse signals, an additional thresholding step is needed to In the high-dimensional sparse case.
Now let us proceed to the spectral initialization for sparse signal. 
Note that the non-zero coordinates of $\bx$ can be identified with the $k$ largest diagonal entries of  $\mathbbm{E}(\hat{\bS}_{\bx})$. We shall therefore estimate the support $\supp(\bx)=\{i\in [n]:x_i\ne 0\}$ by the set of indices, $\calS_{\bx}$, corresponding  to the $k$ largest diagonal entries of $\hat{\bS}_{\bx}$. It then suffices to invoke Algorithm \ref{alg:SI_low} by restricting to the support estimate, leading to the spectral initialization in Algorithm \ref{alg:SI_high}.
		
		\begin{algorithm}
			\caption{Spectral Initialization for  1-bit Sparse Phase Retrieval ({SI-1bSPR})	\label{alg:SI_high}}
			\textbf{Input}:  Model data $(\bA,\by,\tau)$ and sparsity $k$

			\textbf{Estimate $\supp(\bx)$:} Estimate $\supp(\bx)$ by  the index set $\calS_{\bx}$ that corresponds to the $k$ largest entries (in real value) in the diagonal of $\hat{\bS}_{\bx}=  \frac{1}{m}\sum_{i=1}^m y_i \ba_i\ba_i^\top$ 
			
			\textbf{Estimate $\frac{\bx}{\|\bx\|_2}$:} Compute $\hat{\bv}_{\rm SI}\in\Sigma^n_k$ as the  normalized leading eigenvector of $[\hat{\bS}_{\bx}]_{\calS_{\bx}}$, where $[\hat{\bS}_{\bx}]_{\calS_{\bx}}$ denotes the matrix obtained from $\hat{\bS}_{\bx}$ by zeroing out the rows and columns not in $\calS_{\bx}$ 
   %Zero out the entries of $\hat{\bv}_{\rm SI}$ not in ${\calS_{\bx}}$ and set $[\hat{\bv}_{\rm SI}]_{\calS_{\bx}}$ as the leading eigenvector of $[\hat{\bS}_{\bx}]_{\calS_{\bx}}$.		
   
			\textbf{Estimate $\|\bx\|_2$:} Compute $\hat{\lambda}_{\rm SI}=\frac{-\tau}{\Phi^{-1}(\frac{\hat{\lambda}_{\bx}}{2})}$

			\textbf{Output:} Spectral initialization $\hat{\bx}_{\rm SI} = \hat{\lambda}_{\rm SI}\hat{\bv}_{\rm SI}$
		\end{algorithm}

		In the  following theorem, we provide theoretical guarantees for initializing   sparse signals via Algorithm \ref{alg:SI_high}.
		%s guarantee that the spectral methods settle the initialization problem of Algorithm \ref{alg:pgd}.  

		\begin{theorem}
			[Guarantees for {SI-1bSPR}]\label{thm:SI_high} Under standard Gaussian $\bA$, we let  $\hat{\bx}_{\rm SI}$ be the spectral estimator in Algorithm \ref{alg:SI_high}. 
			Given a fixed $\bx\in \Sigma^n_k\cap \mathbbm{A}_{\alpha,\beta}$, if $m\ge C_1k^2\log n$ for sufficiently large $C_1$, then we have $$\mathbbm{P}\left(\dist(\hat{\bx}_{\rm SI},\bx) \le C_2 \big(\frac{k^2\log n}{m}\big)^{1/4}\right) \ge 1-n^{-9}.$$ 
			If $m\ge C_3k^3\log n$ for large enough $C_3$, then we have  a uniform bound $$\mathbbm{P}\left(\sup_{\bx\in \Sigma^n_k\cap\mathbbm{A}_{\alpha,\beta}}\dist(\hat{\bx}_{\rm SI},\bx)\le C_4 \big(\frac{k^3\log n}{m}\big)^{1/4}\right) \ge 1-\exp\big(-c_5k\log\big(\frac{en}{k}\big)\big).$$ 
			\end{theorem}
\begin{rem}
We are now able  to  combine Theorem   \ref{thm:pgd_high} and Theorem \ref{thm:SI_high} to establish a near-optimal guarantee for Algorithm \ref{alg:pgd_high}  initialized by  Algorithm \ref{alg:SI_high}. In particular, for recovering a fixed $k$-sparse signal, the combined algorithm attains the near-optimal error rate $\calO(\frac{k}{m}\log(\frac{mn}{k^2})\log(\frac{m}{k}))$ under the sample complexity $m\ge Ck^2\log(n)\log^2(\frac{m}{k})$ for some sufficiently large constant $C$. It is also worth pointing out that  a higher sample complexity $\calO(k^3\log(n)\log^2(\frac{m}{k}))$ is needed to achieve uniform recovery of all $k$-sparse signals in $\mathbbm{A}_{\alpha,\beta}$. %as dictated by $(\frac{k^3\log n}{m})^{1/4}\lesssim(\log\frac{m}{k})^{-1/2}$.
\end{rem}

\begin{rem}\label{rem:k2}
Note that the sample complexity $\tilde{\calO}(k^2)$ for 1-bit sparse phase retrieval is sub-optimal in light of the information-theoretic bound. As already mentioned in Section \ref{sec:intro}, similar gap between computationally efficient algorithm and information-theoretic limit is widely observed for sparse phase retrieval, whose closing remains open (see, e.g., \cite{li2013sparse,wang2017sparse,cai2016optimal,oymak2015simultaneously,hand2024compressive,jaganathan2017sparse,jaganathan2013sparse}). 
\end{rem}
%\begin{rem}
%We further note that the  non-uniform error rate in Theorem \ref{thm:SI_high} can be improved to the faster rate $\calO((\frac{k}{m}\log(\frac{en}{k}))^{1/2})$, provided that $\bx$ satisfies $ck^{-1/2}\le |x_i|\le Ck^{-1/2}$ for any $i\in\supp(\bx)$ and some given $C\ge c>0$ (e.g., \cite{wang2017sparse}). The reason is that one can prove $\calS_{\bx}=\supp(\bx)$ (w.h.p.) and hence the support estimate error in Lemma \ref{lem:supp_est} no longer plays a role in the final error rate.     
%\end{rem}

\section{Proof Sketches: Information-Theoretic Bounds}\label{sec:ITproof}
In this section, we establish near-matching upper bounds and lower bounds. We first study phaseless hyperplane tessellation that is closely related to (\ref{eq:hdm}). %Let us begin with the definitions for two points being separated by a phaseless hyperplane and the associated probability. 
Given $\bu,\bv\in \mathbb{R}^n$, we say a phaseless hyperplane $\calH_{|\ba|}$ separates $\bu$ and $\bv$ if either $(\bu,\bv)\in \calH_{|\ba|}^+\times  \calH_{|\ba|}^-$ or $(\bu,\bv)\in \calH_{|\ba|}^-\times  \calH_{|\ba|}^+ $ holds.
 For $\bu,\bv \in \mathbb{R}^n$ and $\ba\sim \calN(0,\bI_n)$, we define the   probability of the separation 
	as 
	\begin{align*}
		\sfP _{\bu,\bv} = \mathbbm{P}\big(\calH_{|\ba|}\text{ separates }\bu\text{ and }\bv\big). 
	\end{align*}
The following lemma is the most elementary ingredient that supports our Theorem \ref{thm:local_embed}. It states that $\sfP_{\bu,\bv}$ is order-wise equivalent to $\dist(\bu,\bv)$ for any $(\bu,\bv)$ in  $\mathbbm{A}_{\alpha,\beta}$. %with pre-specified $\beta\ge\alpha>0$.  

%The most elementary ingredient that supports our Theorem \ref{thm:local_embed} (indeed, the whole paper) is the lemma below, which shows that the separation probability $\mathsf{P}_{\bu,\bv}$ is order-wise equivalent to $\dist(\bu,\bv)$ over a given annulus.  

\begin{lem}
	[$\mathsf{P}_{\bu,\bv}\asymp \dist(\bu,\bv)$ over $\mathbbm{A}_{\alpha,\beta}$] \label{lem:Puv} There exist some constants $c,C$ depending only on the given $\beta\ge\alpha>0$, such that $c\dist(\bu,\bv)\le   \sfP_{\bu,\bv}\le C\dist(\bu,\bv),~\forall\bu,\bv\in \mathbbm{A}_{\alpha,\beta}.$
\end{lem}
\begin{proof}[A Proof Sketch]
	We note that many results on 1-bit compressed sensing (e.g., \cite{jacques2013robust,oymak2015near,plan2014dimension}) are built upon $		\mathbbm{P}_{\ba\sim\calN(0,\bI_n)}\left(\sign(\ba^\top\bu)\ne \sign(\ba^\top\bv)\right) = \frac{\arccos (\bu^\top\bv)}{\pi}$,
	a well-documented fact from, e.g., \cite{goemans1995improved}. Nonetheless, we are not aware of prior work implying Lemma \ref{lem:Puv}, thus  an independent proof is needed. We make use of the equivalence $\dist(\bu,\bv)\asymp \dist_{\rm d}(\bu,\bv)+\dist_{\rm n}(\bu,\bv)$ from Lemma \ref{lem:norm_equa}, %\begin{align}
	%  \dist(\bu,\bv)\asymp \dist_{\rm d}(\bu,\bv)+\dist_{\rm n}(\bu,\bv).
	%\end{align}  
	and then 
	decompose $\sfP_{\bu,\bv}$ into the sum of two integrals $\sfI_{\bu,\bv} + \sfJ_{\bu,\bv}$ in (\ref{eq:PtoIJ}).  
	%\begin{align}
	%   \sfP_{\bu,\bv} = \sfI_{\bu,\bv} + \sfJ_{\bu,\bv}.
	%\end{align}
	Then, we separately show $\sfI_{\bu,\bv}\asymp \dist_{\rm n}(\bu,\bv)$ and $\sfJ_{\bu,\bv}\asymp \dist_{\rm d}(\bu,\bv)$ by establishing two-sided bounds.  The complete proof can be found in Appendix \ref{app:sepa_prob_equ}. 
	%while the order-wise equivalence (\ref{eq:norm_equaa}) suffices for our subsequent development. 
\end{proof}

Armed with Lemma \ref{lem:Puv}, the proof of Theorem \ref{thm:local_embed} follows similar courses in \cite{oymak2015near,dirksen2021non}, except that we work with $\dist(\cdot,\cdot)$ rather than $\|\bu-\bv\|_2$ and hence need to adjust certain arguments with cautiousness. We relegate the complete proof to Appendix \ref{app:prove_thm1}. %{\color{blue}[TODO: Continue from here---where to put the following pages?]}

%\subsection{Upper Bound (Theorem \ref{thm:hdm_gua})}

\paragraph{Upper bound:} Our recovery guarantee for (\ref{eq:hdm}) can  be proved by Event $E_l$ in Theorem \ref{thm:local_embed}, which essentially states that large $\dist (\bu,\bv)$   implies large $d_H (\sign(|\bA\bu|-\tau), \sign(|\bA\bv|-\tau))$.

%, this precludes   $\hat{\bx}_{\rm hdm}$ from deviating too much from $\bx$ since its binary measurements best match $\hat{\by}$.
\begin{proof}[Proof of Theorem \ref{thm:hdm_gua}]
	To get started, under (\ref{eq:gua_sample}), Theorem \ref{thm:local_embed} with $r = \epsilon$ gives the following event with the promised probability:
	\begin{align}\label{eq:large2large}
		  d_H(\sign(|\bA\bu|-\tau),\sign(|\bA\bv|-\tau)) \geq c m \dist(\bu,\bv),\quad \forall\bu,\bv\in \calK~\text{obeying}~\dist(\bu,\bv)\geq 2\epsilon
	\end{align}
	%$\dist(\bu,\bv)\geq 2\epsilon$ implies $d_H(\sign(|\bA\bu|-\tau),\sign(|\bA\bv|-\tau)) \geq c m \dist(\bu,\bv)$ 
	for some $c>0$. Note that we only need to consider $\dist(\hat{\bx}_{\rm hdm},\bx)>2\epsilon$; we are done otherwise. Then, applying (\ref{eq:large2large}) to $(\bu,\bv)=(\hat{\bx}_{\rm hdm},\bx)$ gives  $$d_H\big(\sign(|\bA\hat{\bx}_{\rm hdm}|-\tau),\sign(|\bA\bx|-\tau)\big) \geq cm \dist(\hat{\bx}_{\rm hdm},\bx).$$ 
	%Specifically, provided that $\dist(\hat{\bx}_{\rm hdm},\bx)>2\epsilon$, we have 
	Next, we use the assumption $d_H(\hat{\by},\sign(|\bA\bx|-\tau))\le \zeta m$, along with triangle inequality and the optimality of $\hat{\bx}_{\rm hdm}$, to obtain 
    \begin{align*}
        &d_H\big(\sign(|\bA\hat{\bx}_{\rm hdm}|-\tau),\sign(|\bA\bx|-\tau)\big)\\&\leq  d_H\big(\sign(|\bA\hat{\bx}_{\rm hdm}|-\tau),\hat{\by}\big)+d_H\big(\sign(|\bA\bx|-\tau),\hat{\by}\big)\\&\leq  2d_H\big(\sign(|\bA\bx|-\tau),\hat{\by}\big)\\
        &\leq 2\zeta m.
    \end{align*}
    Taken collectively, we obtain $\dist(\hat{\bx}_{\rm hdm},\bx) \leq \frac{2\zeta m}{cm}=\frac{2\zeta}{c}$ and arrive at the desired claim. 
	%Combining with the case that $\dist(\hat{\bx}_{\rm hdm},\bx)\le 2\epsilon$, we always have $\dist(\hat{\bx}_{\rm hdm},\bx)\le 2\epsilon+\frac{2\zeta}{c}$, as claimed.    
\end{proof}

\begin{rem}
By analyzing an intractable program, Theorem \ref{thm:hdm_gua} provides an achievable reconstruction accuracy for 1-bit phase retrieval. Specifically, a measurement number  of order $\calO\big(\frac{\omega^2(\calK_{(3\epsilon'/2)})}{\epsilon^3} +\frac{\scrH(\calK,\epsilon')}{\epsilon}\big)$ is information theoretically sufficient for achieving a uniform accuracy of $\epsilon$.  
\end{rem}

%It is worth mentioning that the defining property of structured signal set is also satisfied by the recently popular generative prior \cite{bora2017compressed}, provided that $\omega^2(\calK)$ in (\ref{eq:structure_set}) is replaced by the latent dimension of the generative model. As a result, our Theorem \ref{thm:hdm_gua} also implies $\frac{1}{m}$ error rate for 1-bit phase retrieval using generative prior.     
%\subsection{Lower Bound (Theorem \ref{thm:lower})}

\paragraph{Lower bound:} We move on to the main techniques for proving the lower bound in Theorem \ref{thm:lower}. Without loss of generality, we can %simply work with $\mathbbm{A}_{\alpha,\beta}=\mathbb{S}^{n-1}$, since the general setting $\beta>\alpha$ can only make the reconstruction harder, and $\alpha =\beta \ne 1$ can be reduced to the considered setting by writing $\by=\sign(|\bA\bx|-\tau)=\sign(|(\alpha \bA)(\alpha^{-1}\bx)|-\tau)$. Moreover, we can 
specify the $\nu$-dimensional subspace $\calV=\calV_0:=\{\bu=(u_1,...,u_{\nu},u_{\nu+1},...,u_n)^\top:u_{i}=0,~\nu+1\le i\le n\}$. %if this is not the case, we can pick an orthogonal matrix $\bP$ such that $\bP\calV=\calV_0$, then  write $\by = \sign(|\bA\bx|-\tau)=\sign(|(\bA\bP^{-1})(\bP\bx)|-\tau)$ and note that  $\bP\bx \in \bP(\calV\cap \mathbbm{A}_{\alpha,\beta})= \calV_0\cap \mathbbm{A}_{\alpha,\beta}$. 
The lower bound is proved by   a counting argument. 
  First we seek an upper bound $\breve{U}$ on the number of quantized measurement vectors. 
  That is, we need to bound  $|\calM_{\bA}|$ where $\calM_{\bA} =  \{\sign(|\bA\bx|-\tau):\bx\in \calV_0 \cap \mathbbm{A}_{\alpha,\beta}\}.$ 
To address this, we first {\it linearize} the phaseless measurements by considering the set of measurements
	\begin{align}\label{eq:tildeM}
		\widetilde{\calM}_{\bA}& = \left\{\begin{bmatrix}
			\sign(\bA\bx+\tau)\\
			\sign(\bA\bx-\tau)
		\end{bmatrix}:\bx\in \calV_0\right\} =\big\{\sign(\widetilde{\bA}\bx + \widetilde{\btau}):\bx\in\calV_0\big\},~\text{where }\widetilde{\bA}=\begin{bsmallmatrix}
			\bA\\\bA
		\end{bsmallmatrix},~\widetilde{\btau} = \begin{bsmallmatrix}
			\tau\mathbf{1}\\
			-\tau \mathbf{1}
		\end{bsmallmatrix}.
	\end{align} 
It obviously holds that $|\widetilde{\calM}_{\bA}|\ge |\calM_{\bA}|$, so it suffices to bound $|\widetilde{\calM}_{\bA}|$.
Since $\widetilde{\bA}\calV_0+\widetilde{\btau}$ is a $\nu$-dimensional affine space in $\mathbb{R}^{2m}$ (here, affine space refers to a linear subspace with possibly non-zero shift), we   seek to bound the maximal number of orthants intersected by an affine space. Specifically the lemma below gives $|\calM_{\bA}|\le |\widetilde{\calM}_{\bA}|\le (\frac{2em}{\nu})^\nu$. 

\begin{lem}\label{lem:intersect_affine}
	Let $I_{m,k}$ be the maximal number of orthants in $\mathbb{R}^m$ intersected by a $k$-dimensional affine space, then we have  $I_{m,k}\le(\frac{em}{k})^k$. 
\end{lem}
\begin{proof}
	The proof can be found in  
	Appendix \ref{appother}.
\end{proof}

    We then construct a packing set contained in $\calV_0\cap \mathbbm{A}_{\alpha,\beta}$ with cardinality larger than $\breve{L}$, such that any two different points  $\bu$ and $\bv$ in this set satisfy $\dist(\bu,\bv)\ge \epsilon$. This can be achieved by a standard volume argument. Therefore, the condition  $\breve{U}\ge\breve{L}$
is necessary for any algorithm to reconstruct all $\bx\in \calV_0\cap \mathbbm{A}_{\alpha,\beta}$ to an accuracy of $\frac{\epsilon}{2}$. This condition immediately gives the lower bound. 

\section{Proof Sketches: Efficient Algorithms}\label{sec:pgdproof}

We introduce a  generic notation 
\begin{align}\label{eq:huv}
	\bh(\bu,\bv) := \frac{1}{2m}\sum_{i=1}^m \big[\sign(|\ba_i^\top\bu|-\tau)-\sign(|\ba_i^\top\bv|-\tau)\big] \sign(\ba_i^\top\bu)\ba_i 
\end{align}
for the gradient $\partial\calL(\bu)$ in (\ref{eq:subgra}). With respect to a specific underlying signal $\bx$, we have $\partial \calL(\bu)=\bh(\bu,\bx)$. 

\subsection{PLL-AIC Implies Convergence}
Our convergence guarantees are built upon a deterministic property on the interplay among several parties: the Gaussian matrix $\bA$, the quantization threshold $\tau$, some cone $\calC$ ($\calC=\mathbb{R}^n$ for GD-1bPR, $\calC=\Sigma^n_k$ for BIHT-1bSPR) and the step size $\eta$. Let us begin with   the formal definition of this property. 
\begin{definition}
	[Phaseless Local Approximate Invertibility Condition (PLL-AIC)] \label{def:plaic}Given $\beta_1\ge\alpha_1>0$ and $\tau>0$, a measurement matrix $\bA=[\ba_1^\top,\cdots,\ba_m^\top]^\top\in \mathbb{R}^{m\times n}$, a cone $\calC$, a step size $\eta$, and certain non-negative scalars $\bdelta=(\delta_1,\delta_2,\delta_3,\delta_4)^\top$, we say $(\bA,\tau,\calC,\eta)$ respects $(\alpha_1,\beta_1,\bdelta)$-PLL-AIC if 
	\begin{subequations}\label{eq:delta1_low}
		\begin{align}
			\|\calP_{\calC_-}(\bu-\bv-\eta\cdot \bh(\bu,\bv))\|_2 \le \delta_1\|\bu-\bv\|_2+\sqrt{\delta_2\cdot\|\bu-\bv\|_2}+ \delta_3,&\\
			\forall \bu,\bv\in\calC\cap\mathbbm{A}_{\alpha_1,\beta_1}\text{ obeying }\|\bu-\bv\|_2\le \delta_4,&
		\end{align}
	\end{subequations}
	where $\bh(\bu,\bv)$ is defined through $\bA$ and $\tau$ in (\ref{eq:huv}).
\end{definition}

%We mention that PLL-AIC is a uniform local property, in that it {\it universally} applies to $\bu,\bv$ in $\mathbbm{A}_{\alpha_1,\beta_1}$ that are {\it not so far apart}. 
Despite the simple ideas behind Algorithm \ref{alg:pgd}--\ref{alg:pgd_high}, the theoretical analysis remains highly challenging due to the presence of  multiple unusual bits in the gradient $\bh(\bu,\bv)$. PLL-AIC will   prove useful in that it  implies the convergence of {GD-1bPR} to certain error rate. %This  phenomenon is precisely characterized in the following lemma. 

%Without loss of generality, we assume here and hereafter that the warm initialization $\bx^{(0)}$ is closer to $\bx$ than $-\bx$,  which gives $\dist(\bx^{(0)},\bx)=\|\bx^{(0)}-\bx\|_2$ and   allows us to simply work with $\{\|\bx^{(t)}-\bx\|_2\}_{t\ge 0}$. 

%All the subsequent attempts are built upon a deterministic analysis of a single iteration:
%\begin{subequations}\label{eq:iter_low}
%   \begin{gather}
	%  \tilde{\bx}^{(t)}=\bx ^{(t-1)} -\eta\cdot \partial \calL(\bx^{(t-1)});\\
	% \bx^{(t)} = \calP_{\mathbbm{A}_{\alpha,\beta}}(\tilde{\bx}^{(t)}). 
	%\end{gather}
	%\end{subequations}
	% 
	\begin{lem}[PLL-AIC Implies Convergence of GD-1bPR]
		\label{lem:low_aic2err} Suppose that $(\bA,\tau,\eta,\mathbb{R}^n)$ respects $(\frac{\alpha}{2},2\beta,\bdelta=(\delta_1,\delta_2,\delta_3,\delta_4)^\top)$-PLL-AIC for some $\delta_1\in[0,1)$,
		with the involved $\{\delta_i\}_{i=1}^4$ obeying \begin{align*}
			E(\delta_1,\delta_2,\delta_3):=\max\Big\{\frac{16\delta_2}{(1-\delta_1)^2},\frac{4\delta_3}{1-\delta_1}\Big\}<\delta_4<\min\Big\{\frac{\alpha}{2},1\Big\}.
		\end{align*}
		Then for any $\bx\in \mathbbm{A}_{\alpha,\beta}$,  Algorithm \ref{alg:pgd} is capable of  entering a neighbourhood that surrounds
		$\pm\bm{x}$ \begin{align}
			\mathbb{B}_2^n\big(\bx;E(\delta_1,\delta_2,\delta_3)\big)\cup\mathbb{B}_2^n\big(-\bx;E(\delta_1,\delta_2,\delta_3)\big)\label{eq:low_neighbour}
		\end{align} with no more than $\lceil \frac{\log E(\delta_1,\delta_2,\delta_3)}{\log((1+\delta_1)/2)}\rceil $ gradient descents and then  staying in (\ref{eq:low_neighbour}) in later iterations, with the proviso that {GD-1bPR} is executed with step sizes $\eta$, the clean phaseless bits $\by=\sign(|\bA\bx|-\tau)$ and a good initialization $\bx^{(0)}$ satisfying $\dist(\bx^{(0)},\bx)\le\delta_4$.
	\end{lem}
	
	\begin{proof}
		The proof can be found in Appendix \ref{app:pllaic_lem}. 
	\end{proof}
	%{\color{blue}[continue from here: Discuss the parameters in pll-aic we need to show]} 
%\subsection{PLL-RAIC Implies Convergence}
	%We propose a deterministic property on the triple $(\bA,\tau,\eta(\bu,\bv))$, termed PLL-RAIC, to govern our convergence analysis for {BIHT-1bSPR}. In contrast to PLL-AIC in Definition \ref{def:plaic}, PLL-RAIC is  restricted to $k$-sparse signals and  tied to the thresholded gradient. In turn, it can be established with a sample complexity proportional to the sparsity $k$.

	%PLL-RAIC ensures the convergence of {BIHT-1bSPR}, but this implication requires $\delta_1<\frac{1}{2}$ that is more restricted than $\delta_1<1$ in Lemma \ref{lem:low_aic2err} due to the thresholding operator $\calT_{(k)}(\cdot)$ incorporated into the iteration.  
	
	\begin{lem}
		[PLL-RAIC Implies Convergence of BIHT-1bSPR] \label{lem:high_aic2err} Suppose that $(\bA,\tau,\eta,\Sigma^n_k)$ respects $(0.99\alpha,1.01\beta,\bdelta=(\delta_1,\delta_2,\delta_3,\delta_4)^\top)$-PLL-AIC for some $\delta_1\in[0,\frac{1}{2})$, with the involved $\{\delta_i\}_{i=1}^4$ obeying
		\begin{align*}
			E(\delta_1,\delta_2,\delta_3):= \max\Big\{\frac{64\delta_2}{(1-2\delta_1)^2},\frac{8\delta_3}{1-2\delta_1}\Big\} <\delta_4< \min\Big\{0.01\alpha,1\Big\}.
		\end{align*}
		Then for any $\bx\in \Sigma^n_k\cap \mathbbm{A}_{\alpha,\beta}$,  Algorithm \ref{alg:pgd_high} is capable of entering a neighbourhood surrounding $\pm\bx$ 
		\begin{align}
			\mathbb{B}_2^n(\bx;E(\delta_1,\delta_2,\delta_3))\cup \mathbb{B}_2^n(-\bx;E(\delta_1,\delta_2,\delta_3))\label{eq:nei_high}
		\end{align}
		with no more than $\lceil\frac{\log E(\delta_1,\delta_2,\delta_3)}{\log( 1/2+\delta_1)}\rceil$ iterations and then staying in (\ref{eq:nei_high}) in later iterations, with the proviso that it is executed with step sizes $\eta$, the clean 1-bit observations $\by=\sign(|\bA\bx|-\tau)$ and a good initialization $\bx^{(0)}$ satisfying $\dist(\bx^{(0)},\bx)\le\delta_4$. 
	\end{lem}
	\begin{proof}
		The proof can be found in Appendix \ref{app:pllraic2conver}. 
	\end{proof}
	
	\subsection{Gaussian Matrix Respects PLL-AIC}\label{sec:gau_pll_aic}
    We have the following theorem which states $(\bA,\tau,\calC,\eta)$ respects local PLL-AIC. {We introduce the shorthand $\calC_{\alpha,\beta}:=\calC\cap \mathbb{A}_{\alpha,\beta}$.}
    %\footnote{In view of Lemmas \ref{lem:low_aic2err}--\ref{lem:high_aic2err} we shall work with $\mathbbm{A}_{\frac{\alpha}{2},2\beta}$, while in the analysis we simply establish PLL-AIC over $\mathbbm{A}_{\alpha,\beta}$; the gap can be easily closed by eventually replacing $(\alpha,\beta)$ with $(\frac{\alpha}{2},2\beta)$.}  
   \begin{theorem}\label{thm:raic}
       Suppose $\bA\sim \calN^{m\times n}(0,1)$,    the   $\alpha,\beta,\tau$ are any given positive numbers,  $\calC$ is a cone in $\mathbb{R}^n$. For some constants $c_i$'s and $C_i$'s depending on $(\alpha,\beta,\tau)$, if $r\in(0,c_1)$ with small enough $c_1$, $mr\ge C_2 [\scrH(\calC_{\alpha,\beta},r)+\omega^2(\calC_{(1)})]$, then with probability at least $1-\exp(-c_3\scrH(\calC_{\alpha,\beta},r))$, $(\bA,\tau,\calC,\eta)$ respects $(\alpha,\beta,\bm{\delta}=(\delta_1,\delta_2,\delta_3,\delta_4)^\top)$-PLL-AIC with $$\delta_1=\sup_{a^2+b^2\in[\alpha^2,\beta^2]}\sqrt{|1-\eta g_\eta(a,b)|^2+|\eta h_\eta(a,b)|^2}+c_3\log^{-1/8}(r^{-1}),$$ $\delta_2= C_4r$, $\delta_3=C_5r\log(r^{-1})$ and $\delta_4=\frac{c_5}{\log^{1/2}(r^{-1})}$, where 
       \begin{gather*}
           g_\eta(a,b)=\sqrt{\frac{2}{\pi}}\exp\Big(-\frac{\tau^2}{2(a^2+b^2)}\Big)\frac{\tau^2a^2+b^2(a^2+b^2)}{(a^2+b^2)^{5/2}},\\
           h_\eta(a,b)=\sqrt{\frac{2}{\pi}}\exp\Big(-\frac{\tau^2}{2(a^2+b^2)}\Big)\frac{ab(a^2+b^2-\tau^2)}{(a^2+b^2)^{5/2}}.
       \end{gather*}
   \end{theorem}   
\begin{rem}
    We note that this result can be established over the so-called ``star-shaped'' $\calC$; see our follow-up work \cite{chen2024optimal}. However, we will not pursue this since we focus on recovering unstructured or sparse signals, as with most prior works in phase retrieval. 
\end{rem}
We only provide an overview for the proof of Theorem \ref{thm:raic} in the main text. Further details are postponed to Appendix \ref{app:prove_pgd}. Before that, we shall also use Theorem \ref{thm:raic} to prove Theorems \ref{thm:pgd_low}--\ref{thm:pgd_high}. 

\begin{proof}
    [Proof of Theorem \ref{thm:pgd_low}] We specialize Theorem \ref{thm:raic} to $\calC=\mathbb{R}^n$, $\eta=\sqrt{\frac{\pi e}{2}}\tau$ and $r=\frac{Cn}{m}\log(\frac{m}{n})$ with large enough $C$ (justified by standard bounds on $\scrH(\mathbbm{A}_{\alpha,\beta},r)$ and $\omega^2(\mathbb{B}_2^n)$) to obtain the following: with probability at least $1-\exp(-c_1n\log(\frac{m}{n}))$, $(\bA,\tau,\mathbb{R}^n,\eta=\sqrt{\frac{\pi e}{2}}\tau)$ respects $(\frac{\alpha}{2},2\beta,\bm{\delta}=(\delta_1,\delta_2,\delta_3,\delta_4)^\top)$-PLL-AIC with $$\delta_1= \sup_{a^2+b^2\in[\alpha^2/4,4\beta^2]}\sqrt{|1-\eta g_\eta(a,b)|^2+|\eta h_\eta(a,b)|^2}+\frac{c_1}{\log^{1/8}(\frac{m}{n})}<1-c_2$$ for some $c_2>0$ when $m\gtrsim n$ (this follows from Lemma \ref{lem:maximum}), $\delta_2=\frac{C_3n}{m}\log(\frac{m}{n})$, $\delta_3=\frac{C_4n}{m}\log^2(\frac{m}{n})$ and $\delta_4=\frac{c_4}{\log^{1/2}(m/n)}$. The statement then follows from Lemma \ref{lem:low_aic2err}.  
\end{proof}
\begin{proof}
    [Proof of Theorem \ref{thm:pgd_high}] We specialize Theorem \ref{thm:raic} to $\calC=\Sigma^n_k$, $\eta=\sqrt{\frac{\pi e}{2}}\tau$ and $r=\frac{Ck}{m}\log(\frac{mn}{k^2})$ with large enough $C$ (justified by standard bounds on $\scrH(\Sigma^n_k\cap \mathbbm{A}_{\alpha,\beta},r)$ and $\omega^2(\Sigma^n_k\cap \mathbb{B}_2^n)$) to obtain the following: with probability at least $1-\exp(-c_1 k\log(\frac{mn}{k^2}))$, $(\bA,\tau,\Sigma^n_k,\eta=\sqrt{\frac{\pi e}{2}}\tau)$ respects $(0.99\alpha,1.01\beta,\bm{\delta}=(\delta_1,\delta_2,\delta_3,\delta_4)^\top)$-PLL-AIC with 
    \begin{align*}
        &\delta_1=\sup_{a^2+b^2\in[(0.99\alpha)^2,(1.01\beta)^2]}\sqrt{|1-\eta g_\eta(a,b)|^2+|\eta h_\eta(a,b)|^2}+\frac{c_1}{\log^{1/8}(r^{-1})}\\
        &=\max\left\{\max_{w\in[\frac{\tau}{1.01\beta},\frac{\tau}{0.99\alpha}]}\Big|1-w^3\exp(\frac{1-w^2}{2})\Big|,\max_{w\in [\frac{\tau}{1.01\beta},\frac{\tau}{0.99\alpha}]}\Big|1-w\exp(\frac{1-w^2}{2})\Big|\right\}+\frac{c_1}{\log(r^{-1})}
    \end{align*}
    (by using Lemma \ref{lem:maximum}), $\delta_2=\frac{C_2k}{m}\log(\frac{mn}{k^2})$, $\delta_3=\frac{C_3k}{m}\log(\frac{mn}{k^2})\log(\frac{m}{k})$ and $\delta_4=\frac{c_4}{\log^{1/2}(m/k)}$. Under (\ref{eq:fixed_ratio}) and small enough $r$ we will have $\delta_1<\frac{1}{2}-c_5$ for some $c_5>0$ depending on $(\alpha,\beta,\tau)$. The statement then follows from Lemma \ref{lem:high_aic2err}.  
\end{proof}

  \subsubsection*{Outline for the proof of Theorem \ref{thm:raic}:} 
 We build a minimal $r$-net of $\calC_{\alpha,\beta}$ that is denoted by $\calN_r$ and satisfies $\log|\calN_r|=\scrH(\calC_{\alpha,\beta},r)$. For any   $\bu,\bv\in \calC_{\alpha,\beta}$ we can pick their closest points in $\calN_r$: $\bu_1  := \text{arg}\min_{\bw \in \calN_r}\|\bw-\bu\|_2 $ and $\bv_1 := \text{arg}\min_{\bw\in\calN_r}\|\bw-\bv\|_2.$ 
	Then by Lemma \ref{lem:pro_closec} we have 
 		\begin{align}\nn
			& \|\calP_{\calC_-}(\bu-\bv-\eta\cdot\bh(\bu,\bv))\|_2\\& \le \|\calP_{\calC_-}(\bu-\bu_1)\|_2 + \|\calP_{\calC_-}(\bv-\bv_1)\|_2+ \|\calP_{\calC_-}(\bu_1-\bv_1-\eta\cdot\bh(\bu,\bv))\|_2\nn\\&\le 2r+  \|\calP_{\calC_-}(\bu_1-\bv_1-\eta\cdot\bh(\bu,\bv))\|_2.\label{eq:use_net1_rapp}
		\end{align} 
	Our further developments to bound $\|\calP_{\calC_-}(\bu_1-\bv_1-\bh(\bu,\bv))\|_2$ are built upon different ideas in two regimes, namely a {\it large-distance regime} where $\|\bu_1-\bv_1\|_2\ge r$, and a {\it small-distance regime} where $\|\bu_1-\bv_1\|_2< r$. Specifically, for $\|\bu_1-\bv_1\|_2\ge r$ we proceed as 
		\begin{align}\label{eq:large_decompose}
			\|\calP_{\calC_-}(\bu_1-\bv_1-\eta\bh(\bu,\bv))\|_2\le \|\calP_{\calC_-}(\bu_1-\bv_1-\eta\bh(\bu_1,\bv_1))\|_2 + \eta \|\calP_{\calC_-}(\bh(\bu,\bv)-\bh(\bu_1,\bv_1))\|_2;
		\end{align}
while for $\|\bu_1-\bv_1\|_2<r$ we begin with a simpler decomposition 
  \begin{align}\label{eq:small_decompose}
			\|\calP_{\calC_-}(\bu_1-\bv_1 -\eta\bh(\bu,\bv))\|_2 \le \|\calP_{\calC_-}(\bu_1-\bv_1)\|_2+\eta \|\calP_{\calC_-}(\bh(\bu,\bv))\|_2 < r + \eta\|\calP_{\calC_-}(\bh(\bu,\bv))\|_2. 
		\end{align}

   \begin{rem}
       Intuitively, the bound in large-distance regime induces a contraction  that shrinks $\dist(\bx^{(t)},\bx)$ to the near-optimal error, then the bound in small-distance regime ensures that  this near-optimal error rate is maintained in subsequent iterations.    
   \end{rem}

	\subsubsection{Large-Distance Regime ($\|\bu_1-\bv_1\|_2\ge r$)}
In the large-distance analysis we need to precisely characterize the gradient $\bh(\bu,\bv)$ to show contraction.

	\subsubsection*{Step 1: Bounding $\|\calP_{\calC_-}(\bu_1-\bv_1-\eta\cdot \bh(\bu_1,\bv_1))\|_2$}
    {Recall that $\delta_4\asymp \frac{1}{\log^{1/2}(r^{-1})}$ and $r$ is small enough, and hence we have $r<\frac{\delta_4}{2}$}, which then yields $\|\bu_1-\bv_1\|_2\le \|\bu_1-\bu\|_2+\|\bu-\bv\|_2+\|\bv-\bv_1\|_2\le 2r+\delta_4\le 2\delta_4$.   
	 Combining with $\|\bu_1-\bv_1\|_2\ge r$, $(\bu_1,\bv_1)$ in the large-distance regime lives in the constraint set 
	\begin{align}\label{eq:Nrxi}
		\calN_{r,\delta_4}^{(2)} = \big\{(\bp,\bq)\in \calN_r\times \calN_r: r\le \|\bp-\bq\|_2 \le 2\delta_4\big\}.
	\end{align}  %and   seek to bound 
	%\begin{align}\label{eq:most_diffi}
	%   \sup_{(\bp,\bq)\in \calN^{(2)}_{r,\xi}}\big\|\bp-\bq-\eta\cdot\bh(\bp,\bq)\big\|_2.
	%\end{align}
	Note that $(\bu_1,\bv_1)$ already lives in the finite set $\calN_{r,\delta_4}^{(2)}$, thus it is enough to bound $\|\bp-\bq-\eta\cdot\bh(\bp,\bq)\|_2$ for a fixed pair $(\bp,\bq)\in \calN^{(2)}_{r,\delta_4}$ and then invoke a union bound. 
	Nonetheless, tightly  bounding $\|\bp-\bq-\eta\cdot\bh(\bp,\bq)\|_2$  for fixed $(\bp,\bq)\in\calN_{r,\delta_4}^{(2)}$ remains challenging due to  the presence of the unusual binary factors. We summarize some of the key ideas as follows:  
	\begin{itemize}
		[leftmargin=2ex,topsep=0.25ex]
		\setlength\itemsep{-0.1em}
		\item We utilize an orthogonal decomposition of the random vector $\bp-\bq-\eta\cdot \bh(\bp,\bq)$, along with some algebraic manipulation on  $X_{i}^{\bp,\bq} :=[\sign(|\ba_i^\top\bp|-\tau)-\sign(|\ba_i^\top\bq|-\tau)]\sign(\ba_i^\top\bp),$ 
		%The algebra on $X_i^{\bp,\bq}$ and the orthogonal decomposition will be introduced in detail shortly, 
		to break down the problem into more manageable pieces. Note that we can write $\bh(\bp,\bq)=\frac{1}{2m}\sum_{i=1}^m X_i^{\bp,\bq}\ba_i$. 
  %(Note that   (\ref{eq:undesired_factor}) allows us to write $\bh(\bp,\bq)=\frac{1}{2m}\sum_{i=1}^mX_{i}^{\bp,\bq}\ba_i$.)
		
		\item  To derive tight concentration bounds for each piece, we observe that a small $\|\bp-\bq\|_2$ implies a small $d_H(\sign(|\ba_i^\top\bp|-\tau),\sign(|\ba_i^\top\bq|-\tau))$, and the number of non-zero contributors to $\bh(\bp,\bq)$ is indeed much fewer than $m$. This inspires us to establish bounds conditioning on the hamming distance ($|\bR_{\bp,\bq}|$ below) first. 
		%With respect to the second dot point above, we must forbid  any ``globally applicable'' concentration inequality that is ignorant of the closeness between $\bp$ and $\bq$ (i.e., $\|\bp-\bq\|_2\le 2\xi$). Rather, we must stick to a localized analysis throughout, with the core spirit being that small $\|\bp-\bq\|_2$ implies small $d_H(\sign(|\ba_i^\top\bp|-\tau),\sign(|\ba_i^\top\bq|-\tau))$ (see Theorem \ref{thm:local_embed}), and hence many of the $m$ summands in $\bh(\bp,\bq)$ vanish. This phenomenon makes it possible to obtain essentially sharper concentration bounds. 
	\end{itemize}

	Let us delve into the details in the following.
	
	\paragraph{Simplifying the gradient:} For the $m$ contributors to $\bh(\bp,\bq)$, only the ones in 
	\begin{align}\label{eq:def_Rpq}
		\bR_{\bp,\bq} = \{i\in[m]:\sign(|\ba_i^\top\bp|-\tau)\neq \sign(|\ba_i^\top\bq|-\tau)\}
	\end{align}
are non-zero. Note that $\bR_{\bp,\bq}$ collect the indexes corresponding to the phaseless separations induced by the phaseless hyperplane $\{\calH_{|\ba_i|}\}_{i=1}^m$. For $i\in \bR_{\bp,\bq}$ we can also simplify $X_i^{\bp,\bq}$ as 
        \begin{align}
			 X_i^{\bp,\bq}& =[\sign(|\ba_i^\top\bp|-\tau)-\sign(|\ba_i^\top\bq|-\tau)]\sign(\ba_i^\top\bp) \nn\\&= 2\sign(|\ba_i^\top\bp|-|\ba_i^\top\bq|)\sign(\ba_i^\top\bp)\nn\\&=2\sign(\ba_i^\top\bp-|\ba_i^\top\bq|\sign(\ba_i^\top\bp))\nn \\&=2 \sign\big(\ba_i^\top\bp - [\sign(\ba_i^\top\bp)\sign(\ba_i^\top\bq)]\ba_i^\top\bq\big).\label{eq:first_simplify}
		\end{align} 
	To allow for further reduction, we shall  define another  index set
 $	\bL_{\bp,\bq} = \big\{i\in[m]:\sign(\ba_i^\top\bp)\ne\sign(\ba_i^\top\bq)\big\}$
	with respect to the separations induced by hyperplanes $\{\calH_{\ba_i,0}\}_{i=1}^m$. Now we have   
	\begin{subequations}\label{eq:simplify_Xipq}
		\begin{align}
			&X_i^{\bp,\bq} = 2\sign(\ba_i^\top(\bp-\bq)),~\forall i\in \bR_{\bp,\bq}\setminus \bL_{\bp,\bq},\\
			&X_i^{\bp,\bq} = 2\sign(\ba_i^\top(\bp+\bq)),~\forall i\in \bR_{\bp,\bq}\cap \bL_{\bp,\bq}.
		\end{align}
	\end{subequations}
	Substituting (\ref{eq:simplify_Xipq}) into $\bh(\bp,\bq)=\frac{1}{2m}\sum_{i=1}^m X_i^{\bp,\bq}\ba_i$ yields
         \begin{align}\nn
			&\bh(\bp,\bq) = \frac{1}{2m}\Big(\sum_{i\in \bR_{\bp,\bq}\setminus \bL_{\bp,\bq}}2\sign(\ba_i^\top(\bp-\bq))\ba_i + \sum_{i\in \bR_{\bp,\bq}\cap \bL_{\bp,\bq}}2\sign(\ba_i^\top(\bp+\bq))\ba_i\Big)\\\nn 
			& = \frac{1}{m}\Big(\sum_{i\in\bR_{\bp,\bq}}\sign(\ba_i^\top(\bp-\bq))\ba_i - \sum_{i\in \bR_{\bp,\bq}\cap\bL_{\bp,\bq}} \sign(\ba_i^\top(\bp-\bq))\ba_i +\sum_{i\in \bR_{\bp,\bq}\cap\bL_{\bp,\bq}}\sign(\ba_i^\top(\bp+\bq))\ba_i \Big)\nn\\
			& = \underbrace{\frac{1}{m}\sum_{i\in\bR_{\bp,\bq}}\sign(\ba_i^\top(\bp-\bq))\ba_i}_{:=\bh_1(\bp,\bq)}+\underbrace{\frac{1}{m}\sum_{i\in \bR_{\bp,\bq}\cap \bL_{\bp,\bq}} \big[\sign(\ba_i^\top(\bp+\bq))-\sign(\ba_i^\top(\bp-\bq))\big]\ba_i}_{:=\bh_2(\bp,\bq)}.\label{eq:main_higher}  
		\end{align}  
	 By $\bh(\bp,\bq)=\bh_1(\bp,\bq)+\bh_2(\bp,\bq)$ and Lemma \ref{lem:pro_closec} we have
	\begin{align}\label{eq:sepa_main_side}
		\|\calP_{\calC_-}(\bp-\bq - \eta\cdot \bh(\bp,\bq))\|_2 \le \|\calP_{\calC_-}(\bp-\bq -\eta\cdot \bh_1(\bp,\bq))\|_2+\eta \cdot\|\calP_{\calC_-}(\bh_2(\bp,\bq))\|_2,
	\end{align}
        Our subsequent development  shows that $\|\calP_{\calC_-}(\bh_2(\bp,\bq))\|_2$ is negligible higher-order term compared to $\|\calP_{\calC_-}(\bp-\bq-\eta\cdot\bh_1(\bp,\bq))\|_2$ because $\bR_{\bp,\bq}\cap \bL_{\bp,\bq}$ is typically much smaller than $\bR_{\bp,\bq}$. We will then sketch the techniques for bounding $\|\calP_{\calC_-}(\bp-\bq-\eta\cdot \bh_1(\bp,\bq))\|_2$.

	\paragraph{Orthogonal decomposition:}
	For a specific pair $(\bp,\bq)\in\calN_{r,\delta_4}^{(2)}$, by Lemma \ref{lem:parameterization} we can pick {\it orthonormal} $\bbeta_1:=\frac{\bp-\bq}{\|\bp-\bq\|_2},~\bbeta_2\in \mathbb{R}^n$, such that 
		\begin{gather}\label{eq:beta1beta2}
			\bp =  u_1\bbeta_1+u_2\bbeta_2\quad\text{and}\quad\bq = v_1\bbeta_1+u_2\bbeta_2
		\end{gather}
	hold for some coordinates $(u_1,u_2,v_1)$ obeying $u_1>v_1$ and $u_2\ge 0$.\footnote{Here we drop the dependence of $(\bbeta_1,\bbeta_2)$ on $(\bp,\bq)$ to avoid cumbersome notation.} This gives the following orthogonal decomposition of $\bh_1(\bp,\bq)$: 
		\begin{align*}
			\bh_1(\bp,\bq) &= \langle\bh_1(\bp,\bq),\bbeta_1\rangle\bbeta_1+ \langle \bh_1(\bp,\bq),\bbeta_2\rangle\bbeta_2 +\underbrace{\big\{\bh_1(\bp,\bq)-\langle\bh_1(\bp,\bq),\bbeta_1\rangle\bbeta_1-\langle \bh_1(\bp,\bq),\bbeta_2\rangle\bbeta_2\big\}}_{:=\bh_1^\bot(\bp,\bq)}. 
		\end{align*}  
	Combining with Lemma \ref{lem:pro_closec} and the symmetry of $\calC_-$, we decompose $\|\calP_{\calC_-}(\bp-\bq-\eta\cdot \bh_1(\bp,\bq))\|_2$ into 
		\begin{align}\nn
			&\|\calP_{\calC_-}(\bp-\bq-\eta\cdot\bh_1(\bp,\bq))\|_2 \\&\nn \le \big\|\bp-\bq-\eta \cdot \langle\bh_1(\bp,\bq),\bbeta_1\rangle \bbeta_1-\eta\cdot\langle \bh_1(\bp,\bq),\bbeta_2\rangle \bbeta_2\big\|_2 + \eta\cdot\|\calP_{\calC_-}(\bh^\bot_1(\bp,\bq))\|_2  \\\nn
			&=\Big\|\Big(\|\bp-\bq\|_2-\eta\cdot\Big<\bh_1(\bp,\bq),\frac{\bp-\bq}{\|\bp-\bq\|_2}\Big>\Big)\bbeta_1 - \eta\cdot \langle \bh_1(\bp,\bq),\bbeta_2\rangle\bbeta_2\Big\|_2 + \eta\cdot\big\|\calP_{\calC_-}(\bh^\bot_1(\bp,\bq))\big\|_2
			\\
			&\le \Big(\Big|\|\bp-\bq\|_2 - \eta\cdot \Big\langle\bh_1(\bp,\bq),\frac{\bp-\bq}{\|\bp-\bq\|_2}\Big\rangle\Big|^2+ \eta^2\cdot \big| \langle \bh_1(\bp,\bq),\bbeta_2\rangle\big|^2\Big)^{1/2}+\eta\cdot\big\|\calP_{\calC_-}(\bh_1^\bot(\bp,\bq))\big\|_2, \nn\\
			:&=((T_1^{\bp,\bq})^2+\eta^2\cdot |T_2^{\bp,\bq}|^2)^{1/2}+\eta \cdot T_3^{\bp,\bq},\label{eq:p_q_hpq}
		\end{align}  
	where in the last line we introduce $ T_1^{\bp,\bq} :=  |\|\bp-\bq\|_2 - \eta\cdot \langle\bh_1(\bp,\bq),\frac{\bp-\bq}{\|\bp-\bq\|_2}\rangle|$, $	T_2^{\bp,\bq}:= \langle \bh_1(\bp,\bq),\bbeta_2\rangle$ and $	T_3^{\bp,\bq}:= \|\calP_{\calC_-}(\bh_1^\bot(\bp,\bq))\|_2$.

	\paragraph{Bounding $\{T_i^{\bp,\bq}\}_{i=1}^3$:} It remains to bound $T_i^{\bp,\bq}~(i=1,2,3)$ separately.  Substituting $\bh_1(\bp,\bq)=\frac{1}{m}\sum_{i\in\bR_{\bp,\bq}}\sign(\ba_i^\top(\bp-\bq))\ba_i$ into $T_i^{\bp,\bq}~(i=1,2,3)$ and conditioning on $|\bR_{\bp,\bq}|$, we observe that they are the sum of $|\bR_{\bp,\bq}|$ i.i.d. random variables that follow certain conditional distributions.  We then proceed to show the sub-Gaussianity of these conditional distributions and establish their concentrations around the means. Here,  showing sub-Gaussianity and calculating the means again prove   technical due to the conditioning on the unusual rare event of ``two close points $\bp,\bq$ being separated by a Gaussian phaseless hyperplane.'' 
	Our mechanism is to first derive the probability density function (P.D.F.) as a handle and then exhaustively examine  every possible relative location  of $(\bp,\bq)$, each of which typically corresponds to  different  conditional distributions. 
	With the conditional concentration ready, we further  analyze $|\bR_{\bp,\bq}|$ via Chernoff bound to establish the unconditional concentration bounds. All the   details are deferred to Appendix \ref{app:large_main_term}. 
	
	\subsubsection*{Step 2: Bounding $\|\calP_{\calC_-}(\bh(\bu,\bv)-\bh(\bu_1,\bv_1))\|_2$} 
 In view of (\ref{eq:large_decompose}) we still need to  uniformly control $\|\calP_{\calC_-}(\bh(\bu,\bv)-\bh(\bu_1,\bv_1))\|_2$ where   $\bu_1,\bv_1$ are the points in $\calN_r$ that are closest to $\bu,\bv$. 
	Compared to $\|\calP_{\calC_-}(\bu_1-\bv_1-\eta\cdot \bh(\bu_1,\bv_1))\|_2$, an evident challenge here is that there are infinitely many pairs of $(\bu,\bv)$, and in turn we need fundamentally different idea.
	More specifically, we seek to establish a tight enough bound through counting the overall number of non-zero contributors, an idea that will be recurring in subsequent development. 
	However, when separately examining $\bh(\bu_1,\bv_1)$ and $\bh(\bu,\bv)$,  they   contain  $|\bR_{\bu_1,\bv_1}|$ and $|\bR_{\bu,\bv}|$ non-zero contributors, which are still overly numerous and preclude us from establishing satisfactory bound using this approach. 
	Observing that $|\bR_{\bu,\bu_1}|$ and $|\bR_{\bv,\bv_1}|$ are notably smaller than $|\bR_{\bu_1,\bv_1}|$ and $|\bR_{\bu,\bv}|$ due to $\|\bu_1-\bu\|_2\le r$ and $\|\bv_1-\bv\|_2\le r$, our remedy is to instead utilize the closeness of $(\bu,\bu_1)$ and $(\bv,\bv_1)$. To achieve this goal, we   first ``decouple'' $(\bu,\bv)$ and $(\bu_1,\bv_1)$,  then ``recouple'' $(\bu,\bu_1)$ and $(\bv,\bv_1)$ to obtain 
	\begin{subequations}
	    	\begin{align}
			\bh(\bu_1,\bv_1)-\bh(\bu,\bv)&=\frac{1}{2m}\sum_{i=1}^m\big[\sign(|\ba_i^\top\bv|-\tau)-\sign(|\ba_i^
			\top\bv_1|-\tau)\big]\sign(\ba_i^\top\bu_1)\ba_i\label{eq:T4}\\
			& + \frac{1}{2m}\sum_{i=1}^m \big[\sign(|\ba_i^\top\bu_1|-\tau)-\sign(|\ba_i^\top\bu|-\tau)\big]\sign(\ba_i^\top\bu_1)\ba_i\label{eq:T5}\\
			& + \frac{1}{2m}\sum_{i=1}^m \big[\sign(\ba_i^\top\bu)-\sign(\ba_i^\top \bu_1)\big]\big[\sign(|\ba_i^\top\bv|-\tau)-\sign(|\ba_i^\top\bu|-\tau)\big]\ba_i\label{eq:T6},
		\end{align}  \label{eq:decom_hminush}
	\end{subequations}
	where we note that the term in the right-hand side of (\ref{eq:T4}) contains no more than $|\bR_{\bv,\bv_1}|$ non-zero contributors, the term in (\ref{eq:T5}) contains no more than $|\bR_{\bu,\bu_1}|$ non-zero contributors, and the term in (\ref{eq:T6}) contains no more than $|\bL_{\bu,\bu_1}|$ non-zero contributors (due to the factor $\sign(\ba_i^\top\bu)-\sign(\ba_i^\top\bu_1)$). Therefore, we deduce that $\bh(\bu_1,\bv_1)-\bh(\bu,\bv)$ involves no more than $|\bR_{\bv,\bv_1}|+|\bR_{\bu,\bu_1}|+|\bL_{\bu,\bu_1}|$ effective contributors, a significant reduction from the trivial bound    {$|\bR_{\bu_1,\bv_1}|+|\bR_{\bu,\bv}|$}. In what follows, 
	Our Theorem \ref{thm:local_embed} and Lemma \ref{lem:binaryembed} immediately imply  uniform bounds on  $|\bR_{\bv,\bv_1}|,~|\bR_{\bu,\bu_1}|$ and $|\bL_{\bu,\bu_1}|$. These bounds are tight enough, hence we are able to further establish satisfactory bounds by triangle inequality and Lemma \ref{lem:max_ell_sum}.  All these technical details for bounding $\|\calP_{\calC_-}(\bh(\bu,\bv)-\bh(\bu_1,\bv_1))\|_2$ are provided in Appendix \ref{app:difference_h}.

	\subsubsection{Small-Distance Regime ($\|\bu_1-\bv_1\|_2<r$)}\label{sec:smalltech}
       In (\ref{eq:small_decompose}) it remains to bound $\|\calP_{\calC_-}(\bh(\bu,\bv))\|_2$ where $\bh(\bu,\bv)$ contains no more than $|\bR_{\bu,\bv}|$ non-zero contributors. Since in the small-distance regime we have $\|\bu-\bv\|_2\le \|\bu-\bu_1\|_2+\|\bu_1-\bv_1\|_2+\|\bv_1-\bv\|_2 <3r,$ $|\bR_{\bu,\bv}|$ here is order-wise as small as $|\bR_{\bv,\bv_1}|$ in bounding $\|\calP_{\calC_-}(\bh(\bu_1,\bv_1)-\bh(\bu,\bv))\|_2$ in large-distance regime.  Therefore we can again apply the techniques developed therein, i.e., first uniformly controlling $|\bR_{\bu,\bv}|$ and then invoke Lemma \ref{lem:max_ell_sum}. The details can be found in Appendix \ref{app:small_distance}.

		\subsection{Technical Comparison with   NBIHT}\label{app:matter}
	%{\color{blue}[continue from here]} We position our major technical contributions in the proofs of Theorems \ref{thm:pgd_low}--\ref{thm:pgd_high}. 
 In this subsection, we provide a technical comparison to the analysis of NBIHT \cite{matsumoto2024binary}. The discussions will be based on the set of notation in the present paper, particularly with $\bh(\bu,\bv)$   denoting the gradient and $r$   denoting the covering radius.  In large-distance regime,   the gradient in NBIHT takes the much simpler form
			$	\bh(\bu,\bv) = \frac{1}{2m}\sum_{i=1}^m \big(\sign(\ba_i^\top\bu)-\sign(\ba_i^\top\bv)\big)\ba_i$ that is amenable to directly work with. In contrast, we encounter a more intricate $\bh(\bu,\bv) = \frac{1}{2m}\sum_{i=1}^m[\sign(|\ba_i^\top\bu|-\tau)-\sign(|\ba_i^\top\bv|-\tau)]\sign(\ba_i^\top\bu)\ba_i$ and will need to first decompose it into $\bh_1(\bu,\bv)+\bh_2(\bu,\bv)$ by examining two index sets $\bR_{\bu,\bv}$ and $\bL_{\bu,\bv}$.  We thus need a set of additional technicalities such as bounding the higher-order term $\bh_2(\bu,\bv)$ and  rearrangement	(\ref{eq:decom_hminush}). Another set of new challenges come from the signals living in an annulus rather than the standard sphere, which   already adds to complication in the orthogonal decomposition of the gradient. More prominently, we have to consider multiple situations when analyzing the conditional distributions,
   since     they essentially hinge on the   ``relative position'' of two points.  
   % compared to \cite[Definition 3.1]{matsumoto2024binary}, we additionally include $\delta_1\|\bu-\bv\|_2$ in our AIC (see (\ref{eq:delta1_low}) and (\ref{eq:delta1_high})), since it is not possible to achieve AIC with $\delta_1=0$ when $\beta>\alpha$. As a result, 
			%the selection of the step size in our algorithms also becomes more tricky, for which the criterion is to lower $\delta_1$ in the AIC.
In the small-distance regime, we seek to control $\|\calP_{\calC_-}(\bh(\bu,\bv))\|_2$ uniformly for all $\bu,\bv$ in $\calK$ obeying $\|\bu-\bv\|_2\le 3r$. {Compared to the small-distance analysis in \cite[Section A.2]{matsumoto2024binary} where the authors utilized a result from \cite{oymak2015near} (see a restatement in our Lemma \ref{lem:binaryembed}) to uniformly bound the number of the non-zero contributors to the gradient and then invoked the conditional concentration bounds from large-distance regime again, we make the following observation that enables a simpler approach: unlike in large-distance regime where   sharp two-sided concentration  of the directional gradient around its mean is necessary, a uniform upper bound (at the order of the optimal reconstruction error) suffices for the small-distance regime, which can be achieved by directly invoking Lemma \ref{lem:max_ell_sum} without resorting to the techniques in large-distance regime again. This argument can be applied to simplify the analysis of the optimality of NBIHT; see \cite{chen2024optimal}. Also note that we instead use   Theorem \ref{thm:local_embed} to establish uniform bound on the number of non-zero contributors in our gradient.}

\section{Numerical Experiments}\label{sec:numerics}
 We provide a few sets of numerical results to complement the theoretical developments. In particular, our goal is to complement the intriguing message that ``phases are not essential for 1-bit compressed sensing''\footnote{{In a nutshell, our experimental results indicate that, when using their respective optimal algorithms, reconstructing unstructured or sparse $\bx$ from $\sign(|\bA\bx|-\tau)$ can achieve smaller reconstruction errors than from $\sign(\bA\bx)$. We emphasize, however, that this observation depends on the specific experimental setup, in particular the choice of the quantization threshold $\tau$ in 1-bit phase retrieval. We also observe that 1-bit sparse phase retrieval generally performs less stably than 1-bit compressed sensing when the number of measurements is relatively small (e.g., $m \le 800$ in Figure~\ref{fig:inessential} (Right)), which is consistent with the higher sample complexity required for 1-bit sparse phase retrieval; see Remark~\ref{rem:k2}.}} and demonstrate the practical merits of the proposed algorithms. All experiments were implemented using Matlab R2022a on a laptop with an Intel CPU up to 2.5 GHz and 32 GB RAM. {The codes are available at}  
 $$\href{https://github.com/junrenchen58/one-bit-phase-retrieval}{\texttt{https://github.com/junrenchen58/one-bit-phase-retrieval}}$$
\paragraph{1-Bit Phase Retrieval v.s. 1-Bit Linear System:} 
We refer to the recovery of an unstructured $\bx\in \mathbb{S}^{n-1}$ from $\by = \sign(\bA\bx)$ as solving 1-bit linear system (1bLS). By introducing the linear constraint $-\frac{\by^\top \bA\bx}{m}= {-\frac{\|\bA\bx\|_1}{m}\le -1}$
that precludes an algorithm from returning $0$, we   solve 1-bit linear system by first finding a feasible point $\hat{\bx}$ satisfying the linear constraint $\diag(\by)\bA\hat{\bx}\le0,~-\frac{\by^\top\bA\hat{\bx}}{m}\le -1$ 
and then use $\frac{\hat{\bx}}{\|\hat{\bx}\|_2}$ as the final estimate. For 1-bit phase retrieval (1bPR) that recovers unstructured $\bx\in \mathbb{S}^{n-1}$ from the 1-bit phaseless measurements $\by=\sign(|\bA\bx|-\tau)$,\footnote{Otherwise stated, we set $\tau=\sqrt{\alpha\beta}$  without tuning in the experiments (e.g., $\tau=1$ is used in Figure \ref{fig:inessential}), and run {GD-1bPR}, {BIHT-1bSPR} and {NBIHT} for 150 iterations.} we execute the provably near-optimal solver {GD-1bPR} with $\bx^{(0)}$ found by {SI-1bPR} and step size $\eta= \sqrt{\frac{\pi e}{2}}\tau$. The experimental results are reported as log-log curves in Figure \ref{fig:inessential}(Left). The results display the theoretically optimal decay rate $\calO(m^{-1})$, and surprisingly the errors in the phaseless setting are smaller  under the same measurement number. For reproducibility, more details on data generation are provided in the caption of Figure \ref{fig:inessential}. 
\begin{figure}[ht!]
	\begin{centering}
		\includegraphics[width=0.4\columnwidth]{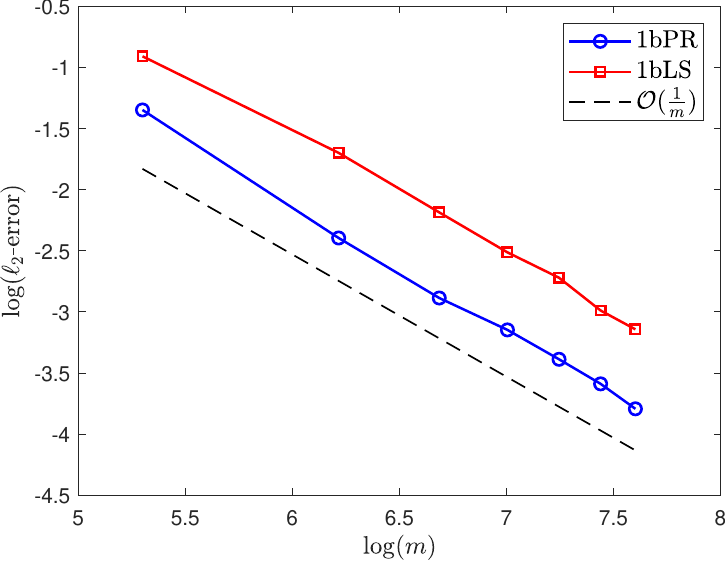} \quad \includegraphics[width=0.4\columnwidth]{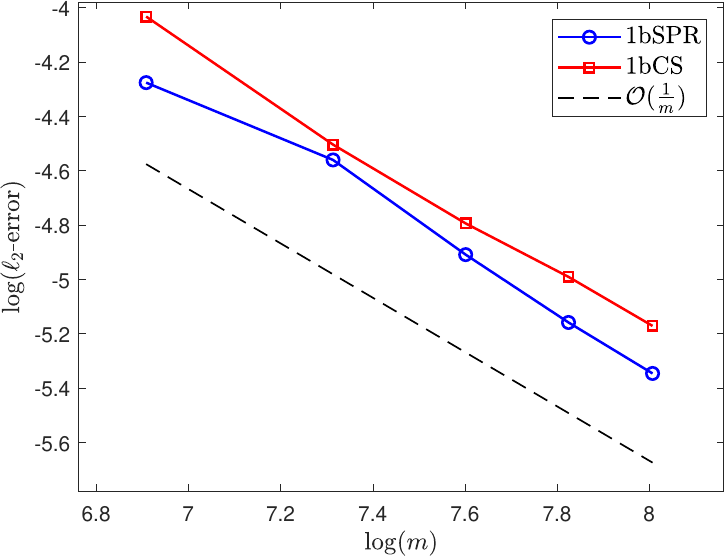}
		\par
	\end{centering}
	
	\caption{\label{fig:inessential}\small {\it Phases are non-essential  in solving 1-bit linear system (Left) and in 1-bit compressed sensing (Right).} The reported data points are averaged over 50 independent trials.  In the left figure we recover $\bx\in \mathbb{R}^{30}$ uniformly distributed over $\mathbb{S}^{29}$ from $m=200:300:2000$ bits produced by Gaussian ensemble (the same below if not specified). The right figure is for  the recovery of $\bx\in\Sigma^{500,*}_3$ from $m=1000:500:3000$  bits, whereas we feed {BIHT-1bSPR} and {SI-1bSPR} with a (slightly) looser sparsity $k=4$ to simulate an actual setting where $k$ is not precisely known. The signals have support uniformly drawn from $\binom{500}{3}$ possibilities and non-zero entries uniformly distributed over $\mathbb{S}^2$.} %Intriguingly, the estimation errors are notably lower in the phaseless settings, thereby empirically supporting the message that ``phases are non-essential in 1-bit compressed sensing''.}
\end{figure}

\paragraph{1-Bit Sparse Phase Retrieval v.s. 1-Bit Compressed Sensing:} 
We simulate sparse recovery and compare 1-bit compressed sensing (1bCS) and 1-bit sparse phase retrieval (1bSPR). The unique provably near-optimal 1bCS solver is NBIHT \cite{matsumoto2024binary}, which is initialized   by the PBP estimator \cite{plan2017high,xu2020quantized} for the sake of fairness. To solve 1-bit sparse phase retrieval, we perform {SI-1bSPR} to get $\bx^{(0)}$ and then execute {BIHT-1bSPR} with $\eta=\sqrt{\frac{\pi e}{2}}\tau$.  We show the experimental results in Figure \ref{fig:inessential}(Right) and provide further details in the caption. Again, the estimation errors in our phaseless setting are clearly smaller, which experimentally demonstrates {phases are non-essential for 1-bit compressed sensing}. %We observed that $\eta=\sqrt{\frac{\pi e}{2}}\tau$ and $\eta=\sqrt{\frac{\pi e}{2}}\|\bx^{(t-1)}\|_2$ oftentimes lead to similar numerical results, with the former being  slightly stabler. Thus, we simply adhere to $\eta=\sqrt{\frac{\pi e}{2}}\tau$ later.

\paragraph{The Impact of $\tau$ on 1b(S)PR:} The quantization threshold $\tau$ is an important tuning parameter, particularly in the sparse setting (cf. Remark~\ref{rem:tau}). Here we simulate 1b(S)PR under different values of $\tau$ to empirically assess its impact on reconstruction performance, focusing on signals $\bx \in \mathbb{S}^{n-1}$. The results in Figure~\ref{fig:tau} show that both overly small $\tau$ (e.g., $\tau \le 0.1$ in the unstructured case and $\tau \le 0.2$ in the sparse case) and overly large $\tau$ (e.g., $\tau \ge 2$ in the unstructured case and $\tau \ge 1.1$ in the sparse case) lead to significant performance degradation. At the same time, accurate reconstruction is achieved over a relatively wide range of $\tau$, notably $\tau \in [0.2,1.5]$ for unstructured signals and $\tau \in [0.3,1]$ for sparse signals. We also observe that sparse recovery is generally more sensitive to the choice of $\tau$.

\begin{figure}[ht!]
	\begin{centering}
		\includegraphics[width=0.4\columnwidth]{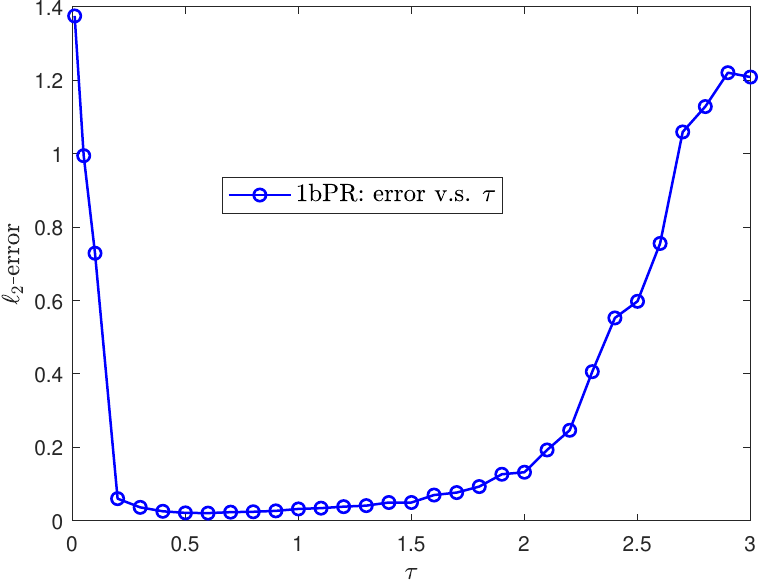} \quad \includegraphics[width=0.4\columnwidth]{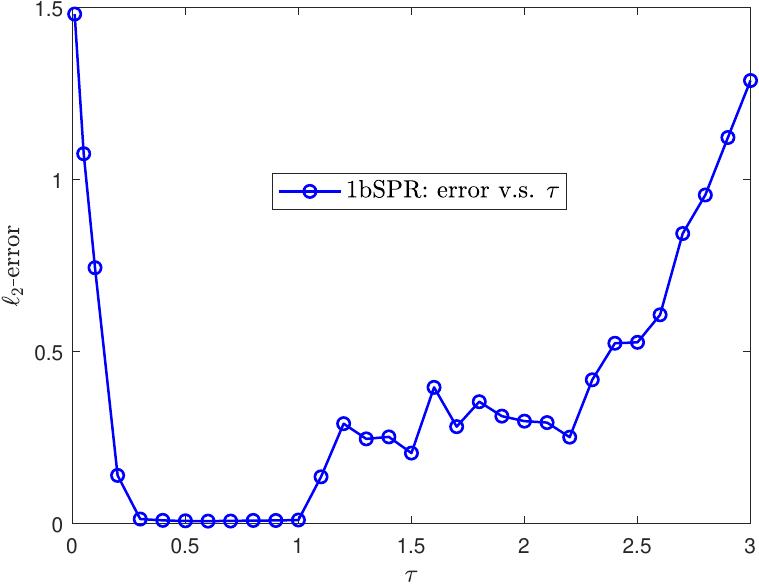}
		\par
	\end{centering}
	
	\caption{\label{fig:tau}\small {\it The impact of $\tau$ on 1bPR (Left) and 1bSPR (Right).} We stick with the experimental settings in Figure \ref{fig:inessential} but fix $m=1500$. We test  $\tau=0.1:0.1:3$ and two additional points $\tau=0.01,0.05$.}
\end{figure}
\paragraph{Recovering Signals  in an Annulus:}
We further demonstrate one's ability to recover signal norm in 1-bit phase retrieval. Our numerical results of recovering   signals in $\mathbbm{A}_{\alpha,\beta}$ with $\beta>\alpha$ are presented in Figure \ref{fig:ratio}. Fixing $\alpha=1$, we observe that increasing $\beta$ leads to larger estimation errors.  This appears natural as a greater number of phaseless hyperplanes are needed to achieve a fine tessellation over a larger signal set $\calK$. More interestingly, {GD-1bPR} maintains the optimal error rate $\calO(m^{-1})$   even when $\beta\gg\alpha$, as observed from $\beta=10\alpha$ and $\beta=15\alpha$ in Figure \ref{fig:ratio}(Left); in contrast,  the performance of   {BIHT-1bSPR} noticeably deteriorates even if $\frac{\beta}{\alpha}$ is only moderately larger than $1$, see $\beta=2.5\alpha$ and $\beta=3.5\alpha$ in  Figure \ref{fig:ratio}(Right).  These results are consistent with our Theorems \ref{thm:pgd_low}--\ref{thm:pgd_high}: compared to {GD-1bPR}, additional   conditions on the ratios $\frac{\alpha}{\tau}$ and $\frac{\beta}{\tau}$ are required to guarantee the near-optimality of {BIHT-1bSPR}. %It needs further investigation whether the ratio conditions in Theorem \ref{thm:pgd_high} is an essential need. %We further argue such necessity in Appendix \ref{app:ratio_con_nece}.  

\begin{figure}[ht!]
	\begin{centering}
		\includegraphics[width=0.4\columnwidth]{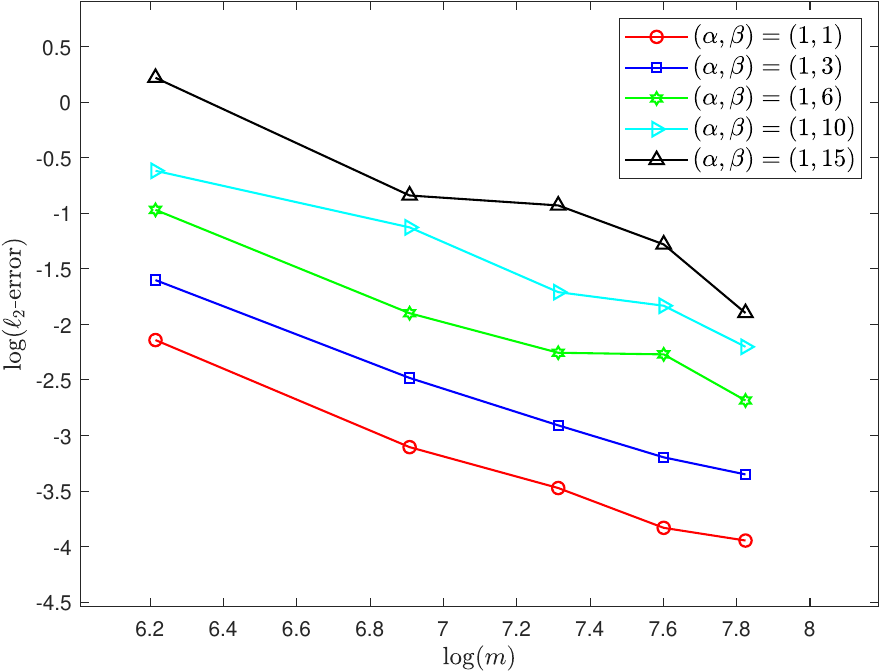} \quad \includegraphics[width=0.4\columnwidth]{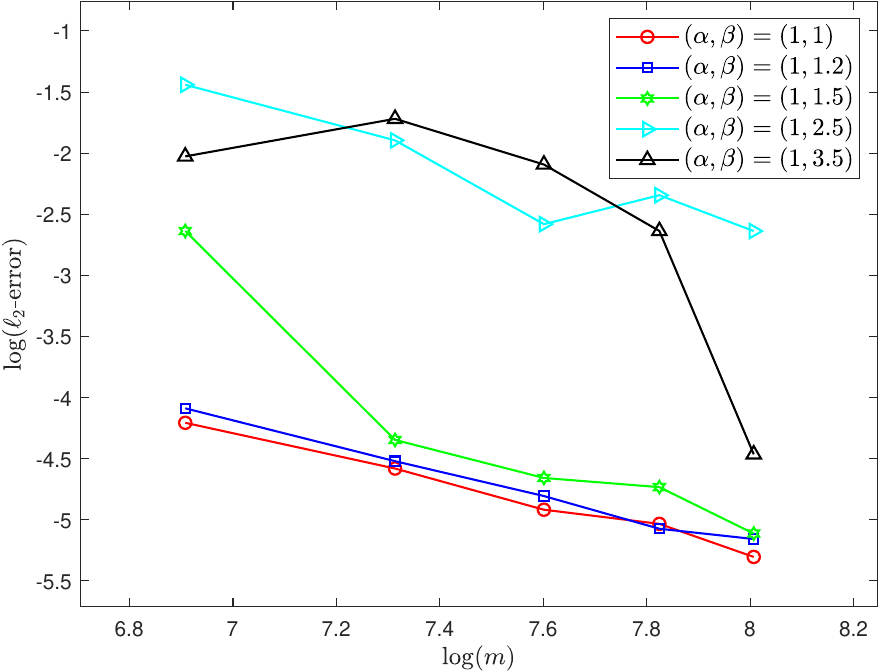}
		\par
	\end{centering}
	
	\caption{\label{fig:ratio}\small {\it Full Signal Reconstruction over $\mathbbm{A}_{\alpha,\beta}$ in 1-bit phase retrieval (Left) and 1-bit sparse phase retrieval (Right).} Each 
		data point is averaged over in $50$ independent trials. 
		In the left figure, we recover $\bx\in \mathbb{R}^{30}$ uniformly distributed over  $\mathbbm{A}_{\alpha,\beta}$ from $m=500:500:2500$ phaseless   bits under $\alpha=1$ and $\beta\in \{1,3,6,10,15\}$. In the right figure, we consider the same setting as in Figure \ref{fig:inessential}(Right) except that the signal $\bx$ is rescaled to $\lambda \bx$ with $\lambda$ uniformly distributed over $[\alpha,\beta]$, under $\alpha=1$ and $\beta\in\{1,1.2,1.5,2.5,3.5\}$. 
  }
\end{figure}

%^\subsection{Testing $\eta= c_\eta\sqrt{\frac{\pi e}{2}}\tau$ with Different $c_\eta$} 
%We present simulation results of the algorithms with $\eta = c_\eta \sqrt{\frac{\pi e}{2}}\tau$ under different $c_\eta$ in Figure \ref{fig:eta}. We run $500$ iterations to ensure the algorithms with $c_\eta<1$ well converge. Since the reconstruction errors under $c_\eta = 1,~0.75,~0.5,~0.1$ are comparable, our results empirically suggest that {GD-1bPR} and {BIHT-1bSPR} are   insensitive to using $c_\eta$ lower than $1$. Nonetheless, we note that  using $c_\eta<1$   takes more iterations to achieve the same error, and that the reconstruction notably worsens when using $c_\eta\ll 1$, as seen from $c_\eta = 0.02$ in Figure \ref{fig:eta}. In contrast,  the performance of both algorithms deteriorates noticeably when $c_\eta$ exceeds $1$. For instance, the curves under $c_\eta = 1.8,~1.6$ are much higher or even lose the optimal decaying rate of $\calO(m^{-1})$. Overall, the recommended 
%$c_\eta = 1$ represents a reasonably good choice.

\paragraph{Recovering Natural Images:} To demonstrate the applicability of our developed theories to  large-scale real-world problems, we move on to testing {GD-1bPR} in the recovery of natural images. We use the ``Stanford Main Quad'' image containing $320\times 1280$ pixels \cite{candes2015phase,chen2017solving} and the ``Milky Way Galaxy'' image containing $1080\times 1920$ pixels \cite{candes2015phase,wang2017solving,zhang2017nonconvex} that are color images with RGB three color bands. 
It is obviously prohibitive in computational and memory burden to use Gaussian sensing ensemble. As in \cite{chen2017solving,zhang2017nonconvex,wang2017solving},  we resort to a type of structured measurement called coded diffraction patterns (CDP) \cite{candes2015cdp}, setting 
\begin{align}\label{eq:cdp}
	\by^{(l)} = \sign(|\bF \bD^{(l)}\bx|-\tau),~~1\le l\le L,
\end{align}
with discrete Fourier transform (DFT) matrix $\bF$ and   diagonal matrices $\bD^{(l)}$ as random masks. Denote the complex unit by $\mathsf{j}$, we let the diagonal entries of the masks $\bD^{(l)}$ be i.i.d. uniformly drawn from $\{1,-1,\mathsf{j},-\mathsf{j}\}$. Note that using $L$ random patterns  generates $m=nL$ 1-bit measurements in total. We execute the sensing and reconstruction   separately for each color band, under the quantization threshold $\tau = \frac{1}{3}\cdot (\|\bx^{(r)}\|_2+\|\bx^{(g)}\|_2+\|\bx^{(b)}\|_2)$ where $\bx^{(r)},\bx^{(g)},\bx^{(b)}$ denote the (vectorized) RGB color bands.
Let us denote the $nL$ phaseless bits from CDP by 
$\{y_i = \sign(|\ba_i^*\bx|-\tau)\}_{i=1}^{nL}$ 
 for certain complex sensing vectors $\{\ba_i\in \mathbb{C}^n:i=1,\cdots,nL\}$.
By Wirtinger calculus \cite{candes2015phase,kreutz2009complex}, the (sub-)gradient of the one-sided $\ell_1$-loss $\calL(\bu)=\frac{1}{2nL}\sum_{i=1}^{nL}[||\ba_i^*\bu|-\tau|-y_i(|\ba_i^*\bu|-\tau)]$ reads
$\partial \calL(\bu) =\frac{1}{4nL}\sum_{i=1}^{nL} (\sign(|\ba_i^*\bu|-\tau)-y_i) \sign(\ba_i^* \bz)\ba_i.$ We execute $\bx^{(t)}=\bx^{(t-1)}-\eta\cdot \partial \calL(\bx^{(t-1)})$  with step size $\eta=\sqrt{2\pi e}\cdot\tau$, which can be regarded as the complex Wirtinger flow counterpart of {GD-1bPR}. As before, these gradient descent refinements are seeded with  spectral method, in which the leading eigenvector of $\hat{\bS}_{\bx}=\frac{1}{m}\sum_{i=1}^m y_i\ba_i\ba_i^*$ is approximately found by power method.  The images are displayed in Figure \ref{fig:stanford} for ``Stanford Main Quad'' and Figure  \ref{fig:galaxy} for ``Milky Way Galaxy''.  We achieve fairly accurate reconstruction from only $L=32$ or $L=64$ random patterns. Remarkably, the recovered ``Milky Way Galaxy'' using $L=64$ patterns and the original image are almost indistinguishable to our eyes, as seen by comparing (a) and (e) in Figure \ref{fig:galaxy}. These results are highly impressive because, compared to the parallel experiments  conducted in prior works that used {\it full-precision} phaseless measurements from $L=20$ \cite{candes2015phase}, $L=12$ \cite{chen2017solving,zhang2017nonconvex}, $L=6$ \cite{zhang2017nonconvex,wang2017solving} random patterns, our reasonably accurate reconstruction is achieved from a much lower number of phaseless bits.\footnote{Theoretically, it is not immediately clear how one can acquire an exact full-precision measurement via finite bits; while experimentally, each measurement in these works is stored as a double-precision floating point number (in Matlab) and requires $64$ bits of storage, and thus using $L=32,~64$ in our 1-bit setting amounts to using $L=0.5,~1$ in theirs in terms of the number of bits in measurements!} 
Furthermore, while our near-optimal guarantee for {GD-1bPR} is proved under Gaussian measurements, it is surprising that the relative error of recovering natural images from the 1-bit CDP measurements also  precisely follows the optimal rate; see Figure \ref{fig:cdprate}. This seems to reflect a type of  universality. %on the sensing vectors. 

\begin{figure}[ht!]
	\centering
	\begin{tabular}{c}
		\includegraphics[width=0.78\textwidth]{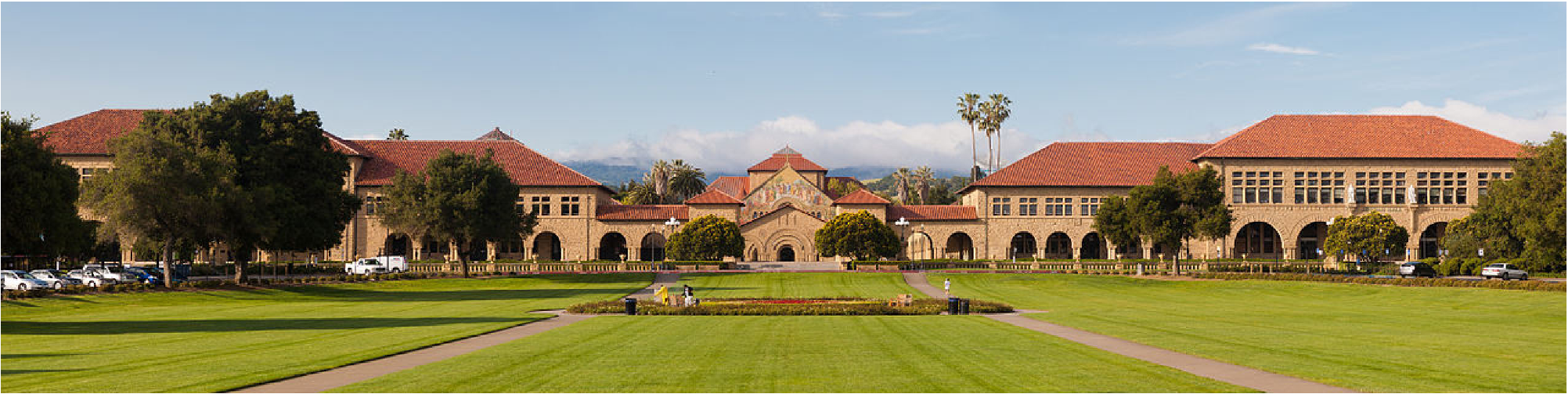} \tabularnewline
		{\footnotesize (a) Original image: Stanford Main Quad.} 
		\tabularnewline
		\includegraphics[width=0.78\textwidth]{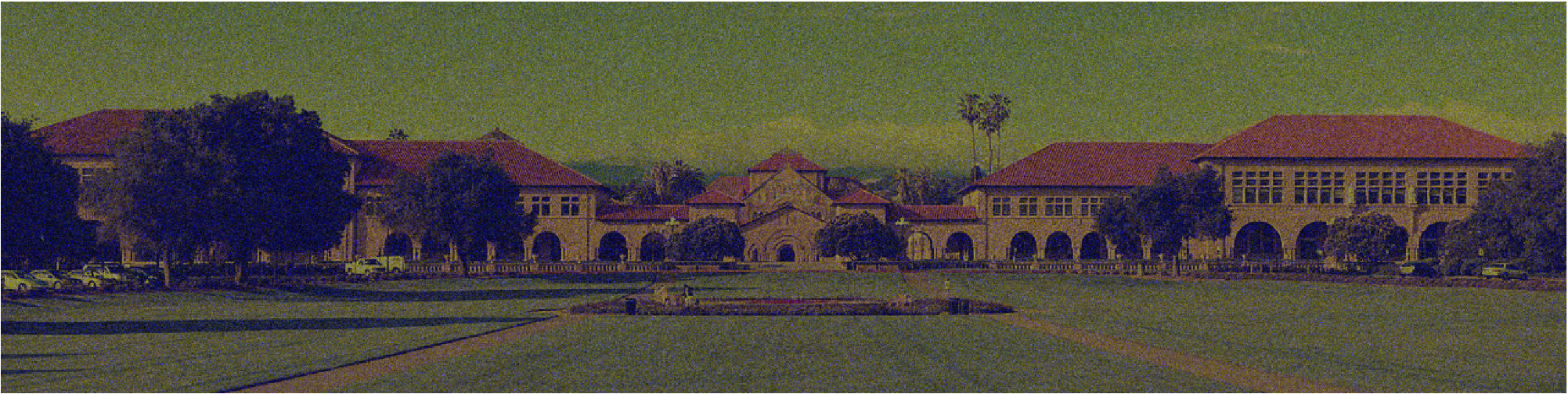} \tabularnewline
		{\footnotesize (b) Recovered image after {SI-1bPR} ($L=32$): relative error = 0.478, PSNR = 11.15.}   
		\tabularnewline
		\includegraphics[width=0.78\textwidth]{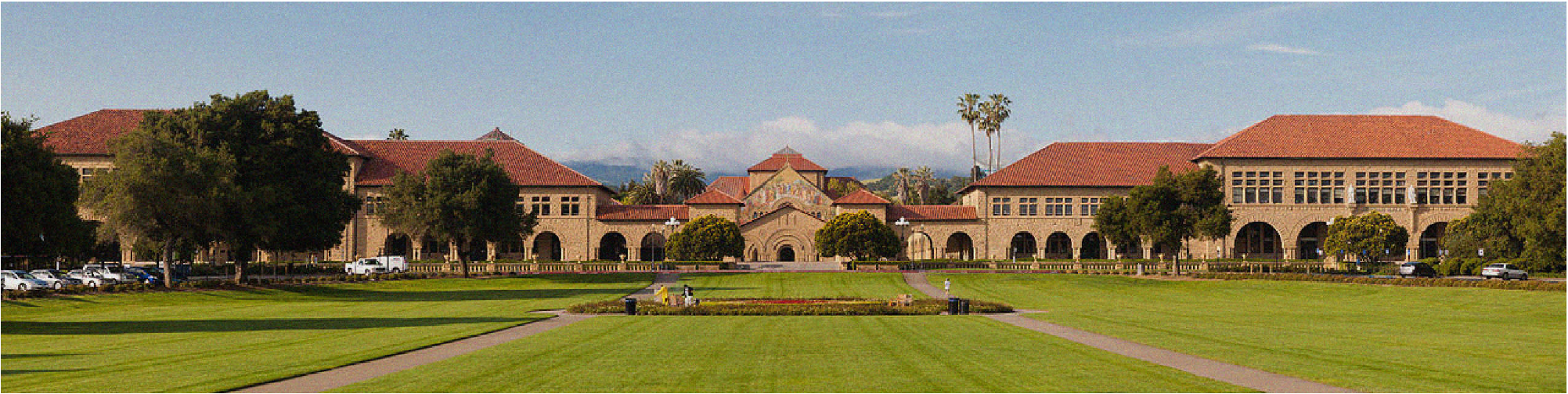} \tabularnewline
		{\footnotesize(c) Recovered image after {GD-1bPR} ($L=32$): relative error = 0.052, PSNR = 30.42.} 
		\tabularnewline
		\includegraphics[width=0.78\textwidth]{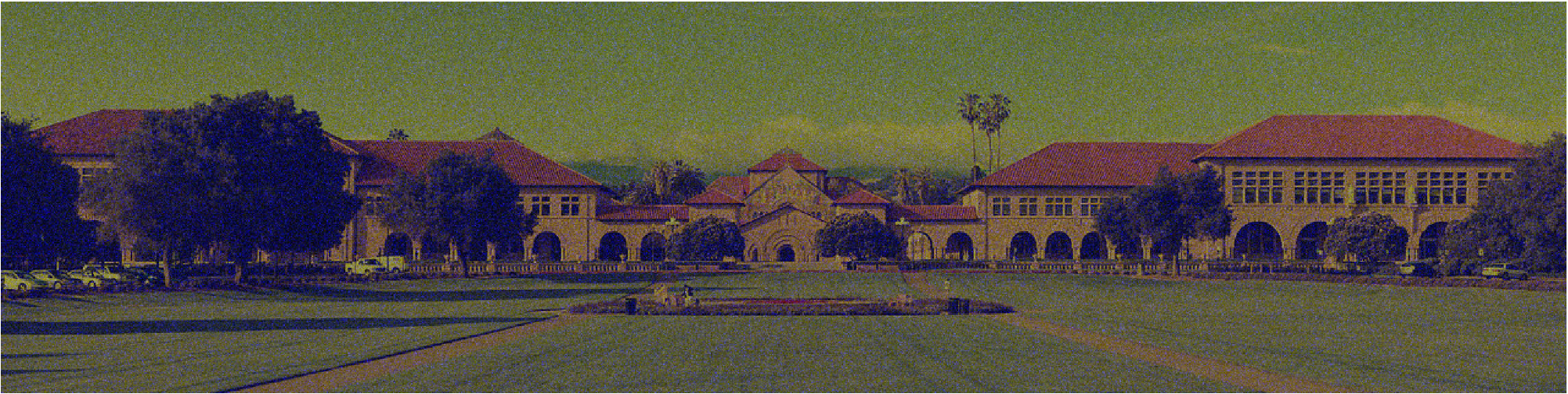} 
		\tabularnewline
		{\footnotesize(d) Recovered image after {SI-1bPR} ($L=64$):  relative error = 0.469, PSNR = 11.32.} 
		\tabularnewline
		\includegraphics[width=0.78\textwidth]{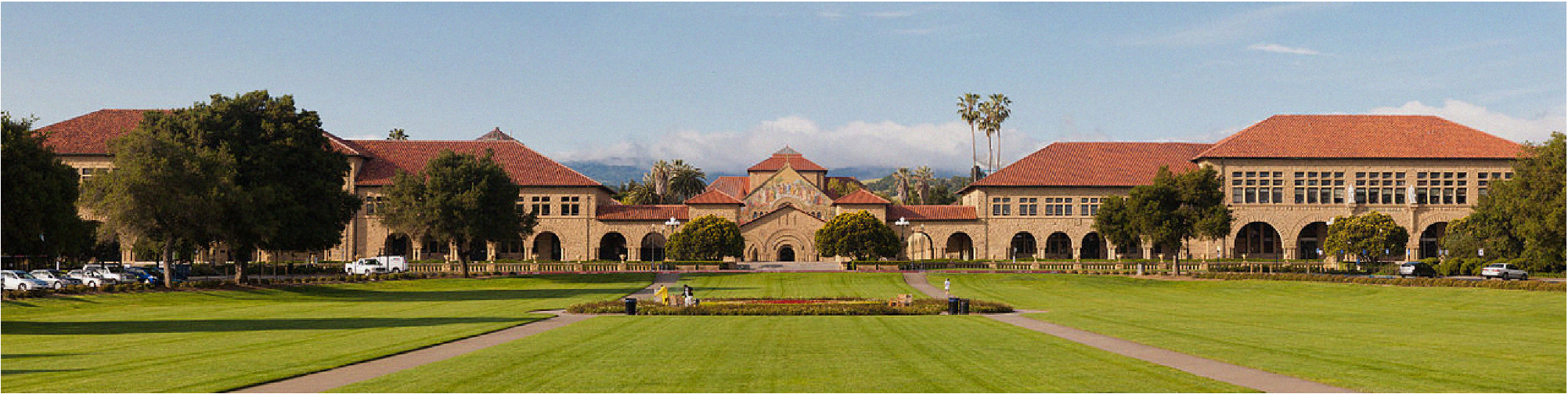} 
		\tabularnewline
		{\footnotesize(e) Recovered image after {GD-1bPR} ($L=64$):  relative error = 0.026, PSNR = 36.42.} 
	\end{tabular}
	\par
	
	\caption{ \small {\it Recovering the $320\times 1280\times 3$ Stanford Main Quad image from phaseless bits produced by CDP with $L=32,~64$ random patterns.} We run $50$ power method iterations for {SI-1bPR} and $100$ iterations for {GD-1bPR}, which involves $32\times 2 \times (50+100) = 9600,~64\times 2 \times (50+100) =19200$ FFTs (see \cite{candes2015phase,chen2017solving}) for each color band and completes in   a few minutes.    
		\label{fig:stanford}}
\end{figure}
\begin{figure}[ht!]
	\centering
	\begin{tabular}{c}
		\includegraphics[width=0.69\textwidth]{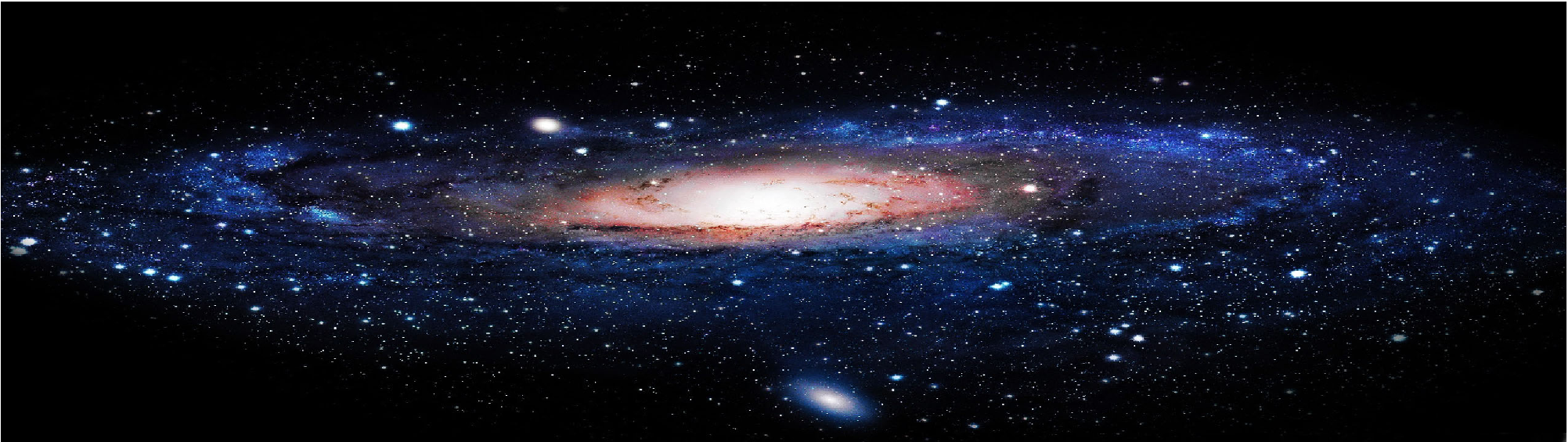} \tabularnewline
		{\footnotesize (a) Original image: Milky Way Galaxy.} 
		\tabularnewline
		\includegraphics[width=0.69\textwidth]{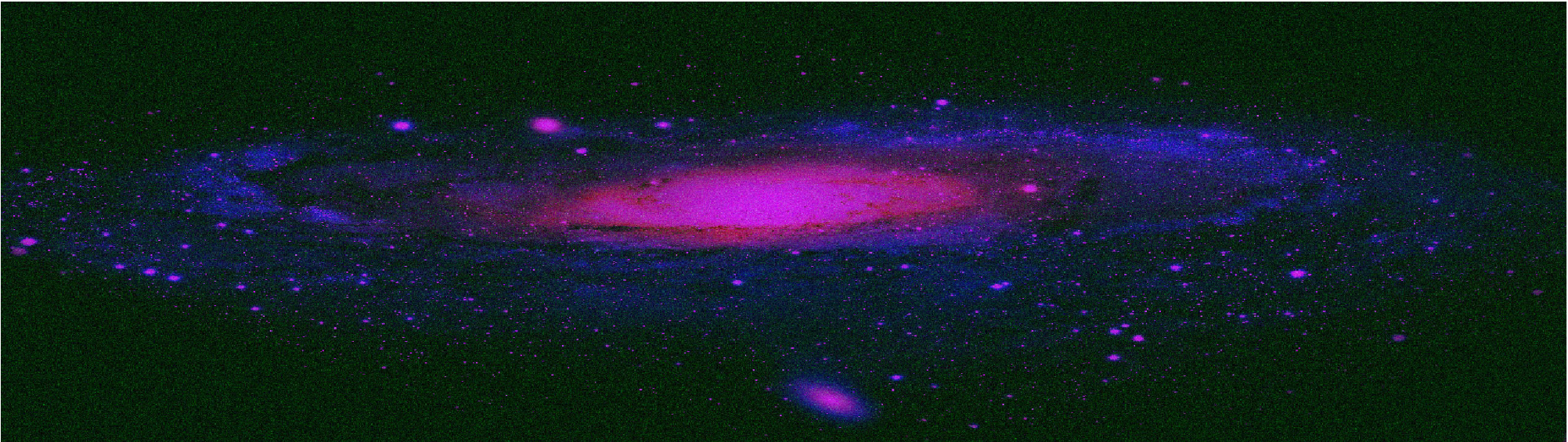} \tabularnewline
		{\footnotesize (b) Recovered image after \text{SI-1bPR} ($L=32$):  relative error = 0.652, PSNR = 17.49.}   
		\tabularnewline
		\includegraphics[width=0.69\textwidth]{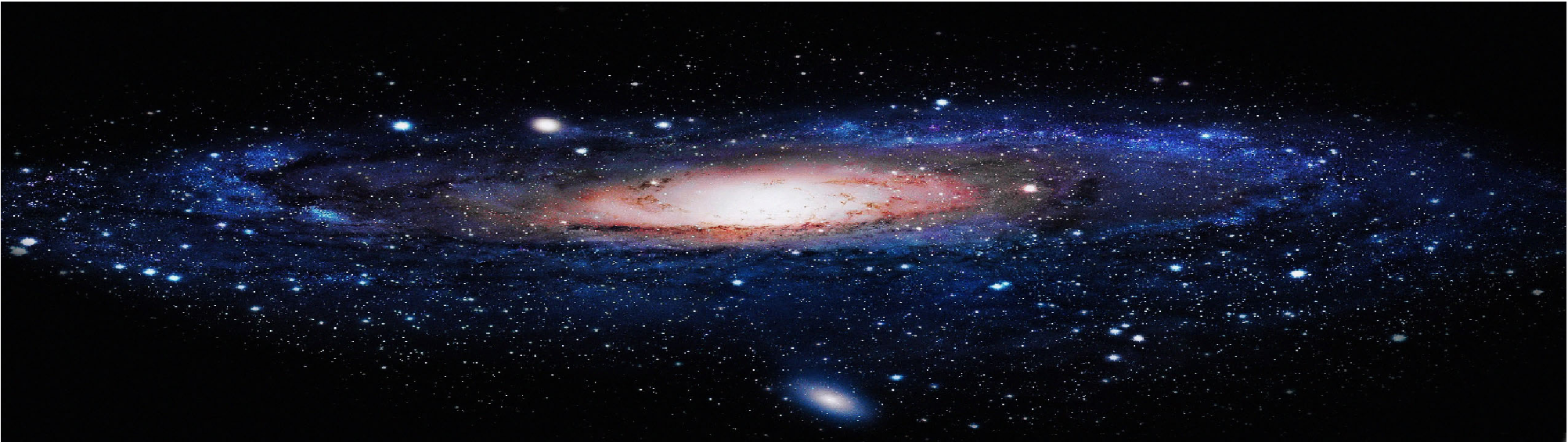} \tabularnewline
		{\footnotesize(c) Recovered image after {GD-1bPR} ($L=32$):  relative error = 0.058, PSNR = 38.52.} 
		\tabularnewline
		\includegraphics[width=0.69\textwidth]{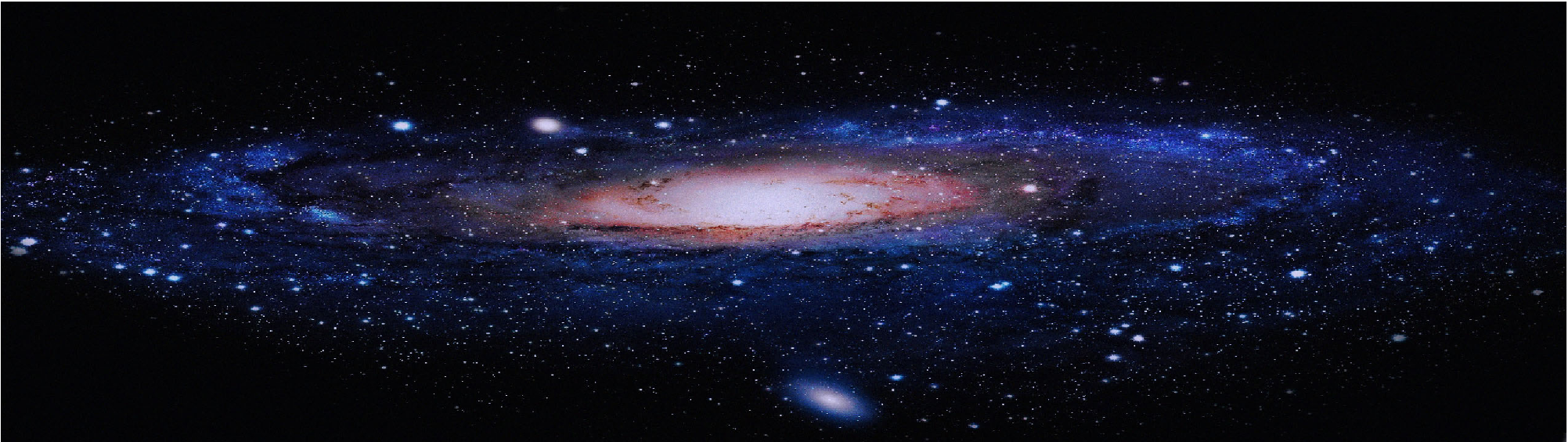} 
		\tabularnewline
		{\footnotesize(d) Recovered image after {SI-1bPR} ($L=64$):  relative error = 0.270, PSNR = 25.14.} 
		\tabularnewline
		\includegraphics[width=0.69\textwidth]{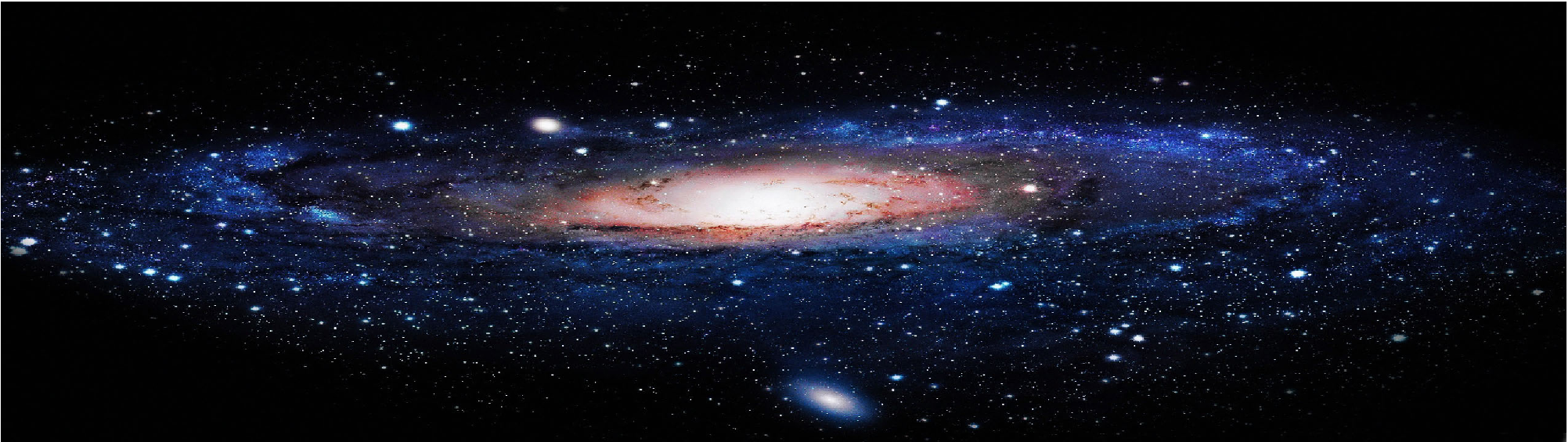} 
		\tabularnewline
		{\footnotesize(e) Recovered image after {GD-1bPR} ($L=64$):  relative error = 0.029, PSNR = 44.65.} 
	\end{tabular}
	
	\caption{ \small {\it Recovering the $1080\times 1980\times 3$ Milky Way Galaxy image from phaseless bits produced by CDP with $L=32,~64$ random patterns.} We run $50$ power method iterations for {SI-1bPR} and $100$ iterations for {GD-1bPR}, which only takes   a matter of minutes to complete. The images are subsampled to $540\times 1980\times 3$ for display.     
		\label{fig:galaxy}}
\end{figure}

\begin{figure}[h]
	\centering
	\includegraphics[width=0.35\textwidth]{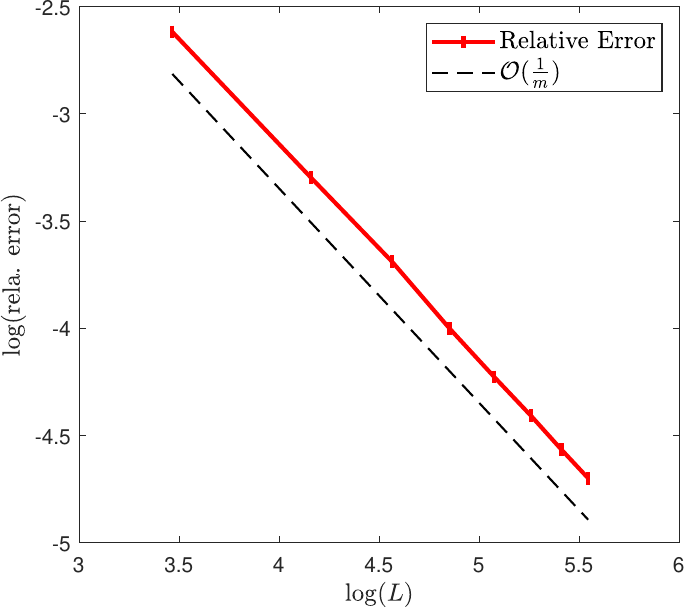}
	\caption{\small The relative error  of recovering the ``Stanford Main Quad'' from the 1-bit CDP measurements (\ref{eq:cdp}) decays with the optimal rate $\calO(m^{-1})$. We subsample the image to $160\times 640\times 3$ for experiments and test $L=32:32:256$.    
		\label{fig:cdprate}}
\end{figure}

\section{Concluding Remarks}\label{sec:conclusion} 
We studied the reconstruction of a signal $\bx\in \mathbb{R}$ from $m$ phaseless bits $\{y_i=\sign(|\ba_i^\top\bx|-\tau)\}_{i=1}^m$, referred to as the 1-bit phase retrieval problem. By developing a theory for Gaussian  phaseless hyperplane tessellation, we showed that generally structured signals can be uniformly and accurately recovered by constrained hamming distance minimization. 
This result specializes to the upper bounds $\calO(\frac{n}{m}\log(\frac{m}{n}))$ for recovering unstructured signals and $\calO(\frac{k}{m}\log(\frac{mn}{k^2}))$ for recovering $k$-sparse signals, which are tight up to logarithmic factors. 
We also investigated the computational aspect  by developing efficient near-optimal algorithms {SI-1bPR} and {GD-1bPR} for 1bPR of unstructured signals, and {SI-1bSPR} and {BIHT-1bSPR} for 1-bit sparse phase retrieval. We further demonstrated the practical merits of these algorithms through a set of numerical experiments.

%{\color{blue}[add robustness; and also state this earlier; star-shaped set?]}  
While we analyzed the efficient algorithms in a noiseless case for succinctness,  it should be noted that the robustness to adversarial bit flips can be obtained by slightly more work. Suppose we observe $\hat{\by}=\by+\be$ obeying $d_H(\hat{\by},\by=\sign(|\bA\bx|-\tau))=\|\be\|_0\le \zeta m$, then the corrupted gradient takes the form $\hat{\bh}(\bx^{(t-1)},\bx)$ where $\hat{\bh}(\bu,\bv)=\frac{1}{2m}\sum_{i=1}^m [\sign(|\ba_i^\top\bu|-\tau)- \sign(|\ba_i^\top\bv|-\tau)- e_i]\sign(\ba_i^\top\bu)\ba_i$. On the other hand, a corrupted version of the PLL-AIC that controls $\|\calP_{\calC_-}(\bu-\bv- \eta\cdot\hat{\bh}(\bu,\bv))\|_2$ similarly leads to convergence of Algorithms \ref{alg:pgd}--\ref{alg:pgd_high}. Since by Lemma \ref{lem:pro_closec} we have $\|\calP_{\calC_-}(\bu-\bv-\eta\cdot\hat{\bh}(\bu,\bv))\|_2\le\|\calP_{\calC_-}(\bu-\bv-\eta\cdot\bh(\bu,\bv))\|_2 + \|\calP_{\calC_-}(\frac{1}{2m}\sum_{i=1}^m e_i\sign(\ba_i^\top\bu)\ba_i)\|_2$, the corruption increments the reconstruction error by $\|\calP_{\calC_-}(\frac{1}{2m}\sum_{i=1}^m e_i\sign(\ba_i^\top\bu)\ba_i)\|_2$; this can be bounded as $\calO(\sqrt{\frac{\zeta \omega^2(\calC_{(1)})}{m}}+\zeta\sqrt{\log(\zeta^{-1})})$ by Lemma \ref{lem:max_ell_sum}, thus the final reconstruction error is incremented by $\calO(\zeta\sqrt{\log(\zeta^{-1})})$. Similarly, one can also show the spectral methods are robust to adversarial bit flips.

Our work is among the first rigorous studies of 1-bit phase retrieval, and points to many interesting open questions. First, can these results be extended to complex-valued measurement vectors and signals? We believe this extension is possible but also needs new techniques. Second, our practical experience seems to suggest that {GD-1bPR} with {\it random initialization} achieves the optimal rates. It would be of great interest to further explore this phenomenon by investigating the optimization landscape \cite{sun2018geometric} or the iteration dynamics \cite{chen2019gradient}. Also, our simulation results suggest that (for certain parameter regimes) 1-bit (sparse) phase retrieval can even outperform 1-bit (compressed) sensing. Explaining this empirical phenomenon would require a more precise error analysis of the algorithms, which we leave for future work (see, e.g., \cite{chandrasekher2023sharp}). Furthermore,  in what way can we go beyond the Gaussian design to preserve the signal recoverability or even the computational feasibility? What quantizer should be used if collecting multiple bits from each measurement is allowed? Can we potentially gain some privileges by replacing the {\it fixed} threshold $\tau$ with {\it random} dithers? We shall leave these tantalizing questions, among many others, for the future. %{\color{blue}[TODO: Sharp analysis]}
 
\bibliographystyle{siamplain}
\bibliography{libr}

%% enabling this for separating
%% ****************************
%\end{document}
%% ****************************

\newpage  

{\centering \huge \bf Appendix \par}
\vspace{3mm} 
\paragraph{Outline:} We prove our information-theoretic results in Appendix \ref{app:proof_embed}; we prove that  PLL-AIC leads to local convergence of GD-1bPR and BIHT-1bSPR in Appendix \ref{app:aic2conver}; we establish PLL-AIC for Gaussian matrix in Appendix \ref{app:prove_pgd}; the technical lemmas which support our proofs are collected in Appendix \ref{app:tools}; the proofs with standard arguments that are not regarded as our main technical contributions
are relegated to Appendix \ref{app:standard}. 

\begin{appendix}
\section{Proofs for Information-Theoretic Bounds}\label{app:proof_embed}
This appendix includes the proofs of   Theorem \ref{thm:local_embed}, Lemma \ref{lem:Puv} and Theorem \ref{thm:lower}. 

\subsection{The Proof of Theorem \ref{thm:local_embed} (Local Binary Phaseless Embedding)}\label{app:prove_thm1}
We first provide some preliminaries. 
For fixed $\bu,\bv$, the performance of $d_H(\sign(|\bA\bu|-\tau),\sign(|\bA\bv|-\tau))\sim \text{Bin}(m,\sfP_{\bu,\bv})$ is precisely characterized by Chernoff bound. Since the high-probability events in Theorem \ref{thm:local_embed} hold universally for all pairs of $(\bu,\bv)$ obeying certain conditions, so we  invoke a covering argument to pursue the uniformity. The major challenge arises from the discontinuity of hamming distance, which makes it hard to assert the closeness between $d_H(\sign(|\bA\bu|-\tau),\sign(|\bA\bv|-\tau))$ and $d_H(\sign(|\bA\bp|-\tau),\sign(|\bA\bq|-\tau))$  even though   $\|\bu-\bp\|_2+\|\bv-\bq\|_2$ is fairly small. To overcome the difficulty, we  soften the phaseless separation as follows. 
\begin{definition}
	[$\theta$-Well-Separation] \label{def:well_separate} Given $\theta> 0$, we say a phaseless hyperplane $\calH_{|\ba|}$ $\theta$-well-separates $\bu,\bv\in \mathbb{R}^n$ if the following hold: 
 \begin{itemize}
		[leftmargin=2ex,topsep=0.25ex]
		\setlength\itemsep{-0.1em}
		\item $\calH_{|\ba|}$ separates $\bu,\bv$, i.e., $\sign(|\ba^\top\bu|-\tau)\ne\sign(|\ba^\top\bv|-\tau)$; 
		\item $|\ba^\top\bu|$ and $|\ba^\top\bv|$ are bounded away from the quantization threshold $\tau$ in that \begin{align}\label{eq:bound_away_tau}
		    ||\ba^\top\bu|-\tau|\ge\theta \cdot \dist(\bu,\bv)~~\text{and}~~||\ba^\top\bv|-\tau|\ge\theta \cdot \dist(\bu,\bv)
		\end{align} hold. 
	\end{itemize} 
\end{definition}
With the additional second dot point, $\theta$-well-separation is stable in the following sense. 

\begin{lem}
	[Stability of $\theta$-Well-Separation] \label{lem:sepa_im_sepa}
	Given a phaseless hyperplane $\calH_{|\ba|}$, we have the following two statements:
	\begin{itemize}
		[leftmargin=2ex,topsep=0.25ex]
		\setlength\itemsep{-0.1em}
		\item If $||\ba^\top\bu|-\tau|>2\min\big\{|\ba^\top(\bp-\bu)|,|\ba^\top(\bp+\bu)|\big\}$ holds,
		then $\bp$ and $\bu$ are not separated by  $\calH_{|\ba|}$. 
		
		\item As a consequence, if $\calH_{|\ba|}$ $\theta$-well-separate $\bu$ and $\bv$, and for some $\bp,\bq$ we have $\min\big\{|\ba^\top(\bp-\bu)|,|\ba^\top(\bp+\bu)|\big\} < \frac{\theta}{2}\cdot \dist(\bu,\bv)$ and $\min\big\{|\ba^\top(\bq-\bv)|,|\ba^\top(\bq+\bv)|\big\}<\frac{\theta}{2}\cdot \dist(\bu,\bv),$ 
		then $\calH_{|\ba|}$ separates $\bp$ and $\bq$.
	\end{itemize}
\end{lem}
\begin{proof}
	The proof can be found in Appendix \ref{appother}.
\end{proof}

Note that $\theta$-well-separation is   stricter   than separation due to the additional conditions in (\ref{eq:bound_away_tau}). However, the probability of $\theta$-well-separation is order-wise the same as $\sfP_{\bu,\bv}$ if $\theta$ is sufficiently small.
%Note that $\theta$-well-separation is obviously more  strict than separation. Fortunately, the probability of $\theta$-well-separation is order-wise the same as $\sfP_{\bu,\bv}$ 
%given that  $\theta$ is small enough. 
\begin{lem}\label{lem:Pthetauv}
	Given $\theta>0$, $\bu,\bv\in\mathbb{R}^n$ and $\ba\sim\calN(0,\bI_n)$, we let $\sfP_{\theta,\bu,\bv}=\mathbbm{P}(\calH_{|\ba|}~\theta\text{-well-separates}~\bu,\bv).$
	Then there exist some $c_1,c_2,C_3>0$ depending on $(\alpha,\beta,\tau)$, such that if $0\le\theta<c_1$, it holds that 
	\begin{align*}
		c_2 \dist(\bu,\bv)\le \sfP_{\theta,\bu,\bv} \le C_3\dist(\bu,\bv),~\forall \bu,\bv\in \mathbbm{A}_{\alpha,\beta}.
	\end{align*}
\end{lem}
\begin{proof}
The proof can be found in Appendix \ref{appother}. 
\end{proof}

\begin{proof}[Proof of Theorem \ref{thm:local_embed}] 
Given $r>0$ and let $r'=\frac{cr}{\log^{1/2}(r^{-1})}$ for some small enough $c$, we construct $\calN_{r'}$ as a minimal $r'$-net of $\calK$ with cardinality satisfying $\log|\calN_{r'}|=\scrH(\calK,r')$. For any $\bu$ and $\bv$ in $\calK$, we find $\bp$ and $\bq$ as the points in $\calN_{r'}$ being closest to $(\bu,\bv)$:  
	\begin{gather} \label{eq:foundpq}
		\bp = \text{arg}\min_{\bw\in \calN_{r'}}~\|\bw-\bu\|_2~~\text{and}~~
		\bq=\text{arg}\min_{\bw\in\calN_{r'}}~\|\bw-\bv\|_2.
	\end{gather} 
Note that $\|\bu-\bp\|_2\le r'$ and $\|\bv-\bq\|_2\le r'$ hold.

\paragraph{Deterministic arguments for establishing $E_s$:}   We work with $(\bu,\bv,\bp)$ satisfying $\|\bu-\bp\|_2\le r'$ and $\dist(\bv,\bp)\le \dist(\bv,\bu)+\dist(\bu,\bp)\le \frac{3r'}{2}$, then we proceed as
\begin{subequations}
	\begin{align}\nn
		& m-d_H(\sign(|\bA\bu|-\tau),\sign(|\bA\bp|-\tau)) = \sum_{i=1}^m\mathbbm{1}\big(\calH_{|\ba_i|}\text{ does not separate }\bu,\bp\big)\\\label{eq:imply_not_sepa}
		& \ge \sum_{i=1}^m \mathbbm{1}\Big(||\ba_i^\top \bp|-\tau|>r,~|\ba_i^\top(\bu-\bp)|\le \frac{r}{2} \Big) \\
		&\ge \sum_{i=1}^m \mathbbm{1}(||\ba_i^\top\bp|-\tau|>r) -\sum_{i=1}^m \mathbbm{1}\Big(|\ba_i^\top(\bu-\bp)|>\frac{r}{2}\Big); \label{eq:union3}
	\end{align}\label{eq:bound_u_p}
\end{subequations}
and 
\begin{subequations}
	\begin{align}\nn
		&m-d_H(\sign(|\bA\bv|-\tau),\sign(|\bA\bp|-\tau))  = \sum_{i=1}^m\mathbbm{1}\big(\calH_{|\ba_i|}\text{ does not separate }\bv,\bp\big) \\\label{eq:imply_not_sepa2}
		&\ge  \sum_{i=1}^m \mathbbm{1}\Big(||\ba_i^\top \bp|-\tau|>r,~\min\{|\ba_i^\top(\bv-\bp)|,|\ba_i^\top(\bv+\bp)|\}<\frac{r}{2}\Big)\\
		&\ge \sum_{i=1}^m \mathbbm{1}(||\ba_i^\top\bp|-\tau|>r) -\sum_{i=1}^m \mathbbm{1}\Big(\min\{|\ba_i^\top(\bv-\bp)|,|\ba_i^\top(\bv+\bp)|\}>\frac{r}{2}\Big),\label{eq:union4}
	\end{align}\label{eq:bound_v_p}
\end{subequations}
where (\ref{eq:imply_not_sepa}) and (\ref{eq:imply_not_sepa2}) hold   because $||\ba_i^\top \bp|-\tau|>r$ and $\min\{|\ba_i^\top(\bw-\bp)|,|\ba_i^\top(\bw+\bp)|\}\le\frac{r}{2}$   imply $||\ba_i^\top\bp|-\tau|>2\min\{|\ba_i^\top(\bw-\bp)|,|\ba_i^\top(\bw+\bp)|\}$, and hence further imply $\calH_{|\ba_i|}$ does not separate $\bp$ and $\bw$ due to the first statement in Lemma \ref{lem:sepa_im_sepa}. In  (\ref{eq:union3}) and (\ref{eq:union4})   we use union bound. 
Now we are able to deduce 
\begin{subequations}\label{eq:dete_upper_bound}
	\begin{align}\nn
		&d_H(\sign(|\bA\bu|-\tau),\sign(|\bA\bv|-\tau))
		\\\label{eq:use_tri_3}
		&\le d_H(\sign(|\bA\bu|-\tau),\sign(|\bA\bp|-\tau))+d_H(\sign(|\bA\bv|-\tau),\sign(|\bA\bp|-\tau))\\\nn
		&\le m- \sum_{i=1}^m \mathbbm{1}(||\ba_i^\top\bp|-\tau|>r) +\sum_{i=1}^m \mathbbm{1}\Big(|\ba_i^\top(\bu-\bp)|>\frac{r}{2}\Big)\\\label{eq:substitute1}
		&\quad + m -\sum_{i=1}^m \mathbbm{1}(||\ba_i^\top\bp|-\tau|>r) +\sum_{i=1}^m \mathbbm{1}\Big(\min\{|\ba_i^\top(\bv-\bp)|,|\ba_i^\top(\bv+\bp)|\}>\frac{r}{2}\Big)\\\nn
		& = 2 \Big(m-\sum_{i=1}^m \mathbbm{1}(||\ba_i^\top\bp|-\tau|> r)\Big)
		\\\nn&\quad +\sum_{i=1}^m \mathbbm{1}\Big(|\ba_i^\top(\bu-\bp)|>\frac{r}{2}\Big) + \sum_{i=1}^m \mathbbm{1}\Big(\min\{|\ba_i^\top(\bv-\bp)|,|\ba_i^\top(\bv+\bp)|\}>\frac{r}{2}\Big)\\
		&\le 2\underbrace{\sup_{\bp\in\calN_{r'}}\sum_{i=1}^m \mathbbm{1}(||\ba_i^\top\bp|-\tau|\le r)}_{:=\breve{T}_1} + 2 \underbrace{\sup_{\bw\in \calK_{(3r'/2)}}\sum_{i=1}^m \mathbbm{1}\Big(|\ba_i^\top\bw|>\frac{r}{2}\Big)}_{:=\breve{T}_2},\label{eq:sup_out1}
	\end{align}
\end{subequations}
where  (\ref{eq:use_tri_3}) is due to triangle inequality, in (\ref{eq:substitute1}) we substitute (\ref{eq:bound_u_p}) and (\ref{eq:bound_v_p}), in (\ref{eq:sup_out1}) we take the supremum over $\bp\in\calN_{r'}$ for $\breve{T}_1$, over 
$\bu-\bp\in \calK_{(3r'/2)}$ and  $\bv-\bp~\text{or }\bv+\bp\in \calK_{(3r'/2)}$ for $\breve{T}_2$ (note that  $\dist(\bv,\bp)\le\frac{3r'}{2}$ and $\calK$ is symmetric).

\paragraph{Deterministic arguments for establishing $E_l$:}    For given $\bu,\bv$ satisfying $\dist(\bu,\bv)\ge 2r$ the corresponding $\bp,~\bq\in \calN_{r'}$ found in (\ref{eq:foundpq}) satisfy \begin{align}\label{eq:tri_angle_important}
	\dist(\bp,\bq)\ge \dist(\bu,\bv)-\|\bu-\bp\|_2-\|\bv-\bq\|_2\ge \frac{1}{2}\dist(\bu,\bv)\ge r,
\end{align} where the last inequality holds because $\frac{1}{2}\dist(\bu,\bv)-\|\bu-\bp\|_2-\|\bv-\bq\|_2\ge r-2r'\ge 0$. By Lemma \ref{lem:Pthetauv}, we can pick a small enough   constant $\theta_0$ such that for some   constant $c_0$, it holds that 
\begin{align}\label{eq:choosetheta0}
	\sfP_{\theta_0,\bw,\bz}\ge c_0\dist(\bw,\bz),\quad\forall \bw,\bz\in \mathbbm{A}_{\alpha,\beta}.
\end{align} %holding for some absolute constant $c_0>0$ and for all $\bw,\bz\in \mathbbm{A}_{\alpha,\beta}$. 
Then we proceed as
\begin{subequations}\label{eq:dete_lower_bound}
	\begin{align}\nn
		&d_H(\sign(|\bA\bu|-\tau),\sign(|\bA\bv|-\tau))\\
		& \ge \sum_{i=1}^m \mathbbm{1}\Big(\calH_{|\ba_i|}~\theta_0\text{-well-separates }\bp\text{ and }\bq,~\label{eq:use_well_sepa} |\ba_i^\top(\bu-\bp)|<\frac{\theta_0r}{2},~|\ba_i^\top(\bv-\bq)|<\frac{\theta_0r}{2}\Big)\\\nn
		&\ge \sum_{i=1}^m\mathbbm{1}(\calH_{|\ba_i|}~\theta_0\text{-well-separates }\bp\text{ and }\bq)  -\sum_{i=1}^m \mathbbm{1}\Big(|\ba_i^\top(\bu-\bp)|\ge \frac{\theta_0r}{2}\Big) -\sum_{i=1}^m\mathbbm{1}\Big(|\ba_i^\top(\bv-\bq)|\ge \frac{\theta_0r}{2}\Big)\\
		&\ge  \underbrace{\sum_{i=1}^m\mathbbm{1}(\calH_{|\ba_i|}~\theta_0\text{-well-separates }\bp\text{ and }\bq)}_{:=\breve{T}_3}-2\underbrace{\sup_{\bw\in\calK_{(r')}}\sum_{i=1}^m \mathbbm{1}\Big(|\ba_i^\top\bw|\ge \frac{\theta_0r}{2}\Big)}_{:=\breve{T}_4},\label{eq:sup_out_3}
	\end{align}
\end{subequations} 
where (\ref{eq:use_well_sepa}) holds due to Lemma \ref{lem:sepa_im_sepa}, since $\calH_{|\ba_i|}$ $\theta_0$-well-separates $\bp,\bq$ and $
	\max\big\{|\ba_i^\top(\bu-\bp)|,|\ba_i^\top(\bv-\bq)|\big\} < \frac{\theta_0r}{2}\le\frac{\theta_0\dist(\bp,\bq)}{2}$ jointly imply $\bu,\bv$ being separated by $\calH_{|\ba_i|}$. Note that  in (\ref{eq:sup_out_3}), we take the supremum over $\bu-\bp\in \calK_{(r')}$ and $\bv-\bq\in \calK_{(r')}$.

With the above deterministic arguments, all that remains is to bound $\breve{T}_1$ and $\breve{T}_2$ in (\ref{eq:sup_out1}), and $\breve{T}_3$ and $\breve{T}_4$ in (\ref{eq:sup_out_3}): we bound $\breve{T}_1$ and $\breve{T}_3$ by Chernoff bound along with a union bound over $\bp\in \calN_{r'}$ and $\{(\bp,\bq)\in \calN_{r'}\times \calN_{r'}:\dist(\bp,\bq)\ge r\}$, then control $\breve{T}_2$ and $\breve{T}_4$ by Lemma \ref{lem:max_ell_sum}.

%, and bound $T_2$ by Lemma \ref{lem:max_ell_sum}.
%   In the remaining work, we lower bound $T_3$ by Chernoff bound for a fixed pair of $(\bp,\bq)$ and then a union bound over $\{(\bp,\bq)\in \calN_{r'}\times \calN_{r'}:\dist(\bp,\bq)\ge r\}$ (the property defining the set is due to (\ref{eq:tri_angle_important})). We bound $T_4$ by Lemma \ref{lem:max_ell_sum}. 
%The proof of Theorem \ref{thm:local_embed} will be further completed
%in Appendix \ref{app:prove_thm1}. 

\paragraph{Bounding $\breve{T}_1$:} Recall   $\calN_{r'}$ is a minimal $r'$-net of $\calK$ and   $\breve{T}_1 := \sup_{\bp\in\calN_{r'}}\sum_{i=1}^m \mathbbm{1}\big(||\ba_i^\top\bp|-\tau|\le r\big).$ For each given $\bp\in \calN_{r'}$, we have $$
    \sum_{i=1}^m \mathbbm{1}\big(||\ba_i^\top\bp|-\tau|\le r\big)\sim\text{Bin}\Big(m,\mathbbm{P}\big(||\ba_i^\top\bp|-\tau|\le r\big)\Big)$$ with the probability being bounded by
$$\mathbbm{P}\big(||\ba_i^\top\bp|-\tau|\le r\big)\le \mathbbm{P}\left(\Big||\calN(0,1)|-\frac{\tau}{\|\bp\|_2}\Big|\le \frac{r}{\alpha}\right)\le \sqrt{\frac{8}{\pi}}\frac{r}{\alpha}.$$ Therefore, for a  fixed $\bp \in\calN_{r'}$, the weakened Chernoff bound in Lemma \ref{lem:chernoff} gives  $$
    \mathbbm{P}\Big(\sum_{i=1}^m \mathbbm{1}\big(||\ba_i^\top\bp|-\tau|\le r\big) \ge \frac{6mr}{\alpha\sqrt{2\pi}} \Big)  \le \mathbbm{P}\Big(\text{Bin}\Big(m,\sqrt{\frac{8}{\pi}}\frac{r}{\alpha}\Big)\ge \frac{6mr}{\alpha\sqrt{2\pi}}\Big) \le \exp\Big(-\frac{mr}{\sqrt{18\pi}\alpha}\Big).$$ 
 We further take a union bound over $\bp\in \calN_{r'}$ to obtain  $$
    \mathbbm{P}\Big(\breve{T}_1\ge \frac{6mr}{\alpha\sqrt{2\pi}}\Big)\le \exp\Big(\scrH(\calK,r')-\frac{mr}{\sqrt{18\pi}\alpha}\Big)\le \exp\big(-\frac{mr}{6\sqrt{\pi}\alpha}\big),$$ where the last inequality holds due to $m\gtrsim \frac{\scrH(\calK,r')}{r}$ assumed in (\ref{eq:local_emb_sample}). Therefore, $\breve{T}_1<\frac{6mr}{\alpha\sqrt{2\pi}}$ holds with probability at least $1-\exp(-\frac{mr}{6\sqrt{\pi}\alpha})$.

\paragraph{Bounding $\breve{T}_2$ and $\breve{T}_4$:} Recall $$\breve{T}_2:= \sup_{\bw\in\calK_{(3r'/2)}}\sum_{i=1}^m\mathbbm{1}\big(|\ba_i^\top\bw|>\frac{r}{2}\big)$$ and $$\breve{T}_4:=\sup_{\bw\in \calK_{(r')}}\sum_{i=1}^m \mathbbm{1}\big(|\ba_i^\top\bw|\ge\frac{\theta_0r}{2}\big).$$ Let $c_3=\min\{\frac{1}{2},\frac{\theta_0}{2}\}$,  then it follows that $\max\{\breve{T}_2,\breve{T}_4\}\le \sup_{\bw\in \calK_{(3r'/2)}}\sum_{i=1}^m \mathbbm{1}(|\ba_i^\top\bw|\ge c_3r).$ We seek to prove   $$\sup_{\bw\in \calK_{(3r'/2)}}\sum_{i=1}^m \mathbbm{1}(|\ba_i^\top\bw|\ge c_3r)\le 2\tilde{c}rm$$ with $\tilde{c}:=\frac{c_0}{16}$, where $c_0$ is the   constant in (\ref{eq:choosetheta0}). Since we assume $rm\ge C$ for some constant $C$, we can proceed with $\tilde{c}rm\ge 1$. 
Now we invoke Lemma \ref{lem:max_ell_sum} with $\calW=\calK_{(3r'/2)}$ and $\ell = \lceil \tilde{c}rm\rceil$, yielding that 
\begin{align}\label{eq:apply_max_ell1}
    \sup_{\bw \in\calK_{(3r'/2)}}\max_{\substack{I\subset [m]\\|I|\le \lceil \tilde{c}rm\rceil}} \Big(\frac{1}{\lceil \tilde{c}rm\rceil}\sum_{i\in I}|\ba_i^\top\bw|^2\Big)^{1/2} \le C_4 \Big(\frac{\omega(\calK_{(3r'/2)})}{\sqrt{\tilde{c}rm}}+\frac{3r'}{2}\sqrt{\log \frac{e}{r}}\Big)
\end{align}
holds with probability at least $1-2\exp(-c_5rm \log(\frac{e}{r}))$. In our setting,
we note the following on the two sides of (\ref{eq:apply_max_ell1}): since (\ref{eq:local_emb_sample}) with sufficiently large $C_2$ implies  $C_4\frac{\omega(\calK_{(3r'/2)})}{\sqrt{\tilde{c}rm}}<\frac{c_3r}{2}$, $r'=\frac{cr}{\sqrt{\log(r^{-1})}}$ with small enough $c$ implies $\frac{3C_4r'}{2}\sqrt{\log(\frac{e}{r})}<\frac{c_3r}{2}$, we have that     the right-hand side of (\ref{eq:apply_max_ell1}) is smaller than $c_3r$;  it is not difficult to observe that   the left-hand side of (\ref{eq:apply_max_ell1}) is a uniform upper bound  on the $\lceil \tilde{c}rm\rceil$-th largest entry of the vector $|\bA\bw|$ for all $\bw\in \calK_{(3r'/2)}$. 
Therefore, universally for all $\bw \in\calK_{(3r'/2)}$, the $\lceil \tilde{c}rm\rceil$-th largest entry of $|\bA\bw|=(|\ba_1^\top\bw|,\cdots,|\ba_m^\top\bw|)^\top$ is smaller than $c_3r$ with probability at least $1-2\exp(-c_5rm\log(\frac{e}{r}))$. Phrasing differently, we obtain $$
    \sum_{i=1}^m\mathbbm{1}(|\ba_i^\top\bw|\ge c_3r)\le \lceil  \tilde{c}rm\rceil\le 2\tilde{c}rm$$ for any $\bw\in \calK_{(3r'/2)}$. 
Therefore, $\max\{\breve{T}_2,\breve{T}_4\}\le 2\tilde{c}rm$ holds with probability at least $1-2\exp(-c_5rm\log(\frac{e}{r}))$.    

\paragraph{Bounding $\breve{T}_3$:} 
Recall $\breve{T}_3 := \sum_{i=1}^m\mathbbm{1}(\calH_{|\ba_i|}~\theta_0\text{-well-separates }\bp\text{ and }\bq).$ By (\ref{eq:tri_angle_important}), $(\bp,\bq)$ in $\breve{T}_3$ live in the constraint set $$
    \calU_{r,r'}:=\{(\bp,\bq)\in \calN_{r'}\times \calN_{r'}:\dist(\bp,\bq)\ge r\}.$$ Hence, for a given pair of $(\bp,\bq)\in\calU_{r,r'}$ we have $\breve{T}_3 \sim \text{Bin}(m,\sfP_{\theta_0,\bp,\bq})$ with $\sfP_{\theta_0,\bp,\bq}\ge c_0\dist(\bp,\bq)$, and thus by Chernoff bound (the weakened version in Lemma \ref{lem:chernoff}) we obtain 
    \begin{align}
   & \nn\mathbbm{P}\Big(\breve{T}_3 \le \frac{c_0m\dist(\bp,\bq)}{2}\Big)\\
  \nn &\le\mathbbm{P}\Big(\text{Bin}(m,c_0\dist(\bp,\bq))\le \frac{c_0m\dist(\bp,\bq)}{2}\Big)\\&\le \exp\Big(-\frac{c_0m\dist(\bp,\bq)}{8}\Big)\le \exp\Big(-\frac{c_0mr}{8}\Big), \label{eq:fixed_pq}
\end{align} 
where the last inequality follows from $\dist(\bp,\bq)\ge r$. 
While (\ref{eq:fixed_pq}) is for a fixed pair of $(\bp,\bq)$, it can be extended to all $(\bp,\bq)$ in $\calU_{r,r'}$ by  union bound, which yields that the event 
\begin{align}
    \sum_{i=1}^m\mathbbm{1}(\calH_{|\ba_i|}~\theta_0\text{-well-separates }\bp\text{ and }\bq) \ge \frac{c_0m\dist(\bp,\bq)}{2},~\forall\,(\bp,\bq)\in \calU_{r,r'} \label{eq:T3_uniform}
\end{align}
holds with probability at least  
   $$1-|\calU_{r,r'}|\exp\Big(-\frac{c_0mr}{8}\Big)\ge 1-|\calN_{r'}|^2\exp\Big(-\frac{c_0mr}{8}\Big) =1-\exp\Big(2\scrH(\calK,r')-\frac{c_0mr}{8}\Big)\ge 
   1-\exp\Big(-\frac{c_0mr}{16}\Big),$$ where the last inequality follows from (\ref{eq:local_emb_sample}).

Now we are ready to establish the two statements in Theorem \ref{thm:local_embed}.   Under the sample complexity and promised probability    in   Theorem \ref{thm:local_embed}, we have $\breve{T}_1<\frac{6mr}{\alpha\sqrt{2\pi}}$ and $\max\{\breve{T}_2,\breve{T}_4\}\le 2\tilde{c}rm=\frac{c_0rm}{8}$. Then,   substituting $\breve{T}_1+\breve{T}_2=\calO(rm)$ into (\ref{eq:dete_upper_bound}) establishes the desired event $E_s$. To establish $E_l$,  for any given $\bu,\bv\in\calK$ with $\dist(\bu,\bv)\ge 2r$ and the corresponding $(\bp,\bq)\in \calU_{r,r'}$ found by (\ref{eq:foundpq}), we    substitute  (\ref{eq:T3_uniform})  and   $\breve{T}_4 \le \frac{c_0mr}{8}$ into (\ref{eq:dete_lower_bound}) to obtain 
   $$d_H(\sign(|\bA\bu|-\tau),\sign(|\bA\bv|-\tau))\ge \breve{T}_3-2\breve{T}_4\ge \frac{c_0m\dist(\bp,\bq)}{2}-\frac{c_0mr}{4} \ge \frac{c_0m\dist(\bp,\bq)}{4}.$$
Combining with $\dist(\bp,\bq)\ge \frac{1}{2}\dist(\bu,\bv)$ from  (\ref{eq:tri_angle_important}), we arrive at $$
    d_H(\sign(|\bA\bu|-\tau),\sign(|\bA\bv|-\tau))\ge \frac{c_0m\dist(\bu,\bv)}{8}.$$ Note that with the promised probability, the above arguments work for any $\bu,\bv\in\calK$ satisfying $\dist(\bu,\bv)\ge 2r$, thus we have shown the event $E_l$. The proof is complete.  
 \end{proof}
 
\subsection{The Proof of Lemma \ref{lem:Puv} ($\sfP_{\bu,\bv}\asymp \dist(\bu,\bv)$)}\label{app:sepa_prob_equ}
\begin{proof}
Since $\sfP_{\bu,\bv}=\sfP_{\bu,-\bv}=\sfP_{\bv,\bu}$ and $\dist(\bu,\bv)=\dist(\bu,-\bv)=\dist(\bv,\bu)$ hold, we can assume $\bu^\top\bv\ge 0$ and $\|\bu\|_2\ge\|\bv\|_2$ without loss of generality. Let us begin with 
          \begin{align}\nn
        \sfP_{\bu,\bv}&=  \mathbbm{P}\big(|\ba^\top\bu|\ge\tau,|\ba^\top\bv|<\tau\big)+ \mathbbm{P}\big(|\ba^\top\bu|<\tau , |\ba^\top\bv|\ge\tau\big)\nn\\\nn
        &\leq 2\mathbbm{P}\big(|\ba^\top\bu|\ge\tau,|\ba^\top\bv|<\tau\big)\\
        &=4\mathbbm{P}\big(\ba^\top\bu\geq \tau,|\ba^\top\bv|<\tau\big) \nn\\
        &\leq 8 \mathbbm{P}\big(\ba^\top\bu\geq \tau,0\leq \ba^\top\bv <\tau\big),\label{eq:uppuv}
    \end{align} 
     where the first inequality holds because $\mathbbm{P}(|\ba^\top\bu|<\tau,|\ba^\top\bv|\geq \tau)\leq \mathbbm{P}(|\ba^\top\bu|\geq \tau,|\ba^\top\bv|<\tau)$ that follows from $\|\bu\|_2\geq \|\bv\|_2$, in the second inequality we use $$\mathbbm{P}(\ba^\top\bu \geq \tau,-\tau\leq \ba^\top\bv <0) \leq \mathbbm{P}(\ba^\top\bu\geq\tau,0\leq \ba^\top\bv<\tau)$$ due to $\bu^\top\bv\geq 0$.
         Evidently, these two inequalities remain true with ``$\le$'' reversed to ``$\ge$'' if we do not introduce the multiplicative factor $2$. This observation provides \begin{align}\label{eq:lowpuv}
             \sfP_{\bu,\bv}\ge 2\mathbbm{P}(\ba^\top\bu\geq \tau,0\leq \ba^\top\bv <\tau).
         \end{align} 
         Combining (\ref{eq:uppuv}) and (\ref{eq:lowpuv})
         gives $\sfP_{\bu,\bv} \asymp \mathbbm{P}(\ba^\top\bu\ge\tau,~0\le \ba^\top\bv<\tau)$, thus  it suffices to show  $
             \mathbbm{P}(\ba^\top\bu\geq \tau,0\leq \ba^\top\bv <\tau)\asymp\dist(\bu,\bv)$  over $\mathbbm{A}_{\alpha,\beta}$. We let $\rho:= \frac{\bu^\top\bv}{\|\bu\|_2\|\bv\|_2}$ and deal with $\rho \in[0,1)$ and $\rho =1$ separately.

\subsubsection*{Case 1: $0\leq \rho<1$}
  Note that $w:= \ba^\top\frac{\bu}{\|\bu\|_2}$ and $z:= \ba^\top\frac{\bv}{\|\bv\|_2}$ have a joint P.D.F. $
        f(w,z) = \frac{1}{2\pi \sqrt{1-\rho^2}}\exp\big(-\frac{w^2-2\rho wz +z^2}{2(1-\rho^2)}\big)$, thus we can decompose the probability as 
    \begin{subequations}\label{eq:PtoIJ}
         \begin{align}
        &\mathbbm{P}\big(\ba^\top\bu\geq \tau,~0\leq \ba^\top\bv<\tau\big)  
        =  \mathbbm{P}\Big(w\geq \frac{\tau}{\|\bu\|_2},~0\leq z<\frac{\tau}{\|\bv\|_2}\Big)\nn\\\nn
        =& \frac{1}{2\pi \sqrt{1-\rho^2}}\int_0^{\frac{\tau}{\|\bv\|_2}}\int_{\frac{\tau}{\|\bu\|_2}}^\infty \exp\Big(-\frac{w^2-2\rho wz+z^2}{2(1-\rho^2)}\Big)~\text{d}w\text{d}z\\\nn
        =& \frac{1}{2\pi\sqrt{1-\rho^2}}\int_0^{\frac{\tau}{\|\bv\|_2}}e^{-\frac{z^2}{2}}\int_{\frac{\tau}{\|\bu\|_2}}^\infty \exp\Big(-\frac{(w-\rho z)^2}{2(1-\rho^2)}\Big)~\text{d}w\text{d}z \\
        =& \frac{1}{2\pi}\int_0^{\frac{\tau}{\|\bv\|_2}}e^{-\frac{z^2}{2}}\int_{\frac{\tau/\|\bu\|_2-\rho z}{\sqrt{1-\rho^2}}}^\infty e^{-\frac{w_1^2}{2}}~\text{d}w_1\text{d}z \label{eq:cha_var_1}\\
        =& \underbrace{\frac{1}{\sqrt{2\pi}}\int_{\tau/\|\bu\|_2}^{\frac{\tau}{\|\bv\|_2}}e^{-\frac{z^2}{2}}\cdot I (z)~\text{d}z}_{:=\sfI_{\bu,\bv}} +\underbrace{\frac{1}{\sqrt{2\pi}}\int_{0}^{\frac{\tau}{\|\bu\|_2}}e^{-\frac{z^2}{2}}\cdot I (z)~\text{d}z}_{:=\sfJ_{\bu,\bv}}, \label{eq:nota1}
    \end{align}
    \end{subequations}
where in (\ref{eq:cha_var_1}) we use a change of variable $w_1 = \frac{w-\rho z}{\sqrt{1-\rho^2}}$, in (\ref{eq:nota1}) we introduce the shorthand 
\begin{align}\label{eq:def_Iz}
    I (z): = \int_{\frac{\tau/\|\bu\|_2-\rho z}{\sqrt{1-\rho^2}}}^\infty \frac{1}{\sqrt{2\pi}}e^{-\frac{w^2}{2}}~\text{d}w =\mathbbm{P}\Big(\mathcal{N}(0,1)>\frac{\tau/\|\bu\|_2 - \rho z}{\sqrt{1-\rho^2}}\Big).
\end{align}   
In the remainder of the proof, our strategy is to show $\sfI_{\bu,\bv}\asymp \dist_{\rm n}(\bu,\bv)$ and $\sfJ_{\bu,\bv}\asymp \dist_{\rm d}(\bu,\bv)$ separately, then the desired statement follows from Lemma \ref{lem:norm_equa}.

\paragraph{Showing $\sfI_{\bu,\bv}\asymp \dist_{\rm n}(\bu,\bv)$:} The integral variable $z$ in $\sfI_{\bu,\bv}$ satisfies $z\geq \frac{\tau}{\|\bu\|_2}$, which provides lower bound on the lower limit of the integration $I(z)$: {$$
    \frac{\tau/\|\bu\|_2-\rho z}{\sqrt{1-\rho^2}} \leq \frac{{\tau}(1-\rho)/{\|\bu\|_2}}{\sqrt{1-\rho^2}}=\frac{\tau\sqrt{1-\rho}}{\|\bu\|_2\sqrt{1+\rho}}\leq \frac{\tau}{\alpha}.$$ Hence, we have  $ 1\geq    I(z)\geq \int_{\tau/\alpha}^\infty \frac{1}{\sqrt{2\pi}} e^{-\frac{w^2}{2}}~\text{d}w := C_0$ for some   constant $C_0$. Moreover, when $z\in [\frac{\tau}{\|\bu\|_2},\frac{\tau}{\|\bv\|_2}]\subset [\frac{\tau}{\beta},\frac{\tau}{\alpha}]$ we have $ e^{-\frac{z^2}{2}}\in [\exp(-\frac{\tau^2}{2\alpha^2}),\exp(-\frac{\tau^2}{2\beta^2})]$.} Taken collectively, we have shown that the integral function of $\sfI_{\bu,\bv}$ takes value in $[C_1,C_2]$ for some   constants $C_2>C_1>0$. Therefore, we have $$
    \sfI_{\bu,\bv} \asymp \Big|\frac{\tau}{\|\bu\|_2}-\frac{\tau}{\|\bv\|_2}\Big|= \frac{\tau}{\|\bu\|_2\|\bv\|_2}|\|\bu\|_2-\|\bv\|_2|\asymp \dist_{\rm n}(\bu,\bv).$$

%combining with $I(z)=\Theta(1)$ from (\ref{eq:Iz_Theta1}) gives 
%\begin{align}
 %   I_1 &=\int_{\frac{\tau}{\|\bu\|_2}}^{\frac{\tau}{\|\bv\|_2}}\frac{e^{-\frac{z^2}{2}}I(z)}{\sqrt{2\pi}}~\text{d}z=\Theta\Big(\frac{\tau}{\|\bv\|_2}-\frac{\tau}{\|\bu\|_2}\Big)\\
  %  &=\Theta \big(\|\bu\|_2-\|\bv\|_2\big)= \Theta(\dnor(\bu,\bv)). \label{eq:I1_scale}
%\end{align}

 \paragraph{Showing $\sfJ_{\bu,\bv}\asymp \dist_{\rm d}(\bu,\bv)$:}  Using Lemma \ref{lem:gaussian_tail} to bound $I(z)$ in (\ref{eq:def_Iz}), we obtain 
\begin{equation}\label{eq:low_up}
   \frac{1}{4}\exp \Big(-\frac{(\tau/ \|\bu\|_2- \rho z)^2}{1-\rho^2}\Big) \leq I(z) \leq \frac{1}{2} \exp \Big(-\frac{(\tau/\|\bu\|_2-\rho z)^2}{2(1-\rho^2)}\Big).
\end{equation}
Next, we bound $\sfJ_{\bu,\bv}$ from both sides via calculations.
 
     \paragraph{(i) Lower Bound:} Substituting the lower bound in (\ref{eq:low_up}) into $\sfJ_{\bu,\bv}$ along with some calculations, we proceed as \begin{subequations}
        \begin{align}\nn
    &\sfJ_{\bu,\bv} \geq \frac{1}{4\sqrt{2\pi}}\int_0^{\tau/\|\bu\|_2} e^{-\frac{z^2}{2}} \exp \Big(-\frac{(\tau/\|\bu\|_2-\rho z)^2}{1-\rho^2}\Big)~\text{d}z \\&\label{eq:alge1}= \frac{1}{4\sqrt{2\pi}}\exp\left(-\frac{\tau^2}{\|\bu\|_2^2 (1+\rho^2)}\right)\int_0^{\tau/\|\bu\|_2}\exp \left(-\frac{(1+\rho^2)[z-\frac{\tau}{\|\bu\|_2}\frac{2\rho}{(1+\rho^2)}]^2}{2(1-\rho^2)}\right)~\text{d}z \\
    &\geq \frac{1}{4\sqrt{2\pi}}\exp\Big(-\frac{\tau^2}{\alpha^2}\Big) \int_0^{\tau/\|\bu\|_2}\exp\left(-\Big[\frac{z- \frac{\tau}{\|\bu\|_2}\frac{2\rho}{1+\rho^2}}{\sqrt{1-\rho^2}}\Big]^2\right)~\text{d}z \label{eq:Theta1_1}\\
    & \label{eq:chan_vari_2}= \frac{\sqrt{1-\rho^2}}{4\sqrt{2\pi}} \exp\Big(-\frac{\tau^2}{\alpha^2}\Big)\int_{-\frac{2\rho\tau}{\|\bu\|_2 (1+\rho^2)\sqrt{1-\rho^2}}}^{\frac{\tau(1-\rho)^{3/2}}{\|\bu\|_2\sqrt{1+\rho}(1+\rho^2)}} e^{-w^2}~\text{d}w \\
    &\geq \label{eq:rho2} \frac{\sqrt{1-\rho^2}}{4\sqrt{2\pi}}\exp\Big(-\frac{\tau^2}{\alpha^2}\Big)\int_{-\frac{\rho \tau}{\beta}}^{\frac{\tau(1-\rho)^{3/2}}{2\sqrt{2}\beta}}e^{-w^2}~\text{d}w \\& \label{eq:Omega1}  \ge   \frac{1}{4\sqrt{2\pi}} \exp\Big(-\frac{\tau^2}{\alpha^2}\Big) \exp\Big(-\frac{\tau^2}{\beta^2}\Big)
    \Big(\frac{\tau(1-\rho)^{3/2}}{2\sqrt{2}\beta}+ \frac{\rho\tau}{\beta}\Big) \sqrt{1-\rho^2}\\&\ge c_3\sqrt{1-\rho^2},\label{eq:final_J_lower}
\end{align}\label{eq:lower_J}
    \end{subequations}
where  in (\ref{eq:alge1}) we complete the square; in (\ref{eq:Theta1_1}) we use $\|\bu\|_2\geq \alpha$ and $1+\rho^2\in [1,2]$; in (\ref{eq:chan_vari_2}) we apply a change of the variable $w= \frac{z-2\rho\tau/[\|\bu\|_2(1+\rho^2)]}{\sqrt{1-\rho^2}}$; in (\ref{eq:rho2}) we ``shrink'' the integral limits by $\sqrt{1+\rho}(1+\rho^2)\leq 2\sqrt{2}$ and $\frac{2}{(1+\rho^2)\sqrt{1-\rho^2}}\geq 1$; (\ref{eq:Omega1}) holds because $e^{-w^2}\ge \exp(-\frac{\tau^2}{\beta^2})$ when $w\in [-\frac{\rho\tau}{\beta},\frac{\tau(1-\rho)^{3/2}}{2\sqrt{2}\beta}]$; the final bound in (\ref{eq:final_J_lower}) holds for a constant $c_3$ independent of $\rho$ since $\inf_{\rho\in [0,1]} \frac{\tau(1-\rho)^{3/2}}{2\sqrt{2}\beta}+\frac{\rho\tau}{\beta}>0$. 
          
 \paragraph{(ii) Upper Bound:} Substituting the upper bound in (\ref{eq:low_up}) into $\sfJ_{\bu,\bv}$ yields 
\begin{subequations}
    \begin{align}\nn
    & \sfJ_{\bu,\bv}\leq \frac{1}{2\sqrt{2\pi}}\int_0^{\tau/\|\bu\|_2}\exp\Big(-\frac{z^2}{2}\Big)\exp\left(-\frac{(\tau/\|\bu\|_2-\rho z)^2}{2(1-\rho^2)}\right)~\text{d}z\\
    & \label{eq:alge2}= \frac{1}{2\sqrt{2\pi}}\exp\Big(-\frac{\tau^2}{2\|\bu\|_2^2}\Big)\int_0^{\tau/\|\bu\|_2}\exp\left(-\frac{[z-\frac{\rho\tau}{\|\bu\|_2}]^2}{2(1-\rho^2)}\right) ~\text{d}z \\
    & \label{eq:chan_vari_3} =\frac{\sqrt{1-\rho^2}}{2}\exp\Big(-\frac{\tau^2}{2\|\bu\|_2^2}\Big) \int_{-\frac{\rho\tau}{\|\bu\|_2\sqrt{1-\rho^2}}}^{\frac{\tau\sqrt{1-\rho}}{\|\bu\|_2\sqrt{1+\rho}}}\frac{1}{\sqrt{2\pi}}e^{-\frac{w^2}{2}}~\text{d}w\leq \frac{1}{2}\exp\Big(-\frac{\tau^2}{2\alpha^2}\Big)\sqrt{1-\rho^2},
\end{align}\label{eq:upper_J}
\end{subequations}
where (\ref{eq:alge2}) is obtained by completing the square;  in (\ref{eq:chan_vari_3}) we apply a change of the variable $w= \frac{z-\rho\tau/\|\bu\|_2}{\sqrt{1-\rho^2}}$. 
Combining the bounds in (\ref{eq:lower_J}), (\ref{eq:upper_J}), and $$
    \dist_{\rm d}(\bu,\bv)=\dist\big(\frac{\bu}{\|\bu\|_2},\frac{\bv}{\|\bv\|_2}\big)=\big(2-\frac{2|\bu^\top\bv|}{\|\bu\|_2\|\bv\|_2}\big)^{1/2}=\sqrt{2(1-|\rho|)},$$ we arrive at 
$
    \sfJ_{\bu,\bv}\asymp\sqrt{1-\rho^2}=\sqrt{1+|\rho|}\sqrt{1-|\rho|}\asymp \sqrt{1-|\rho|}\asymp \dist_{\rm d}(\bu,\bv).$ 
Further combining the bounds on $\sfI_{\bu,\bv}$, $\sfJ_{\bu,\bv}$ and (\ref{eq:PtoIJ}), we come to 
\begin{align*}
    \sfP_{\bu,\bv}\asymp \mathbbm{P}(\ba^\top\bu\ge \tau,0\le\ba^\top\bv<\tau)\asymp \sfI_{\bu,\bv}+\sfJ_{\bu,\bv} \asymp \dist_{\rm n}(\bu,\bv)+\dist_{\rm d}(\bu,\bv)\asymp\dist(\bu,\bv).  
\end{align*}

%\begin{align}
 %   I_2& =\Theta\big(\sqrt{1-\rho^2}\big)\\& =\Theta(\sqrt{1-|\rho|})= \Theta(\ddir(\bu,\bv)),\label{eq:212}
%\end{align} 
%where (\ref{eq:212}) follows from (\ref{eq:expand_rhod}). Moreover, recall from our initial analyses that $P_{\|\bu\|_2,\|\bv\|_2,\rho}=\Theta\big(I_1+I_2\big)$, then combining $I_1= \Theta(\dnor(\bu,\bv))$ and $I_2=\Theta(\ddir(\bu,\bv))$ and applying Lemma \ref{lem:norm_equ} we obtain \begin{align}
   % P_{\|\bu\|_2,\|\bv\|_2,\rho} = \Theta\big(\dnor(\bu,\bv)+\ddir(\bu,\bv)\big) = \Theta\big(\dist(\bu,\bv)\big),
%\end{align}
%the result follows. 

%\subsubsection*{Upper bound}

\subsubsection*{Case 2: $\rho=1$}
Recall that we assume $\|\bu\|_2\ge\|\bv\|_2$. By Cauchy-Schwarz inequality we know that  $\rho  = \frac{\bu^\top\bv}{\|\bu\|_2\|\bv\|_2}=1$ holds if and only if $\bu= \lambda \bv$ for some $\lambda\geq 1$, with $\lambda = \frac{\|\bu\|_2}{\|\bv\|_2}$. Thus, we can directly calculate the separation probability in this simpler case:   
\begin{subequations}
\begin{align}
    &\sfP_{\bu,\bv}= \mathbbm{P}\Big(\sign\big(\lambda|\ba^\top\bv|-\tau\big) \neq \sign\big(|\ba^\top\bv|-\tau\big)\Big)\label{eq:comp_2}=\mathbbm{P}\Big(\lambda|\ba^\top\bv|-\tau \geq 0,~|\ba^\top\bv|-\tau<0\Big)\\
    & = \mathbbm{P}\Big(\frac{\tau}{\lambda\|\bv\|_2}\leq \Big|\ba^\top\frac{\bv}{\|\bv\|_2}\Big|<\frac{\tau}{\|\bv\|_2}\Big) \label{eq:plug_lam} \asymp\frac{\tau}{\|\bv\|_2}\Big(1-\frac{1}{\lambda}\Big)\asymp \big|\|\bu\|_2-\|\bv\|_2\big|,
\end{align}
\end{subequations}
where (\ref{eq:comp_2}) holds because $\lambda|\ba^\top\bv|\geq |\ba^\top\bv|$; (\ref{eq:plug_lam}) holds because the P.D.F. of $\ba^\top\frac{\bv}{\|\bv\|_2}\sim\calN(0,1)$ satisfies $$\frac{1}{\sqrt{2\pi}}\exp\big(-\frac{\tau^2}{2\alpha^2}\big)\le p(z)=\frac{1}{\sqrt{2\pi}}\exp\big(-\frac{z^2}{2}\big)\le\frac{1}{\sqrt{2\pi}}$$ when $z\in [-\frac{\tau}{\alpha},\frac{\tau}{\alpha}]$. Note that in this case we have $\dist(\bu,\bv) = |\|\bu\|_2-\|\bv\|_2|$, thus we obtain $\sfP_{\bu,\bv}\asymp \dist(\bu,\bv)$ again. The proof is complete. 
\end{proof}
 
 \subsection{The Proof of Theorem \ref{thm:lower} (Lower Bound)}\label{app:prove_lower}
\begin{proof}
We continue from the proof sketch provided in Section \ref{sec:ITproof}, where we let $\calV_0:=\{(u_1,...,u_\nu,u_{\nu+1},...,u_n)^\top:u_i=0,\forall \,\nu+1\le i\le n\}$ be a specific $\nu$-dimensional subspace in $\mathbb{R}^n$ and seek to construct a packing set of $\calV_0\cap \mathbbm{A}_{\alpha,\beta}$ with large enough cardinality. Since the packing is only relevant to the Euclidean metric, we can  identify $\calV_0 \cap\mathbbm{A}_{\alpha,\beta}$ with  $
	\calK_\nu:=  \{\bu\in\mathbb{R}^\nu:\alpha\le\|\bu\|_2\le \beta\}.$ Restricting our attention to the metric $\dist(\cdot,\cdot)$, we say $\calS$ is an $\epsilon$-packing set of $\calK_\nu$ if $\calS\subset \calK_{\nu}$ and any two distinct points $\bu,\bv$ in $\calS$ satisfy $\dist(\bu,\bv)\ge \nu$. We    
denote by $\scrP(\calK_{\nu},\epsilon;\dist)$   the $\epsilon$-packing number of $\calK_{\nu}$ under the metric $\dist(\cdot,\cdot)$ that is defined as the maximal cardinality of an $\epsilon$-packing set of $\calK_{\nu}$. By \cite[Lem. 4.2.8]{vershynin2018high}, we can first lower bound the packing number by the covering number: $\scrP(\calK_{\nu},\epsilon;\dist)\ge \scrN(\calK_{\nu},\epsilon;\dist)$. Then, we seek a lower bound on the $\scrN(\calK_{\nu},\epsilon;\dist)$ by  a standard volume argument. We let $\calN_\epsilon$ be a minimal $\epsilon$-net of $\calK_\nu$ satisfying $|\calN_\epsilon|=\scrN(\calK_{\nu},\epsilon;\dist)$. By  the equivalence between $\dist(\bu,\bv)\le \epsilon$ and $\bu\in \mathbb{B}_2^\nu(\bv;\epsilon)\cup \mathbb{B}_2^\nu(-\bv;\epsilon)$, we can   return to the $\ell_2$-distance metric by 
$\calK_{\nu}\subset\bigcup_{\bu\in \calN_\epsilon}\{\mathbb{B}_2^\nu(\bu;\epsilon)\cup \mathbb{B}_2^\nu(-\bu;\epsilon)\}.$ Next, comparing the volume yields 
	$$
		\text{Vol}(\calK_\nu)\le \sum_{\bu\in \calN_\epsilon}\big(\text{Vol}(\mathbb{B}_2^\nu(\bu;\epsilon))+\text{Vol}(\mathbb{B}_2^\nu(-\bu;\epsilon))\big)=2\scrN(\calK_\nu,\epsilon;\dist)\cdot \text{Vol}(\mathbb{B}_2^\nu(\epsilon)),$$ which leads to $(\beta^\nu-\alpha^\nu)\cdot\text{Vol}(\mathbb{B}_2^\nu)\le 2\scrN(\calK_{\nu},\epsilon;\dist)\cdot \epsilon^\nu\cdot \text{Vol}(\mathbb{B}_2^\nu)$, and hence $
      \scrN(\calK_{\nu},\epsilon;\dist) \ge \frac{\beta^\nu-\alpha^\nu}{2\epsilon^\nu}.$ Combining with $|\calM_{\bA}|\le (\frac{2em}{\nu})^\nu$ established in the main text, for the existence of an estimator $\hat{\bx}$ to achieve $\sup_{\bx\in\calK_{\nu}}\dist(\hat{\bx},\bx)\le\frac{\epsilon}{2}$, we must have $$(\frac{2em}{\nu})^\nu \ge \frac{\beta^\nu-\alpha^\nu}{2\epsilon^\nu},$$ yielding the desired claim. 
 \end{proof}

 \section{PLL-AIC Implies Convergence of Algorithms \ref{alg:pgd}--\ref{alg:pgd_high}} \label{app:aic2conver}
This appendix includes the proofs of Lemma \ref{lem:low_aic2err}  and  Lemma \ref{lem:high_aic2err}. 
\subsection{The Proof of Lemma \ref{lem:low_aic2err} (PLL-AIC Implies Convergence of {GD-1bPR})} \label{app:pllaic_lem}
\begin{proof}
Without loss of generality, we assume   $\bx^{(0)}$ is closer to $\bx$ than $-\bx$,  namely   $\dist(\bx^{(0)},\bx)=\|\bx^{(0)}-\bx\|_2\le \delta_4$. We consider GD-1bPR with step size $\eta$ and define $$\hat{t} = \inf\big\{t\in\mathbb{Z}_{\ge 0}:\|\bx^{(t)}-\bx\|_2\le E(\delta_1,\delta_2,\delta_3)\big\}$$ as the smallest $t$ such  that $\bx^{(t)}$ enters $\mathbb{B}_2^n(\bx;E(\delta_1,\delta_2,\delta_3))$, with the convention that $\hat{t}=\infty$ if such $t$ does not exist. We claim that $\hat{t}<\infty$ and the iterates  enter $\mathbb{B}_2^n(\bx;E(\delta_1,\delta_2,\delta_3))$ with no more than $\lceil \frac{\log E(\delta_1,\delta_2,\delta_3)}{\log((1+\delta_1)/2)}\rceil$ iterations. If $\|\bx^{(0)}-\bx\|_2\le E(\delta_1,\delta_2,\delta_3)$, then $\hat{t}=0$ and the claim is trivial, thus we can simply assume $\|\bx^{(0)}-\bx\|_2>E(\delta_1,\delta_2,\delta_3)$. This along with $E(\delta_1,\delta_2,\delta_3)=\max\{\frac{16\delta_2}{(1-\delta_1)^2},\frac{4\delta_3}{1-\delta_1}\}$ implies $  \delta_3\le \frac{1-\delta_1}{4}\|\bx^{(0)}-\bx\|_2 $ and $\sqrt{\delta_2\|\bx^{(0)}-\bx\|_2}\le\frac{1-\delta_1}{4}\|\bx^{(0)}-\bx\|_2$.  Combining $\delta_4<\frac{\alpha}{2}$, $\|\bx^{(0)}-\bx\|_2\le \delta_4$ and $\bx\in\mathbbm{A}_{\alpha,\beta}$, we obtain $\bx^{(0)}\in \mathbbm{A}_{\alpha/2,2\beta}$, hence we can use PLL-AIC to deduce
      \begin{align}\label{eq:low_021}
      \|\bx^{(1)}-\bx\|_2  = \|\bx^{(0)}-\bx-\eta\bh(\bx^{(0)},\bx)\|_2 \le \delta_1 \|\bx^{(0)}-\bx\|_2 + \sqrt{\delta_2\|\bx^{(0)}-\bx\|_2}+\delta_3 \le \frac{1+\delta_1}{2} \|\bx^{(0)}-\bx\|_2, 
  \end{align} 
  If $\|\bx^{(1)}-\bx\|_2 >E(\delta_1,\delta_2,\delta_3)$, since (\ref{eq:low_021}) gives $\|\bx^{(1)}-\bx\|_2\le \|\bx^{(0)}-\bx\|_2\le \delta_4$, by re-iterating the arguments in (\ref{eq:low_021}) we obtain $$
     \|\bx^{(2)}-\bx\|_2\le \frac{1+\delta_1}{2}\|\bx^{(1)}-\bx\|_2\le \Big(\frac{1+\delta_1}{2}\Big)^2\|\bx^{(0)}-\bx\|_2\le \Big(\frac{1+\delta_1}{2}\Big)^2.$$ 
     By induction, if $\hat{t}>\hat{t}_l$ holds (i.e., $\|\bx^{(t)}-\bx\|_2>E(\delta_1,\delta_2,\delta_3)$ holds for any $t\le \hat{t}_l$), then we have
$
 E(\delta_1,\delta_2,\delta_3) < \|\bx^{(\hat{t}_l)}-\bx\|_2\le (\frac{1+\delta_1}{2})^{\hat{t}_l}.$
 This leads to 
 $ \hat{t}_l< \frac{\log E(\delta_1,\delta_2,\delta_3)}{\log((1+\delta_1)/2)},$ which further implies $\hat{t}\le \lceil \frac{\log E(\delta_1,\delta_2,\delta_3)}{\log((1+\delta_1)/2)}\rceil$, 
   stating that the iterates enter  $\mathbb{B}_2^n(\bx;E(\delta_1,\delta_2,\delta_3))$ with no more than $\lceil \frac{\log E(\delta_1,\delta_2,\delta_3)}{\log((1+\delta_1)/2)}\rceil$ iterations. It remains to show $\|\bx^{(t)}-\bx\|_2\le E(\delta_1,\delta_2,\delta_3)$ for any $t\ge \hat{t}$. Note that $\|\bx^{(\hat{t})}-\bx\|\le E(\delta_1,\delta_2,\delta_3)<\delta_4<\frac{\alpha}{2}$ gives $\bx^{(\hat{t})}\in \mathbbm{A}_{\alpha/2,2\beta}$, together with $\delta_3 \le \frac{1-\delta_1}{4}E(\delta_1,\delta_2,\delta_3)$ and $\delta_2\le \frac{(1-\delta_1)^2}{16}E(\delta_1,\delta_2,\delta_3)$, we come to
      \begin{align*}
      &\|\bx^{(\hat{t}+1)} - \bx\|_2  = \|\bx^{(\hat{t})}-\bx-\eta\cdot \bh(\bx^{(\hat{t})},\bx)\|_2 \\& \le \delta_1 \|\bx^{(\hat{t})}-\bx\|_2 +\sqrt{\delta_2\cdot \|\bx^{(\hat{t})}-\bx\|_2}+\delta_3  \\
     & \le \frac{1+3\delta_1}{4}E(\delta_1,\delta_2,\delta_3) + \sqrt{\delta_2\cdot E(\delta_1,\delta_2,\delta_3)}\\& \le E(\delta_1,\delta_2,\delta_3),  
 \end{align*}
  By induction we can show $\bx^{(t)}\in \mathbb{B}_2^n(\bx;E(\delta_1,\delta_2,\delta_3))$ for any $t\ge \hat{t}$, which completes the proof.  
\end{proof}

\subsection{The Proof of Lemma \ref{lem:high_aic2err} (PLL-AIC Implies Convergence of {BIHT-1bSPR})}\label{app:pllraic2conver}

\begin{proof}
    Without loss of generality, we assume $\dist(\bx^{(0)},\bx)=\|\bx^{(0)}-\bx\|_2\le\delta_4$. We consider {BIHT-1bSPR} with step size $\eta$ and define $$\hat{t} = \inf\big\{t\in\mathbb{Z}_{\ge 0}:\|\bx^{(t)}-\bx\|_2\le E(\delta_1,\delta_2,\delta_3)\big\}$$ as the smallest $t$ such that $\bx^{(t)}$ enters $\mathbb{B}_2^n(\bx;E(\delta_1,\delta_2,\delta_3))$, with the convention that $\hat{t}=\infty$ if such $t$ does not exist. 
    We first show $\hat{t}\le \lceil \frac{\log E(\delta_1,\delta_2,\delta_3)}{\log(1/2+\delta_1)}\rceil$, which means that {BIHT-1bSPR} enters $\mathbb{B}_2^n(\bx;E(\delta_1,\delta_2,\delta_3))$ with no more than $ \lceil \frac{\log E(\delta_1,\delta_2,\delta_3)}{\log(1/2+\delta_1)}\rceil$ iterations. We assume $\|\bx^{(0)}-\bx\|_2>E(\delta_1,\delta_2,\delta_3)$, otherwise we have $\hat{t}=0$ and  we are done. Note that $\|\bx^{(0)}-\bx\|_2>E(\delta_1,\delta_2,\delta_3)=\max\{\frac{64\delta_2}{(1-2\delta_1)^2},\frac{8\delta_3}{1-2\delta_1}\}$ gives $ 2\sqrt{\delta_2\|\bx^{(0)}-\bx\|_2}\le \frac{1-2\delta_1}{4}\|\bx^{(0)}-\bx\|_2$ and $2\delta_3\le \frac{1-2\delta_1}{4}\|\bx^{(0)}-\bx\|_2$. By $\delta_4<0.01\alpha$, $\|\bx^{(0)}-\bx\|_2\le\delta_4$ and $\bx\in \mathbbm{A}_{\alpha,\beta}$, we have $\bx^{(0)}\in\mathbbm{A}_{0.99,1.01\beta}$. 
    By identifying $\calT_{(k)}$ as the projection onto $\Sigma^n_k$ and then use Lemma \ref{lem:pro_contain} and PLL-AIC, we proceed as 
        \begin{align} 
           \nn &\|\bx^{(1)}-\bx\|_2 = \|\calT_{(k)}(\tilde{\bx}^{(1)})-\bx\|_2 = \|\calP_{\Sigma^n_k}(\tilde{\bx}^{(1)})-\bx\|_2\\&\nn=\|\calP_{\Sigma^n_k-\bx}(\bx^{(0)}-\bx-\eta\cdot\bh(\bx^{(0)},\bx))\|_2\\&\nn\le 2\|\calP_{\Sigma^n_{2k}}(\bx^{(0)}-\bx- \eta\cdot \bh(\bx^{(0)},\bx))\|_2\\
            &\le 2\delta_1\|\bx^{(0)}-\bx\|_2 + 2\sqrt{\delta_2\cdot \|\bx^{(0)}-\bx\|_2}+2\delta_3
            \nn \\& \le \Big(\frac{1}{2}+\delta_1\Big)\|\bx^{(0)}-\bx\|_2,\label{eq:0to1sparse}
        \end{align}  
Suppose $\hat{t}>1$, then we have $\|\bx^{(1)}-\bx\|_2>E(\delta_1,\delta_2,\delta_3)$, and by (\ref{eq:0to1sparse}) we have $\bx^{(1)}\in \mathbbm{A}_{0.99\alpha,1.01\beta}$. Hence, we can re-iterate the arguments in (\ref{eq:0to1sparse}) to obtain $$
    \|\bx^{(2)}-\bx\|_2 \le \Big(\frac{1}{2}+\delta_1\Big)\|\bx^{(1)}-\bx\|_2\le \Big(\frac{1}{2}+\delta_1\Big)^2\|\bx^{(0)}-\bx\|_2\le \Big(\frac{1}{2}+\delta_1\Big)^2.$$ 
    By induction, if $\hat{t}>\hat{t}_l$ holds (that is, $\|\bx^{(t)}-\bx\|_2>E(\delta_1,\delta_2,\delta_3)$ holds for any $t\le\hat{t}_l$), then we have $E(\delta_1,\delta_2,\delta_3)<\|\bx^{(\hat{t}_l)}-\bx\|_2 \le \Big(\frac{1}{2}+\delta_1\Big)^{\hat{t}_l}$ which leads to $\hat{t}_l < \frac{\log E(\delta_1,\delta_2,\delta_3)}{\log(\frac{1}{2}+\delta_1)}$. 
Since this holds for any $\hat{t}_l$ smaller than $\hat{t}$, we have $\hat{t}\le \lceil \frac{\log (\delta_1,\delta_2,\delta_3)}{\log(1/2+\delta_1)}\rceil$. All that remains is to show $\|\bx^{(t)}-\bx\|_2\le E(\delta_1,\delta_2,\delta_3)$ for any $t\ge \hat{t}$. Note that $\|\bx^{(\hat{t})}-\bx\|_2\le E(\delta_1,\delta_2,\delta_3)<0.01\alpha$ gives $\bx^{(\hat{t})}\in\mathbbm{A}_{0.99\alpha,1.01\beta}$. Taken collectively with $\bx,\bx^{(\hat{t})}\in\Sigma^n_k$, Lemma \ref{lem:pro_contain} and PLL-AIC, we obtain 
    \begin{align*}
       & \|\bx^{(\hat{t}+1)}-\bx\|_2=\|\calT_{(k)}(\tilde{\bx}^{(\hat{t})})-\bx\|_2\\&
       =\|\calP_{\Sigma^n_k}(\bx^{(\hat{t})}-\eta\cdot\bh(\bx^{(\hat{t})},\bx))-\bx\|_2
       \\&= \|\calP_{\Sigma^n_k-\bx}(\bx^{(\hat{t})}-\bx-\eta\cdot\bh(\bx^{(\hat{t})},\bx))\|_2\\&\le 2\|\calP_{\Sigma^n_{2k}}(\bx^{(\hat{t})}-\bx- \eta \cdot \bh(\bx^{(\hat{t})},\bx))\|_2  \\
       &\le 2\delta_1\|\bx^{(\hat{t})}-\bx\|_2 + 2\sqrt{\delta_2\cdot \|\bx^{(\hat{t})}-\bx\|_2} + 2\delta_3\\&\le \frac{1+6\delta_1}{4}E(\delta_1,\delta_2,\delta_3)+ 2\sqrt{\delta_2\cdot E(\delta_1,\delta_2,\delta_3)}\\& \le E(\delta_1,\delta_2,\delta_3),  
    \end{align*} 
By induction, it follows that $\|\bx^{(t)}-\bx\|_2\le E(\delta_1,\delta_2,\delta_3)$ holds for any $t\ge \hat{t}$. The proof is complete. 
\end{proof}
\section{Gaussian Matrix Respects PLL-AIC}\label{app:prove_pgd}
Based on the proof outline in Section \ref{sec:pgdproof}, this appendix is devoted to proving Theorem \ref{thm:raic}. 
\subsection{Bounding $\|\calP_{\calC_-}(\bu_1-\bv_1-\eta\cdot\bh(\bu_1,\bv_1))\|_2$ (Step 1 in Large-distance regime)} \label{app:large_main_term}
It is sufficient to first separately bound    $T_1^{\bp,\bq},T_2^{\bp,\bq},T_3^{\bp,\bq}$ and the higher order term $\|\bh_2(\bp,\bq)\|_2$  for  fixed $(\bp,\bq)\in \calN_{r,\delta_4}^{(2)}$  and then take  union bound over $(\bp,\bq)\in \calN_{r,\delta_4}^{(2)}$. 

\subsubsection*{Step 1.1: Bounding $T_1^{\bp,\bq}$ over $\calN_{r,\delta_4}^{(2)}$}
Substituting $\bh_1(\bp,\bq)=\frac{1}{m}\sum_{i\in\bR_{\bp,\bq}}\sign(\ba_i^\top(\bp-\bq))\ba_i$  into $T_1^{\bp,\bq}$   yields $$
    T_1^{\bp,\bq} = \Big|\|\bp-\bq\|_2- \frac{\eta}{m}\sum_{i\in\bR_{\bp,\bq}}\frac{|\ba_i^\top(\bp-\bq)|}{\|\bp-\bq\|_2}\Big|:=\big|\|\bp-\bq\|_2 - \eta\cdot T_{11}^{\bp,\bq}\big|,$$ 
    where in the second equality we introduce $$T_{11}^{\bp,\bq}:=\frac{1}{m}\sum_{i\in \bR_{\bp,\bq}}\frac{|\ba_i^\top(\bp-\bq)|}{\|\bp-\bq\|_2}=\frac{1}{m}\sum_{i\in \bR_{\bp,\bq}}|\ba_i^\top\bbeta_1|.$$  
Conditioning on  $|\bR_{\bp,\bq}| = |\{i\in[m]:\sign(|\ba_i^\top\bp|-\tau)\neq \sign(|\ba_i^\top\bq|-\tau)\}|=r_{\bp,\bq},$ 
 the distribution of $T_{11}^{\bp,\bq}$ is identical with 
$T_{11}^{\bp,\bq}\big|\{|\bR_{\bp,\bq}|=r_{\bp,\bq}\} \sim   \frac{1}{m}\sum_{i=1}^{r_{\bp,\bq}}Z_i^{\bp,\bq}$, 
where the i.i.d. random variables $Z_1^{\bp,\bq},...,Z_{r_{\bp,\bq}}^{\bp,\bq}$ follow the conditional distribution 
\begin{align}\label{eq:Zipq}
    Z_i^{\bp,\bq} \sim  |\ba_i^\top\bbeta_1|\Big|\big\{\sign(|\ba_i^\top\bp|-\tau)\ne \sign(|\ba_i^\top\bq|-\tau)\big\} 
\end{align}
with standard Gaussian vector $\ba_i\sim \calN(0,\bI_n)$. We proceed to derive the P.D.F. of $Z_i^{\bp,\bq}$.

\paragraph{P.D.F. of $Z_i^{\bp,\bq}$:}
Recall (\ref{eq:beta1beta2}) that holds for some orthonormal $(\bbeta_1=\frac{\bp-\bq}{\|\bp-\bq\|_2},\bbeta_2)$ and coordinates  $(u_1,u_2,v_1)$ satisfying $u_1>v_1,u_2\ge 0$.   Combining with the rotational invariance of $\ba$, we can let $a_1,a_2$  be independent $\calN(0,1)$ variables and deduce that $Z_i^{\bp,\bq}$ in (\ref{eq:Zipq}) has the same distribution as  
    \begin{align*}
 &|\ba_i^\top\bbeta_1|\Big|\{\sign(|u_1\ba_i^\top\bbeta_1+u_2\ba_i^\top\bbeta_2|-\tau)\ne \sign(|v_1\ba_i^\top\bbeta_1+u_2\ba_i^\top\bbeta_2|-\tau)\}\\
   &\sim |a_1| \Big| \{\sign(|a_1u_1+a_2u_2|-\tau)\neq \sign(|a_1v_1+a_2u_2|-\tau)\}, 
\end{align*}    Define the event $\hat{E}:=\{\sign(|a_1u_1+a_2u_2|-\tau)\ne\sign(|a_1v_1+a_2u_2|-\tau)\} $  that happens with probability $  \mathbbm{P}(\hat{E})=\mathbbm{P}(\calH_{|\ba|}\text{ separates }\bp,\bq)=\sfP_{\bp,\bq},$
 Bayes' Theorem establishes the P.D.F. of $Z_i^{\bp,\bq}$ as 
\begin{align}\label{eq:pdf_a1}
    f_{Z_i^{\bp,\bq}}(z) = \frac{f_{|a_1|}(z)\cdot\mathbbm{P}(\hat{E}||a_1|=z)}{\mathbbm{P}(\hat{E})}= \sqrt{\frac{2}{\pi}}\frac{\exp(-\frac{z^2}{2})\cdot\mathbbm{P}(\hat{E}||a_1|=z)}{\sfP_{\bp,\bq}},~z\ge 0. 
\end{align}
Since $|a_1||\hat{E}$ and $Z_i^{\bp,\bq}$ are identically distributed,
we shall further calculate $\mathbbm{P}(\hat{E}||a_1|=z)$ to get the P.D.F. of $Z_i^{\bp,\bq}$. %To preserve the presentation flow, the proofs of the lemmas in this appendix are deferred to Appendix \ref{app:cal_dis}.

\begin{lem}[The P.D.F. of $Z_i^{\bp,\bq}$]\label{lem:T1pdfprob}
   For any $z>0$, we have $\mathbbm{P}\big(\hat{E}||a_1|=z\big) = P_1+P_2$ where \footnote{Note that this might be inexact at $z=\frac{\tau}{|u_1|},\frac{\tau}{|v_1|}$  when $u_2=0$ but does not affect our subsequent use of (\ref{eq:lem_pdf_zipq}) as the P.D.F. of $Z_i^{\bp,\bq}$. A similar note should also be made for Lemma \ref{lem:prob_cal_T2} below.}   $$P_1:=\mathbbm{P}\Big(\max\{\tau-zu_1,-\tau-zv_1\}<\calN(0,u_2^2)<\tau-zv_1\Big)$$ and $$P_2:=\mathbbm{P}\Big(\max\{zu_1-\tau,zv_1+\tau\}<\calN(0,u_2^2)<\tau+zu_1\Big).$$ 
   Substituting this into (\ref{eq:pdf_a1}) yields the P.D.F. of $Z_i^{\bp,\bq}$:\begin{align}
       \label{eq:lem_pdf_zipq}f_{Z_i^{\bp,\bq}}(z) = \sqrt{\frac{2}{\pi}} \frac{\exp(-\frac{z^2}{2})[P_1+P_2]}{\sfP_{\bp,\bq}},~~z\ge 0. 
   \end{align}
\end{lem}
\begin{proof}In view of (\ref{eq:pdf_a1}),
    it is enough to calculate $\mathbbm{P}(\hat{E}||a_1|=z)$: 
    \begin{subequations}
        \begin{align}\nn
       & \mathbbm{P}\big(\sign(|a_1u_1+a_2u_2|-\tau)\ne\sign(|a_1v_1+a_2u_2|-\tau)\big||a_1|=z\big)\\\nn
       & = \mathbbm{P}\big(\sign(||a_1|u_1+\sign(a_1)a_2u_2|-\tau)\ne \sign(||a_1|v_1+\sign(a_1)a_2u_2|-\tau)\big||a_1|=z\big)\\
       & = \mathbbm{P}\big(\sign(|zu_1+a_2u_2|-\tau)\ne \sign(|zv_1+a_2u_2|-\tau)\big)\label{eq:remove_con_1}\\
       & = \mathbbm{P}\big(|zu_1+a_2u_2|>\tau,|zv_1+a_2u_2|<\tau\big)+ \mathbbm{P}\big(|zu_1+a_2u_2|<\tau,|zv_1+a_2u_2|>\tau\big)\label{eq:leave_out_equal}:=P_1+P_2,
    \end{align}
    \end{subequations}
    where (\ref{eq:remove_con_1}) holds because $\sign(a_1)a_2| \{|a_1|=z\}$ and $a_2$ are identically distributed,
            (\ref{eq:leave_out_equal}) might not hold exactly when $u_2=0$ due to the cases of 
            $|zu_1|-\tau=0$ and $|zv_1|-\tau = 0$, but such inexactness could only occur at $z=\frac{\tau}{|u_1|}$ and $z= \frac{\tau}{|v_1|}$ and hence will not 
            affect our subsequent use of  the P.D.F. of the continuous random variable $Z_i^{\bp,\bq}$. 
       It remains to evaluate $P_1$ and $P_2$. With respect to $P_1$, since $zu_1+a_2u_2> zv_1+a_2u_2$, we have $$P_1 = \mathbbm{P}\big(zu_1+a_2u_2 >\tau,-\tau<zv_1+a_2u_2<\tau\big)=\mathbbm{P}\big(\max\{\tau-zu_1,-\tau-zv_1\}<a_2u_2<\tau-zv_1\big).$$ Similarly, $P_2$ simplifies to 
       \begin{align*}
           &P_2 = \mathbbm{P}\big(-\tau<zu_1+a_2u_2<\tau,~zv_1+a_2u_2<-\tau\big)\\&= \mathbbm{P}\big(-\tau-zu_1<a_2u_2< \min\{\tau-zu_1,-\tau-zv_1\}\big)\\&= \mathbbm{P}\big(\max\{zu_1-\tau,zv_1+\tau\}<a_2u_2<\tau+zu_1\big).
       \end{align*}
       By further noting $a_2u_2\sim \calN(0,u_2^2)$, we complete the proof. 
\end{proof}
\paragraph{Sub-Gaussianity of $Z_i^{\bp,\bq}$:} Next, we seek to show $Z_i^{\bp,\bq}$ is $\calO(1)$-sub-Gaussian. We accomplish this by showing the tail of $Z_i^{\bp,\bq}$ is dominated by the one of some Gaussian variable. To this end, we need to meticulously discuss all the cases of $(u_1,u_2,v_1)$ since the distribution of $Z_i^{\bp,\bq}$     essentially hinges on the relative position of $(\bp,\bq)$, with some typical cases depicted in  Figure \ref{fig:relaposi_D2}. In particular, we need to separately treat $u_2$ bounded away from $0$ and $u_2$ being close to $0$ here. %See Figure \ref{fig:relaposi_D2}.  

%Note that this is also the case for Lemma \ref{lem:SGT2hatZ} below.
%{\color{blue}[TODO: Think more.]}  

%the proof involves some meticulous discussions to cover all the possible cases of $(u_1,u_2,v_1)$.  

%These goals are accomplished in   the forthcoming lemmas, whose proofs are deferred Appendix \ref{app:cal_dis}.
\begin{lem}[$Z_i^{\bp,\bq}$ is $\calO(1)$-Sub-Gaussian] \label{lem:T1SG}
    There exists some  constant $C_0$ depending on $(\alpha,\beta,\tau)$ such that  $\|Z_i^{\bp,\bq}\|_{\psi_2}\le C_0$ holds for any $(\bp,\bq)\in \calN_{r,\delta_4}^{(2)}$.
\end{lem}

\begin{proof}
    By Lemma \ref{lem:T1pdfprob} the P.D.F. of $Z_i^{\bp,\bq}$ is given by $$
        f_{Z_i^{\bp,\bq}}(z) = \sqrt{\frac{2}{\pi}}\frac{\exp(-\frac{z^2}{2})[P_1+P_2]}{\sfP_{\bp,\bq}},\qquad\forall \,z>0,$$ and we have
$P_1\le \mathbbm{P}(\tau-zu_1<\calN(0,u_2^2)<\tau-zv_1)$ and $P_2\le \mathbbm{P}(\tau+zv_1<\calN(0,u_2^2)<\tau+zu_1)$.   
   Recall $u_2\ge 0$, we separately  treat  several cases in the following. 
   \begin{figure}[ht!]
    \begin{centering}
        \includegraphics[width=0.20\columnwidth]{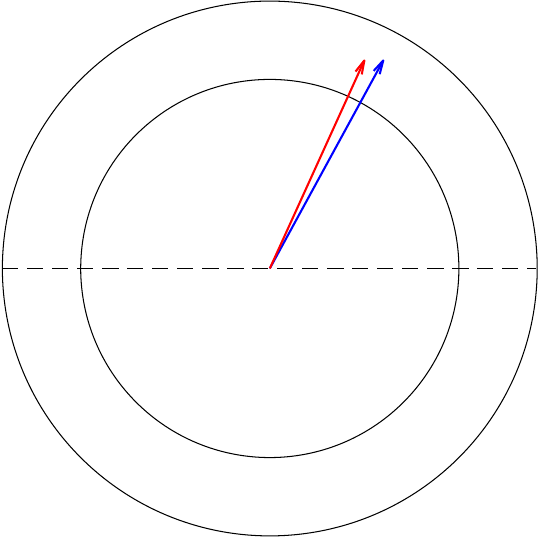} \quad\quad \includegraphics[width=0.20\columnwidth]{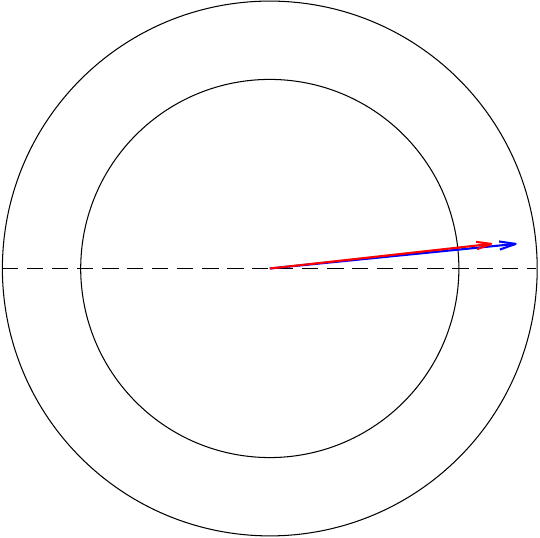}\quad \quad
        \includegraphics[width=0.20\columnwidth]{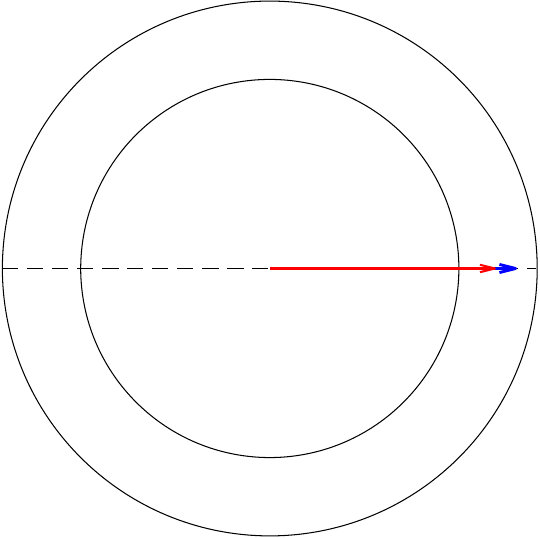}
        \par
        \vspace{1mm}
       \hspace{41mm}(a) \hspace{35mm} (b) \hspace{35mm} (c)
    \end{centering}
    
     \caption{\label{fig:relaposi_D2} \small We parametrize $\bp,\bq\in \mathbbm{A}_{\alpha,\beta}$ as $(u_1,u_2)$ and $(v_1,u_2)$ that obey $u_1>v_1$ and $u_2\ge0$, with $(u_1,u_2)$ and $(v_1,u_2)$ being alluded by the blue arrow and red arrow respectively in this figure. (a): $u_2$ is bounded away from $0$. (b): $u_2$ is close to $0$. (c): The ``degenerate'' case where $u_2=0$.}
\end{figure}
  We first deal with $u_2>0.$
    In this case, we have the following bound: 
        \begin{align}\nn
             P_1+P_2&\le \mathbbm{P}\Big(\frac{\tau-zu_1}{u_2}<\calN(0,1)< \frac{\tau-zv_1}{u_2}\Big)+\mathbbm{P}\Big(\frac{\tau+zv_1}{u_2}<\calN(0,1)<\frac{\tau+zu_1}{u_2}\Big)\\\nn
            & \le  \frac{1}{\sqrt{2\pi}}\exp\Big(-\frac{\min\{(\tau-zu_1)^2,(\tau-zv_1)^2\}}{2u_2^2}\Big) \frac{z(u_1-v_1)}{u_2}\\& +   \frac{1}{\sqrt{2\pi}}\exp\Big(-\frac{\min\{(\tau+zu_1)^2,(\tau+zv_1)^2\}}{2u_2^2}\Big) \frac{z(u_1-v_1)}{u_2}; \label{eq:P1P2upper}
        \end{align} 
     Since $\|\bp-\bq\|_2= u_1-v_1$, by Lemma \ref{lem:Puv} there exists some $c_0$ depending only on $(\alpha,\beta,\tau)$ such that $\sfP_{\bp,\bq}\ge c_0(u_1-v_1)$.  Substituting this and (\ref{eq:P1P2upper}) into the P.D.F. yields 
     \begin{align}\label{eq:pdf_relax1}
         f_{Z_i^{\bp,\bq}}(z)\le \frac{z\exp(-\frac{z^2}{2})}{c_0\pi u_2} \Big[\exp\Big(-\frac{\min\{(\tau-zu_1)^2,(\tau-zv_1)^2\}}{2u_2^2}\Big) +\exp\Big(-\frac{\min\{(\tau+zu_1)^2,(\tau+zv_1)^2\}}{2u_2^2}\Big)\Big].
     \end{align}
    It remains to show the right-hand side of (\ref{eq:pdf_relax1}) is dominated by some  Gaussian variable for large enough $z$. %Note that we can simply consider large enough $z$, since the part of $z=\calO(1)$ only makes $\calO(1)$ contribution to $\|Z_i^{\bp,\bq}\|_{\psi_2}$. 
    To this end, we separately treat the cases that $u_2$ bounded away from $0$ and $u_2$ being close to $0$. 
 
        \paragraph{(Case 1: $u_2\ge \frac{\alpha}{2}$, e.g., Figure \ref{fig:relaposi_D2}(a))} For any $z\in \mathbb{R}$, (\ref{eq:pdf_relax1}) implies $ f_{Z_i^{\bp,\bq}}(z)\le \frac{4 }{c_1\alpha\pi}z\exp\big(-\frac{z^2}{2}\big).$ When $z$ is larger than some  constant, $z\exp(-\frac{z^2}{2})\le \exp(-\frac{z^2}{4})$ holds,  implying that $\|Z_i^{\bp,\bq}\|_{\psi_2}\le C_1$ for some  constant $C_1$ which depends on $(\alpha,\beta,\tau)$.  
    \paragraph{(Case 2: $u_2\in(0,\frac{\alpha}{2})$, e.g., Figure \ref{fig:relaposi_D2}(b))}   From $\bp \in \mathbbm{A}_{\alpha,\beta}$ we have  $u_1^2+u_2^2\ge \alpha^2$, leading to $|u_1|>\frac{\alpha}{2}$. Note that $|v_1|\ge \frac{\alpha}{2}$ follows similarly. Thus, when $z\ge \frac{4\tau}{\alpha}$, we have $|\tau \pm zu_1|\ge |zu_1|-\tau \ge \tau$ and $|\tau \pm zv_1|  \ge |zv_1|-\tau \ge \tau$. Substituting into (\ref{eq:pdf_relax1}) gives $$
           f_{Z_i^{\bp,\bq}}(z)\le \frac{2}{c_0\pi}\frac{1}{u_2}\exp\big(-\frac{\tau^2}{2u_2^2}\big)\cdot z\exp\big(-\frac{z^2}{2}\big).$$  
           Note that for large enough $z$ we have $z\exp(-\frac{z^2}{2})\le \exp(-\frac{z^2}{4})$, and also $\sup_{u_2>0}\frac{1}{u_2}\exp(-\frac{\tau^2}{2u_2^2})\le C_2$ for some  constant $C_2$. Therefore, the probability tail of $Z_i^{\bp,\bq}$ is dominated by some Gaussian variable, implying the desired conclusion $\|Z_i^{\bp,\bq}\|_{\psi_2}=\calO(1)$.

    \paragraph{(Case 3: $u_2=0$, e.g., Figure \ref{fig:relaposi_D2}(c))} All that remains is to deal with $u_2=0$.    By Lemma \ref{lem:parameterization} $u_2=0$ occurs if and only if $\bp \parallel \bq$. By Lemma \ref{lem:T1pdfprob} we have $$P_1\le \mathbbm{P}(z>0,-\tau-zv_1<0<\tau-zv_1)=\mathbbm{P}(z>0,z|v_1|<\tau)$$ and $$P_2\le \mathbbm{P}(z>0,-\tau+zu_1<0<\tau+zu_1)=\mathbbm{P}(z>0,z|u_1|<\tau).$$ 
          Noticing $\|\bp\|_2 = |u_1|\ge \alpha$ and  $\|\bq\|_2 = |v_1| \ge\alpha$, hence we arrive 
at $P_1\le \mathbbm{1}\big(0<z<\frac{\tau}{|v_1|}\big)\le \mathbbm{1}\big(0<z<\frac{\tau}{\alpha}\big)$ and $P_2\le \mathbbm{1}\big(0<z<\frac{\tau}{|u_1|}\big)\le \mathbbm{1}\big(0<z<\frac{\tau}{\alpha}\big)$. Therefore, $Z_i^{\bp,\bq}$ is a random variable bounded as $|Z_i^{\bp,\bq}|\le \frac{\tau}{\alpha}$, which evidently possesses $\calO(1)$ sub-Gaussian norm.

    Note that the above discussions cover all possible cases of $(u_1,u_2,v_1)$. Hence, there exists some  constant $C_3$ such that $\|Z_i^{\bp,\bq}\|_{\psi_2}\le C_3$ for any $(\bp,\bq)\in\calN_{r,\delta_4}^{(2)}$, as claimed.   
\end{proof}

\paragraph{Concentration Bound on $T_1^{\bp,\bq}$:} We seek to establish tight and uniform concentration bound on $T_1^{\bp,\bq}$ by the following strategy:
\begin{itemize}
	[leftmargin=2ex,topsep=0.25ex]
	\setlength\itemsep{-0.1em}
    \item  Armed with Lemma \ref{lem:T1SG}, we first establish the   concentration   of $\frac{1}{m}\sum_{i\in r_{\bp,\bq}}Z_i^{\bp,\bq}$ about its mean, referred to as the   concentration bound conditioning on $\{|\bR_{\bp,\bq}|=r_{\bp,\bq}\}$;
    \item 
 We   further analyze the binomial random variable $|\bR_{\bp,\bq}|$ to get the unconditional concentration bound, and then take a union bound over $\calN_{r,\delta_4}^{(2)}$. 
\end{itemize}
\begin{lem}
    [Bounding $T_1^{\bp,\bq}$ over $\calN_{r,\delta_4}^{(2)}$] \label{lem:final_T11}Suppose $mr\ge C_0\scrH(\calC_{\alpha,\beta},r)$ for some sufficiently large $C_0$. Recall that $\calN_{r,\delta_4}^{(2)}$ is the constraint set defined in (\ref{eq:Nrxi}). Then for some  constants $C_1$ and $c$, the event 
    \begin{align}\label{eq:finaluniformT1}
         T_1^{\bp,\bq}\le C_1\eta\sqrt{\frac{ \|\bp-\bq\|_2\scrH(\calC_{\alpha,\beta},r)}{m}}+\mathsf{T}_\eta^{\bp,\bq},\quad\forall (\bp,\bq)\in \calN_{r,\delta_4}^{(2)}
    \end{align}
     holds   with probability at least $1-4\exp(-2\scrH(\calC_{\alpha,\beta},r))$, where $  \mathsf{T}_\eta^{\bp,\bq} := \big|\|\bp-\bq\|_2-\eta \sfP_{\bp,\bq}\mathbbm{E}(Z_i^{\bp,\bq})\big|.$  
\end{lem}
\begin{proof}
      We have argued that $T_{11}^{\bp,\bq}|\{|\bR_{\bp,\bq}|=r_{\bp,\bq}\}$ has the same distribution as $\frac{1}{m}\sum_{i=1}^{r_{\bp,\bq}}Z_{i}^{\bp,\bq}$, where $Z_{1}^{\bp,\bq},\cdots, Z_{r^{\bp,\bq}}^{\bp,\bq}$ are i.i.d. random variables satisfying $\|Z_i^{\bp,\bq}\|=\calO(1)$. By centering \cite[Lem. 2.6.8]{vershynin2018high}, we have $\|Z_i^{\bp,\bq}-\mathbbm{E}Z_i^{\bp,\bq}\|_{\psi_2}=\calO(1)$, and further   \cite[Prop. 2.6.1]{vershynin2018high} gives  
         \begin{align}
         \Big\|\frac{1}{m}\sum_{i=1}^{r_{\bp,\bq}}(Z_i^{\bp,\bq}-\mathbbm{E}Z_i^{\bp,\bq})\Big\|_{\psi_2}^2=\frac{1}{m^2}\Big\|\sum_{i=1}^{r_{\bp,\bq}}(Z_i^{\bp,\bq}-\mathbbm{E}Z_i^{\bp,\bq})\Big\|_{\psi_2}^2\lesssim \frac{1}{m^2}\sum_{i=1}^{r_{\bp,\bq}}\|Z_i^{\bp,\bq}-\mathbbm{E}Z_i^{\bp,\bq}\|_{\psi_2}^2\lesssim \frac{r_{\bp,\bq}}{m^2}.\label{eq:psi2sum}
     \end{align} 
The standard tail bound for sub-Gaussian variable implies   $$\mathbbm{P} \Big(\frac{1}{m}\sum_{i=1}^{r_{\bp,\bq}}(Z_i^{\bp,\bq}-\mathbbm{E}Z_i^{\bp,\bq})\ge \frac{r_{\bp,\bq}t}{m}\Big)\le 2\exp(-cr_{\bp,\bq}t^2)$$ for some  constant $c$ and any $t\ge 0$. Conditioning on $\{|\bR_{\bp,\bq}|=r_{\bp,\bq}\}$, by triangle inequality, $$\mathbbm{E}(T_{11}^{\bp,\bq}|\{|\bR_{\bp,\bq}|=r_{\bp,\bq}\})=\frac{r_{\bp,\bq}\cdot \mathbbm{E}(Z_i^{\bp,\bq})}{m}\,,\qquad \sfT_\eta^{\bp,\bq}=|\|\bp-\bq\|_2-\eta\sfP_{\bp,\bq}\mathbbm{E}(Z_i^{\bp,\bq})|$$ and that $T_{11}^{\bp,\bq}|\{|\bR_{\bp,\bq}|=r_{\bp,\bq}\}$ has the same distribution as $\frac{1}{m}\sum_{i=1}^{r_{\bp,\bq}}Z_i^{\bp,\bq}$, we obtain  
        \begin{align*}
       T_{1}^{\bp,\bq}  &\le \Big|\|\bp-\bq\|_2- \frac{\eta r_{\bp,\bq}\cdot\mathbbm{E}(Z_i^{\bp,\bq})}{m}\Big| +\eta \cdot\big|T_{11}^{\bp,\bq}-\mathbbm{E}(T_{11}^{\bp,\bq}|\{|\bR_{\bp,\bq}|= r_{\bp,\bq}\})\big| \\
        &\le \sfT_\eta^{\bp,\bq}+\frac{\eta \mathbbm{E}(Z_i^{\bp,\bq})|r_{\bp,\bq}-m\sfP_{\bp,\bq}|}{m} + \frac{\eta r_{\bp,\bq}t}{m},
    \end{align*} 
 where the second inequality holds with probability at least $1-2\exp(-cr_{\bp,\bq}t^2)$. Therefore, conditioning on $|\bR_{\bp,\bq}|=r_{\bp,\bq}>0$, 
 we set $t= C_1\sqrt{\frac{\scrH(\calC_{\alpha,\beta},r)}{r_{\bp,\bq}}}$ with sufficiently large $C_1$ to obtain 
        \begin{align}  
  & \mathbbm{P}\Big( T_1^{\bp,\bq} \le \sfT_\eta^{\bp,\bq}+\frac{\eta \mathbbm{E}(Z_i^{\bp,\bq})|r_{\bp,\bq}-m\sfP_{\bp,\bq}|}{m} + \frac{C_1\eta\sqrt{r_{\bp,\bq}\scrH(\calC_{\alpha,\beta},r)}}{m}\Big| |\bR_{\bp,\bq}|=r_{\bp,\bq}\Big)\nn\\&\ge 1-2\exp(-4\scrH(\calC_{\alpha,\beta},r)).  \label{eq:condi_devi_chosent}
    \end{align} 
    We then analyze the behaviour of $|\bR_{\bp,\bq}|\sim \text{Bin}(m,\sfP_{\bp,\bq})$. For any  $\delta\in(0,1)$, the weakened Chernoff bound in Lemma \ref{lem:chernoff} gives  $$\mathbbm{P}\big(\big||\bR_{\bp,\bq}|-m\sfP_{\bp,\bq}\big|\ge \delta m\sfP_{\bp,\bq}\big)\le 2\exp\big(-\frac{\delta^2m\sfP_{\bp,\bq}}{3}\big).$$ 
    Because $\sfP_{\bp,\bq}\ge c_0\|\bp-\bq\|_2\ge c_0r$ for some  constant $c_0$ (by Lemma \ref{lem:Puv} and $(\bp,\bq)\in \calN_{r,\delta_4}^{(2)}$) and $mr\gtrsim\scrH(\calC_{\alpha,\beta},r)$ with large enough implied constant, we have $ \sqrt{\frac{12\scrH(\calC_{\alpha,\beta},r)}{m\sfP_{\bp,\bq}}}\in (0,1)$. Thus, we can set $\delta= \sqrt{\frac{12\scrH(\calC_{\alpha,\beta},r)}{m\sfP_{\bp,\bq}}}$ in the Chernoff bound to obtain 
    \begin{align}\label{eq:chernoff_twosided}
       \mathbbm{P}\Big( ||\bR_{\bp,\bq}|-m\sfP_{\bp,\bq}| < \sqrt{12m \sfP_{\bp,\bq}\scrH(\calC_{\alpha,\beta},r)} \Big)\ge 1- 2\exp(-4\scrH(\calC_{\alpha,\beta},r)). 
    \end{align} 
    Also note that this high-probability event implies $|\bR_{\bp,\bq}|\le 2m\sfP_{\bp,\bq}$.
    Given $(\bp,\bq)\in\calN_{r,\delta_4}^{(2)}$, we combine (\ref{eq:condi_devi_chosent}) and (\ref{eq:chernoff_twosided}), together with $\mathbbm{E}(Z_i^{\bp,\bq})\lesssim 1$ implied by Lemma \ref{lem:T1SG} and $\sfP_{\bp,\bq}\lesssim \|\bp-\bq\|_2$ from Lemma \ref{lem:Puv}, we come to the unconditional bound $$T_1^{\bp,\bq}\le C_2\eta\sqrt{\frac{\|\bp-\bq\|_2\scrH(\calC_{\alpha,\beta},r)}{m}}+\sfT_\eta^{\bp,\bq}$$ that holds with probability at least $1-4\exp(-4\scrH(\calC_{\alpha,\beta},r))$. Taking a union bound over $(\bp,\bq)\in \calN_{r,\delta_4}^{(2)}$ completes the proof.     
\end{proof}

\paragraph{Calculating $\mathbbm{E}(Z_i^{\bp,\bq})$ and $\sfT_\eta^{\bp,\bq}$:} 
%The impact of the step size $\eta$ on the second term $\mathsf{T}_\eta^{\bp,\bq}$ as per (\ref{eq:eta_chosen}) needs further investigation. Specifically, due to the projection onto the non-convex set $\mathbbm{A}_{\alpha,\beta}$, there appears a multiplicative factor of $2$ in the right-hand side of (\ref{eq:single_bound}). Because $T_1^{\bp,\bq}$ is only one piece in our decomposition of $\|\bu-\bv-\eta \cdot \bh(\bu,\bv)\|_2$,     (approximately) there will be a term $2\mathsf{T}_\eta^{\bx^{(t-1)},\bx}$ appearing in the final bound on the right-hand side of (\ref{eq:single_bound}). Therefore, to establish an effective contraction, we must choose suitable step size $\eta=\eta_t$ such that $\sfT_\eta^{\bx^{(t-1)},\bx}<\frac{1}{2}\|\bx^{(t-1)}-\bx\|_2$ (approximately) holds. 
In the right-hand side of (\ref{eq:finaluniformT1}), $\eta$   has minimal impact on the first term since it appears as a leading factor. In   contrast, $\eta$ should be chosen such that another term $\sfT_\eta^{\bp,\bq}$ is small enough 
to establish an effective contraction. We note in advance that the selection of $\eta$ in our problem is indeed more delicate: in our analysis of $T_2^{\bp,\bq}$,   another term proportional to $\|\bp-\bq\|_2$ arises and   competes with $\sfT_\eta^{\bp,\bq}$ (see $\sfP_{\bp,\bq}|\mathbbm{E}\hat{Z}_i^{\bp,\bq}|$ in Lemma \ref{lem:integral2}). Here, it is imperative to accurately evaluate $\sfT_\eta^{\bp,\bq}=|\|\bp-\bq\|_2-\eta\sfP_{\bp,\bq}\mathbbm{E}(Z_i^{\bp,\bq})|$. We thus proceed to the calculations of  $\sfP_{\bp,\bq}\mathbbm{E}(Z_i^{\bp,\bq})$ that contain a number of technical ingredients: first,  it looks difficult (if not impossible) to seek an exact formula for $\sfP_{\bp,\bq}\mathbbm{E}(Z_i^{\bp,\bq})$ due to the intricate P.D.F. of $Z_i^{\bp,\bq}$, and our remedy is to allow for the imprecision of $o(\|\bp-\bq\|_2)$, which will not affect the selection of $\eta$ and can deliver several notable reductions; second, since the distribution of $Z_i^{\bp,\bq}$ essentially varies with the relative position of $(\bp,\bq)$,  we  have to provide separate treatments   to the two cases of $0<u_2<\hat{\xi}$ and $u_2>\hat{\xi}$ where   $\hat{\xi}=\frac{\tau}{2\sqrt{\log(u_1-v_1)^{-1}}}$. 

\begin{lem}
    [Calculating $\mathbbm{E}(Z_i^{\bp,\bq})$]\label{lem:cal_Tpq} We consider $(\bp,\bq)\in\calN_{r,\delta_4}^{(2)}$ parameterized as $\bp=u_1\bbeta_1+u_2\bbeta_2,\bq=v_1\bbeta_1+u_2\bbeta_2$ for some orthonormal $(\bbeta_1=\frac{\bp-\bq}{\|\bp-\bq\|_2},\bbeta_2)$ and some $u_1>v_1$ and $u_2\ge 0$, where $\|\bp-\bq\|_2=u_1-v_1$ is sufficiently small. By  (\ref{eq:explicit_uuv}), $u_1,v_1,u_2$ are determined by $(\bp,\bq)$ through $u_1=\frac{\langle \bp,\bp-\bq\rangle}{\|\bp-\bq\|_2}$, $v_1=\frac{\langle\bq, \bp-\bq\rangle}{\|\bp-\bq\|_2}$ and $u_2=\frac{(\|\bp\|_2^2\|\bq\|_2^2-(\bp^\top\bq)^2)^{1/2}}{\|\bp-\bq\|_2}$.  We further define $$g_\eta(a,b)=\sqrt{\frac{2}{\pi}}\exp\Big(-\frac{\tau^2}{2(a^2+b^2)}\Big)\frac{\tau^2a^2+b^2(a^2+b^2)}{(a^2+b^2)^{5/2}}.$$ Then we have $$\sfP_{\bp,\bq}\mathbbm{E}(Z_i^{\bp,\bq})=(u_1-v_1)\cdot(g_\eta(u_1,u_2)+\sfE_{\bp,\bq})$$ 
    where $\sfE_{\bp,\bq}$ is a higher order term satisfying $  |\sfE_{\bp,\bq}|\le C_0(u_1-v_1)\log((u_1-v_1)^{-1})$
    for some  constant $C_0$ only depending on $(\alpha,\beta,\tau)$.
\end{lem}
\begin{rem}
 Since $\|\bp-\bq\|_2=u_1-v_1$ is sufficiently small, we can write $\sfE_{\bp,\bq}=\calO(\sqrt{\|\bp-\bq\|_2})$ and then substitute $\sfP_{\bp,\bq}\mathbbm{E}(Z_i^{\bp,\bq})=(u_1-v_1)\cdot(g_\eta(u_1,u_2)+\sfE_{\bp,\bq})$ to obtain 
    \begin{align}\label{eq:Teta_calculate}
        \sfT_\eta^{\bp,\bq} = \|\bp-\bq\|_2\cdot \Big|1-\eta \big(g_\eta(u_1,u_2)+\calO(\sqrt{\|\bp-\bq\|_2})\big)\Big|. 
    \end{align}    
\end{rem}
\begin{proof}
    [Proof of Lemma \ref{lem:cal_Tpq}] To establish the desired $$\sfP_{\bp,\bq}\mathbbm{E}(Z_i^{\bp,\bq})=(u_1-v_1)\cdot(g_\eta(u_1,u_2)+\sfE_{\bp,\bq}),$$ it is sufficient to prove $$\sfP_{\bp,\bq}\mathbbm{E}(Z_i^{\bp,\bq}) = (u_1-v_1)\cdot \big(g_\eta(v_1,u_2)+\sfE_{\bp,\bq}\big).$$
This observation follows from Taylor's theorem: with sufficiently small $u_1-v_1$ and $\alpha^2\le u_1^2+u_2^2,v_1^2+u_2^2\le\beta^2$, we have $\|(\theta u_1+(1-\theta)v_1,u_2)\|_2\in[\frac{\alpha}{2},\beta]$ holding for all $\theta\in[0,1]$; we further observe that $|\frac{\partial g_\eta(u_1,u_2)}{\partial u_1}|$ is uniformly bounded over $\frac{\alpha^2}{4}\le u_1^2+u_2^2\le\beta^2$, then invoke Taylor's theorem to obtain $|g_\eta(v_1,u_2) - g_\eta(u_1,u_2)| = \calO(u_1-v_1)$. Note that $\calO(u_1-v_1)$ can be absorbed into $\sfE_{\bp,\bq}$, hence substituting this into $\sfP_{\bp,\bq}\mathbbm{E}(Z_i^{\bp,\bq}) = (u_1-v_1)\cdot \big(g_\eta(v_1,u_2)+\sfE_{\bp,\bq}\big)$ yields the desired $\sfP_{\bp,\bq}\mathbbm{E}(Z_i^{\bp,\bq})=(u_1-v_1)\cdot(g_\eta(u_1,u_2)+\sfE_{\bp,\bq})$.

Therefore, it remains to show $\sfP_{\bp,\bq}\mathbbm{E}(Z_i^{\bp,\bq}) = (u_1-v_1)\cdot \big(g_\eta(v_1,u_2)+\sfE_{\bp,\bq}\big)$. 
We first consider $u_2>0$ and will separately deal with $u_2=0$. By Lemma \ref{lem:T1pdfprob} we have 
   \begin{align}\label{eq:integral1_1}
       \sfP_{\bp,\bq}\mathbbm{E}(Z_i^{\bp,\bq}) = \sqrt{\frac{2}{\pi}}\int_0^\infty z\exp\Big(-\frac{z^2}{2}\Big)\big(P_1+P_2\big)\text{d}z
   \end{align}
 where under $u_2>0$ we have $$P_1=\mathbbm{P}\Big(\frac{\max\{\tau-zu_1,-\tau-zv_1\}}{u_2}<\calN(0,1)<\frac{\tau-zv_1}{u_2}\Big) $$ and $$P_2=\mathbbm{P}\Big(\frac{\max\{zu_1-\tau,zv_1+\tau\}}{u_2}<\calN(0,1)<\frac{\tau+zu_1}{u_2}\Big).$$  
   Now we introduce $$ P_1'  = \mathbbm{P}\Big(\frac{\tau-zu_1}{u_2}<\calN(0,1)<\frac{\tau-zv_1}{u_2}\Big)$$ and $$ P_2' = \mathbbm{P}\Big(\frac{zv_1+\tau}{u_2}<\calN(0,1)< \frac{zu_1+\tau}{u_2}\Big)$$ 
   and show that we can proceed with the simpler $P_1'+P_2'$ (instead of $P_1+P_2$) without affecting the desired claim. Specifically, observe that $-\tau-zv_1>\tau-zu_1$ holds if and only if $z>\frac{2\tau}{u_1-v_1}$, hence we have $P_1=P_1'$ for $z\le\frac{2\tau}{u_1-v_1}$. Similarly, we have $P_2=P_2'$ for $z\le \frac{2\tau}{u_1-v_1}$. Therefore, we can bound the deviation induced by substituting $P_1+P_2$ with $P_1'+P_2'$ in (\ref{eq:integral1_1}) as follows: 
        \begin{align*}
       &\Big|\sqrt{\frac{2}{\pi}}\int_0^\infty z\exp\Big(-\frac{z^2}{2}\Big)\big([P_1+P_2]-[P_1'+P_2']\big)\text{d}z\Big| \\&\le 2\sqrt{\frac{2}{\pi}}\int_{\frac{2\tau}{u_1-v_1}}^\infty z\exp\Big(-\frac{z^2}{2}\Big)\text{d}z\\& =2\sqrt{\frac{2}{\pi}}\exp\Big(-\frac{2\tau^2}{(u_1-v_1)^2}\Big).
   \end{align*}
    For sufficiently small $u_1-v_1$, such deviation is of order $\calO(u_1-v_1)$ and hence can be absorbed into $\sfE_{\bp,\bq}$. Therefore, we only need to show $$\sqrt{\frac{2}{\pi}}\int_0^\infty z\exp\Big(-\frac{z^2}{2}\Big)(P_1'+P_2')~\text{d}z=(u_1-v_1)\cdot(g_\eta(v_1,u_2)+\sfE_{\bp,\bq}).$$ To this end, we separately deal with  ``small $u_2$'' and ``large $u_2$''. 

   \subsubsection*{Case 1: $0\le u_2<\frac{\tau}{2\sqrt{\log(u_1-v_1)^{-1}}}$}
  In this case $u_2$ is sufficiently small, thus we can assume $|u_1|$ and $|v_1|$ are bounded away from $0$, say, $|u_1|\ge \frac{\alpha}{2}$ and $|v_1|>\frac{\alpha}{2}$. We can also assume $u_1$ and $v_1$ have the same sign since $u_1-v_1$ is small enough (see Figure \ref{fig:relaposi_D2}(b) for a graphical illustration). We will only consider $u_1>v_1\ge \frac{\alpha}{2}$, and arguments for the case $-\frac{\alpha}{2}\ge u_1>v_1$ are parallel. 
   We first show that $\sqrt{\frac{2}{\pi}}\int_0^\infty z\exp(-\frac{z^2}{2})P_2'~\text{d}z$ can be absorbed into $\sfE_{\bp,\bq}$. By Lemma \ref{lem:gaussian_tail}, for $z\ge 0$ we have  $$ P_2' \le \mathbbm{P}\Big(\calN(0,1)>\frac{\tau}{u_2}\Big)\le \frac{1}{2}\exp\Big(-\frac{\tau^2}{2u_2^2}\Big)\le \frac{(u_1-v_1)^2}{2},$$ 
   and thus $$
       \sqrt{\frac{2}{\pi}}\int_0^\infty z\exp\Big(-\frac{z^2}{2}\Big)P_2'~\text{d}z\le \frac{(u_1-v_1)^2}{\sqrt{2\pi}}\int_0^\infty z\exp\Big(-\frac{z^2}{2}\Big)\text{d}z=\frac{(u_1-v_1)^2}{\sqrt{2\pi}},$$ which can be incorporated into $(u_1-v_1)\sfE_{\bp,\bq}$.  
   Therefore, all that remains is to show $$\sqrt{\frac{2}{\pi}}\int_0^\infty z\exp\Big(-\frac{z^2}{2}\Big)P_1'~\text{d}z=(u_1-v_1)\cdot(g_\eta(v_1,u_2)+\sfE_{\bp,\bq}).$$ To this end, we proceed as follows: 
   \begin{subequations}
       \begin{align}\nn
           &\sqrt{\frac{2}{\pi}}\int_0^\infty z\exp\Big(-\frac{z^2}{2}\Big)P_1'~\text{d}z= \frac{1}{\pi}\int_0^\infty z\exp\Big(-\frac{z^2}{2}\Big)\Big(\int_{\frac{\tau-zu_1}{u_2}}^{\frac{\tau-zv_1}{u_2}}\exp\Big(-\frac{w^2}{2}\Big)\text{d}w\Big)\text{d}z\nn\\
           \label{eq:symmetry_e}&= \frac{1}{\pi}\int_0^\infty z\exp\Big(-\frac{z^2}{2}\Big)\Big(\int_{\frac{zv_1-\tau}{u_2}}^{\frac{zu_1-\tau}{u_2}}\exp\Big(-\frac{w^2}{2}\Big)\text{d}w\Big)\text{d}z \\
           \label{eq:exchange1}&= \frac{1}{\pi}\int_{-\frac{\tau}{u_2}}^\infty \exp\Big(-\frac{w^2}{2}\Big)\Big(\int_{\frac{wu_2+\tau}{u_1}}^{\frac{wu_2+\tau}{v_1}}z\exp\Big(-\frac{z^2}{2}\Big)\text{d}z\Big)\text{d}w\\&\label{eq:use_zez_inte}= 
           \frac{1}{\pi}\int_{-\frac{\tau}{u_2}}^\infty \exp\Big(-\frac{w^2}{2}\Big)\Big[\exp\Big(-\frac{(wu_2+\tau)^2}{2u_1^2}\Big)-\exp\Big(-\frac{(wu_2+\tau)^2}{2v_1^2}\Big)\Big]\text{d}w \\
          \label{eq:approxi2} &\approx \frac{1}{\pi}\int_{-\infty}^\infty \exp\Big(-\frac{w^2}{2}\Big)\Big[\exp\Big(-\frac{(wu_2+\tau)^2}{2u_1^2}\Big)-\exp\Big(-\frac{(wu_2+\tau)^2}{2v_1^2}\Big)\Big]\text{d}w\\
          \label{eq:shorthandFtu2}:&=  F_1(u_1,u_2)-F_1(v_1,u_2)
       \end{align}\label{eq:inte1_difference}
   \end{subequations}
   where (\ref{eq:symmetry_e}) is due to the symmetry of $\exp(-\frac{w^2}{2})$; in (\ref{eq:exchange1}) we formulate the domain of integration $\{z>0,~\frac{zv_1-\tau}{u_2}<w<\frac{zu_1-\tau}{u_2}\}$ as $\{w>-\frac{\tau}{u_2},~\frac{wu_2+\tau}{u_1}<z<\frac{wu_2+\tau}{v_1}\}$, and then exchange the order of integration; in (\ref{eq:use_zez_inte}) we use the closed-form integral $\int_a^b z\exp(-\frac{z^2}{2})\text{d}z = \exp(-\frac{a^2}{2})-\exp(-\frac{b^2}{2})$; (\ref{eq:approxi2}) is an imprecise step (thus we slightly abuse the notation ``$\approx$'') but will not affect our desired claim, since the difference between (\ref{eq:use_zez_inte}) and (\ref{eq:approxi2}) is bounded by $\frac{1}{\pi}\int_{-\infty}^{-\frac{\tau}{u_2}}\exp\big(-\frac{w^2}{2}\big)\text{d}w\le \sqrt{\frac{2}{\pi}}\exp\big(-\frac{\tau^2}{2u_2^2}\big)\le \sqrt{\frac{2}{\pi}}(u_1-v_1)^2,$ and can hence be incorporated into the higher order term $(u_1-v_1)\sfE_{\bp,\bq}$; in (\ref{eq:shorthandFtu2}) we introduce the shorthand $F_1(t,u_2):=\frac{1}{\pi}\int_{-\infty}^\infty \exp\big(-\frac{w^2}{2}\big)\exp\big(-\frac{(wu_2+\tau)^2}{2t^2}\big)\text{d}w,$ and since $t$ will be substituted by $u_1$ and $v_1$ we can proceed with the constraint $\frac{\alpha}{2}\le t\le \beta$. 
   We proceed to calculate $F_1(t,u_2)$ as follows: 
        \begin{align*}
           &F_1(t,u_2)= \frac{1}{\pi}\int_{-\infty}^\infty \exp\Big(-\frac{w^2}{2}-\frac{(wu_2+\tau)^2}{2t^2}\Big)\text{d}w\\
           &= \frac{1}{\pi}\exp\Big(-\frac{\tau^2}{2(t^2+u_2^2)}\Big)\int_{-\infty}^{\infty}\exp \Big(-\frac{t^2+u_2^2}{2t^2}\Big(w+\frac{\tau u_2}{t^2+u_2^2}\Big)^2\Big)\text{d}w\\
           & =  \frac{1}{\pi}\exp\Big(-\frac{\tau^2}{2(t^2+u_2^2)}\Big)\int_{-\infty}^{\infty}\exp\Big(-\frac{(t^2+u_2^2)w^2}{2t^2}\Big)\text{d}w =  \sqrt{\frac{2}{\pi}}\frac{t}{\sqrt{t^2+u_2^2}}\exp\Big(-\frac{\tau^2}{2(t^2+u_2^2)}\Big).
       \end{align*}
   Then we compute the derivatives to find $   \frac{\partial F_1(t,u_2)}{\partial t}  = g_\eta(t,u_2)$, and also 
$|\frac{\partial^2F_1(t,u_2)}{\partial t^2}|\le C_1$ holds uniformly for some   constant $C_1$ depending on $(\alpha,\beta,\tau)$, since $t\in [\frac{\alpha}{2},\beta]$ and $u_2\le\beta$. 
  Now we continue from (\ref{eq:inte1_difference}) and use Taylor's theorem to obtain 
       \begin{align*}
          &\sqrt{\frac{2}{\pi}}\int_0^\infty z\exp\Big(-\frac{z^2}{2}\Big)P_1'~\text{d}z = F_1(u_1,u_2)-F_1(v_1,u_2)=  \frac{\partial F_1(v_1,u_2)}{\partial t}(u_1-v_1)+ O \big((u_1-v_1)^2\big)\\
          &= (u_1-v_1)\cdot \Big(g_\eta(v_1,u_2)+\calO(u_1-v_1)\Big),
      \end{align*}
 as desired.

  \paragraph{The case  that $u_2=0$:} In this case, without loss of generality we assume $u_1>v_1\ge\frac{\alpha}{2}$ (see Figure \ref{fig:relaposi_D2}(c)), then we have 
\begin{align*}
    &\mathbbm{P}\big(\sign(|a_1u_1+a_2u_2|-\tau)\ne\sign(|a_1v_1+a_2u_2|-\tau)\big||a_1|=z\big)\\
    &=\mathbbm{P}\big(\sign(zu_1-\tau)\ne \sign(zv_1-\tau)\big)\\&= \mathbbm{1}\big(\frac{\tau}{u_1}\le z<\frac{\tau}{v_1}\big)
\end{align*}
   Thus, using the P.D.F. of $Z_i^{\bp,\bq}$ and applying Taylor's theorem to $F_1(t):=\exp(-\frac{\tau^2}{2t^2})$ which possesses uniformly bounded second derivative, we obtain  
      \begin{align*}
     &\sfP_{\bp,\bq}\mathbbm{E}(Z_i^{\bp,\bq})=\sqrt{\frac{2}{\pi}}\int_{\frac{\tau}{u_1}}^{\frac{\tau}{v_1}}z\exp\Big(-\frac{z^2}{2}\Big)\text{d}z\\&=\sqrt{\frac{2}{\pi}}\Big[\exp\Big(-\frac{\tau^2}{2u_1^2}\Big)-\exp\Big(-\frac{\tau^2}{2v_1^2}\Big)\Big]\\
     &= \sqrt{\frac{2}{\pi}}(u_1-v_1)\cdot\Big[\exp\Big(-\frac{\tau^2}{2v_1^2}\Big)\frac{\tau^2}{v_1^3}+\calO(u_1-v_1)\Big]\\&=(u_1-v_1)\cdot \big(g_\eta(v_1,0)+\calO(u_1-v_1)\big),
 \end{align*}
as desired.

       \subsubsection*{Case 2: $u_2\ge\frac{\tau}{2\sqrt{\log(u_1-v_1)^{-1}}}$}
See Figure \ref{fig:relaposi_D2}(b) for a graphical illustration of this case. 
As we pointed out, it suffices to show $$\sqrt{\frac{2}{\pi}}\int_0^\infty z\exp\Big(-\frac{z^2}{2}\Big)(P_1'+P_2')~\text{d}z=(u_1-v_1)(g_\eta(v_1,u_2)+\calO(u_1-v_1)).$$ Note that we can write $P_1'= \Phi\big(\frac{\tau-zv_1}{u_2}\big)-\Phi\big(\frac{\tau-zu_1}{u_2}\big)=\Phi\big(\frac{zu_1-\tau}{u_2}\big)-\Phi\big(\frac{zv_1-\tau}{u_2}\big)$ and $ P_2'=\Phi\big(\frac{zu_1+\tau}{u_2}\big)-\Phi\big(\frac{zv_1+\tau}{u_2}\big)$, where $\Phi(t)$ is the C.D.F. of $\calN(0,1)$. Then because 
$\Phi''(t)=-\frac{t}{\sqrt{2\pi}}\exp(-\frac{t^2}{2})$ is uniformly bounded by  $1$, 
Taylor's theorem gives $$P_1' = \Phi'\Big(\frac{zv_1-\tau}{u_2}\Big)\frac{z(u_1-v_1)}{u_2} + P_{1,r}'\quad \text{and} \quad P_2' = \Phi'\Big(\frac{zv_1+\tau}{u_2}\Big)\frac{z(u_1-v_1)}{u_2} + P_{2,r}',$$
 where the remainders satisfy $|P'_{1,r}|\le \frac{z^2(u_1-v_1)^2}{2u_2^2}$ and $|P'_{2,r}|\le \frac{z^2(u_1-v_1)^2}{2u_2^2}$. 
 We now show that the remainders only have minimal impact in the sense that it can only contribute to the higher order term $(u_1-v_1)\sfE_{\bp,\bq}$. Specifically, we can bound the effect of $P_{1,r}'$ by $$\sqrt{\frac{2}{\pi}}\int_0^\infty z\exp\big(-\frac{z^2}{2}\big)|P'_{1,r}|~\text{d}z\le  \sqrt{\frac{2}{\pi}}\frac{(u_1-v_1)^2}{2u_2^2}\int_0^\infty z^3\exp\big(-\frac{z^2}{2}\big)\text{d}z=(u_1-v_1)\cdot O\Big(\frac{u_1-v_1}{\log(u_1-v_1)^{-1}}\Big),$$ where in the last equality we substitute $u_2\ge \frac{\tau}{2\sqrt{\log(u_1-v_1)^{-1}}}$. Analogously, the same bound holds similarly for $P'_{2,r}$. Therefore, we can safely replace $P_1'+P_2'$ with $\Phi'(\frac{zv_1-\tau}{u_2})\frac{z(u_1-v_1)}{u_2} + \Phi'(\frac{zv_1+\tau}{u_2})\frac{z(u_1-v_1)}{u_2}$, and all that remains is to prove 
\begin{align*}
    \sqrt{\frac{2}{\pi}}\frac{u_1-v_1}{u_2}\int_0^\infty z^2\exp\Big(-\frac{z^2}{2}\Big)\Big[\Phi'\Big(\frac{zv_1-\tau}{u_2}\Big)+\Phi'\Big(\frac{zv_1+\tau}{u_2}\Big)\Big]\text{d}z = (u_1-v_1)\cdot \Big(g_\eta(v_1,u_2)+\calO(u_1-v_1)\Big)
\end{align*} 
To this end,
by substituting $\Phi'(t)=\frac{1}{\sqrt{2\pi}}\exp(-\frac{t^2}{2})$ we proceed as follows: 
\begin{subequations}
    \begin{align}\nn
    &\frac{u_1-v_1}{\pi u_2} \int_0^\infty z^2\exp\Big(-\frac{z^2}{2}\Big)\Big[\exp \Big(-\frac{(zv_1-\tau)^2}{2u_2^2}\Big)+\exp\Big(-\frac{(zv_1+\tau)^2}{2u_2^2}\Big)\Big]\text{d}z\\
    &= \frac{u_1-v_1}{\pi u_2} \int_{-\infty}^\infty z^2\exp\Big(-\frac{z^2}{2}\Big)\exp\Big(-\frac{(zv_1+\tau)^2}{2u_2^2}\Big)\text{d}z \nn\\\nn
    & = \frac{u_1-v_1}{\pi u_2}\exp \Big(-\frac{\tau^2}{2(u_2^2+v_1^2)}\Big)\int_{-\infty}^\infty z^2 \exp \Big(-\frac{u_2^2+v_1^2}{2u_2^2}\Big(z+\frac{\tau v_1}{u_2^2+v_1^2}\Big)^2\Big)\text{d}z \\
    & \label{eq:change_va_z_zminus}= \frac{u_1-v_1}{\pi u_2}\exp \Big(-\frac{\tau^2}{2(u_2^2+v_1^2)}\Big)\int_{-\infty}^\infty \Big(z-\frac{\tau v_1}{u_2^2+v_1^2}\Big)^2 \exp \Big(-\frac{(u_2^2+v_1^2)z^2}{2u_2^2}\Big)\text{d}z  \\\nn
    & = \frac{u_1-v_1}{\pi u_2}\exp \Big(-\frac{\tau^2}{2(u_2^2+v_1^2)}\Big)\int_{-\infty}^\infty \Big(z^2+\frac{\tau ^2v_1^2}{(u_2^2+v_1^2)^2}\Big) \exp \Big(-\frac{(u_2^2+v_1^2)z^2}{2u_2^2}\Big)\text{d}z \\
    & \label{eq:change_va_w_wminus} = \frac{u_1-v_1}{\pi \sqrt{u_2^2+v_1^2}}\exp\Big(-\frac{\tau^2}{2(u_2^2+v_1^2)}\Big) \int_{-\infty}^\infty \Big(\frac{u_2^2w^2}{u_2^2+v_1^2}+\frac{\tau^2v_1^2}{(u_2^2+v_1^2)^2}\Big)\exp\Big(-\frac{w^2}{2}\Big)\text{d}w\\\nn
    & = (u_1-v_1)\cdot \Big(\sqrt{\frac{2}{\pi}}\exp \Big(-\frac{\tau^2}{2(u_2^2+v_1^2)}\Big)\frac{u_2^2(u_2^2+v_1^2)+\tau^2v_1^2}{(u_2^2+v_1^2)^{5/2}}\Big)=(u_1-v_1)\cdot g_\eta(v_1,u_2),  
\end{align}
\end{subequations}
where in (\ref{eq:change_va_z_zminus}) we employ a change of the variable that substitutes $z+\frac{\tau v_1}{u_2^2+v_1^2}$ with $z$; in (\ref{eq:change_va_w_wminus}) we utilize a change of the variable $w=\frac{\sqrt{u_2^2+v_1^2}\cdot z}{u_2}$ to facilitate the comparison with $\calN(0,1)$. 
Thus, we obtain the claim again. 
Note that the above discussions cover all the situations, completing the proof. 
\end{proof}

% As  we shall see, the selection of the step size in our problem is rather delicate:   there will be an additional term of similar scaling competing with $\sfT_\eta^{\bp,\bq}$, thus to achieve a useful contraction   
% we must take $\eta$ that draws a reasonable trade-off between $\sfT_\eta^{\bp,\bq}$ and this additional term.
\subsubsection*{Step 1.2: Bounding $|T_2^{\bp,\bq}|$ over $\calN_{r,\delta_4}^{(2)}$}\label{app:boundT2}
Next, we proceed to bounding $|T_2^{\bp,\bq}|$. The strategy is similar to that for bounding $T_1^{\bp,\bq}$ but many details differ. 
Substituting $\bh_1(\bp,\bq)=\frac{1}{m}\sum_{i\in\bR_{\bp,\bq}}\sign(\ba_i^\top(\bp-\bq))\ba_i$ into $T_2^{\bp,\bq}$ yields $$T_2^{\bp,\bq} = \frac{1}{m}\sum_{i\in\bR_{\bp,\bq}}\sign(\ba_i^\top(\bp-\bq))\ba_i^\top\bbeta_2= \frac{1}{m}\sum_{i\in\bR_{\bp,\bq}}\sign(\ba_i^\top\bbeta_1)\ba_i^\top\bbeta_2$$ where $\bbeta_1=\frac{\bp-\bq}{\|\bp-\bq\|_2}$ and  $\bbeta_2$ are orthonormal, and they parameterize $(\bp,\bq)$ as $\bp=u_1\bbeta_1+u_2\bbeta_2$ and $\bq=v_1\bbeta_1+u_2\bbeta_2$ for some $(u_1,u_2,v_1)$ satisfying $u_1>v_1,~u_2\ge 0$.
Conditioning on $|\bR_{\bp,\bq}|=|\{i\in [m]:\sign(|\ba_i^\top\bp|-\tau)\ne\sign(|\ba_i^\top\bq|-\tau)\}|=r_{\bp,\bq},$ we identify the distribution of $T_2^{\bp,\bq}$ with     
$
    T_2^{\bp,\bq}|\{|\bR_{\bp,\bq}|=r_{\bp,\bq}\} \sim \frac{1}{m}\sum_{i=1}^{r_{\bp,\bq}}\hat{Z}_i^{\bp,\bq}$ 
where the i.i.d. random variables $\hat{Z}_1^{\bp,\bq},\cdots,\hat{Z}_{r_{\bp,\bq}}^{\bp,\bq}$ follow the conditional distribution $$\hat{Z}_i^{\bp,\bq} {\sim}\sign(\ba^\top \bbeta_1)\ba^\top\bbeta_2\big|\big\{\sign(|\ba^\top\bp|-\tau)\ne \sign(|\ba^\top\bq|-\tau)\big\}$$
with standard Gaussian vector $\ba\sim\calN(0,\bI_n)$.

\paragraph{The P.D.F. of $\hat{Z}_i^{\bp,\bq}$:} By $\bp=u_1\bbeta_1+u_2\bbeta_2$ and $\bq=v_1\bbeta_1+u_2\bbeta_2$ and the rotational invariance of $\ba\sim \calN(0,\bI_n)$ we have 
    \begin{align*}
    &\hat{Z}_i^{\bp,\bq} 
    =\sign(a_1)a_2 \Big| \big\{\sign(|a_1u_1+a_2u_2|-\tau)\ne\sign(|a_1v_1+a_2u_2|-\tau)\big\}\\
    &= \sign(a_1)a_2\Big| \big\{\sign(||a_1|u_1+\sign(a_1)a_2u_2|-\tau)\ne \sign(||a_1|v_1+\sign(a_1)a_2u_2|-\tau)\big\}\\
    &\sim a_2\Big|\big\{\sign(||a_1|u_1+a_2u_2|-\tau)\ne \sign(||a_1|v_1+a_2u_2|-\tau)\big\}, 
\end{align*} 
where $a_1,a_2$ are independent $\calN(0,1)$ variables, and we can replace $\sign(a_1)a_2$ by $a_2$ in the last line because $(a_1,a_2)$ and $(a_1,\sign(a_1)a_2)$ have the same distribution. We define the event $\tilde{E}=\{\sign(||a_1|u_1+a_2u_2|-\tau)\ne \sign(||a_1|v_1+a_2u_2|-\tau)\}$ and evaluate its probability by 
    \begin{align*}
        &\mathbbm{P}(\tilde{E}) = \mathbbm{P}\big(\sign(||a_1|u_1+a_2u_2|-\tau)\ne \sign(||a_1|v_1+a_2u_2|-\tau)\big)\\
        & = \mathbbm{P}\big(\sign(|a_1u_1+a_2\sign(a_1)u_2|-\tau)\ne\sign(|a_1v_1+a_2\sign(a_1)u_2|-\tau)\big) \\
       & =   \mathbbm{P}\big(\sign(|a_1u_1+a_2u_2|-\tau)\ne\sign(|a_1v_1+a_2u_2|-\tau)\big) =\sfP_{\bp,\bq}.
    \end{align*} 
 Now we can formulate the P.D.F. of $\hat{Z}_i^{\bp,\bq}$ by Bayes' Theorem as  
        \begin{align}
        &f_{\hat{Z}_i^{\bp,\bq}}(z)  = \frac{f_{a_2}(z)\cdot\mathbbm{P}(\tilde{E}|a_2=z)}{\mathbbm{P}(\tilde{E})}  = \frac{1}{\sqrt{2\pi}}\frac{\exp(-\frac{z^2}{2})\cdot\mathbbm{P}(\tilde{E}|a_2=z)}{\sfP_{\bp,\bq}}\nn\\&=  \frac{\exp(-\frac{z^2}{2})\cdot\mathbbm{P}(\sign(||a_1|u_1+zu_2|-\tau)\ne \sign(||a_1|v_1+zu_2|-\tau))}{\sqrt{2\pi}\cdot\sfP_{\bp,\bq}},~z\in\mathbb{R}.\label{eq:T2pdf}
    \end{align}  
%which is also the P.D.F. of $\hat{Z}_i^{\bp,\bq}$. 
\begin{lem}
    [The P.D.F. of $\hat{Z}_i^{\bp,\bq}$] \label{lem:prob_cal_T2}Given $(\bp,\bq)\in\calN_{r,\delta_4}^{(2)}$, for $z\in \mathbb{R}$   we have $\mathbbm{P}(\tilde{E}|a_2=z) = P_3+P_4$ where $$  P_3:= \mathbbm{P}\Big(|a_1|u_1+zu_2 >\tau,~-\tau< |a_1|v_1+zu_2 <\tau\Big)$$ and $$ P_4:= \mathbbm{P}\Big(-\tau<|a_1|u_1+zu_2<\tau,~|a_1|v_1+zu_2<-\tau\Big)$$ with $a_1\sim \calN(0,1)$. Substituting  into (\ref{eq:T2pdf}) yields the P.D.F. of $\hat{Z}_i^{\bp,\bq}$ as $$f_{\hat{Z}_i^{\bp,\bq}}(z) = \frac{\exp(-\frac{z^2}{2})\cdot[P_3+P_4]}{\sqrt{2\pi}\cdot\sfP_{\bp,\bq}},~~z\in\mathbb{R}.$$ 
\end{lem}
\begin{proof}
 We only need to calculate $$\mathbbm{P}(\tilde{E}|a_2=z)=\mathbbm{P}\big(\sign(||a_1|u_1+zu_2|-\tau)\ne \sign(||a_1|v_1+zu_2|-\tau)\big)$$  
    where $a_1\sim\calN(0,1)$. Then, $|a_1|u_1>|a_1|v_1$ holds almost surely. Thus, without affecting the subsequent probability calculation the event $\{\sign(||a_1|u_1+zu_2|-\tau)\ne \sign(||a_1|v_1+zu_2|-\tau)\}$ holds if and only if either 
    $$E_3:=\{|a_1|u_1+zu_2>\tau,~-\tau<|a_1|v_1+zu_2<\tau\}$$ or $$E_4:=\{-\tau<|a_1|u_1+zu_2<\tau,~|a_1|v_1+zu_2<-\tau\}$$ holds. Note that $E_3$ and $E_4$ are disjoint, and thus we can write $\mathbbm{P}(\tilde{E}|a_2=z) = \mathbbm{P}(E_3)+\mathbbm{P}(E_4):=P_3+P_4,$ as desired. 
 %Note that we ignore the cases $|a_1|u_1+zu_2=\pm\tau$  and   $|a_1|v_1+zu_2=\pm\tau$ in (\ref{eq:E3_E41}). 
 %This has no effect when $u_1\ne 0$ and $v_1\ne 0$, since the probability of the continuous random variables $|a_1|u_1$ and $|a_1|v_1$ taking specific value is zero. When $u_1=0$ or $v_1=0$, this might affect the correctness of $\mathbbm{P}(\tilde{E}|a_2=z)=P_3+P_4$ at most at two points, namely $z=\pm\frac{\tau}{u_2}$. Thus, this  
  %  will not affect our subsequent use of (\ref{eq:lem_hatZpdf}) as the P.D.F. of $\hat{Z}_i^{\bp,\bq}$. 
\end{proof}

\paragraph{Sub-Gaussianity of $\hat{Z}_i^{\bp,\bq}$:} Next, we   show in Lemma \ref{lem:SGT2hatZ} that $\hat{Z}_i^{\bp,\bq}$ is $\calO(1)$-sub-Gaussian. Again, a careful examination of all possible $(u_1,u_2,v_1)$ is in order. 
\begin{lem}
    [$\hat{Z}_i^{\bp,\bq}$ is $\calO(1)$-Sub-Gaussian]\label{lem:SGT2hatZ} Suppose $\delta_4\le\frac{\alpha}{2}$, there exists an  constant $C_0$ depending on $(\alpha,\beta,\tau)$ such that    $\|\hat{Z}_i^{\bp,\bq}\|_{\psi_2}\le C_0$ holds for any $(\bp,\bq)\in \calN_{r,\delta_4}^{(2)}$. 
\end{lem}
\begin{proof}
  By Lemma \ref{lem:prob_cal_T2}  the P.D.F. of $\hat{Z}_i^{\bp,\bq}$ is given by
    \begin{align}
        f_{\hat{Z}_i^{\bp,\bq}}(z)= \frac{\exp(-\frac{z^2}{2})(P_3+P_4)}{\sqrt{2\pi}\sfP_{\bp,\bq}}\le \frac{C_0\exp(-\frac{z^2}{2})(P_3+P_4)}{u_1-v_1},~z\in \mathbbm{R}, \label{eq:usePpq}
    \end{align}
    where  the inequality holds for some  constant $C_0$ because $\sfP_{\bp,\bq}\asymp \dist(\bp,\bq)=\|\bp-\bq\|_2$   for $(\bp,\bq)\in \calN_{r,\delta_4}^{(2)}$. We show $\|\hat{Z}_i^{\bp,\bq}\|_{\psi_2}=\calO(1)$  by discussing different cases of $(u_1,u_2,v_1)$.  
       \begin{figure}[ht!]
    \begin{centering}
        \includegraphics[width=0.2\columnwidth]{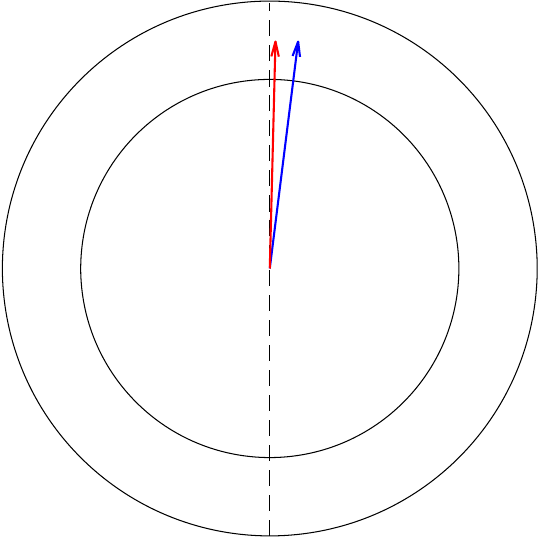} \quad \includegraphics[width=0.2\columnwidth]{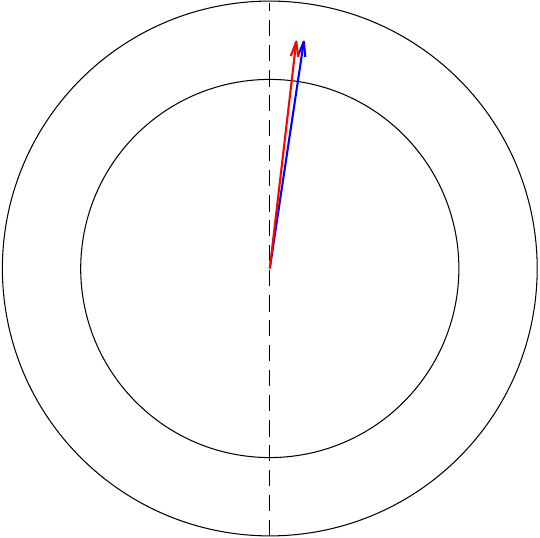}\quad  
        \includegraphics[width=0.2\columnwidth]{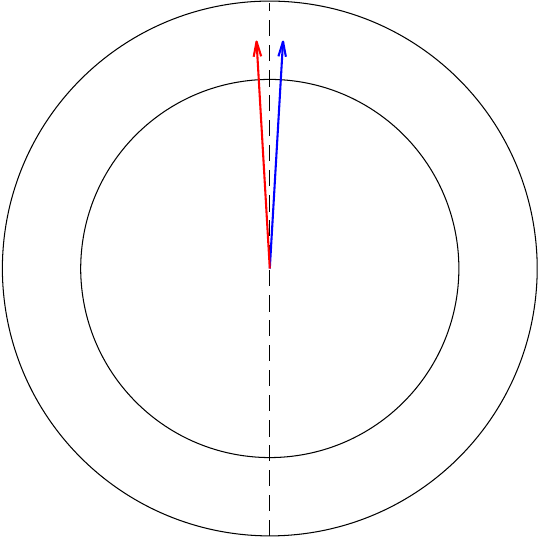}\quad
           \includegraphics[width=0.2\columnwidth]{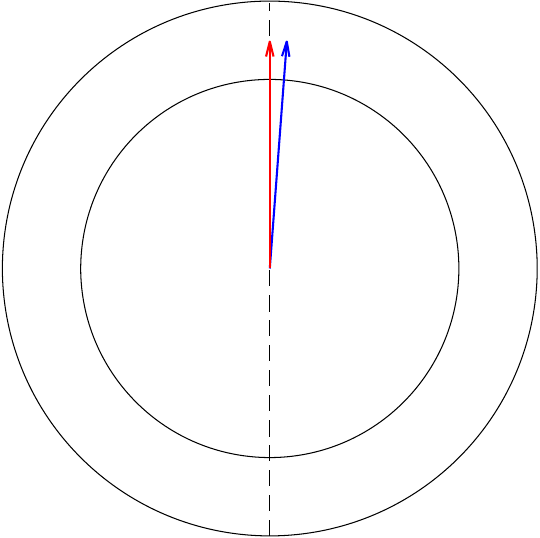}
     \par
     \par 
     \par
     \par
     \vspace{1mm}
    ~(a) \hspace{30mm} ~(b) \hspace{30mm} ~(c)  \hspace{32mm}(d)
           
        \par
    \end{centering}
    
     \caption{\label{fig:relaposi_D9} \small We parametrize $\bp,\bq\in \mathbbm{A}_{\alpha,\beta}$ as $(u_1,u_2)$ and $(v_1,u_2)$ that obey $u_1>v_1$ and $u_2\ge0$, with $(u_1,u_2)$ and $v_1,u_2$ being alluded to by the blue arrow and red arrow respectively in this figure. The above four sub-figures are concerned with different cases of ``$u_2$ being bounded away from $0$'' (say, $u_2\ge\frac{\alpha}{2}$): (a)  $0<v_1<\frac{u_1}{2}$; (b) $\frac{u_1}{2}\le v_1<u_1$; (c) $v_1<0<u_1$; (d) $v_1=0<u_1$.} %With $u_1-v_1=\|\bp-\bq\|_2$ being small enough,  we must have $|u_1|$ and $|v_1|$ are also sufficiently small in (a), (c) and (d) --- this might not be true for (b), see Figure \ref{fig:relaposi_D2}(a) for instance.
\end{figure}

    \subsubsection*{Case 1: $u_1>v_1>0$ or $v_1<u_1<0$}
    We  only analyze $u_1>v_1>0$ since the case $v_1<u_1<0$ can be done analogously. Under $u_1>v_1>0$, $P_3$ and $P_4$ in Lemma \ref{lem:prob_cal_T2} simplify to $  P_3   = \mathbbm{P}\big(|a_1|>\frac{\tau-zu_2}{u_1},~|a_1|>\frac{-\tau-zu_2}{v_1},~|a_1|<\frac{\tau-zu_2}{v_1}\big)$ and $ P_4   = \mathbbm{P}\big(|a_1|>\frac{-\tau-zu_2}{u_1},~|a_1|<\frac{\tau-zu_2}{u_1},~|a_1|<\frac{-\tau-zu_2}{v_1}\big)$.  
   We proceed to analyze the several cases separately. 
 
     \paragraph{(Case 1.1: $0\le u_2\le \frac{\alpha}{2}$)} In this case, $u_2$ might be very small or even exactly equal to $0$, see Figure \ref{fig:relaposi_D2}(b) and Figure \ref{fig:relaposi_D2}(c). Due to $u_1^2+u_2^2\ge \alpha^2$ and $v_1^2+u_2^2\ge \alpha^2$, we in turn have $u_1\ge \frac{\alpha}{2}$ and $v_1\ge \frac{\alpha}{2}$. Since the P.D.F. of $|a_1|$ is given by $f_{|a_1|}(z) = \sqrt{\frac{2}{\pi}}e^{-\frac{z^2}{2}}$ for $z\ge 0$, we have    
           \begin{gather*}
               P_3 \le \sqrt{\frac{2}{\pi}}\Big|\frac{\tau-zu_2}{v_1}-\frac{\tau-zu_2}{u_1}\Big|\le \sqrt{\frac{2}{\pi}}\frac{(u_1-v_1)|\tau-zu_2|}{u_1v_1}\le \sqrt{\frac{32}{\pi}}\frac{(\tau+\alpha|z|)(u_1-v_1)}{\alpha^2};\\
               P_4   \le \sqrt{\frac{2}{\pi}}\Big|\frac{-\tau-zu_2}{v_1}-\frac{-\tau-zu_2}{u_1}\Big|\le \sqrt{\frac{2}{\pi}}\frac{(u_1-v_1)|\tau+zu_2|}{u_1v_1}\le \sqrt{\frac{32}{\pi}}\frac{(\tau+\alpha|z|)(u_1-v_1)}{\alpha^2}. 
           \end{gather*}
        Substituting these into (\ref{eq:usePpq}) yields $
           f_{\hat{Z}_i^{\bp,\bq}}(z)\le \sqrt{\frac{2}{\pi}}\frac{8C_0}{\alpha^2}(\tau+\alpha|z|)\exp\big(-\frac{z^2}{2}\big)\le \exp\big(-\frac{z^2}{4}\big)$ where the second inequality holds when $z$ is larger than some  constant. This establishes $\|\hat{Z}_i^{\bp,\bq}\|_{\psi_2}=\calO(1)$.

       \paragraph{(Case 1.2: $u_2>\frac{\alpha}{2},~0<v_1<\frac{u_1}{2}$)} See Figure \ref{fig:relaposi_D9}(a) for a graphical illustration of this case. Note that $0<v_1<\frac{u_1}{2}$ gives $u_1-v_1\ge \frac{u_1}{2}$, hence (\ref{eq:usePpq}) implies \begin{align}\label{eq:case1_2pdf}
           f_{\hat{Z}_i^{\bp,\bq}}(z)\le \frac{2C_0\exp(-\frac{z^2}{2})(P_3+P_4)}{u_1},\quad z\in \mathbb{R}.
       \end{align} 
       Since $\frac{\tau-zu_2}{v_1}\le \frac{\tau-\frac{\alpha z}{2}}{v_1}\le 0$, we have $P_3 =0$ when $z\ge \frac{2\tau}{\alpha}$. When $z\le 0$, we use Lemma \ref{lem:gaussian_tail} to obtain $$ P_3\le \mathbbm{P}\big(|a_1|>\frac{\tau-zu_2}{u_1}\big)\le \mathbbm{P}\big(|a_1|>\frac{\tau}{u_1}\big)\le \exp\big(-\frac{\tau^2}{2u_1^2}\big).$$ Thus we come to $   P_3\le \exp\big(-\frac{\tau^2}{2u_1^2}\big)$ when $|z|>\frac{2\tau}{\alpha}$. 
       Similarly, we have $P_4=0$ when $z\ge 0$. When $z\le -\frac{4\tau}{\alpha}$, we have $$P_4 \le \mathbbm{P}\big(|a_1|>\frac{-\tau-zu_2}{u_1}\big)\le \mathbbm{P}\big(|a_1|> \frac{-\tau +\frac{4\tau}{\alpha}\cdot \frac{\alpha}{2}}{u_1}\big)= \mathbbm{P}\big(|a_1|>\frac{\tau}{u_1}\big)\le \exp\big(-\frac{\tau^2}{2u_1^2}\big).$$ Therefore, we arrive at $ P_4 \le \exp \big(-\frac{\tau^2}{2u_1^2}\big)$ when $|z|>\frac{4\tau}{\alpha}. $
       Substituting the bounds on $P_3,P_4$ into  (\ref{eq:case1_2pdf}) yields   $$   f_{\hat{Z}_i^{\bp,\bq}}(z) \le  \frac{4C_0}{u_1}\exp\big(-\frac{\tau^2}{2u_1^2}\big)\exp\big(-\frac{z^2}{2}\big)\le C_1e^{-\frac{z^2}{2}}$$ when $|z|>\frac{4\tau}{\alpha}$,       
       where the second inequality holds because $\frac{1}{u_1}\exp(-\frac{\tau^2}{2u_1^2})$ is uniformly bounded for $u_1>0$. Thus we obtain  $\|\hat{Z}_i^{\bp,\bq}\|_{\psi_2}=\calO(1)$.

        \paragraph{(Case 1.3: $u_2>\frac{\alpha}{2},~v_1\ge \frac{u_1}{2}$)} Readers can have Figure \ref{fig:relaposi_D9}(b) in mind. 
       As shown in Case 1.2, $P_3=P_4=0$ when $z>\frac{2\tau}{\alpha}$, we can hence focus on   $z<0$. Suppose $z\le-\frac{4\tau}{\alpha}$, then $\tau-zu_2>-\tau-zu_2 \ge -\tau + \frac{4\tau}{\alpha}\frac{\alpha}{2}\ge \tau$, and moreover
      \begin{subequations}
          \begin{align}\nn
              &P_3 \le \mathbbm{P}\left(\frac{\tau-zu_2}{u_1}<|a_1|<\frac{\tau-zu_2}{v_1}\right)\le\sqrt{\frac{2}{\pi}}\exp \Big(-\frac{(\tau-zu_2)^2}{2u_1^2}\Big)\Big|\frac{\tau-zu_2}{v_1}-\frac{\tau-zu_2}{u_1}\Big|\\
              &\le \sqrt{\frac{2}{\pi}}\exp\Big(-\frac{\tau^2}{2u_1^2}\Big)\frac{(u_1-v_1)|\tau-zu_2|}{u_1v_1} \le \sqrt{\frac{8}{\pi}}\exp \Big(-\frac{\tau^2}{2u_1^2}\Big)\frac{(u_1-v_1)(\tau+|z|\beta)}{u_1^2};\label{eq:usev1lower1}\\\nn
              &P_4 \le \mathbbm{P}\left(\frac{-\tau-zu_2}{u_1}<|a_1|<\frac{-\tau-zu_2}{v_1}\right)\le \sqrt{\frac{2}{\pi}}\exp\Big(-\frac{(\tau+zu_2)^2}{2u_1^2}\Big)\Big|\frac{-\tau-zu_2}{v_1}-\frac{-\tau-zu_2}{u_1}\Big|\\
              &\le \sqrt{\frac{2}{\pi}}\exp \Big(-\frac{\tau^2}{2u_1^2}\Big)\frac{(u_1-v_1)|\tau+zu_2|}{u_1v_1}\le\sqrt{\frac{8}{\pi}}\exp \Big(-\frac{\tau^2}{2u_1^2}\Big)\frac{(u_1-v_1)(\tau+|z|\beta)}{u_1^2}, \label{eq:usev1lower2}
          \end{align}
      \end{subequations}
      where in (\ref{eq:usev1lower1}) and (\ref{eq:usev1lower2}) we use $v_1\ge\frac{u_1}{2}$, $|u_2|\le\beta$ and triangle inequality. Combining these pieces, we arrive at $$  P_3+P_4\le 8\sqrt{\frac{2}{\pi}}\frac{\beta|z|}{u_1^2}\exp\big(-\frac{\tau^2}{2u_1^2}\big)(u_1-v_1)\le C_2 |z|(u_1-v_1)$$ when $|z|>\frac{4\tau}{\alpha}$,
          where we use $\tau+|z|\beta\le 2|z|\beta$ in the first inequality, and then the second inequality holds for some  constant $C_2$ since $\frac{1}{u_1^2}\exp(-\frac{\tau^2}{2u_1^2})$ is uniformly bounded  for $u_1>0$. Now we substitute this bound into (\ref{eq:usePpq}) to obtain 
          $$ f_{\hat{Z}_i^{\bp,\bq}}(z)\le C_0C_2|z|\exp\big(-\frac{z^2}{2}\big)\le \exp\big(-\frac{z^2}{4}\big)$$ when $|z|$ is large enough. We thus obtain $\|\hat{Z}_i^{\bp,\bq}\|_{\psi_2}=\calO(1)$.  
 
   \subsubsection*{Case 2: $u_1\ge 0\ge v_1$}
  In this case we have $\frac{\alpha}{2}\ge\delta_4\ge \|\bp-\bq\|_2=u_1-v_1 \ge \max\{u_1,|v_1|\}$, and hence we have $u_2>\frac{\alpha}{2}$ due to $u_1^2+u_2^2\ge \alpha^2$ and $v_1^2+u_2^2\ge\alpha^2$. Combining with (\ref{eq:usePpq}) yields 
  \begin{align}
      \label{eq:t2pdfdifferentsign}
      f_{\hat{Z}_i^{\bp,\bq}}(z)\le \frac{C_0\exp(-\frac{z^2}{2})(P_3+P_4)}{\max\{u_1,|v_1|\}},\quad z\in\mathbb{R}. 
  \end{align}
  We proceed to discuss the several cases.  
 
      \paragraph{(Case 2.1: $u_1>0>v_1$)} Readers can find a graphical illustration in Figure \ref{fig:relaposi_D9}(c). In this case we have $$  P_3:=\mathbbm{P}\Big(|a_1|>\frac{\tau-zu_2}{u_1},~|a_1|>\frac{zu_2-\tau}{|v_1|},~|a_1|<\frac{\tau+zu_2}{|v_1|}\Big)$$ and $$ P_4:=\mathbbm{P}\Big(|a_1|>\frac{-\tau-zu_2}{u_1},~|a_1|>\frac{\tau+zu_2}{|v_1|},~|a_1|<\frac{\tau-zu_2}{u_1}\Big).$$ Since $\frac{\tau+zu_2}{|v_1|}\le 0$,  it follows that $P_3=0$ when $z<-\frac{2\tau}{\alpha}$ . When $z\ge \frac{4\tau}{\alpha}$ by Lemma \ref{lem:gaussian_tail} we have $$ P_3\le \mathbbm{P}\big(|a_1|>\frac{zu_2-\tau}{|v_1|}\big)\le \mathbbm{P}\big(|a_1|>\frac{\tau}{|v_1|}\big)\le \exp\big(-\frac{\tau^2}{2v_1^2}\big).$$ Therefore, we come to $  P_3\le \exp \big(-\frac{\tau^2}{2v_1^2}\big)$  when $|z|\ge\frac{4\tau}{\alpha}.$ 
      Similarly, $P_4=0$ when $z\ge\frac{2\tau}{\alpha}$ because  $\frac{\tau-zu_2}{u_1}\le 0$. When $z\le -\frac{4\tau}{\alpha}$, by Lemma \ref{lem:gaussian_tail} we proceed as $$ P_4 \le\mathbbm{P}\big(|a_1|>\frac{-\tau-zu_2}{u_1}\big)\le \mathbbm{P}\big(|a_1|\ge \frac{\tau}{u_1}\big)\le \exp\big(-\frac{\tau^2}{2u_1^2}\big).$$ Thus, we arrive at $P_4\le \exp\big(-\frac{\tau^2}{2u_1^2}\big)$ when $|z|\ge \frac{4\tau}{\alpha}$. 
      Combining these bounds we obtain $$ P_3+P_4 \le 2\exp\big(-\frac{\tau^2}{2(\max\{u_1,|v_1|\})^2}\big),\qquad \forall |z|\ge\frac{4\tau}{\alpha}.$$   
    Substituting this into (\ref{eq:t2pdfdifferentsign}) yields $$f_{\hat{Z}_i^{\bp,\bq}}(z)\le \frac{2C_0\exp(-\frac{z^2}{2})}{\max\{u_1,|v_1|\}}\exp \big(-\frac{\tau^2}{2(\max\{u_1,|v_1|\})^2}\big)\le C_3\exp\big(-\frac{z^2}{2}\big)$$ for $|z|\ge \frac{4\tau}{\alpha}$,
    where the second inequality is due to the uniform boundedness of the function $g(t)=\frac{1}{t}\exp(-\frac{\tau^2}{2t^2})$. We come to $\|\hat{Z}_{i}^{\bp,\bq}\|_{\psi_2}=\calO(1)$. 
      
       \paragraph{(Case 2.2: $u_1>0=v_1$)} See Figure \ref{fig:relaposi_D9}(d) for an intuitive illustration. When $|z|>\frac{2\tau}{\alpha}$ we have $|zu_2|=|z(v_1^2+u_2^2)^{1/2}|\ge |\alpha z|>\tau$ and hence $P_3=0$. Also,  when $z\ge 0$ we have $zu_2\ge 0$, which leads to $P_4=0$. Moreover, when $z<-\frac{4\tau}{\alpha}$ it follows that $$   P_4\le \mathbbm{P}\big(|a_1|>\frac{-\tau-zu_2}{u_1}\big)\le \mathbbm{P}\big(|a_1|>\frac{\tau}{u_1}\big)\le \exp\big(-\frac{\tau^2}{2u_1^2}\big).$$ Therefore, we have shown $P_3+P_4 \le \exp(-\frac{\tau^2}{2u_1^2})$ when $|z|>\frac{4\tau}{\alpha}$, then analogously to the previous cases, we can show the desired claim by substituting this into (\ref{eq:t2pdfdifferentsign}). 
       \paragraph{(Case 2.3: $u_1=0>v_1$)} Note that $|zu_2|>\tau$ holds when $|z|>\frac{2\tau}{\alpha}$, which implies $P_4=0$. Also observe that $P_3=0$ when $z\le 0$. Moreover, when $z\ge \frac{4\tau}{\alpha}$ we can bound $P_3$ by Lemma \ref{lem:gaussian_tail} as $$P_3\le \mathbbm{P}\big(|a_1|v_1+zu_2<\tau\big) = \mathbbm{P}\Big(|a_1|>\frac{zu_2-\tau}{|v_1|}\Big)\le \mathbbm{P}\Big(|a_1|>\frac{\tau}{|v_1|}\Big)\le \exp\big(-\frac{\tau^2}{2v_1^2}\big).$$ Combining all the pieces, we arrive at $P_3+P_4\le \exp(-\frac{\tau^2}{2v_1^2})$ when $|z|\ge\frac{4\tau}{\alpha}$, and substituting this into (\ref{eq:t2pdfdifferentsign}) proves the desired claim. 
  Note that we have established  the claim in all the cases that might appear, which completes the proof. 
\end{proof}

%A direct outcome is the following conditional concentration of $T_2^{\bp,\bq}$ about its mean, which immediately implies Lemma \ref{lem:con_devi_T2}. 
%Furthermore, we remove the conditioning  by analyzing   $|\bR_{\bp,\bq}|$ with Chernoff bound, which yields the  unconditional bound in Lemma \ref{lem:uncondevi_T2}. Then we can get the final uniform bound in Lemma \ref{lem:final_bound_T2} by union bounding. 

\paragraph{Concentration Bound on $|T_2^{\bp,\bq}|$:} We have the following that is obtained by a proof strategy similar to that of  Lemma \ref{lem:final_T11}. 
\begin{lem}
    [Bounding $|T_2^{\bp,\bq}|$ over $\calN_{r,\delta_4}^{(2)}$] \label{lem:final_bound_T2}
Suppose $mr \ge C \scrH(\calC_{\alpha,\beta},r)$ holds for some sufficiently large $C$, then for some  constants $C_1$  the event 
       \begin{align}
        \label{eq:uncon_bound_T2_forall}
        |T_2^{\bp,\bq}|\le C_1\sqrt{\frac{ \|\bp-\bq\|_2 \scrH(\calC_{\alpha,\beta},r)}{m}} + \sfP_{\bp,\bq}|\mathbbm{E}(\hat{Z}_i^{\bp,\bq})|,\quad\forall (\bp,\bq)\in \calN_{r,\delta_4}^{(2)}
    \end{align}
holds with probability at least $1-4\exp(-2\scrH(\calC_{\alpha,\beta},r))$. 
\end{lem}
\begin{proof}
   We have shown that $T_2^{\bp,\bq}-\mathbbm{E}[T_2^{\bp,\bq}|\{|\bR_{\bp,\bq}|=r_{\bp,\bq}\}]$ has the same distribution as $\frac{1}{m}\sum_{i=1}^{r_{\bp,\bq}}\big(\hat{Z}_i^{\bp,\bq}-\mathbbm{E}[\hat{Z}_i^{\bp,\bq}]\big)$ where $\hat{Z}_1^{\bp,\bq},\cdots, \hat{Z}_{r_{\bp,\bq}}^{\bp,\bq}$ are i.i.d. random variables with $\calO(1)$ sub-Gaussian norm.  By centering (see \cite[Lem. 2.6.8]{vershynin2018high}) we obtain $$\|\hat{Z}_i^{\bp,\bq}-\mathbbm{E}[\hat{Z}_i^{\bp,\bq}]\|_{\psi_2}=\calO(1),$$ and similarly to (\ref{eq:psi2sum}) we can use \cite[Prop. 2.6.1]{vershynin2018high} to show $$
        \big\|\frac{1}{m}\sum_{i=1}^{r_{\bp,\bq}}(\hat{Z}_i^{\bp,\bq}-\mathbbm{E}[\hat{Z}_i^{\bp,\bq}])\big\|_{\psi_2}\lesssim\frac{\sqrt{r_{\bp,\bq}}}{m}.$$ By the tail bound of sub-Gaussian variable and $\mathbbm{E}[T_2^{\bp,\bq}|\{|\bR_{\bp,\bq}|=r_{\bp,\bq}\}]=\frac{r_{\bp,\bq}}{m}\mathbbm{E}[\hat{Z}_i^{\bp,\bq}]$,  we have 
        $$\mathbbm{P}\Big(\big|T_2^{\bp,\bq}-\frac{r_{\bp,\bq}\mathbbm{E}[\hat{Z}_i^{\bp,\bq}]}{m}\big|\le \frac{r_{\bp,\bq}t}{m}\Big| \{|\bR_{\bp,\bq}|=r_{\bp,\bq}\}\Big)\ge 1-2 \exp(-cr_{\bp,\bq}t^2),$$ which leads to $$|T_2^{\bp,\bq}|\le \frac{r_{\bp,\bq}}{m}(|\mathbbm{E}\hat{Z}_i^{\bp,\bq}|+t)\le \sfP_{\bp,\bq}|\mathbbm{E}\hat{Z}_i^{\bp,\bq}|+\Big|\frac{r_{\bp,\bq}}{m}-\sfP_{\bp,\bq}\Big|\cdot |\mathbbm{E}\hat{Z}_i^{\bp,\bq}|+\frac{r_{\bp,\bq}t}{m}$$ holding with the same probability when conditioning on $\{|\bR_{\bp,\bq}|=r_{\bp,\bq}\}$. For every positive $r_{\bp,\bq}$, we set $t=C_1\sqrt{\frac{\scrH(\calC_{\alpha,\beta},r)}{r_{\bp,\bq}}}$ with large enough $C_1$ to come to 
        \begin{align}
            &\mathbbm{P}\Big(|T_2^{\bp,\bq}|\le \sfP_{\bp,\bq}|\mathbbm{E}\hat{Z}_i^{\bp,\bq}|+
            \Big|\frac{|\bR_{\bp,\bq}|}{m}-\sfP_{\bp,\bq}\Big||\mathbbm{E}\hat{Z}_i^{\bp,\bq}|
            +\frac{C_1\sqrt{|\bR_{\bp,\bq}|\scrH(\calC_{\alpha,\beta},r)}}{m}\Big||\bR_{\bp,\bq}|\Big)\label{hphp}\nn\\&\ge 1-2\exp(-4\scrH(\calC_{\alpha,\beta},r)). 
        \end{align}
    Next, we analyze the behaviour of $|\bR_{\bp,\bq}|\sim \text{Bin}(m,\sfP_{\bp,\bq})$, which has been done in the proof of Lemma \ref{lem:final_T11}. Specifically, under $mr\ge C_2 \scrH(\calC_{\alpha,\beta},r)$ with sufficiently large $C_2$, $$||\bR_{\bp,\bq}|-m\sfP_{\bp,\bq}|<\sqrt{12m\sfP_{\bp,\bq}\scrH(\calC_{\alpha,\beta},r)}\quad\text{and}\quad|\bR_{\bp,\bq}|<2m\sfP_{\bp,\bq}$$ hold with probability at least $1-2\exp(-4\scrH(\calC_{\alpha,\beta},r))$. Taken collectively with (\ref{hphp}) and $\sfP_{\bp,\bq}\lesssim \|\bp-\bq\|_2$ and $\mathbbm{E}|\hat{Z}_i^{\bp,\bq}|\lesssim 1$, we arrive at the unconditional bound $$|T_2^{\bp,\bq}|\le C_3\sqrt{\frac{\|\bp-\bq\|_2\scrH(\calC_{\alpha,\beta},r)}{m}}+\sfP_{\bp,\bq}|\mathbbm{E}\hat{Z}_i^{\bp,\bq}|$$ holds with probability at least $1-4\exp(-4\scrH(\calC_{\alpha,\beta},r))$. Taking a union bound over $\calN_{r,\delta_4}^{(2)}$ yields the claim. 
\end{proof}

\paragraph{Calculating $\sfP_{\bp,\bq}|\mathbbm{E}\hat{Z}_i^{\bp,\bq}|$:} We provide
a closed-form expression for $\sfP_{\bp,\bq}|\mathbbm{E}\hat{Z}_i^{\bp,\bq}|$ which possibly contains an imprecision of $\calO(\|\bp-\bq\|_2^2)=o(\|\bp-\bq\|_2)$.   

\begin{lem}
    [Calculating $\sfP_{\bp,\bq}|\mathbbm{E}\hat{Z}_i^{\bp,\bq}|$] \label{lem:integral2}In the setting of Lemma \ref{lem:cal_Tpq}, we define 
$$h_\eta(a,b)=\sqrt{\frac{2}{\pi}}\exp\Big(-\frac{\tau^2}{2(a^2+b^2)}\Big)\frac{ab(a^2+b^2-\tau^2)}{(a^2+b^2)^{5/2}},$$
then we have $$\sfP_{\bp,\bq}|\mathbbm{E}\hat{Z}_i^{\bp,\bq}| = \|\bp-\bq\|_2\big| h_\eta(u_1,u_2)+\calO(\|\bp-\bq\|_2)\big|.$$ 
\end{lem}
\begin{proof}
       By the P.D.F. of $\hat{Z}_i^{\bp,\bq}$ and $P_3,P_4$ given in Lemma \ref{lem:prob_cal_T2} we have 
        \begin{align*} 
             \sfP_{\bp,\bq}|\mathbbm{E}(\hat{Z}_i^{\bp,\bq}) | =\Big|\int_{-\infty}^\infty \frac{z\exp(-\frac{z^2}{2})[P_3+P_4]}{\sqrt{2\pi}} ~\text{d}z \Big|= \Big|\int_0^\infty\int_{-\infty}^\infty \frac{z\exp(-\frac{z^2+w^2}{2})[\mathbbm{1}(E_3)+\mathbbm{1}(E_4)]}{\pi}~\text{d}z\text{d}w\Big|
        \end{align*}  
    where $E_3$ and $E_4$ are disjoint events given by $ E_3  = \big\{wu_1+zu_2>\tau,-\tau<wv_1+zu_2<\tau\big\}  $ and $E_4 =\big\{-\tau<wu_1+zu_2<\tau,wv_1+zu_2<-\tau\big\} $.
We first consider $u_2>0$ and then provide separate treatment to $u_2=0$. 

\subsubsection*{Case 1: $u_2>0$}
In this case we have $E_3     =\big\{\frac{\max\{\tau-wu_1,-\tau-wv_1\}}{u_2}<z<\frac{\tau-wv_1}{u_2}\big\}$ and $E_4 = \big\{\frac{-\tau - wu_1}{u_2}<z<\frac{\min\{\tau-wu_1,-\tau-wv_1\}}{u_2}\big\}$. 
   Now we show that it suffices to prove $$\Big|\int_0^\infty\int_{-\infty}^\infty \frac{z}{\pi}\exp\Big(-\frac{z^2+w^2}{2}\Big)[\mathbbm{1}(E_3')+\mathbbm{1}(E_4')]~\text{d}z\text{d}w\Big|=\|\bp-\bq\|_2|h_\eta(u_1,u_2)+\calO(\|\bp-\bq\|_2)|$$ where $E_3'$ and $E_4'$ are the simpler events defined as 
   $ E_3'   = \big\{\frac{\tau-wu_1}{u_2}<z<\frac{\tau-wv_1}{u_2}\big\}$ and $ E_4'   = \big\{\frac{-\tau-wu_1}{u_2}<z<\frac{-\tau-wv_1}{u_2}\big\}$.  
    To show this claim, we observe that $E_3=E_3'$ and $E_4=E_4'$ when $w<\frac{2\tau}{u_1-v_1}$. Hence, by  Lemma \ref{lem:gaussian_tail} the difference induced by replacing $\mathbbm{1}(E_3)+\mathbbm{1}(E_4)$ with $\mathbbm{1}(E_3')+\mathbbm{1}(E_4')$ is bounded by 
         \begin{align*}
            &\Big|\frac{1}{\pi}\int_0^\infty e^{-\frac{w^2}{2}}\int_{-\infty}^\infty ze^{-\frac{z^2}{2}}[\mathbbm{1}(E_3)+\mathbbm{1}(E_4)-\mathbbm{1}(E_3')-\mathbbm{1}(E_4')]\text{d}z\text{d}w\Big|\\
            &= \Big|\frac{1}{\pi}\int_{\frac{2\tau}{u_1-v_1}}^\infty e^{-\frac{w^2}{2}}\int_{-\infty}^\infty ze^{-\frac{z^2}{2}}[\mathbbm{1}(E_3)+\mathbbm{1}(E_4)-\mathbbm{1}(E_3')-\mathbbm{1}(E_4')]\text{d}z\text{d}w\Big|\\
            &\le \frac{2}{\pi}\int_{\frac{2\tau}{u_1-v_1}}^\infty e^{-\frac{w^2}{2}}\Big(\int_{-\infty}^\infty |ze^{-\frac{z^2}{2}}|\text{d}z\Big)\text{d}w  \\&= \frac{4}{\pi}\int_{\frac{2\tau}{u_1-v_1}}^\infty e^{-\frac{w^2}{2}}\text{d}w \\&\le 2\sqrt{\frac{2}{\pi}}\exp\Big(-\frac{\tau^2}{2(u_1-v_1)^2}\Big),
        \end{align*}
     With $\|\bp-\bq\|_2=u_1-v_1$ being sufficiently small, the difference   can be incorporated into $\calO(\|\bp-\bq\|_2^2)$. 
    Therefore, all that remains is to calculate  $|\int_0^\infty\int_{-\infty}^\infty \frac{z}{\pi}\exp(-\frac{z^2+w^2}{2})[\mathbbm{1}(E_3')+\mathbbm{1}(E_4')]~\text{d}z\text{d}w|$, and we proceed as
        \begin{align*}
       & \Big|\frac{1}{\pi}\int_0^\infty e^{-\frac{w^2}{2}} \int_{-\infty}^\infty ze^{-\frac{z^2}{2}}[\mathbbm{1}(E_3')+\mathbbm{1}(E_4')]~\text{d}z\text{d}w\Big|\\
       &=\Big|\frac{1}{\pi}\int_0^\infty e^{-\frac{w^2}{2}}\Big[\exp\Big(-\frac{(\tau-wu_1)^2}{2u_2^2}\Big)+\exp\Big(-\frac{(\tau+wu_1)^2}{2u_2^2}\Big) \nn \\
       &\quad\quad\quad\quad\quad\quad - \exp\Big(-\frac{(\tau-wv_1)^2}{2u_2^2}\Big)-\exp\Big(-\frac{(\tau+wv_1)^2}{2u_2^2}\Big)\Big]\text{d}w \Big| \\ 
       &=\Big|\frac{1}{\pi}\int_{-\infty}^\infty e^{-\frac{w^2}{2}}\Big[\exp\Big(-\frac{(\tau+wu_1)^2}{2u_2^2}\Big)-\exp\Big(-\frac{(\tau+wv_1)^2}{2u_2^2}\Big)\Big]\text{d}w \Big|\\
       &= |F_2(u_1,u_2)-F_2(v_1,u_2)|, 
    \end{align*}
    where in the last equality we introduce the short hand $F_2(t,u_2)$ and now further calculate it as 
             \begin{align*}
                &F_2(t,u_2) := \frac{1}{\pi}\int_{-\infty}^\infty \exp\Big(-\frac{w^2}{2}-\frac{(\tau+wt)^2}{2u_2^2}\Big)\text{d}w\\
                & = \frac{1}{\pi}\exp\Big(-\frac{\tau^2}{2(u_2^2+t^2)}\Big)\int_{-\infty}^\infty \exp\Big(-\frac{u_2^2+t^2}{2u_2^2}\Big(w+\frac{\tau t}{u_2^2+t^2}\Big)^2\Big)\text{d}w
                \\
                &= \sqrt{\frac{2}{\pi}}\frac{u_2}{\sqrt{u_2^2+t^2}}\exp\Big(-\frac{\tau^2}{2(u_2^2+t^2)}\Big).
            \end{align*}
            Note that we will use $t=u_1,v_1$ in $F_2(t,u_2)$, hence we can consider $F_2(t,u_2)$ under the constraint $t^2+u_2^2\in[\alpha^2,\beta^2]$. By calculating the  derivatives of $F_2(t,u_2)$ we find $  \frac{\partial F_2(t,u_2)}{\partial t}= -h_\eta(t,u_2)$, and also  $\frac{\partial^2F_2(t,u_2)}{\partial t^2}$ is uniformly bounded under the constraint $t^2+u_2^2\in[\alpha^2,\beta^2]$.   Hence,  Taylor's theorem gives $$ |F_2(u_1,u_2)-F_2(v_1,u_2) |=\big|\frac{\partial F_2(u_1,u_2)}{\partial t}(u_1-v_1) +\calO((u_1-v_1)^2)\big|=\|\bp-\bq\|_2 |h_\eta(u_1,u_2)+\calO(\|\bp-\bq\|_2)|,$$  as desired. 

           \subsubsection*{Case 2: $u_2=0$} 
        When  $u_2=0$, we have $\sfP_{\bp,\bq}|\mathbbm{E}(\hat{Z}_i^{\bp,\bq})|=|\int_{-\infty}^\infty \frac{z\exp(-\frac{z^2}{2})[P_3+P_4]}{\sqrt{2\pi}}~\text{d}z|=0$ since 
        $P_3= \mathbbm{P}_{a_1\sim\calN(0,1)}(|a_1|>\tau,~-\tau<|a_1|v_1<\tau)$ and $P_4= \mathbbm{P}_{a_1\sim\calN(0,1)}(-\tau<|a_1|u_1<\tau,~|a_1|v_1<-\tau)$ do not depend on $z$. Note that this case is accommodated by our statement since $h_\eta(u_1,0)=0$. The proof is complete.
\end{proof}

\subsubsection*{Step 1.3: Bounding $T_3^{\bp,\bq}$ over $\calN_{r,\delta_4}^{(2)}$}
Recall that $(\bbeta_1=\frac{\bp-\bq}{\|\bp-\bq\|_2},\bbeta_2)$ are orthonormal and can express $(\bp,\bq)$ as $\bp = u_1\bbeta_1+u_2\bbeta_2$ and $\bq = v_1\bbeta_1+u_2\bbeta_2$, and that we have $\bh_1(\bp,\bq)=\frac{1}{m}\sum_{i\in\bR_{\bp,\bq}}\sign(\ba_i^\top\bbeta_1)\ba_i$ and $\bh_1^\bot(\bp,\bq):= 
        \bh_1(\bp,\bq) - \langle \bh_1(\bp,\bq),\bbeta_1\rangle\bbeta_1 -\langle \bh_1(\bp,\bq),\bbeta_2\rangle\bbeta_2$. Some algebra finds
        \begin{align*}
    \bh_1^\bot(\bp,\bq) = \frac{1}{m}\sum_{i\in \bR_{\bp,\bq}}\sign(\ba_i^\top\bbeta_1)\Big[\ba_i-(\ba_i^\top\bbeta_1)\bbeta_1- (\ba_i^\top\bbeta_2)\bbeta_2\Big]:= \frac{1}{m}\sum_{i\in\bR_{\bp,\bq}} \bG_i^{\bp,\bq}; 
\end{align*}
          our goal is to control $T_3^{\bp,\bq}=\|\calP_{\calC_-}(\bh_1^\bot(\bp,\bq))\|_2$ uniformly over $(\bp,\bq)\in \calN_{r,\delta_4}^{(2)}$. 
Conditioning on $|\bR_{\bp,\bq}|=r_{\bp,\bq}$, the random vector $\bh_1^\bot(\bp,\bq)$ has the same distribution as $\frac{1}{m}\sum_{i=1}^{r_{\bp,\bq}}\tilde{\bZ}^{\bp,\bq}_i,$ where the i.i.d. random vectors $\tilde{\bZ}_1^{\bp,\bq},\cdots, \tilde{\bZ}_{r_{\bp,\bq}}^{\bp,\bq}$  follows the conditional distribution $\bG_{i}^{\bp,\bq}\big|\big\{\sign(|\ba_i^\top\bp|-\tau)\ne \sign(|\ba_i^\top\bq|-\tau)\big\}$.

\paragraph{Distribution of $\tilde{\bZ}_i^{\bp,\bq}$:} There exists an orthogonal matrix $\bP$ whose first two rows are $\bbeta_1$ and $\bbeta_2$ respectively, such that $\bP \bp = u_1\be_1+u_2\be_2$ and $\bP\bq = v_1\be_1+u_2\be_2$.   
We then let $\tilde{\ba}_i=\bP\ba_i$ to arrive at $\bG_{i}^{\bp,\bq} = \bP^\top \sign(\tilde{\ba}_i^\top\be_1)\big[\tilde{\ba}_i -(\tilde{\ba}_i^\top\be_1)\be_1-(\tilde{\ba}_i^\top\be_2)\be_2\big]$ and $$\{\sign(|\ba_i^\top\bp|-\tau)\ne \sign(|\ba_i^\top\bq|-\tau)\} = \{\sign(|\tilde{\ba}_i^\top (u_1\be_1+u_2\be_2)|-\tau)\ne \sign(|\tilde{\ba}_i^\top (v_1\be_1+u_2\be_2)|-\tau)\}.$$
By rotational invariance we have $\tilde{\ba}_i:=(a_i)_{i\in[n]}\sim\calN(0,\bI_n)$, 
then  $\tilde{\bZ}_i^{\bp,\bq}$ follows the conditional distribution  
    \begin{align*}
    &\bP^\top \sign(a_1)(0,0,a_3,\cdots,a_n)^\top\big|\{\sign(|a_1u_1+a_2u_2|-\tau)\ne\sign(|a_1v_1+a_2u_2|-\tau)\}
    \\ 
        &\sim\bP^\top (0,0,a_3,\cdots,a_n)^\top \big|\{\sign(|a_1u_1+a_2u_2|-\tau)\ne\sign(|a_1v_1+a_2u_2|-\tau)\}
        \\&\sim \bP^\top(0,0,a_3,\cdots ,a_n)^\top. 
    \end{align*}  
Letting $\ba_i$ be i.i.d. $\calN(0,\bI_n)$ vectors and $\ba_i^{[3:n]}$ be the vector obtained by setting the first two entries of $\ba_i$ to $0$, we come to 
\begin{align}\label{eq:h1bot_dis}
 \bh_1^\bot(\bp,\bq)\big|\{|\bR_{\bp,\bq}|=r_{\bp,\bq}\}\sim \frac{1}{m}\sum_{i=1}^{r_{\bp,\bq}}\bP^\top\ba_i^{[3:n]} \sim \bP^\top\begin{bmatrix}
        0 \\
        0 \\
        \calN(0,\frac{r_{\bp,\bq}}{m^2}\bI_{n-2})
    \end{bmatrix}. 
\end{align}
\paragraph{Concentration Bound on $T_3^{\bp,\bq}$:} We have the following.  
%Equipped with (\ref{eq:h1bot_dis}), we are now able to establish a conditional bound on $T_3^{\bp,\bq}=\|\bh_1^\bot(\bp,\bq)\|_2$. 
\begin{lem}
    [Bounding $T_3^{\bp,\bq}$ over $\calN_{r,\delta_4}^{(2)}$] \label{lem:T3_final} Suppose $mr\ge C_0 \scrH(\calC_{\alpha,\beta},r)$ for some large enough $C_0$, then for some  constant $C$ the event $$
          T_3^{\bp,\bq}\le C_1\sqrt{\frac{\|\bp-\bq\|_2[\omega^2(\calC_{(1)})+\scrH(\calC_{\alpha,\beta},r)]}{m}},~~\forall \,(\bp,\bq)\in \calN_{r,\delta_4}^{(2)}$$ 
    holds   with probability at least $1-4\exp(-2\scrH(\calC_{\alpha,\beta},r))$. 
\end{lem}
\begin{proof}
    For some orthogonal matrix $\bP$, the conditional distribution   $\bh_1^\bot(\bp,\bq)|\{|\bR_{\bp,\bq}|=r_{\bp,\bq}\}$ is identical to  $\frac{\sqrt{r_{\bp,\bq}}}{m}\bP^\top \ba'$ with $\ba'\sim(0,0,\calN(0,\bI_{n-2}))^\top$ in (\ref{eq:h1bot_dis}). By Lemma \ref{lem:pro_closec} we can write $$T_3^{\bp,\bq}=\|\calP_{\calC_-}(\bh_1^\bot(\bp,\bq))\|_2=\sup_{\bw\in \calC_-\cap \mathbb{B}_2^n}\bw^\top \bh_1^\bot(\bp,\bq)=\sup_{\bw\in\calC_{(1)}}\bw^\top\bh_1^\bot(\bp,\bq).$$ Thus, when conditioning on $\{|\bR_{\bp,\bq}|=r_{\bp,\bq}\}$ we have $$\|\calP_{\calC_-}(\bh_1^\bot(\bp,\bq))\|_2 \sim \frac{\sqrt{r_{\bp,\bq}}}{m}\sup_{\bw\in\calC_{(1)}}\bw^\top \bP^\top\ba' = \frac{\sqrt{r_{\bp,\bq}}}{m}\sup_{\bw\in \bP\calC_{(1)}}\bw^\top \ba'. $$
For any $\bw_1,\bw_2 \in \bP\calC_{(1)}$ we have $\|\bw_1^\top\ba'-\bw_2^\top\ba'\|_{\psi_2} \le \|\ba'\|_{\psi_2}\|\bw_1-\bw_2\|_2\le C \|\bw_1-\bw_2\|_2$ for some absolute constant $C.$
  Therefore, for any $t\ge 0$, Lemma \ref{lem:tala} yields    $$\mathbbm{P}\Big(\sup_{\bw\in\bP\calC_{(1)}}\bw^\top\ba' \le C_1 [\omega(\calC_{(1)})+t]|\{|\bR_{\bp,\bq}|=r_{\bp,\bq}\}\Big)\ge 1-2\exp(-t^2).$$ Taken collectively,  we set $t=2\sqrt{\scrH(\calC_{\alpha,\beta},r)}$ to arrive at 
  \begin{align*}
      \mathbbm{P}\Big(T_3^{\bp,\bq}\le \frac{C_1\sqrt{r_{\bp,\bq}}\cdot\big[\omega(\calC_{(1)})+2\sqrt{\scrH(\calC_{\alpha,\beta},r)}\big]}{m}\Big|\{|\bR_{\bp,\bq}|=r_{\bp,\bq}\}\Big)\ge 1-2\exp(-4\scrH(\calC_{\alpha,\beta},r)).
  \end{align*}
  The behaviour of $|\bR_{\bp,\bq}|\sim \text{Bin}(m,\sfP_{\bp,\bq})$ has been done in the proof of Lemma \ref{lem:final_T11}. Particularly, if $mr\ge C_0 \scrH(\calC_{\alpha,\beta},r)$ with sufficiently large $C_0$, then     $$||\bR_{\bp,\bq}|-m\sfP_{\bp,\bq}|\le\sqrt{12m\sfP_{\bp,\bq}\scrH(\calC_{\alpha,\beta},r)}\quad\text{and}\quad|\bR_{\bp,\bq}|\le2m\sfP_{\bp,\bq}$$ hold with probability at least $1-2\exp(-4\scrH(\calC_{\alpha,\beta},r))$. Combining the conditional bound on $T_3^{\bp,\bq}$ and the behaviour of $|\bR_{\bp,\bq}|$, along with $\sfP_{\bp,\bq}\lesssim \|\bp-\bq\|_2$, we obtain   the unconditional bound $$T_3^{\bp,\bq} \le C_2\sqrt{\frac{\|\bp-\bq\|_2[\omega^2(\calC_{(1)})+\scrH(\calC_{\alpha,\beta},r)]}{m}}$$ which holds with probability at least $1-4\exp(-4\scrH(\calC_{\alpha,\beta},r))$. Taking a union bound over $(\bp,\bq)\in\calN_{r,\delta_4}^{(2)}$ completes the proof. 
\end{proof}

 \subsubsection*{Step 1.4: Bounding $\|\calP_{\calC_-}(\bh_2(\bp,\bq))\|_2$ over $\calN_{r,\delta_4}^{(2)}$}
We seek to control $\sup_{(\bp,\bq)\in\calN_{r,\delta_4}^{(2)}}\|\calP_{\calC_-}(\bh_2(\bp,\bq))\|_2$ where $  \bh_2(\bp,\bq):=\frac{1}{m}\sum_{i\in \bR_{\bp,\bq}\cap \bL_{\bp,\bq}} \big[\sign(\ba_i^\top(\bp+\bq))-\sign(\ba_i^\top(\bp-\bq))\big]\ba_i$ with the index sets being given by $  \bR_{\bp,\bq}: = \big\{i\in[m]:\sign(|\ba_i^\top\bp|-\tau)\neq \sign(|\ba_i^\top\bq|-\tau)\big\}$ and $ \bL_{\bp,\bq}: = \big\{i\in[m]:\sign(\ba_i^\top\bp)\ne\sign(\ba_i^\top\bq)\big\}.$
As we shall see, $\|\bh_2(\bp,\bq)\|_2$ is a negligible higher-order term compared to $\|\bp-\bq-\eta\cdot \bh_1(\bp,\bq)\|_2$.  
The key intuition is that $|\bR_{\bp,\bq}\cap\bL_{\bp,\bq}|$ is essentially smaller than $|\bR_{\bp,\bq}|$. More precisely, note that $(\bp,\bq)\in \calN_{r,\delta_4}^{(2)}$ has small $\|\bp-\bq\|_2$; compared to $\mathbbm{P}(i\in\bR_{\bp,\bq}) =\sfP_{\bp,\bq}$ that is order-wise equivalent to $\|\bp-\bq\|_2$ (Lemma \ref{lem:Puv}), the ``double separation probability'' $\mathbbm{P}(i\in \bR_{\bp,\bq}\cap\bL_{\bp,\bq})$ decays with $\|\bp-\bq\|_2$ in a much faster exponential way (Lemma \ref{lem:double_sepa_prob}).  

\paragraph{Probability of Double Separation:}
 For $\bu,\bv\in \mathbb{R}^n$ and $\ba\sim\calN(0,\bI_n)$, the probability of $\bu$ and $\bv$ being separated by both the hyperplane $\calH_{\ba,0}$ and the phaseless hyperplane $\calH_{|\ba|}$ can be formulated as 
    \begin{align*}
        \sfP^{(2)}_{\bu,\bv} = \mathbbm{P}\begin{pmatrix}
            \sign(|\ba^\top\bu|-\tau)\ne \sign(|\ba^\top\bv|-\tau)\\
            \sign(\ba^\top\bu)\ne \sign(\ba^\top\bv)
        \end{pmatrix}, 
    \end{align*}
    referred to as double separation probability henceforth. 

\begin{lem} \label{lem:double_sepa_prob}  For any $\bu,\bv\in\mathbb{R}^n$ we have $ \sfP_{\bu,\bv}^{(2)}\le 4\exp\big(-\frac{\tau^2}{2\|\bu-\bv\|_2^2}\big).$ 
\end{lem}
\begin{proof}
     We proceed with the following estimations: 
    \begin{subequations}
        \begin{align}\nn
              &\sfP^{(2)}_{\bu,\bv} = \mathbbm{P}\begin{pmatrix}
            \sign(|\ba^\top\bu|-\tau)\ne \sign(|\ba^\top\bv|-\tau)\\
            \sign(\ba^\top\bu)\ne \sign(\ba^\top\bv)
        \end{pmatrix} \\
        &= \mathbbm{P}\begin{pmatrix}
            \sign(\ba^\top\bu-\tau)\ne \sign(-\ba^\top\bv-\tau)\\
           \ba^\top\bu \ge 0,~\ba^\top\bv <0
        \end{pmatrix}  +\mathbbm{P}\begin{pmatrix}
            \sign(-\ba^\top\bu-\tau)\ne \sign(\ba^\top\bv-\tau)\\
           \ba^\top\bu < 0,~\ba^\top\bv \ge0
        \end{pmatrix}\label{eq:lawtotal1} \\
        & = \mathbbm{P}\begin{pmatrix}
            \ba^\top\bu\ge \tau,~\ba^\top\bv>-\tau\\
           \ba^\top\bu \ge 0,~\ba^\top\bv <0
        \end{pmatrix}+\mathbbm{P}\begin{pmatrix}
           \ba^\top\bu< \tau,~\ba^\top\bv\le -\tau\\
           \ba^\top\bu \ge 0,~\ba^\top\bv <0
        \end{pmatrix} \nn
        \\\label{eq:lawtotal2}&\quad\quad+\mathbbm{P}\begin{pmatrix}
          \ba^\top\bu\le -\tau,~\ba^\top\bv<\tau\\
           \ba^\top\bu < 0,~\ba^\top\bv \ge0
        \end{pmatrix} +\mathbbm{P}\begin{pmatrix}
            \ba^\top\bu> -\tau,~\ba^\top\bv\ge\tau\\
           \ba^\top\bu < 0,~\ba^\top\bv \ge0
        \end{pmatrix} \\
        &\le  \mathbbm{P}\begin{pmatrix}
            \ba^\top\bu\ge\tau\\
            \ba^\top\bv<0
        \end{pmatrix}+\mathbbm{P}\begin{pmatrix}
            \ba^\top\bu\ge0\\
            \ba^\top\bv\le -\tau
        \end{pmatrix}+\mathbbm{P}\begin{pmatrix}
            \ba^\top\bu\le -\tau\\
            \ba^\top\bv\ge 0
        \end{pmatrix}+\mathbbm{P}\begin{pmatrix}
            \ba^\top\bu<0\\
            \ba^\top\bv\ge \tau
        \end{pmatrix}\\
        &\le 4 \mathbbm{P}(|\ba^\top(\bu-\bv)|\ge \tau)  \le 8 \mathbbm{P}\Big(\calN(0,1) \ge \frac{\tau}{\|\bu-\bv\|_2}\Big) \le 4\exp\Big(-\frac{\tau^2}{2\|\bu-\bv\|_2^2}\Big), \label{eq:usegautail} 
         \end{align} 
    \end{subequations}
    where (\ref{eq:lawtotal1}) and (\ref{eq:lawtotal2}) are due to the law of total probability, and in (\ref{eq:usegautail}) we use the Gaussian tail bound in Lemma \ref{lem:gaussian_tail}. The proof is complete.  
\end{proof}

\paragraph{Bounding $\|\calP_{\calC_-}(\bh_2(\bp,\bq))\|_2$:} We shall show that a logarithmically small $\delta_4$ ensures $\sfP_{\bp,\bq}^{(2)}\ll \sfP_{\bp,\bq}$ for $(\bp,\bq)\in\calN_{r,\delta_4}^{(2)}$ and consequently allows for a tight enough bound on $\|\calP_{\calC_-}(\bh_2(\bp,\bq))\|_2$.

%Inspired by Claim \ref{claim:aic_low}, we already specified $\xi\asymp (\log(\frac{m}{n}))^{-1/2}$ when formulating the desired PLL-AIC in (\ref{eq:desired_aic_low}). Nonetheless, it is not immediately obvious why we use such logarithmically small $\xi$, for instance, one may ask whether    $\xi$  smaller than some  constant (adopted in \cite{candes2015phase,zhang2017nonconvex}) would be sufficient here. In fact, such logarithmically small $\xi$ is enlightened by our attempts to control $\|\bh_2(\bp,\bq)\|_2$, and it can not be relaxed to $\xi\asymp 1$. To elucidate this point,  we  first establish the desired (\ref{eq:desired_h2bound}) under an additional (unjustified) assumption (\ref{eq:capscaling}). 

\begin{lem}\label{lem:final_h2}Suppose $mr\ge C [\scrH(\calC_{\alpha,\beta},r)+\omega^2(\calC_{(1)})]$ and $\delta_4 \le \frac{c}{\log^{1/2}(r^{-1})}$ hold for some large enough $C$ and small enough $c$, then the event 
$$\|\calP_{\calC_-}(\bh_2(\bp,\bq))\|_2 \le C_1r,\qquad\forall\,  (\bp,\bq)\in\calN_{r,\delta_4}^{(2)}$$ holds with probability at least $1-3\exp(-cmr)$.   
\end{lem}
\begin{proof}
 We first bound $|\bR_{\bp,\bq}\cap\bL_{\bp,\bq}|$ for a fixed $(\bp,\bq)\in \calN_{r,\delta_4}^{(2)}$, which satisfies $\|\bp-\bq\|_2\le 2\delta_4$ for small enough $c$. Note that  $|\bR_{\bp,\bq}\cap \bL_{\bp,\bq}|\sim \text{Bin}(m,\sfP_{\bp,\bq}^{(2)})$ where the double separation probability is bounded by  Lemma \ref{lem:double_sepa_prob} as $$  \sfP_{\bp,\bq}^{(2)}\le 4\exp\big(-\frac{\tau^2(2\delta_4)^{-2}}{2}\big)\le r^{C_2}$$      
     for some  constant $C_2$ that can be made sufficiently large,   due to    $\delta_4\le\frac{c}{\log^{1/2}(r^{-1})}$ with small enough $c$. By the first statement in Lemma \ref{lem:chernoff} we obtain 
       \begin{align*}
           &\mathbbm{P}\Big(|\bR_{\bp,\bq}\cap\bL_{\bp,\bq}|\ge \frac{C_1mr}{\sqrt{\log( r^{-1})}}\Big)\\&\le \mathbbm{P}\Big(\text{Bin}(m,r^{C_2})\ge \frac{C_1mr}{\sqrt{\log( r^{-1})}}\Big)  \\ 
           &\le \exp\Big(-\frac{C_1mr}{2\sqrt{\log (r^{-1})}}\log \Big(\frac{1}{r^{C_2-2}}\Big)\Big) \\&\le \exp (-c_3mr\sqrt{\log (r^{-1})}). 
       \end{align*} 
 We then take a union bound over $(\bp,\bq)\in\calN_{r,\delta_4}^{(2)}$, showing that  $$\sup_{(\bp,\bq)\in\calN_{r,\delta_4}^{(2)}}|\bR_{\bp,\bq}\cap\bL_{\bp,\bq}| < \frac{C_1mr}{\sqrt{\log(r^{-1})}}$$
   holds with probability at least $$1- \exp(2\scrH(\calC_{\alpha,\beta},r)-c_3mr\sqrt{\log (r^{-1})})\ge 1-\exp(-c_4mr\sqrt{\log (r^{-1})}),$$ where we use $mr\ge C\scrH(\calC_{\alpha,\beta},r)$ with large enough $C$.   Armed with the uniform bound on $|\bR_{\bp,\bq}\cap\bL_{\bp,\bq}|$, we use Lemma \ref{lem:pro_closec} to get started and proceed as 
         \begin{align}\nn
             &\sup_{(\bp,\bq)\in\calN_{r,\delta_4}^{(2)}}\|\bh_2(\bp,\bq)\|_2 =\sup_{(\bp,\bq)\in\calN_{r,\delta_4}^{(2)}}\sup_{\bw\in \calC_{(1)}} \bw^\top \bh_2(\bp,\bq) \\
             &\quad\nn = \sup_{(\bp,\bq)\in\calN_{r,\delta_4}^{(2)}}\sup_{\bw\in \calC_{(1)}}\frac{1}{m}\sum_{i\in \bR_{\bp,\bq}\cap\bL_{\bp,\bq}}\big[\sign(\ba_i^\top(\bp+\bq))-\sign(\ba_i^\top(\bp-\bq))\big]\ba_i^\top\bw \\
             &\quad \nn \le   \sup_{(\bp,\bq)\in\calN_{r,\delta_4}^{(2)}}\sup_{\bw\in \calC_{(1)}}\frac{2}{m}\sum_{i\in \bR_{\bp,\bq}\cap\bL_{\bp,\bq}}|\ba_i^\top\bw|\le \sup_{\bw \in \calC_{(1)}}\max_{\substack{\calS\subset [m]\\ |\calS|\le \frac{C_1mr}{ \log ^{1/2}(r^{-1})}}} \frac{2}{m}\sum_{i\in\calS}|\ba_i^\top\bw| \\\nn
             &\quad = \frac{2C_1r}{\sqrt{\log(r^{-1})}}\sup_{\bw \in \calC_{(1)}}\max_{\substack{\calS\subset [m]\\ |\calS|\le  \frac{C_1mr}{\log^{1/2}(r^{-1})}}} \frac{\sqrt{\log( r^{-1})}}{C_1mr}\sum_{i\in\calS}|\ba_i^\top\bw| \\
             &\quad \nn\le \frac{2C_1r}{\sqrt{\log (r^{-1})}}\sup_{\bw \in \calC_{(1)}}\max_{\substack{\calS\subset [m]\\ |\calS|\le  \frac{C_1mr}{\log^{1/2}(r^{-1})}}}\Big(\frac{\sqrt{\log (r^{-1})}}{C_1mr}\sum_{i\in\calS}|\ba_i^\top\bw|^2\Big)^{1/2}\\
             &\quad \label{eq:boundh2_ranvec}\le C_2r \Big(\Big(\frac{\omega^2(\calC_{(1)})}{mr\sqrt{\log r^{-1}}}\Big)^{1/2}+C_3\Big)\le C_2r(1+C_3)
         \end{align} 
     where    in the last line, the first inequality  holds with probability at least $1-2\exp(-cmr )$ due to Lemma \ref{lem:max_ell_sum},  the second inequality follows from $mr\ge C\omega^2(\calC_{(1)})$ with large enough $C$. The proof is complete. 
\end{proof}
\subsection{Bounding $\|\calP_{\calC_-}(\bh(\bu,\bv)-\bh(\bu_1,\bv_1))\|_2$ (Step 2 in Large-distance regime)}\label{app:difference_h}

We make use of (\ref{eq:decom_hminush}), Lemma \ref{lem:pro_closec} and triangle inequality to obtain
\begin{align}\label{startminus}
    \|\calP_{\calC_-}(\bh(\bu,\bv)-\bh(\bu_1,\bv_1))\|_2= \sup_{\bw\in\calC_{(1)}}~\bw^\top\big(\bh(\bu,\bv)-\bh(\bu_1,\bv_1)\big)\le\sup_{\bw\in\calC_{(1)}}~\frac{3}{m}\sum_{i\in I_{\bu,\bu_1,\bv,\bv_1}}|\ba_i^\top\bw|
\end{align}
where we let $I_{\bu,\bu_1,\bv,\bv_1}:=\bR_{\bv,\bv_1}\cup \bR_{\bu,\bu_1}\cup \bL_{\bu,\bu_1}$. We then use technique similar to bounding $\|\calP_{\calC_-}(\bh_2(\bp,\bq))\|_2$, that is, we first establish the uniform bound on $|I_{\bu,\bu_1,\bv,\bv_1}|$ that can be achieved by uniformly controlling $|\bR_{\bv,\bv_1}|$, $|\bR_{\bu,\bu_1}|$ and $|\bL_{\bu,\bu_1}|$, and then invoke Lemma \ref{lem:max_ell_sum}. Interestingly, bounding $|\bL_{\bp,\bq}|$ is due to the {\it local binary embedding for 1-bit compressed sensing} in \cite{oymak2015near}, which we restate in Lemma \ref{lem:binaryembed}. 
 \begin{lem}
    \label{lem:final_hdifference}  Suppose $mr\ge C[\omega^2(\calC_{(1)})+\scrH(\calC_{\alpha,\beta},r)]$ holds with large enough $C_1$, then the event $$\|\calP_{\calC_-}(\bh(\bu,\bv)-\bh(\bu_1,\bv_1)) \|_2\le C_2r\log(r^{-1}),~~\forall \bu,\bv\in \calC_{\alpha,\beta}$$
    holds with probability at least $1- \exp(-c_3mr)$.  
\end{lem}
\begin{proof}
We first bound $|\bR_{\bv,\bv_1}|$ and $|\bR_{\bu,\bu_1}|$. 
By Theorem \ref{thm:local_embed} with $\calK=\calC_{\alpha,\beta}=\calC\cap \mathbbm{A}_{\alpha,\beta}$ and parameters $(r',r)$ therein set as $(2r,C_0r\sqrt{\log (r^{-1})})$ for large enough constant $C_0$, 
 we come to the following: if $mr \ge C_1 [\omega^2(\calC_{(1)})+\scrH(\calC_{\alpha,\beta},r)]$ for sufficiently large $C_1$, which is assumed in our statement, then the event $$|\bR_{\bp,\bq}| \le C_2mr\sqrt{\log (r^{-1})},\qquad\forall \,\bp,\bq\in\calC_{\alpha,\beta}~\text{ obeying }~\|\bp-\bq\|_2\le r$$ 
holds with probability at least $1-\exp(-c_0mr)$. Next, we proceed to bound $|\bL_{\bu,\bu_1}|$.  By  $\|\bu_1-\bu\|_2\le r$ with small enough $r$ and Lemma \ref{lem:norm_equa}, we obtain $$
         \Big\|\frac{\bu_1}{\|\bu_1\|_2} - \frac{\bu}{\|\bu\|_2}\Big\|_2 =\dist_{\rm d}(\bu_1,\bu)\le \frac{1}{\alpha}\dist(\bu_1,\bu) \le \frac{r}{\alpha}.$$ Then we invoke Lemma \ref{lem:binaryembed} with $\calK=\calC^*:=\calC\cap \mathbb{S}^{n-1}$ and parameters $(r',r)$ therein set as $(\frac{r}{\alpha},C_4r\sqrt{\log (r^{-1})})$ with sufficiently large $C_4$, yielding the following: if $mr\ge C_5[\omega^2(\calC_{(1)})+\scrH(\calC^*,\frac{r}{\alpha})]$ with sufficiently large $C_5$, which is assumed in our statement, then 
    $$ d_H\big(\sign(\bA\bp),\sign(\bA\bq)\big) \le C_6 r\sqrt{\log (r^{-1})}m,\qquad\forall \,\bp,\bq\in \calC^*~\text{ obeying }~\|\bp-\bq\|_2 \le \frac{r}{\alpha}$$
     holds with probability at least $1- \exp(-c_7r\sqrt{\log (r^{-1})}\cdot m)$.  Combined with $\|\frac{\bu_1}{\|\bu_1\|_2}-\frac{\bu}{\|\bu\|_2}\|_2\le \frac{r}{\alpha}$ this event implies $$|\bL_{\bu,\bu_1}|=d_H\big(\sign\big(\bA\frac{\bu_1}{\|\bu_1\|_2}\big),\sign\big(\bA\frac{\bu}{\|\bu\|_2}\big)\big)\le C_6mr\sqrt{\log (r^{-1})},\qquad \forall\,\bu\in\calC_{\alpha,\beta}.$$ Therefore, universally for all $\bu,\bv\in\calC_{\alpha,\beta}$ we have $$|I_{\bu,\bu_1,\bv,\bv_1}|\le |\bR_{\bv,\bv_1}|+|\bR_{\bu,\bu_1}|+|\bL_{\bu,\bu_1}|\le C_8mr\sqrt{\log(r^{-1})}$$ holds with promised probability. With this bound we continue from (\ref{startminus}) to proceed as   
     \begin{align}
        &\|\calP_{\calC_-}(\bh(\bu,\bv)-\bh(\bu_1,\bv_1))\|_2\le \sup_{\bw\in\calC_{(1)}}~\frac{3}{m}\sum_{i\in I_{\bu,\bu_1,\bv,\bv_1}}|\ba_i^\top\bw|\nn \\
        & \le 3 C_8r\sqrt{\log(r^{-1})}\cdot\sup_{\bw\in \calC_{(1)}}\max_{\substack{\calS\subset [m]\\|\calS|\le C_8mr\sqrt{\log(r^{-1})}}} \frac{1}{C_8mr\sqrt{\log(r^{-1})}}\sum_{i\in \calS}|\ba_i^\top\bw| \nn\\
        & \le 3 C_8r\sqrt{\log(r^{-1})}\cdot\sup_{\bw\in \calC_{(1)}}\max_{\substack{\calS\subset [m]\\|\calS|\le C_8mr\sqrt{\log (r^{-1})}}} \Big(\frac{1}{C_8mr\sqrt{\log(r^{-1})}}\sum_{i\in \calS}|\ba_i^\top\bw|^2\Big)^{1/2} \nn\\ 
        &\le C_9 r\sqrt{\log(r^{-1})}\Big(\frac{\omega(\calC_{(1)})}{\sqrt{mr}}+\sqrt{\log(r^{-1})}\Big) \le C_{10} r\log (r^{-1}),\label{eq:boundT4vec}
    \end{align} 
where in the last line, the first inequality holds with the promised probability due to Lemma \ref{lem:max_ell_sum}, the second inequality follows from $mr \gtrsim \omega^2(\calC_{(1)})$ that we assume. The proof is complete.   
\end{proof}

\subsection{Bounding $\|\calP_{\calC_-}(\bh(\bu,\bv))\|_2$ (Small-distance regime)}\label{app:small_distance}
For controlling $\|\bh(\bu,\bv)\|_2$ in the small-distance regime, 
a tight enough bound can again be derived by the techniques  in (\ref{eq:boundh2_ranvec}) and (\ref{eq:boundT4vec}). 
\begin{lem} 
\label{lem:smalldis} 
Suppose $mr\ge C_1\scrH(\calC_{\alpha,\beta},r)$ holds with large enough $C_1$, then  $$
      \|\calP_{\calC_-}(\bh(\bu,\bv))\|_2\le C_2r\log(r^{-1}),\qquad\forall\, \bu,\bv\in\calC_{\alpha,\beta}~\,\text{obeying}~\,\|\bu-\bv\|_2\le 3r$$
    holds with probability at least $1- \exp(-c_3mr).$ 
\end{lem}
 \begin{proof}
 We use Lemma \ref{lem:pro_closec} and triangle inequality to get started: $$\|\calP_{\calC_-}(\bh(\bu,\bv))\|_2=\sup_{\bw\in \calC_{(1)}}\bw^\top\bh(\bu,\bv)\le \sup_{\bw\in\calC_{(1)}}\frac{1}{m}\sum_{i\in\bR_{\bu,\bv}}|\ba_i^\top\bw|.$$
The remainder of the proof is parallel to the ones for Lemmas \ref{lem:final_h2}--\ref{lem:final_hdifference}. First we use Theorem \ref{thm:local_embed} to establish $|\bR_{\bu,\bv}|\le C_1mr\sqrt{\log(r^{-1})} ~(\forall \bu,\bv\in\calC_{\alpha,\beta}\text{ obeying }\|\bu-\bv\|_2\le 3r)$ that holds with the promised probability. Second, we utilize Lemma \ref{lem:max_ell_sum} to show the claim. We omit further details. 
 \end{proof}
\subsection{The Proof of Theorem \ref{thm:raic}}\label{app:combine_conclude}
 \begin{proof}  
    Our goal is to control $\|\calP_{\calC_-}(\bu-\bv-\eta \cdot \bh(\bu,\bv))\|_2$ uniformly for $\bu,\bv\in\calC_{\alpha,\beta}=\calC\cap \mathbbm{A}_{\alpha,\beta}$ obeying $\|\bu-\bv\|_2\le\delta_4$. In the main text, we construct $\calN_r$ as a minimal $r$-net of $\calC_{\alpha,\beta}$ and let $\bu_1,\bv_1$ be the points in $\calN_r$ that are closest to $\bu,\bv$, respectively, and then arrive at   (\ref{eq:use_net1_rapp}). We shall bound $\|\calP_{\calC_-}(\bu_1-\bv_1-\eta\cdot\bh(\bu,\bv))\|_2$ separately for   large-distance regime and small-distance regime. 

    \paragraph{Large-distance regime:} Combining (\ref{eq:large_decompose}), (\ref{eq:sepa_main_side}) and (\ref{eq:p_q_hpq}) we come to 
    $$\|\calP_{\calC_-}(\bu_1-\bv_1-\eta\bh(\bu,\bv))\|_2\le T_1^{\bu_1,\bv_1}+\eta|T_2^{\bu_1,\bv_1}|+ \eta T_3^{\bu_1,\bv_1}+ \eta\|\calP_{\calC_-}(\bh_2(\bu_1,\bv_1))\|_2 + \eta\big\|\calP_{\calC_-}(\bh(\bu,\bv)-\bh(\bu_1,\bv_1))\big\|_2.$$ By Lemmas \ref{lem:final_T11}--\ref{lem:cal_Tpq} and (\ref{eq:Teta_calculate}), we have $$T_1^{\bu_1,\bv_1}\le  C_1\eta\sqrt{\frac{\|\bu_1-\bv_1\|_2\scrH(\calC_{\alpha,\beta},r)}{m}}+\|\bu_1-\bv_1\|_2\cdot\big|1-\eta\cdot[g_\eta(a,b)+c_1\sqrt{\|\bu_1-\bv_1\|_2}]\big|,$$ where $(a,b)$ depends on $\bu_1,\bv_1$ and satisfies $a^2+b^2\in[\alpha^2,\beta^2]$, holds uniformly for all $(\bu_1,\bv_1)$ under consideration with the promised probability. Similarly, with the promised probability,  combining Lemmas \ref{lem:final_bound_T2}--\ref{lem:integral2} yields the uniform bound $$\eta\cdot|T_2^{\bu_1,\bv_1}|\le C_2\eta\sqrt{\frac{\|\bu_1-\bv_1\|_2\scrH(\calC_{\alpha,\beta},r)}{m}}+\|\bu_1-\bv_1\|_2\cdot \eta\cdot |h_\eta(a,b)+c_2\|\bu_1-\bv_1\|_2|;$$   Lemma \ref{lem:T3_final} gives the uniform bound $$\eta\cdot T_3^{\bu_1,\bv_1}\le C_3\eta \sqrt{\frac{\|\bu_1-\bv_1\|_2[\omega^2(\calC_{(1)})+\scrH(\calC_{\alpha,\beta},r)]}{m}}\,;$$ under $\delta_4\le \frac{c}{\log^{1/2}(r^{-1})}$ Lemma \ref{lem:final_h2} gives the uniform bound $\eta\cdot \|\calP_{\calC_-}(\bh_2(\bu_1,\bv_1))\|_2\le C_4\eta r$;  and finally Lemma \ref{lem:final_hdifference} delivers the uniform bound $\eta\cdot \|\calP_{\calC_-}(\bh(\bu,\bv)-\bh(\bu_1,\bv_1))\|_2\le C_5\eta r\log(r^{-1})$. Combining these pieces and (\ref{eq:use_net1_rapp}), and further taking supremum on $(a,b)$, we arrive at the uniform bound for large-distance regime:
    \begin{align}\nn
        \|\calP_{\calC_-}(\bu-&\bv-\eta\cdot\bh(\bu,\bv))\|_2  \le C_6\eta \sqrt{\frac{\|\bu_1-\bv_1\|_2[\omega^2(\calC_{(1)})+\scrH(\calC_{\alpha,\beta},r)]}{m}}+C_7\eta r\log(r^{-1})\\
        &+\|\bu_1-\bv_1\|_2\left(\sup_{a^2+b^2\in[\alpha^2,\beta^2]}\sqrt{|1-\eta g_\eta (a,b)|^2+|\eta h_\eta(a,b)|^2}+c_8\eta \|\bu_1-\bv_1\|_2^{1/4}\right).\label{largedis}
    \end{align}

    \paragraph{Small-distance regime:} Since $\|\bu_1-\bv_1\|_2\le r$, we have $\|\bu-\bv\|_2\le \|\bu-\bu_1\|_2+\|\bu_1-\bv_1\|_2+\|\bv_1-\bv\|_2\le 3r$. Thus by Lemma \ref{lem:smalldis} we have $\|\calP_{\calC_-}(\bh(\bu,\bv))\|_2<C_9r\log(r^{-1})$ holds for all $\bu,\bv$ in small-distance regime. Substituting this into  (\ref{eq:small_decompose}) and combining with (\ref{eq:use_net1_rapp}), we  obtain the uniform bound for small-distance regime: $$\|\calP_{\calC_-}(\bu-\bv-\eta\cdot\bh(\bu,\bv))\|_2\le C_{10}r\log(r^{-1}).$$

    \paragraph{Concluding:} We combine the bounds in the two regimes, together with triangle inequality, $\eta=\calO(1)$ and $mr\gtrsim \omega^2(\calC_{(1)})+\scrH(\calC_{\alpha,\beta},r)$, to come to the following final bound: for all $\bu,\bv\in\calC_{\alpha,\beta}$ obeying $\|\bu-\bv\|_2\le\delta_4\le \frac{c}{\log^{1/2}(r^{-1})}$, we have $$\|\calP_{\calC_-}(\bu-\bv-\eta\cdot\bh(\bu,\bv))\|_2\le \delta_1\|\bu-\bv\|_2+\sqrt{\delta_2\|\bu-\bv\|_2}+\delta_3$$ with the promised probability, where $$\delta_1 = \sup_{a^2+b^2\in[\alpha^2,\beta^2]}\sqrt{|1-\eta g_\eta(a,b)|^2+|\eta h_\eta(a,b)|^2}+\frac{c'}{\log^{1/8}(r^{-1})},$$ $\delta_2= \frac{C_{11}[\omega^2(\calC_{(1)})+\scrH(\calC_{\alpha,\beta},r)]}{m}$ and $\delta_3=C_{12}r\log(r^{-1})$. This establishes the claim. 
 \end{proof}

\section{Technical Lemmas}\label{app:tools}
%We list the technical lemmas and useful facts used along our theoretical analysis. A proof will be included if needed. 
\subsection{Auxiliary Facts}

\begin{lem}
 \label{lem:norm_equa} For any $(\bu,\bv)\in \mathbbm{A}_{\alpha,\beta}$  we have $$
        \max\{\alpha\dist_{\rm d}(\bu,\bv), \dist_{\rm n}(\bu,\bv)\}\le \dist(\bu,\bv) \le \beta\dist_{\rm d}(\bu,\bv)+\dist_{\rm n}(\bu,\bv),$$
  which implies $$\dist(\bu,\bv)\asymp \dist_{\rm d}(\bu,\bv)+ \dist_{\rm n}(\bu,\bv).$$
\end{lem}

\begin{lem}
    \label{lem:equ_dist_dmat}For any $(\bu,\bv)\in \mathbbm{A}_{\alpha,\beta}$ we have 
   $$  \frac{1}{2\beta}\|\bu\bu^\top-\bv\bv^\top\|_{\rm F}\le  \dist(\bu,\bv)\le \frac{1}{\alpha}\|\bu\bu^\top-\bv\bv^\top\|_{\rm F}.$$ Note that $\|\bu\bu^\top-\bv\bv^\top\|_{\rm op}\le\|\bu\bu^\top-\bv\bv^\top\|_{\rm F}\le \sqrt{2}\|\bu\bu^\top-\bv\bv^\top\|_{\rm op}$. 
\end{lem}
\begin{lem}
    [A Useful Parameterization] \label{lem:parameterization}Given  two different points $\bu,\bv$ ($\bu\neq \bv$) in $\mathbbm{A}_{\alpha,\beta}$, we let $\bbeta_1= \frac{\bu-\bv}{\|\bu-\bv\|_2}$ and can further find $\bbeta_2$ satisfying $\|\bbeta_2\|_2=1$ and $\langle \bbeta_1,\bbeta_2\rangle=0$, such that $$
             \bu = u_1\bbeta_1 + u_2 \bbeta_2\quad\text{and}\quad\bv = v_1 \bbeta_1 + u_2 \bbeta_2$$
    hold for some coordinates $(u_1,u_2,v_1)$ satisfying $u_1>v_1$ and $u_2\ge 0$. Note that the coordinates $(u_1,u_2,v_1)$ are uniquely determined by $(\bu,\bv)$ through  
    \begin{align}
\label{eq:explicit_uuv}
             u_1 = \frac{\langle\bu,\bu-\bv\rangle}{\|\bu-\bv\|_2}, ~
            v_1 = \frac{\langle \bv,\bu-\bv\rangle}{\|\bu-\bv\|_2},  ~
            u_2 = \frac{\sqrt{\|\bu\|_2^2\|\bv\|_2^2-(\bu^\top\bv)^2}}{\|\bu-\bv\|_2}.
     \end{align}
    More specifically, we have $v_1=-u_1$ when $\|\bu\|_2=\|\bv\|_2$ holds. %we can further ensure $v_1=-u_1$ for some $u_1>0$. 
\end{lem}

\begin{lem}
   \label{lem:maximum}   
    We define     $F(\eta,a,b)=\sqrt{|1-\eta g_\eta(a,b)|^2+|\eta h_\eta (a,b)|^2}$ over $\eta>0$ and $a^2+b^2\in [\alpha^2,\beta^2]$ for some $\beta\ge\alpha>0$, where $g_\eta(a,b)$ and $h_\eta(a,b)$ are  the functions introduced in Lemma \ref{lem:cal_Tpq} and Lemma \ref{lem:integral2}. Then we have 
        \begin{align*}
          \sup_{a^2+b^2\in[\alpha^2,\beta^2]}F\Big(\sqrt{\frac{\pi e}{2}}\tau,a,b\Big)= \max \left\{\max_{w\in[\frac{\tau}{\beta},\frac{\tau}{\alpha}]}\Big|1- w^3\exp\Big(\frac{1-w^2}{2}\Big)\Big|,\max_{w\in[\frac{\tau}{\beta},\frac{\tau}{\alpha}]}\Big|1- w \exp \Big(\frac{1-w^2}{2}\Big)\Big|\right\}.
        \end{align*}
\end{lem}
 
\begin{lem}
    [Gaussian Tail Bound] \label{lem:gaussian_tail} Let $X\sim \calN(0,1)$, then for any $t>0$ we have $
        \frac{1}{4}\exp(-t^2)\le\mathbbm{P}(X\ge t)\le \frac{1}{2}\exp(-\frac{t^2}{2}).$
\end{lem}

      \begin{lem}
    [See, e.g., Lemma 16 in \cite{oymak2017sharp}] \label{lem:pro_closec}Let $\calD\subset \mathbb{R}^n$ be a closed cone, then for any $\bv\in \mathbb{R}^n$ we have 
    \begin{align*}
        \|\calP_{\calD}(\bv)\|_2 = \sup_{\bu\in \calD\cap \mathbb{B}_2^n}\bu^\top\bv.
    \end{align*}
    In particular, for $\bv,\bv'\in\mathbb{R}^n$ and $\lambda\ge 0$, this implies  $ \|\calP_{\calD}(\bv+\bv')\|_2 \le \|\calP_{\calD}(\bv)\|_2+\|\calP_{\calD}(\bv')\|_2$, $ \|\calP_{\calD}(\lambda \bv)\|_2 = \lambda\|\calP_{\calD}(\bv)\|_2$ and $\|\calP_{\calD}(\bv)\|_2\le \|\bv\|_2.$
 \end{lem}
\begin{lem}
    [See, e.g., Lemma 18 in \cite{oymak2017sharp}] \label{lem:pro_contain}Let $\calD$ be a closed and non-empty set that contains $0$, and $\calC$ be a closed cone such that $\calD\subset \calC$. Then for any $\bv\in \mathbb{R}^n$ we have $\|\calP_{\calD}(\bv)\|_2\le 2\|\calP_{\calC}(\bv)\|_2.$
\end{lem}
\begin{lem}
    [Local Binary Embedding, Theorem 3.2 in \cite{oymak2015near}]\label{lem:binaryembed} Suppose that $\ba_1,\cdots,\ba_m$ are i.i.d. $\calN(0,\bI_n)$. Given an arbitrary set $\calK$ contained in $\mathbb{S}^{n-1}$ and some sufficiently small $r$ satisfying $rm \ge C_0$ for some  constant $C _0\ge 1$, we let $r'=\frac{c_1r}{\log ^{-1/2}(r^{-1})}$ for some small enough $c_1$.   If 
    $$
        m\ge C_2\Big(\frac{\omega^2(\calK_{(r')})}{r^3}+\frac{\scrH(\calK,r')}{r}\Big)$$
    holds for some sufficiently large $C_2$, then with probability at least $1-C_3\exp(-c_4rm)$, for any $\bu,\bv\in \calK$ with $\|\bu-\bv\|_2\le r'$ we have $$d_H\big(\sign(\bA\bu),\sign(\bA\bv)\big)\le C_5 rm. $$
\end{lem}
\subsection{Concentration Inequalities}
\begin{lem}
    [Chernoff Bound, e.g., Section 4.1 in \cite{motwani1995randomized}] \label{lem:chernoff} 
    Suppose that $X\sim \text{\rm Bin}(m,q)$ for some $q\in (0,1)$, then we have the following   bounds: 
    \begin{itemize}
    [leftmargin=2ex,topsep=0.25ex]
\setlength\itemsep{-0.1em}
        \item (Chernoff Bound)  For any $\delta>0$ we have $$\mathbbm{P}\big(X \ge (1+\delta)mq\big)\le \exp \Big(-mq \big((1+\delta)\log(1+\delta)-\delta\big)\Big);$$  
        We also have $$  \mathbbm{P}\big(X\le (1-\delta)mq\big)  \le \exp\Big(-mq \big((1-\delta)\log(1-\delta)+\delta\big)\Big)$$ for $\delta\in (0,1);$ 
        \item (Weakened Chernoff Bound)  For any $\delta \in (0,1)$, we have $$ \mathbbm{P}\big(X\ge (1+\delta)mq\big) \le \exp \Big(-\frac{\delta^2 mq}{3}\Big)\quad \text{and} \quad \mathbbm{P}\big(X\le (1-\delta)mq\big) \le \exp \Big(-\frac{\delta^2mq}{3}\Big).$$ 
    \end{itemize}
\end{lem}
\begin{comment}
\begin{lem}
    [Bernstein's Inequality, e.g., \cite{vershynin2018high}] \label{lem:bernstein} Let $X_1,...,X_m$ be independent, zero-mean, sub-exponential random variables. Then for some absolute constant $c$ and any $t\ge 0$, we have 
    \begin{align}
        \mathbbm{P} \left(\Big|\frac{1}{m}\sum_{i=1}^m X_i\Big| \ge t\right) \le 2\exp \left(-c m \min\Big\{\frac{t^2}{K^2},\frac{t}{K}\Big\}\right),
    \end{align}
    where $K= \max_i \|X_i\|_{\psi_1}$. 
\end{lem}
\end{comment}
\begin{lem}[See, e.g., Exercise 8.6.5 in \cite{vershynin2018high}]
    \label{lem:tala}
  Let $(R_{\bu})_{\bu\in\calW}$ be a random process (not necessarily zero-mean) on a subset $\calW\subset \mathbb{R}^n$. Assume that $R_{0}=0$, and  $\|R_{\bu}-R_{\bv}\|_{\psi_2}\leq K\|\bu-\bv\|_2$ holds for all $\bu,\bv\in\calW\cup\{0\}$. Then, for every $t\geq 0$, the event 
    $$
        \sup_{\bu\in \calW}\big|R_{\bu}\big|\leq CK \big(\omega(\calW)+t\cdot \rad(\calW)\big)$$ 
    with probability at least $1-2\exp(-t^2)$. 
\end{lem}
\begin{lem}
    [Concentration of Product Process, \cite{mendelson2016upper} (see also \cite{genzel2023unified})] \label{lem:product_process}Let $\{g_{\bu}\}_{\bu\in \calU}$ and $\{h_{\bv}\}_{\mathbf{b}\in\calV}$ be stochastic processes indexed by two sets $\calU\subset \mathbb{R}^{p}$ and $\calV\subset \mathbb{R}^q$, both defined on a common probability space $(\Omega,A,\mathbbm{P})$. We assume that there exist $K_{\calU},K_{\calV},r_{\calU},r_{\calV}\geq 0$ such that $\|g_{\bu}-g_{\bu'}\|_{\psi_2}\leq K_{\calU}\|\bu-\bu'\|_2,~\|g_{\bu}\|_{\psi_2} \leq r_{\calU}~(\forall\,\bu,\bu'\in \calU)$, and that $\|h_{\bv}-h_{\bv'}\|_{\psi_2} \leq {K}_{\calV}\|\bv-\bv'\|_2,~\|h_{\bv}\|_{\psi_2}\leq r_{\calV}~(\forall\,\bv,\bv'\in \calV)$. Suppose that $\ba_1,...,\ba_m$ are independent copies of a random variable $\ba\sim \mathbbm{P}$, then for every $t\geq 1$ the event\begin{equation}
   \begin{aligned}\nonumber
        &\sup_{\bu\in\calU}\sup_{\bv\in\calV} ~\Big| \frac{1}{m}\sum_{i=1}^m g_{\bu}(\ba_i)h_{\bv}(\ba_i)-\mathbbm{E}\big[g_{\bu}(\ba_i)h_{\bv}(\ba_i)\big]\Big|\\
        &\leq C\left(\frac{(K_{\calU}\cdot\omega(\calU)+t\cdot r_{\calU})\cdot (K_{\calV}\cdot \omega(\calV)+t\cdot r_{\calV})}{m}+\frac{r_{\calU}\cdot K_{\calV}\cdot \omega(\calV)+r_{\calV}\cdot K_{\calU}\cdot \omega(\calU)+t\cdot r_{\calU}r_{\calV}}{\sqrt{m}}\right)
   \end{aligned}
\end{equation}
holds with probability at least $1-2\exp(-ct^2)$. 
\end{lem}
\begin{lem}
    [See, e.g., \cite{dirksen2021non}] \label{lem:max_ell_sum} Let $\ba_1,...,\ba_m$ be independent   random vectors in $\mathbb{R}^n$ satisfying $\mathbbm{E}(\ba_i\ba_i^\top)=\bI_n$ and  $\max_i\|\ba_i\|_{\psi_2}\leq L$. For some given given $\calW\subset \mathbbm{R}^n$ and  $1\leq \ell\leq m$, there exist  constants $C_1,c_2$ depending only on $L$ such that   the event 
    \begin{equation*}     \sup_{\bx\in\calW}\max_{\substack{I\subset [m]\\|I|\leq \ell}}\Big(\frac{1}{\ell}\sum_{i\in I}|\langle\ba_i,\bx\rangle|^2\Big)^{1/2}\leq C_1\Big(\frac{\omega(\calW)}{\sqrt{\ell}}+\rad(\calW)\sqrt{\log\Big(\frac{em}{\ell}\Big)}\Big)
    \end{equation*}
    holds with probability at least $1-2\exp(-c_2\ell\log(\frac{em}{\ell}))$. 
\end{lem}
\section{Proofs with Standard Arguments} \label{app:standard}
This appendix collects the proofs with standard or elementary arguments, including the ones for Theorems \ref{thm:SI_low}--\ref{thm:SI_high}, Lemma \ref{lem:intersect_affine}, Lemmas \ref{lem:sepa_im_sepa}--\ref{lem:Pthetauv} and Lemmas \ref{lem:norm_equa}--\ref{lem:maximum}. 
\subsection{Spectral Initialization (Theorems \ref{thm:SI_low}--\ref{thm:SI_high})}

\paragraph{Population level:} We calculate $\mathbbm{E}(\hat{\bS}_{\bx})$ where $\hat{\bS}_{\bx}=\frac{1}{m}\sum_{i=1}^m y_i\ba_i\ba_i^\top$. 
\begin{lem}
	[Expectation of $\hat{\bS}_{\bx}$] \label{lem:ESx}For any $\bx\in \mathbbm{A}_{\alpha,\beta}$, we let $\lambda_{\bx}=\|\bx\|_2$ and $g\sim \calN(0,1)$, then for $\hat{\bS}_{\bx}=\frac{1}{m}\sum_{i=1}^my_i \ba_i\ba_i^\top$ we have $$\mathbbm{E}(\hat{\bS}_{\bx})= a_{\bx}\bx\bx^\top + b_{\bx} \bI_n$$ where $$a_{\bx} =  \lambda_{\bx}^{-2}\cdot \mathbbm{E}(\sign(\lambda_{\bx}|g|-\tau)(g^2-1))\quad\text{and}\quad b_{\bx} = \mathbbm{E}(\sign(\lambda_{\bx}|g|-\tau)).$$ 
	Moreover, there exists some   constant $c_0>0$   only depending on $(\alpha,\beta,\tau)$ such that
	$$\inf_{\bx\in \mathbbm{A}_{\alpha,\beta}}a_{\bx}\ge c_0.$$ 
	Denote the first two eigenvalues of $\mathbbm{E}(\hat{\bS}_{\bx} )$ by $\lambda_1(\mathbbm{E}\hat{\bS}_{\bx})$ and $\lambda_2(\mathbbm{E}\hat{\bS}_{\bx} )$, then we come to $$\lambda_1(\mathbbm{E}\hat{\bS}_{\bx})-\lambda_2(\mathbbm{E}\hat{\bS}_{\bx})=a_{\bx}\|\bx\|_2^2\ge c_1:=\alpha^2c_0\,,\qquad\forall\,\bx\in \mathbbm{A}_{\alpha,\beta}.$$ %uniformly for all $\bx\in \mathbbm{A}_{\alpha,\beta}.$
\end{lem}
\begin{proof}
 Note that $\mathbbm{E}\bS_{\bx}=\mathbbm{E}(y_i\ba_i\ba_i^\top)= \mathbbm{E}(\sign(|\ba_i^\top\bx|-\tau)\ba_i\ba_i^\top)$. We let $\lambda_\bx := \|\bx\|_2$, then there exists some orthogonal matrix $\bP\in \mathbb{R}^{n\times n}$ such that $\bP\bx = \lambda_{\bx} \be_1$. Further, we let $\tilde{\ba}_i = \bP\ba_i= (\tilde{a}_{i1},...,\tilde{a}_{in})^\top\sim\calN(0,\bI_n)$ and  proceed as 
         \begin{align*}
       & \mathbbm{E}\bS_\bx = \mathbbm{E}\big[\sign(|\ba_i^\top\bx|-\tau)\ba_i\ba_i^\top\big]  \\
       &= \bP^\top \mathbbm{E}\big[\sign(|\tilde{\ba}_i^\top \lambda_{\bx}\be_1|-\tau)\tilde{\ba}_i\tilde{\ba}_i^\top\big]\bP \\
       & = \bP^\top \mathbbm{E}\big[\sign(|\lambda_{\bx}\tilde{a}_{i1}|-\tau)\tilde{\ba}_i\tilde{\ba}_i^\top\big]\bP \\
       &= \bP^\top\begin{bmatrix}
           t_{0,\bx} & 0 \\0 & t_{1,\bx}\bI_{n-1} 
       \end{bmatrix}\bP   \\
       & = \bP^\top\big((t_{0,\bx}-t_{1,\bx})\be_1\be_1^\top+t_{1,\bx}\bI_n\big)\bP \\&= \frac{(t_{0,\bx}-t_{1,\bx})\bx\bx^\top}{\lambda_{\bx}^2}+ t_{1,\bx}\bI_n,  
    \end{align*} 
where we define $ t_{0,\bx}:= \mathbbm{E} (\sign(|\lambda_{\bx}g|-\tau)g^2)$ and $ t_{1,\bx}:=\mathbbm{E} (\sign(|\lambda_{\bx}g|-\tau))$ with $g\sim\calN(0,1)$. Observe that     $a_{\bx}:=\frac{t_{0,\bx}-t_{1,\bx}}{\lambda_{\bx}^2}\ge \frac{t_{0,\bx}-t_{1,\bx}}{\beta^2}$, we further establish a lower bound on $t_{0,\bx}-t_{1,\bx}$: 
    \begin{align*}
&t_{0,\bx} -t_{1,\bx}=\mathbbm{E}_{g\sim \calN(0,1)} \Big[\sign\Big(|g|- \frac{\tau}{\lambda_{\bx}}\Big) (g^2-1)\Big] \\&= \mathbbm{E}_{g\sim \calN(0,1)} \Big[\Big(2\mathbbm{1}\Big(|g|\ge \frac{\tau}{\lambda_{\bx}}\Big)-1\Big)(g^2-1)\Big]  \\  
& =2\mathbbm{E}_{g\sim\calN(0,1)}\Big[(g^2-1) \mathbbm{1}\Big(|g|\ge \frac{\tau}{\lambda_{\bx}}\Big)\Big] \\&\ge 2\min_{t\in [\frac{\tau}{\beta},\frac{\tau}{\alpha}]}\int_{|a|>t} (a^2-1)\cdot\frac{1}{\sqrt{2\pi}}\exp\Big(-\frac{a^2}{2}\Big) ~\text{d}a\\:&=2 \min_{t\in [\frac{\tau}{\beta},\frac{\tau}{\alpha}]} f_1(t),
\end{align*} 
where  in the last equality we take the minimum over $\frac{\tau}{\lambda_{\bx}}\in [\frac{\lambda}{\beta},\frac{\lambda}{\alpha}]$ and introduce the shorthand $$f_1(t)=\int_{|a|>t}\frac{a^2-1}{\sqrt{2\pi}}\exp(-\frac{a^2}{2})~\text{d}a=2 \int_{t}^\infty \frac{a^2-1}{\sqrt{2\pi}}\exp(-\frac{a^2}{2})~\text{d}a.$$ It is easy to see that $\min_{t\in[\frac{\tau}{\beta},\frac{\tau}{\alpha}]}f_1(t)$  is a positive constant only depending on $(\alpha,\beta,\tau)$, as desired. 
\end{proof} 
\paragraph{Concentration bounds:} The main ingredients for proving Theorems \ref{thm:SI_low}--\ref{thm:SI_high} are a set of (uniform) concentration bounds. 
\begin{lem}
    [Concentration of $\hat{\bS}_{\bx}$] \label{lem:con_Sx_low}For a fixed $\bx\in \mathbbm{A}_{\alpha,\beta}$ and any $t\ge 1$, the non-uniform bound  $$\|\hat{\bS}_{\bx}-\mathbbm{E}\hat{\bS}_{\bx}\|_{\rm op}\le C_1 \Big[\frac{n+t}{m}+ \sqrt{\frac{n+t}{m}}\Big]$$
    holds with probability at least $1-\exp(-c_2 t)$. If $m\ge C_3 n $ for some large enough $C_3$, then     the uniform bound  $$ \sup_{\bx\in \mathbbm{A}_{\alpha,\beta}}  \|\hat{\bS}_{\bx}-\mathbbm{E}\hat{\bS}_{\bx}\|_{\rm op}\le \sqrt{\frac{C_4 n}{m}\log(\frac{m}{n})}$$
    holds with probability at least $1-\exp(-c_5n\log(\frac{m}{n}))$. 
\end{lem}
\begin{proof}
    We first establish the {\it non-uniform bound} and then strengthen it to a {\it uniform bound} by covering argument. In the covering argument, we would use Theorem \ref{thm:local_embed} to overcome the discontinuity of the quantizer; we will only prove the non-uniform bounds for the remaining concentration lemmas if their uniform counterparts can be obtained by similar technique. 
    
   \noindent{\it Non-uniform bound:}
   % It is standard to show the non-uniform concentration (\ref{eq:non_uni_Sx_low}) by a covering argument (e.g., \cite{vershynin2018high}), whereas it is more convenience to accomplish the goal by the in-depth result of concentration of product process (Lemma \ref{lem:product_process}).
   To get started, we write $$    \|\hat{\bS}_{\bx}- \mathbbm{E}\hat{\bS}_{\bx}\|_{\rm op}  = \sup_{\bu\in \mathbb{S}^{n-1}}\big|\bu^\top \hat{\bS}_{\bx}\bu - \mathbbm{E}(\bu^\top \hat{\bS}_{\bx}\bu)\big| = \sup_{\bu\in\mathbb{S}^{n-1}}\Big|\frac{1}{m}\sum_{i=1}^n y_i(\ba_i^\top\bu)^2-\mathbbm{E}\big[y_i(\ba_i^\top\bu)^2\big]\Big|.$$ Note that this can be viewed as a product process with two factors being $g_{\bu}(\ba_i)=y_i \ba_i^\top \bu$ and $h_{\bu}(\ba_i)=\ba_i^\top\bu$, thus the   non-uniform bound follows from a straightforward application of Lemma \ref{lem:product_process}.  

\noindent{\it Uniform   bound:} 
We shall precisely write $y_i=\sign(|\ba_i^\top\bx|-\tau)$ to indicate its dependence on $\bx$. For some sufficiently small  $\gamma$ to be chosen, we let $\calN_\gamma$  be a minimal $\gamma$-net of  $\mathbbm{A}_{\alpha,\beta}$ with cardinality $|\calN_\gamma|\le(\frac{C}{\gamma})^n$ for some $C$ that only depends on $(\alpha,\beta)$. We invoke the non-uniform bound  over all $\bx\in \calN_\gamma$ and then take a union bound, yielding that the event $$\sup_{\bx\in\calN_\gamma}\|\hat{\bS}_{\bx}-\mathbbm{E}\hat{\bS}_{\bx}\|_{\rm op}\le C_1 \Big[\frac{n+t}{m}+ \sqrt{\frac{n+t}{m}}\Big]$$
holds with probability at least $1-\exp(n\log(\frac{C}{\gamma})-c_2t)$. 
 For any $\bx\in \mathbbm{A}_{\alpha,\beta}$ we let $\bx_1 = \text{arg}\min_{\bu\in \calN_\gamma}\|\bu-\bx\|_2$ and then $\|\bx_1-\bx\|_2\le \gamma$, and moreover  triangle inequality yields $$\|\hat{\bS}_{\bx}-\mathbbm{E}\hat{\bS}_{\bx}\|_{\rm op}\le \|\hat{\bS}_{\bx} - \hat{\bS}_{\bx_1}\|_{\rm op} + \|\hat{\bS}_{\bx_1}-\mathbbm{E}\hat{\bS}_{\bx_1}\|_{\rm op} + \|\mathbbm{E}\hat{\bS}_{\bx_1}- \mathbbm{E}\hat{\bS}_{\bx}\|_{\rm op}.$$ Note that we have controlled $\|\hat{\bS}_{\bx_1}-\mathbbm{E}\hat{\bS}_{\bx_1}\|_{\rm op}$, and we will need to further bound $\|\hat{\bS}_{\bx}-\hat{\bS}_{\bx_1}\|_{\rm op}$ and $\|\mathbbm{E}\hat{\bS}_{\bx_1}-\mathbbm{E}\hat{\bS}_{\bx}\|_{\rm op}$.

 To uniformly bound $\|\hat{\bS}_{\bx}-\hat{\bS}_{\bx_1}\|_{\rm op}$, we first invoke Theorem \ref{thm:local_embed}  with $\calK=\mathbbm{A}_{\alpha,\beta}$ and $(r,r')=(C_3\gamma\sqrt{\log \gamma^{-1}},2\gamma)$ for some large enough $C_3$, yielding the following: if $\gamma\gtrsim \frac{n}{m}\sqrt{\log(\frac{m}{n})}$, then the event $$  d_H\big(\sign(|\bA\bu|-\tau),\sign(|\bA\bv|-\tau)\big)\le C_4\gamma\sqrt{\log\gamma^{-1}}m,\quad \forall\, \bu,\bv\in \mathbbm{A}_{\alpha,\beta}\,\text{ obeying }\,\|\bu-\bv\|_2\le \gamma$$
holds with probability at least $1-\exp(-c_5\gamma m)$. Due to $\|\bx-\bx_1\|_2\le \gamma$, the  event implies  $$ d_H\big(\sign(|\bA\bx|-\tau),\sign(|\bA\bx_1|-\tau)\big)\le C_4\gamma\sqrt{\log\gamma^{-1}}m,\qquad\forall \bx\in\mathbbm{A}_{\alpha,\beta}.$$ 
Now we are able to obtain  
    \begin{align*}
        &\|\hat{\bS}_{\bx}-\hat{\bS}_{\bx_1}\|_{\rm op} = \Big\|\frac{1}{m}\sum_{i=1}^m \big[\sign(|\ba_i^\top\bx|-\tau)-\sign(|\ba_i^\top\bx_1|-\tau)\big]\ba_i\ba_i^\top\Big\|_{\rm op} \\
        &=  \sup_{\bu\in \mathbb{S}^{n-1}} \Big|\frac{1}{m}\sum_{i=1}^m \big[\sign(|\ba_i^\top\bx|-\tau)-\sign(|\ba_i^\top\bx_1|-\tau)\big](\ba_i^\top\bu)^2\Big|\\  
        &\le \sup_{\bu\in \mathbb{S}^{n-1}}\sup_{\substack{\calS\subset [m]\\|\calS|\le C_4\gamma\sqrt{\log\gamma^{-1}}m}}\frac{2}{m}\sum_{i\in \calS}(\ba_i^\top\bu)^2 \le C_5 \Big(\frac{n}{m}+\gamma\log^{3/2}\Big(\frac{1}{\gamma}\Big)\Big)
    \end{align*}
where  the last inequality holds with probability $1-\exp(-c_6\gamma\log^{3/2}(\gamma^{-1})m)$ due to Lemma \ref{lem:max_ell_sum}.

 To uniformly bound $\|\mathbbm{E}\hat{\bS}_{\bx}-\mathbbm{E}\hat{\bS}_{\bx_1}\|_{\rm op}$, we use Lemma \ref{lem:ESx} to obtain  $\mathbbm{E}\hat{\bS}_{\bx}=a_{\bx}\bx\bx^\top+b_{\bx}\bI_n$ and then utilize triangle inequality to obtain  \begin{align*}
     &\|\mathbbm{E}\hat{\bS}_{\bx}-\mathbbm{E}\hat{\bS}_{\bx_1}\|_{\rm op} = \|a_{\bx}\bx\bx^\top-a_{\bx_1}\bx_1\bx_1^\top+ (b_{\bx}-b_{\bx_1})\bI_n\|_{\rm op}\\&\le |a_{\bx}-a_{\bx_1}|\|\bx\bx^\top\|_{\rm op}+ |a_{\bx_1}|\|\bx\bx^\top-\bx_1\bx_1^\top\|_{\rm op} + |b_{\bx}-b_{\bx_1}| \\&\lesssim |a_{\bx}-a_{\bx_1}|+ \|\bx\bx^\top-\bx_1\bx_1^\top\|_{\rm op} +|b_{\bx}-b_{\bx_1}|.
 \end{align*} 
By the explicit formulas for $a_{\bx}$ and $b_{\bx}$ in Lemma \ref{lem:ESx}, we let   $g\sim \calN(0,1)$ and have the bound 
    \begin{align*}
   & |a_{\bx}-a_{\bx_1}| = \Big|\frac{\mathbbm{E}(\sign(\|\bx\|_2|g|-\tau)(g^2-1))}{\|\bx\|_2^2} - \frac{\mathbbm{E}(\sign(\|\bx_1\|_2|g|-\tau)(g^2-1))}{\|\bx_1\|_2^2} \Big|\\
   &\lesssim \Big|\frac{1}{\|\bx\|_2^2}-\frac{1}{\|\bx_1\|_2^2}\Big|+\Big|\mathbbm{E}\Big(\big[\sign(\|\bx\|_2|g|-\tau)-\sign(\|\bx_1\|_2|g|-\tau)\big](g^2-1)\Big)\Big|,  
\end{align*}  
Furthermore, we have $$  |\|\bx\|_2^{-2}-\|\bx_1\|_2^{-2}|\le \frac{|\|\bx_1\|_2^2-\|\bx\|_2^2|}{\alpha^4} \le \frac{\|\bx-\bx_1\|_2\|\bx+\bx_1\|_2}{\alpha^4} \le \frac{2\beta \gamma}{\alpha^4}$$ and 
\begin{align*}
    &\big|\mathbbm{E}\big(\big[\sign(\|\bx\|_2|g|-\tau)-\sign(\|\bx_1\|_2|g|-\tau)\big](g^2-1)\big)\big| \\
    &\le 2 \mathbbm{E} \left(|g^2-1|\mathbbm{1}\Big(\frac{\tau}{\max\{\|\bx\|_2,\|\bx_1\|_2\}}\le|g|\le \frac{\tau}{\min\{\|\bx\|_2,\|\bx_1\|_2\}}\Big)\right)\\&\lesssim \big|\frac{\tau}{\|\bx\|_2}-\frac{\tau}{\|\bx_1\|_2}\big| \lesssim \gamma
\end{align*}
Similarly we can show $\sup_{\bx\in \mathbbm{A}_{\alpha,\beta}}|b_{\bx}-b_{\bx_1}|\lesssim\gamma$. All that remains is to  bound $\|\bx\bx^\top-\bx_1\bx_1^\top\|_{\rm op}$: $$   \|\bx\bx^\top-\bx_1\bx_1^\top\|_{\rm op}\le \|\bx(\bx-\bx_1)^\top\|_{\rm op}+\|(\bx-\bx_1)\bx_1^\top\|_{\rm op}=\|\bx-\bx_1\|_2 (\|\bx\|_2+\|\bx_1\|_2) \le2\beta\gamma.$$
Note that these arguments are uniform for 
$\bx\in \mathbbm{A}_{\alpha,\beta}$, thus we obtain the uniform bound $\|\mathbbm{E}\hat{\bS}_{\bx}-\mathbbm{E}\hat{\bS}_{\bx_1}\|_{\rm op}\lesssim \gamma$.

We combine everything to obtain the following:
 if $\gamma\gtrsim \frac{n}{m}\sqrt{\log(\frac{m}{n})}$, then uniformly for all $\bx\in\mathbbm{A}_{\alpha,\beta}$, the bound $$  \|\hat{\bS}_{\bx}-\mathbbm{E}\hat{\bS}_{\bx}\|_{\rm op} \le C_7\Big[\frac{n+t}{m}+\sqrt{\frac{n+t}{m}}+\gamma\log^{3/2}(\gamma^{-1})\Big]$$ 
holds with probability at least $1-\exp(n\log\frac{C}{\gamma}-c_2t)-\exp(-c_8\gamma m)$. Therefore, by setting $\gamma=\frac{C_9n}{m}\sqrt{\log(\frac{m}{n})}$  and $t=C_{10}n\log (\frac{m}{n})$  with large enough $C_9,C_{10}$, the desired statement follows. 
\end{proof}

For a matrix $\bM\in \mathbb{R}^{n\times n}$ and any $\calS\subset [n]$, recall that $[\bM]_{\calS}\in \mathbb{R}^{n\times n}$ is obtained from $\bM$ by setting the rows and columns not in  $\calS$ to zero. 
\begin{lem}
    {\rm (Concentration  of $[\hat{\bS}_{\bx}]_{\calS}$)}  \label{lem:con_Sx_high} For a fixed $\bx\in \Sigma^n_k\cap\mathbbm{A}_{\alpha,\beta}$  and any $t\ge 1$, the non-uniform bound  
    $$
        \sup_{ \calS\subset [n],|\calS|=k}\big\|[\hat{\bS}_{\bx}]_{\calS}-\mathbbm{E}[\hat{\bS}_{\bx}]_{\calS}\big\|_{\rm op} \le C_1 \Big[\frac{k\log (\frac{en}{k})+t}{m}+ \sqrt{\frac{k\log(\frac{en}{k})+t}{m}}\Big]$$ holds with probability at least $1-\exp(-c_2t)$. If $m\ge C_3 k \log(\frac{en}{k})$ for some large enough $C_3$, then   the uniform bound  
    $$
        \sup_{\bx\in \Sigma^n_k\cap \mathbbm{A}_{\alpha,\beta}} \sup_{ \calS\subset [n],|\calS|=k} \big\|[\hat{\bS}_{\bx}]_{\calS}-\mathbbm{E}[\hat{\bS}_{\bx}]_{\calS}\big\|_{\rm op} \le \sqrt{\frac{C_4k}{m}\log(\frac{mn}{k^2})}$$ holds with probability at least $1-\exp(-c_5k\log(\frac{en}{k}))$. 
\end{lem}
\begin{proof}
   As mentioned, we will only prove the non-uniform bound and leave the uniform bound to avid readers. 
    For a fixed $\bx\in\Sigma^n_k\cap \mathbbm{A}_{\alpha,\beta}$ we write $$
        \sup_{ \calS\subset [n],|\calS|=k}\big\|[\hat{\bS}_{\bx}]_{\calS}-\mathbbm{E}[\hat{\bS}_{\bx}]_{\calS}\big\|_{\rm op}= \sup_{\bu\in \Sigma^{n,*}_k}\big|\frac{1}{m}\sum_{i=1}^my_i (\ba_i^\top\bu)^2-\mathbbm{E}\big[y_i(\ba_i^\top \bu)^2\big]\big|.$$
    Then  a straightforward application of Lemma \ref{lem:product_process} yields the non-uniform concentration bound.
\end{proof}
For the sparse case we need to additionally evaluate the quality of the support estimate, that is $\calS_{\bx}$ in Algorithm \ref{alg:SI_high},  measured by $\|\bx-\bx_{\calS_{\bx}}\|_2$. 
\begin{lem}
    [Support estimate bound] \label{lem:supp_est} 
    Let $\calS_{\bx}$ be the estimate for $\supp(\bx)$ as per Algorithm \ref{alg:SI_high}, and suppose that $m\ge C_0 k\log(\frac{en}{k})$ for some large enough $C_0$.
    For a fixed $\bx\in \Sigma^n_k \cap \mathbbm{A}_{\alpha,\beta}$, the non-uniform bound  
  $$
        \|\bx - \bx_{\calS_{\bx}}\|_2^2 \le C_1 \sqrt{\frac{k^2\log n}{m}}$$
    holds with probability at least $1-n^{-10}$. Moreover, the uniform bound 
   $$
        \sup_{\bx\in \Sigma^n_k \cap \mathbbm{A}_{\alpha,\beta}}\|\bx-\bx_{\calS_{\bx}}\|_2^2 \le C_2 \sqrt{\frac{k^3}{m} \log(\frac{mn}{k^2})}$$ 
holds with probability at least $1- \exp(-c_3 k\log(\frac{en}{k}))$. 
\end{lem}
\begin{proof}
   Before bounding $\|\bx-\bx_{\calS_{\bx}}\|_2$, we first establish the entry-wise concentration of the diagonal of $\hat{\bS}_{\bx}$. 
    We consider a fixed $\bx\in \calK:= \Sigma^n_k\cap \mathbbm{A}_{\alpha,\beta}$ with $\by = \sign(|\bA\bx|-\tau)$. 
     We  let $\calU = \{\be_1,\be_2,...,\be_n\}$ and write 
    $$\max _{j\in [n]} \Big|\frac{1}{m}\sum_{i=1}^m y_ia_{ij}^2 - \mathbbm{E}\big(\frac{1}{m}\sum_{i=1}^m y_ia_{ij}^2\big)\Big| \le \sup_{\bu\in \calU}\Big|\frac{1}{m}\sum_{i=1}^my_i (\ba_i^\top\bu)^2-\mathbbm{E}\big[y_i(\ba_i^\top \bu)^2\big]\Big|$$
    and then use Lemma \ref{lem:product_process} to obtain 
 $$
         \max _{j\in [n]} \Big|\frac{1}{m}\sum_{i=1}^m y_ia_{ij}^2 - \mathbbm{E}\Big(\frac{1}{m}\sum_{i=1}^m y_ia_{ij}^2\Big)\Big| \le C_1 \Big[\frac{t+\log n}{m}+\sqrt{\frac{t+\log n}{m}}\Big] $$
holds with probability at least $1-\exp(-c_2t)$. Setting $t\asymp \log n$ with large implied constant and substituting $\mathbbm{E}(y_ia_{ij}^2)=a_{\bx}x_j^2+b_{\bx}$ (see Lemma \ref{lem:ESx}) yields $$\max _{j\in [n]} \Big|\frac{1}{m}\sum_{i=1}^m y_ia_{ij}^2 -  \left(a_{\bx}x_j^2+b_{\bx}\right)\Big| \le C_2\sqrt{\frac{\log n}{m}}$$ with probability at least $1-n^{-10}$. 
    For the uniform bound we directly utilize Lemma \ref{lem:con_Sx_high}. Specifically, if $m\ge C_3 k\log\frac{en}{k}$ for sufficiently large $C_3$, then with probability at least $1-\exp(-c_4k\log(\frac{en}{k}))$  the uniform bound in Lemma \ref{lem:con_Sx_high} holds and evidently implies 
    $$\max _{j\in [n]} \Big|\frac{1}{m}\sum_{i=1}^m y_ia_{ij}^2 - \left(a_{\bx}x_j^2+b_{\bx}\right)\Big| \le  \sqrt{\frac{C_5k}{m}\log(\frac{mn}{k^2})},\qquad\forall\,\bx\in\Sigma^n_k\cap \mathbbm{A}_{\alpha,\beta}$$
     where $y_i = \sign(|\ba_i^\top\bx|-\tau)$ depends on $\bx$.

      Now we are ready to show the desired statements. For a fixed $\bx\in\calK$ we proceed as  
     \begin{subequations}\label{eq:bound_supp_err}
          \begin{align}\nn
          &\|\bx - \bx_{\calS_{\bx}}\|_2 ^2 = \sum_{j\in \supp(\bx)\setminus \calS_{\bx}} x_j^2 = \frac{1}{a_{\bx}^2}\sum_{j\in \supp(\bx)\setminus \calS_{\bx}} a_{\bx}^2x^2_j\\\nn
              &\le \frac{1}{a_{\bx}^2}\sum_{j\in \supp(\bx)\setminus \calS_{\bx}} \Big\{ \Big(a_{\bx}^2x_j^2+b_{\bx}- \frac{1}{m}\sum_{i=1}^m y_i a_{ij}^2 \Big) \\&\quad\quad\quad\quad\quad\quad\quad\quad -\Big(a_{\bx}^2x_{j'}^2+b_{\bx}- \frac{1}{m}\sum_{i=1}^m y_i a_{ij'}^2 \Big)\Big\}\label{eq:intro_jpai}\\
              &\le \frac{2}{a_{\bx}^2}\sum_{j\in \supp(\bx)\setminus \calS_{\bx}} \max_{\ell\in [n]} \Big|\frac{1}{m}\sum_{i=1}^m y_ia_{i\ell}^2 - \left(a_{\bx}x_\ell^2+b_{\bx}\right)\Big| \nn\\\label{eq:count_support}
              &\lesssim  k \cdot \max_{\ell\in [n]} \Big|\frac{1}{m}\sum_{i=1}^m y_ia_{i\ell}^2 - \left(a_{\bx}x_\ell^2+b_{\bx}\right)\Big|,
      \end{align} 
     \end{subequations}
     where  in (\ref{eq:intro_jpai}) we introduce $j'\in\calS_{\bx} \setminus \supp(\bx)$ \footnote{Without loss of generality, we assume $\supp(\bx)\setminus \calS_{\bx}\ne\varnothing$, then by $|\calS_{\bx}|=k$ such $j'$ exists.} satisfying $x_{j'}=0$ (since $j'\notin\supp(\bx)$) and $\frac{1}{m}\sum_{i=1}^m y_i a_{ij'}^2\ge \frac{1}{m}\sum_{i=1}^m y_ia_{ij}^2$ (due to the selection criterion of $\calS_{\bx}$); (\ref{eq:count_support}) holds because $|\supp(\bx)\setminus \calS_{\bx}|\le |\supp(\bx)|\le k$. 
          Then substituting the (non-uniform and uniform) bounds on $ \max_{\ell\in [n]} \big|\frac{1}{m}\sum_{i=1}^m y_ia_{i\ell}^2 - \left(a_{\bx}x_\ell^2+b_{\bx}\right)\big|$ completes the proof. 
\end{proof}
The concentration of $\hat{\lambda}_{\bx} = \frac{1}{m}\sum_{i=1}^m \mathbbm{1}(y_i = 1)$ will prove useful in bounding the norm estimation error.  
\begin{lem}
    [Concentration of $\hat{\lambda}_{\bx}$]  \label{lem:con_lambda_x} 
    For a fixed $\bx\in \mathbbm{A}_{\alpha,\beta}$, for any $t\ge 0$ we have the non-uniform bound $$\mathbbm{P}\left(\big|\hat{\lambda}_{\bx}-2\Phi\big(-\frac{\tau}{\|\bx\|_2}\big)\big|\ge t\right)\le 2\exp(-c_1mt^2).$$ 
    If $m\ge C_2n $ for some large enough $C_2$, then with probability at least $1-\exp(-c_3n\log(\frac{m}{n}))$ we have the uniform bound over $ \mathbbm{A}_{\alpha,\beta}$: $$\sup_{\bx \in \mathbbm{A}_{\alpha,\beta}}\big|\hat{\lambda}_{\bx}-2\Phi\big(-\frac{\tau}{\|\bx\|_2}\big)\big| \le \sqrt{\frac{C_4n}{m}\log(\frac{m}{n})}.$$
    If $m\ge C_5 k\log(\frac{en}{k})$ for some large enough $C_5$, then with probability at least $1-\exp(-c_6k\log(\frac{en}{k}))$ we have the uniform bound over $\Sigma^n_k\cap \mathbbm{A}_{\alpha,\beta}$: $$ \sup_{\bx\in \Sigma^n_k\cap \mathbbm{A}_{\alpha,\beta}}\big|\hat{\lambda}_{\bx}-2\Phi\big(-\frac{\tau}{\|\bx\|_2}\big)\big| \le \sqrt{\frac{C_7k }{m}\log(\frac{mn}{k^2})}.$$
\end{lem}
 \begin{proof}
We first establish the non-uniform   bound. 
    Noticing $$  \mathbbm{E}\hat{\lambda}_{\bx}= \mathbbm{P}\big(|\ba_i^\top\bx|>\tau\big)=\mathbbm{P}\big(|\calN(0,1)|>\frac{\tau}{\|\bx\|_2}\big)=2\Phi\big(-\frac{\tau}{\|\bx\|_2}\big)\in\big [2\Phi(-\frac{\tau}{\beta}),2\Phi(-\frac{\tau}{\alpha})]$$
    and $\|\hat{\lambda}_{\bx}-\mathbbm{E}\hat{\lambda}_{\bx}\|_{\psi_2}=\calO(m^{-1/2})$, thus the sub-Gaussian tail bound gives $$
        \mathbbm{P}\big(|\hat{\lambda}_{\bx}-\mathbbm{E}\hat{\lambda}_{\bx}|\ge t\big)\le 2\exp\big(-c_1mt^2\big)$$ for a fixed $\bx\in \calK$, yielding the non-uniform bound. The uniform bounds follow from a covering argument over $\mathbbm{A}_{\alpha,\beta}$ and $\Sigma^n_k\cap \mathbbm{A}_{\alpha,\beta}$, which we leave for avid readers.  
\end{proof}

We now briefly outline the main ideas in the proofs of Theorems \ref{thm:SI_low}--\ref{thm:SI_high}.
By Lemma \ref{lem:norm_equa} we can first decompose the estimation error into 
\begin{align}
	\dist(\hat{\bx}_{\rm SI}, \bx) \le \beta  \dist_{\rm d}(\hat{\bx}_{\rm SI}, \bx)+\dist_{\rm n}(\hat{\bx}_{\rm SI}, \bx).\label{eq:divide_dir_nor}
\end{align}
Note that in Algorithms \ref{alg:SI_low}--\ref{alg:SI_high}  we have $
		\dist_{\rm d}(\hat{\bx}_{\rm SI},\bx) = \dist\Big(\hat{\bv}_{\rm SI},\frac{\bx}{\|\bx\|_2}\Big)~~\text{and}~~
		\dist_{\rm n}(\hat{\bx}_{\rm SI}, \bx) = |\hat{\lambda}_{\rm SI}-\|\bx\|_2|$.  The proof mechanism is standard in the literature (e.g., \cite{chen2017solving,zhang2017nonconvex,cai2016optimal,jagatap2019sample}). 
Up to multiplicative constant, bounding $\dist(\hat{\bv}_{\rm SI}, \bx/\|\bx\|_2)$ amounts to controlling $\|\hat{\bv}_{\rm SI}\hat{\bv}_{\rm SI}^\top - \|\bx\|_2^{-2}\bx\bx^\top\|_{\rm op}$, and it is standard to achieve the (latter) goal by applying   Davis-Kahan sin$\Theta$ theorem in conjunction with the operator norm concentration of certain data matrix (i.e., Lemmas \ref{lem:con_Sx_low}--\ref{lem:con_Sx_high}). On the other hand, we bound $|\hat{\lambda}_{\rm SI}-\|\bx\|_2|$ by Lemma \ref{lem:con_lambda_x} and the Lipschitzness of $\Phi^{-1}$  over a closed interval contained in $(0,1)$.
\subsubsection{The Proof of Theorem \ref{thm:SI_low}}
\begin{proof} 
    We first establish a deterministic bound  on $\dist(\hat{\bx}_{\rm SI},\bx)$. By (\ref{eq:divide_dir_nor}) we only need to bound $\dist_{\rm d}(\hat{\bx}_{\rm SI},\bx)$ and $\dist_{\rm n}(\hat{\bx}_{\rm SI},\bx)$ separately. 
  Note that $\hat{\bv}_{\rm SI}$ is the normalized leading eigenvector of $\hat{\bS}_{\bx}$, $\frac{\bx}{\|\bx\|_2}$ is the normalized leading eigenvector of $\mathbbm{E}\hat{\bS}_{\bx}$, and $\mathbbm{E}\hat{\bS}_{\bx}$ possesses an eigengap uniformly bounded away from zero  (see Lemma \ref{lem:ESx}). Thus, provided  $\|\hat{\bS}_{\bx}-\mathbbm{E}\hat{\bS}_{\bx}\|_{\rm op}$ being small enough, Davis-Kahan sin$\bTheta$ Theorem (e.g., \cite{davis1970rotation,chi2019nonconvex}) and Lemma \ref{lem:equ_dist_dmat}  give 
$$
    \dist\Big(\hat{\bv}_{\rm SI},\frac{\bx}{\|\bx\|_2}\Big) \le \Big\|\hat{\bv}_{\rm SI}\hat{\bv}_{\rm SI}^\top-\frac{\bx\bx^\top}{\|\bx\|_2^2}\Big\|_{\rm F}\le \sqrt{2}\Big\|\hat{\bv}_{\rm SI}\hat{\bv}_{\rm SI}^\top-\frac{\bx\bx^\top}{\|\bx\|_2^2}\Big\|_{\rm op} \le \frac{2}{c_0}\|\hat{\bS}_{\bx}-\mathbbm{E}\hat{\bS}_{\bx}\|_{\rm op}.$$ Next, we bound  $\dist_{\rm n}(\hat{\bx}_{\rm SI},\bx)=|\hat{\lambda}_{\rm SI}-\|\bx\|_2|$. 
 Since $\Phi(-\frac{\tau}{\|\bx\|_2})\in [\Phi(-\frac{\tau}{\alpha}),\Phi(-\frac{\tau}{\beta})]$, if additionally we have $$
    \Big|\frac{\hat{\lambda}_{\bx}}{2}-\Phi\Big(-\frac{\tau}{\|\bx\|_2}\Big)\Big|\le \min\Big\{\frac{1}{2}\Phi\Big(-\frac{\tau}{\alpha}\Big),\frac{1}{2}\Big(\frac{1}{2}-\Phi\Big(-\frac{\tau}{\beta}\Big)\Big)\Big\},$$ then  it holds that  $
    \frac{\hat{\lambda}_{\bx}}{2}\in \Big[\frac{1}{2}\Phi\Big(-\frac{\tau}{\alpha}\Big),\frac{1}{2}\Big(\frac{1}{2}+\Phi\Big(-\frac{\tau}{\beta}\Big)\Big)\Big].$ Note that on this compact interval contained in $(0,1)$, $\Phi^{-1}$ has first order derivative bounded by some $C_1$ only   only on $(\alpha,\beta,\gamma)$, thus yielding 
    \begin{align}
\label{eq:before_inverse}\Big|\Phi^{-1}\Big(\frac{\hat{\lambda}_{\bx}}{2}\Big)+\frac{\tau}{\|\bx\|_2}\Big| =\Big|\Phi^{-1}\Big(\frac{\hat{\lambda}_{\bx}}{2}\Big)-\Phi^{-1}\Big(\Phi\Big(-\frac{\tau}{\|\bx\|_2}\Big)\Big)\Big|\le C_1\Big|\frac{\hat{\lambda}_{\bx}}{2}-\Phi\Big(-\frac{\tau}{\|\bx\|_2}\Big)\Big|.
\end{align} 
If additionally we have $|\frac{\hat{\lambda}_{\bx}}{2}-\Phi(-\frac{\tau}{\|\bx\|_2})|\le \frac{1}{2C_1}\frac{\tau}{\beta},$ then (\ref{eq:before_inverse}) gives $\Phi^{-1}(\frac{\hat{\lambda}_{\bx}}{2})\le-\frac{\tau}{2\beta}$,   yielding  $$
   \big|\hat{\lambda}_{\rm SI}-\|\bx\|_2\big|=\frac{\|\bx\|_2}{|\Phi^{-1}(\frac{\hat{\lambda}_{\bx}}{2})|}\Big|\Phi^{-1}\Big(\frac{\hat{\lambda}_{\bx}}{2}\Big)+\frac{\tau}{\|\bx\|_2}\Big|\le \frac{2\beta^2C_1}{\tau}\Big|\frac{\hat{\lambda}_{\bx}}{2}-\Phi\Big(-\frac{\tau}{\|\bx\|_2}\Big)\Big|$$  
where in the last inequality we use (\ref{eq:before_inverse}) again.  Combining the above two bounds we arrive at this deterministic bound: if $\|\hat{\bS}_{\bx}-\mathbbm{E}\hat{\bS}_{\bx}\|_{\rm op}$ and $|\frac{\hat{\lambda}_{\bx}}{2}-\Phi(-\frac{\tau}{\|\bx\|_2})|$ are   smaller than some constant depending only on $(\alpha,\beta,\gamma)$, then  $$\dist(\hat{\bx}_{\rm SI}, \bx) \lesssim \big\|\hat{\bS}_{\bx}-\mathbbm{E}\hat{\bS}_{\bx}\big\|_{\rm op}+\Big| \hat{\lambda}_{\bx} -2\Phi\Big(-\frac{\tau}{\|\bx\|_2}\Big)\Big|.$$ We now establish the non-uniform final bound. Recall $m\ge C_2n$ for large enough $C_2$.  For a fixed $\bx\in \mathbbm{A}_{\alpha,\beta}$, we set $t=n$ in the non-uniform bound of Lemma \ref{lem:con_Sx_low} and $t=\sqrt{\frac{n}{m}}$ in the non-uniform bound of Lemma \ref{lem:con_lambda_x} to obtain $$\|\hat{\bS}_{\bx}-\mathbbm{E}\hat{\bS}_{\bx}\|_{\rm op} =\calO\Big(\sqrt{\frac{n}{m}}\Big)\quad\text{and}\quad \Big|\hat{\lambda}_{\bx}-2\Phi\Big(-\frac{\tau}{\|\bx\|_2}\Big)\Big|=\calO\Big(\sqrt{\frac{n}{m}}\Big)$$
    that 
  hold with probability at least $1-\exp(-c_3n)$. Combining with the deterministic bound yields the claim. For the uniform final bound, we only  need to invoke the uniform bounds in Lemma \ref{lem:con_Sx_low} and Lemma \ref{lem:con_lambda_x} instead.
\end{proof}
\subsubsection{The Proof of Theorem \ref{thm:SI_high}}
\begin{proof}
    We first establish a deterministic bound on $\dist(\hat{\bx}_{\rm SI},\bx)$.  By (\ref{eq:divide_dir_nor}) we only need to bound $\dist_{\rm d}(\hat{\bx}_{\rm SI},\bx)$ and $\dist_{\rm n}(\hat{\bx}_{\rm SI},\bx)$ separately.  Compared to the proof of Theorem \ref{thm:SI_low}, bounding $\dist_{\rm d}(\hat{\bx}_{\rm SI},\bx)=\dist(\hat{\bv}_{\rm SI},\frac{\bx}{\|\bx\|_2})$  requires  additional attempts.  We need to first  decompose the error into $$
     \dist\Big(\hat{\bv}_{\rm SI},\frac{\bx}{\|\bx\|_2}\Big)\le \dist\Big(\hat{\bv}_{\rm SI},\frac{\bx_{\calS_{\bx}}}{\|\bx_{\calS_{\bx}}\|_2}\Big)+\dist\Big(\frac{\bx_{\calS_{\bx}}}{\|\bx_{\calS_{\bx}}\|_2},\frac{\bx}{\|\bx\|_2}\Big).$$  
For bounding $\dist(\hat{\bv}_{\rm SI},\frac{\bx_{\calS_{\bx}}}{\|\bx_{\calS_{\bx}}\|_2})$,  we note that $\hat{\bv}_{\rm SI}$ and $\frac{\bx_{\calS_{\bx}}}{\|\bx_{\calS_{\bx}}\|_2}$ are the normalized leading eigenvectors of $[\hat{\bS}_{\bx}]_{\calS_{\bx}}$ and $ [\mathbbm{E}\hat{\bS}_{\bx}]_{\calS_{\bx}}=\mathbbm{E}[\hat{\bS}_{\bx}]_{\calS_{\bx}}$, respectively.  Lemma \ref{lem:ESx} delivers $$[\mathbbm{E}\hat{\bS}_{\bx}]_{\calS_{\bx}}=a_{\bx}[\bx\bx^\top]_{\calS_{\bx}} + b_{\bx}[\bI_n]_{\calS_{\bx}} = a_{\bx}\bx_{\calS_{\bx}}\bx_{\calS_{\bx}}^\top+ b_{\bx}[\bI_n]_{\calS_{\bx}}$$
 which exhibits an eigengap $$\lambda_1\big([\mathbbm{E}\hat{\bS}_{\bx}]_{\calS_{\bx}}\big)-\lambda_2\big([\mathbbm{E}\hat{\bS}_{\bx}]_{\calS_{\bx}}\big)=a_{\bx}\|\bx_{\calS_{\bx}}\|_2^2\gtrsim \|\bx_{\calS_{\bx}}\|_2^2.$$ Thus, if we have $\|\bx-\bx_{\calS_{\bx}}\|_2\le \frac{\alpha}{2}$, then $\|\bx_{\calS_{\bx}}\|_2\ge \|\bx\|_2-\|\bx-\bx_{\calS_{\bx}}\|_2\ge \frac{\alpha}{2}$, which further implies $   \lambda_1\big([\mathbbm{E}\hat{\bS}_{\bx}]_{\calS_{\bx}}\big)-\lambda_2\big([\mathbbm{E}\hat{\bS}_{\bx}]_{\calS_{\bx}}\big)\ge c_0$ for some $c_0>0$ depending only on $(\alpha,\beta,\tau)$.  Moreover, if additionally we have $
     \max_{\calS:|\calS|=k}\big\|[\hat{\bS}_{\bx}]_{\calS}-\mathbbm{E}[\hat{\bS}_{\bx}]_{\calS}\big\|_{\rm op}\le \frac{c_0}{2}$ that  ensures $\|[\hat{\bS}_{\bx}]_{\calS_{\bx}}-\mathbbm{E}[\hat{\bS}_{\bx}]_{\calS_{\bx}}\|_{\rm op}\le \frac{c_0}{2}$, then   Davis-Kahan sin$\bTheta$ Theorem and Lemma \ref{lem:equ_dist_dmat} yield $$ \dist\Big(\hat{\bv}_{\rm SI},\frac{\bx}{\|\bx\|_2}\Big)\le \Big\|\hat{\bv}_{\rm SI}\hat{\bv}_{\rm SI}^\top-\frac{\bx\bx^\top}{\|\bx\|_2^2}\Big\|_{\rm F}\le\sqrt{2}\Big\|\hat{\bv}_{\rm SI}\hat{\bv}_{\rm SI}^\top-\frac{\bx\bx^\top}{\|\bx\|_2^2}\Big\|_{\rm op}\le \frac{2}{c_0}\max_{{\calS\subset [n],|\calS|=k}} \big\|[\hat{\bS}_{\bx}]_\calS-\mathbbm{E}[\hat{\bS}_{\bx}]_\calS\big\|_{\rm op}.$$ Next, let us Bound $\dist\big(\frac{\bx_{\calS_{\bx}}}{\|\bx_{\calS_{\bx}}\|_2},\frac{\bx}{\|\bx\|_2}\big)$, which follows from some algebra: 
$$    \dist\left(\frac{\bx_{\calS_{\bx}}}{\|\bx_{\calS_{\bx}}\|_2},\frac{\bx}{\|\bx\|_2}\right) = \Big\|\frac{\bx_{\calS_{\bx}}}{\|\bx_{\calS_{\bx}}\|_2}-\frac{\bx}{\|\bx\|_2}\Big\|_2   \le \left\|\frac{\bx}{\|\bx\|_2}-\frac{\bx_{\calS_{\bx}}}{\|\bx\|_2}\right\|_{2} +\left\|\frac{\bx_{\calS_{\bx}}}{\|\bx\|_2}-\frac{\bx_{\calS_{\bx}}}{\|\bx_{\calS_{\bx}}\|_2}\right\|_2\le  \frac{2\|\bx-\bx_{\calS_{\bx}}\|_2}{\alpha}.$$ We shall now turn to bounding   $\dist_{\rm n}(\hat{\bx}_{\rm SI},\bx)=|\hat{\lambda}_{\rm SI}-\|\bx\|_2|$. 
 This can be done via exactly the same arguments as  (\ref{eq:before_inverse}), thus we     state the conclusion without re-iterating the technical details: if $|\hat{\lambda}_{\bx}-2\Phi(-\frac{\tau}{\|\bx\|_2})|\le c_1$ for some $c_1$ only depending on $(\alpha,\beta,\tau)$, then we have $$
     \big|\hat{\lambda}_{\rm SI}-\|\bx\|_2\big| \le C_2 \Big|\frac{\hat{\lambda}_{\bx}}{2}-\Phi\Big(-\frac{\tau}{\|\bx\|_2}\Big)\Big|.$$ Combining these pieces, we obtain the deterministic bound: we have 
 \begin{align}\label{eq:deter_bound_SI_high}
     \dist(\hat{\bx}_{\rm SI},\bx) \lesssim\max_{\substack{\calS\subset [n]\\|\calS|=k}} \big\|[\hat{\bS}_{\bx}]_\calS-\mathbbm{E}[\hat{\bS}_{\bx}]_\calS\big\|_{\rm op}+\big\|\bx-\bx_{\calS_{\bx}}\big\|_2 + \Big|\frac{\hat{\lambda}_{\bx}}{2}-\Phi\Big(-\frac{\tau}{\|\bx\|_2}\Big)\Big|, 
 \end{align}
 provided that each term on the right-hand side of (\ref{eq:deter_bound_SI_high}) is smaller than some   constant depending on $(\alpha,\beta,\tau)$. Now we establish the non-uniform final bound. 
  Suppose $m\gtrsim k\log(\frac{en}{k})$ with large enough implied constant.  
 For a fixed $\bx\in \Sigma^n_k\cap \mathbbm{A}_{\alpha,\beta}$, we set $t=k\log(\frac{en}{k})$ in  the non-uniform bounds of Lemma \ref{lem:con_Sx_high} and Lemma \ref{lem:con_lambda_x} to obtain 
 $$
     \max_{{\calS\subset [n],|\calS|=k}} \big\|[\hat{\bS}_{\bx}]_\calS-\mathbbm{E}[\hat{\bS}_{\bx}]_\calS\big\|_{\rm op}+ \Big|\frac{\hat{\lambda}_{\bx}}{2}-\Phi\Big(-\frac{\tau}{\|\bx\|_2}\Big)\Big|\lesssim 
 \sqrt{\frac{k\log(\frac{en}{k})}{m}}$$  that holds with probability at least $1-\exp(-c_2 k\log(\frac{en}{k}))$. Then, the non-uniform bound in Lemma \ref{lem:supp_est} gives $
   \big\|\bx-\bx_{\calS_{\bx}}\big\|_2\lesssim \big(\frac{k^2\log n}{m}\big)^{1/4}$ that holds with probability at least $1-n^{-10}$. Therefore, if $m\ge C_3 k^2\log n$ with large enough $C_3$, then the terms in the right-hand side of (\ref{eq:deter_bound_SI_high})   are small enough, and substituting these bounds into (\ref{eq:deter_bound_SI_high}) yields  non-uniform final bound.  For the uniform final bound, we can similarly apply the uniform bounds in Lemmas \ref{lem:con_Sx_high}--\ref{lem:con_lambda_x}. The proof is complete. 
\end{proof}

 \subsection{Other Standard Proofs}\label{appother}

 \subsubsection{The Proof of Lemma \ref{lem:intersect_affine}}
\begin{proof}
       We first claim this: $I_{m,k}$ can be equivalently interpreted as the maximal number of (open) regions that $\mathbb{R}^k$ can be decomposed by $m$ ($(k-1)$-dimensional) affine hyperplanes. To see this, let $\calV$ be a $k$-dimensional affine space embedded in $\mathbb{R}^m$ not contained in any of the following hyperplanes: $
        \{\bu\in \mathbb{R}^m:u_i=a\},~i\in[m],~a\in \mathbb{R},$ 
    %$\{\bu\in \mathbb{R}^m:u_i=a\}$ with $i\in [m]$ and $a\in \mathbb{R}$, 
    then the orthants intersected by $\calV$ stand in one-to-one corrspondence with the regions of $\calV$ (that is now viewed as $\mathbb{R}^k$) decomposed by the $m$ affine hyperplanes $\{\bu\in \mathbb{R}^m:u_i=a\}$ with $i\in[m]$. Such rephrasing lowers the problem dimension  from $m$ to $k$. It is easy to see $I_{1,k}=2$ since  $1$ affine hyperplane separates $\mathbb{R}^k$ into $2$ regions, and  $I_{m,1}=m+1$ since $m$ distinct points separate $\mathbb{R}$ into $m+1$ intervals. 
    Next, let us establish a recursive inequality
    \begin{align}\label{eq:recurine}
        I_{m,k}\le I_{m-1,k}+I_{m-1,k-1},\quad\forall m,k\ge 2.
    \end{align}
    Let $\calH_1,...,\calH_m$ be $m$ affine hyperplanes in $\mathbb{R}^k$ that decompose $\mathbb{R}^k$ into $I_{m,k}$ regions, then it can be computed by adding the following two quantities together:
    \begin{subequations}
        \begin{gather}
            R_{1:m-1} := \text{the number of regions separated by the first $m-1$ hyperplanes $\calH_1,...,\calH_{m-1}$},\label{eq:R1m-1}\\
            R_m:= \text{the number of the $R_{1:m-1}$ regions in (\ref{eq:R1m-1}) intersected by the $m$-th hyperplane $\calH_{m}$} . 
        \end{gather}
    \end{subequations}
    %``the number of regions separated by the first $m-1$ hyperplanes $\calH_1,...,\calH_{m-1}$'', denoted $R_{1:m-1}$, and 
  %  ``the number of these $R_{1:m-1}$ regions intersected by the $m$-th hyperplane $\calH_{m}$'', denoted $R_m$. 
  We thus obtain $I_{m,k} = R_{1:m-1}+R_m,$
and all that remains is to bound $R_{1:m-1}$ and $R_m.$ It is evident that $R_{1:m-1}\le I_{m-1,k}$. To bound $R_{m}$, we again utilize the dimension reduction argument at the beginning of the proof. Specifically, due to one-to-one correspondence, $R_m$ can be equivalently interpreted as the number of regions that $\calH_m$ is decomposed by its $m-1$ affine hyperplanes $\calH_1\cap \calH_m,~\calH_2\cap \calH_m,~...,~\calH_{m-1}\cap \calH_m$. Identifying $\calH_m$ with $\mathbb{R}^{k-1}$  immediately yields the bound $R_m\le I_{m-1,k-1}$. Taken collectively, we   arrive at (\ref{eq:recurine}). We define $\{\tilde{I}_{m,k}\}_{k,m\ge 1}$ by the recursion and the initial values  as follows: 
\begin{gather*}
    \tilde{I}_{m,k}= \tilde{I}_{m-1,k}+\tilde{I}_{m-1,k-1},\quad\forall\, m,k\ge 2;\\
    \tilde{I}_{1,k}=2,~\tilde{I}_{m,1}=m+1,\quad\forall\,m,k\ge 1.
\end{gather*} 
Then, some algebra finds a closed form formula $$
     \tilde{I}_{m,k}=\sum_{i=0}^{\min\{k, m\}}\binom{m}{i}.$$ 
 Therefore, an upper bound on the $\{I_{m,k}\}_{k,m\ge 1}$ of interest  follows: $$
     I_{m,k}\le \tilde{I}_{m,k}\le \sum_{i=0}^{\min\{k,m\}}\binom{m}{i}.$$ 
 Hence, it holds for $k\le m$ that $$
     I_{m,k}\le \sum_{i=0}^k\binom{m}{i}\le \big(\frac{em}{k}\big)^k,$$ where the last inequality can be shown by induction or the existing tighter bound \cite[Lem. 16.19]{flum2006subexponential}.  
\end{proof}

\subsubsection{The Proof of Lemma \ref{lem:sepa_im_sepa}}
 \begin{proof}
     We first show the statement in the first dot point of Lemma \ref{lem:sepa_im_sepa}, and then show that it immediately implies the second dot point therein. If $\bu\in \calH_{|\ba|}^+$, then $|\ba^\top\bu|\ge\tau$, which along with our assumption gives $
     |\ba^\top\bu|>\tau+2\min\big\{|\ba^\top(\bp-\bu)|,|\ba^\top(\bp+\bu)|\big\}.$ 
 Moreover, triangle inequality yields $
         |\ba^\top\bp| \ge |\ba^\top\bu|-\min\big\{|\ba^\top(\bp-\bu)|,|\ba^\top(\bp+\bu)|\big\} >\tau,$ 
 hence we also have $\bp\in \calH_{|\ba|}^+$. Analogously,  $\bu\in\calH_{|\ba|}^{-}$ and our assumption taken collectively yield $\bp\in\calH_{|\ba|}^{-}$.

 We then prove the second statement. 
By $\calH_{|\ba|}$ $\theta$-well-separates $\bu$ and $\bv$, we have $||\ba^\top\bu|-\tau|\ge \theta\dist(\bu,\bv)$ and $||\ba^\top\bv|-\tau|\ge\theta\dist(\bu,\bv)$, which along with our assumption imply $||\ba^\top\bu|-\tau|>2 \min\big\{|\ba^\top(\bp-\bu)|,|\ba^\top(\bp+\bu)|\big\}$ and $||\ba^\top\bv|-\tau|>2\min\big\{|\ba^\top(\bq-\bv)|,|\ba^\top(\bq+\bv)|\big\}.$ 
     Therefore, by the first statement that we proved, we have that $\bu$ and $\bp$ live in the same side of $\calH_{|\ba|}$, and that $\bv$ and $\bq$ also live in the same side of $\calH_{|\ba|}$. Because $\theta$-well-separation implies separation, $\bu$ and $\bv$ live in different sides of $\calH_{|\ba|}$ (i.e., they are separated by $\calH_{|\ba|}$). Taken collectively, it follows that $\bp$ and $\bq$ live in different sides of $\calH_{|\ba|}$, which completes the proof. 
 \end{proof}

 \subsubsection{The Proof of Lemma \ref{lem:Pthetauv}}
\begin{proof} 
     $\sfP_{\theta,\bu,\bv}\le C_3\dist(\bu,\bv)$ is directly implied by Lemma \ref{lem:Puv} due to $\sfP_{\theta,\bu,\bv}\le\sfP_{\bu,\bv}$, hence we only need to prove the converse $\sfP_{\theta,\bu,\bv}\ge c_2\dist(\bu,\bv)$. We accomplish this by proceeding as follows:
    \begin{subequations}
        \begin{align}\nn
        &\sfP_{\theta,\bu,\bv}=\mathbbm{P}\big(\calH_{|\ba|}\text{ separates }\bu\text{ and }\bv,~(\ref{eq:bound_away_tau})\text{ holds}\big) \\\label{eq:union1}
        &\ge \mathbbm{P}\big(\calH_{|\ba|}\text{ separates }\bu\text{ and }\bv\big)-\mathbbm{P}((\ref{eq:bound_away_tau})\text{ fails})\\\label{eq:union2}
        &\ge \sfP_{\bu,\bv}-\mathbbm{P}(||\ba^\top\bu|-\tau|<\theta\dist(\bu,\bv))-\mathbbm{P}(||\ba^\top\bv|-\tau|<\theta\dist(\bu,\bv))\\\label{eq:usePuv}
        &\ge c_{2}'\dist(\bu,\bv)-\frac{4\theta\cdot \dist(\bu,\bv)}{\sqrt{2\pi}\alpha}-\frac{4\theta\cdot \dist(\bu,\bv)}{\sqrt{2\pi}\alpha} \\&\ge c_2 \dist(\bu,\bv), \label{eq:smallou}
    \end{align}
    \end{subequations}
    where (\ref{eq:union1}) and (\ref{eq:union2}) are due to union bound; in (\ref{eq:usePuv}) we use $\sfP_{\bu,\bv}\ge c_2'\dist(\bu,\bv)$ for some constant $c_2'$ from Lemma \ref{lem:Puv}, and then  apply $$ \mathbbm{P}\big(||\ba^\top\bu|-\tau|<\theta\dist(\bu,\bv)\big)\le \mathbbm{P}\left(\Big||\calN(0,1)|-\frac{\tau}{\|\bu\|_2}\Big|<\frac{\theta}{\alpha}\dist(\bu,\bv)\right)\le \frac{4\theta\dist(\bu,\bv)}{\sqrt{2\pi}\alpha}$$
    and the same upper bound to $\mathbbm{P}(||\ba^\top\bv|-\tau|<\theta\dist(\bu,\bv))$;  (\ref{eq:smallou}) holds with $c_2=\frac{c_2'}{2}$ as long as $\theta<\frac{\sqrt{2\pi}c_2'\alpha}{16}$. The proof is complete. 
\end{proof}
\subsubsection{The Proof of Lemma \ref{lem:norm_equa}}
\begin{proof} 
    The second inequality follows from \begin{align*}
        &\dist(\bu,\bv) = \|\bu\|_2\cdot \dist \left(\frac{\bu}{\|\bu\|_2},\frac{\bv}{\|\bu\|_2}\right)
        \\&\le \|\bu\|_2 \cdot \left[\dist \left(\frac{\bu}{\|\bu\|_2},\frac{\bv}{\|\bv\|_2}\right)+\dist\left(\frac{\bv}{\|\bv\|_2},\frac{\bv}{\|\bu\|_2}\right)\right]\\&\le \beta \dist_{\rm d}(\bu,\bv) +  \|\bu\|_2\cdot \left\|\frac{\bv}{\|\bv\|_2}- \frac{\bv}{\|\bu\|_2}\right\|_2  \\
        &= \beta \dist_{\rm d}(\bu,\bv) + \dist_{\rm n}(\bu,\bv).
    \end{align*} 
   Also, for any $\bu,\bv\in \mathbbm{A}_{\alpha,\beta}$ we have \begin{align*}
       &\dist(\bu,\bv) = \sqrt{\min\big\{\|\bu-\bv\|_2^2,\|\bu+\bv\|_2^2\big\}}  \\
       &= \sqrt{\|\bu\|_2^2 + \|\bv\|_2^2 -2|\bu^\top\bv|}  \\&= \sqrt{\|\bu\|_2\|\bv\|_2}\cdot\sqrt{\frac{\|\bu\|_2}{\|\bv\|_2}+\frac{\|\bv\|_2}{\|\bu\|_2}-\frac{2|\bu^\top\bv|}{\|\bu\|_2\|\bv\|_2}} \\&\ge\alpha \sqrt{2-\frac{2|\bu^\top\bv|}{\|\bu\|_2\|\bv\|_2}}\\&= \alpha \dist_{\rm d}(\bu,\bv).
   \end{align*}
  Moreover,  $\dist(\bu,\bv)$ is evidently lower bounded by $\dist_{\rm n}(\bu,\bv)$, since we can  apply Cauchy-Schwarz to obtain
  $$\dist(\bu,\bv) = \sqrt{\|\bu\|_2^2+\|\bv\|_2^2 - 2|\bu^\top\bv|}  \ge \sqrt{\|\bu\|_2^2+\|\bv\|_2^2- 2\|\bu\|_2\|\bv\|_2} = \big|\|\bu\|_2-\|\bv\|_2\big| = \dist_{\rm n}(\bu,\bv).$$
   The proof is complete.
\end{proof}
\subsubsection{The Proof of Lemma \ref{lem:equ_dist_dmat}}
\begin{proof} 
     The result follows from some simple algebra: 
        \begin{align*}
            & \|\bu\bu^\top-\bv\bv^\top\|_{\rm F}^2 = \|\bu\|_2^4+\|\bv\|_2^4 - 2 (\bu^\top\bv)^2  \\
            &\ge \frac{1}{2}\left[(\|\bu\|_2^2+\|\bv\|_2^2)^2-4(\bu^\top\bv)^2\right]\\
            &= \frac{1}{2}\left[\|\bu\|_2^2+\|\bv\|_2^2-2|\bu^\top\bv|\right]\left[\|\bu\|_2^2+\|\bv\|_2^2+2|\bu^\top\bv|\right]\\&\ge \alpha^2 \left[\|\bu\|_2^2+\|\bv\|_2^2-2|\bu^\top\bv|\right]  =\alpha^2\dist^2(\bu,\bv);
            \end{align*}
            and
            \begin{align*}
            & \|\bu\bu^\top-\bv\bv^\top\|_{\rm F}^2 = \|\bu\|_2^4+\|\bv\|_2^4 - 2 (\bu^\top\bv)^2  \\&
            =( \|\bu\|_2^2 + \|\bv\|_2^2 )^2-4(\bu^\top\bv)^2+2(\bu^\top\bv)^2 - 2\|\bu\|_2^2\|\bv\|_2^2 \\
            &\le ( \|\bu\|_2^2 + \|\bv\|_2^2 )^2-4(\bu^\top\bv)^2   \\&= \left[\|\bu\|_2^2+\|\bv\|_2^2 -2|\bu^\top\bv|\right]\left[\|\bu\|_2^2+\|\bv\|_2^2 +2|\bu^\top\bv|\right]\\&\le 4\beta^2\dist^2(\bu,\bv). 
        \end{align*}
        The proof is complete.
\end{proof}
\subsubsection{The Proof of Lemma \ref{lem:parameterization}}
\begin{proof}
    This lemma consists of some  simple facts in linear algebra.   We separately deal with the two cases:   $\bu,\bv$ are not parallel; $\bu = \lambda\bv$ for some $\lambda\ne0$. Note that  the last statement follows immediately from (\ref{eq:explicit_uuv}): 
    If  $\|\bu\|_2=\|\bv\|_2$ holds, it is easy to verify $v_1=-u_1$ by $u_1=\frac{\|\bu\|_2^2-\bu^\top\bv}{\|\bu-\bv\|_2}$ and $v_1=\frac{\bu^\top\bv-\|\bv\|_2^2}{\|\bu-\bv\|_2}$, then combining  with the condition $u_1>v_1$ gives $u_1>0$.   
    \paragraph{(Case 1: Non-Parallel $\bu,\bv$)}
    In this case, $\bu\ne \lambda\bv$ for any $\lambda\neq 0$, then we consider $$\bbeta_2 = \frac{\langle \bv,\bv-\bu\rangle \bu+ \langle \bu,\bu-\bv\rangle\bv}{\|\langle \bv,\bv-\bu\rangle \bu+ \langle \bu,\bu-\bv\rangle\bv\|_2},$$ which is meaningful since some algebra verifies $ \|\langle \bv,\bv-\bu\rangle \bu+ \langle \bu,\bu-\bv\rangle\bv\|_2^2 = \|\bu-\bv\|_2^2\cdot \big[\|\bu\|_2^2\|\bv\|_2^2-(\bu^\top\bv)^2\big]>0.$ 
    It is   straightforward to verify the orthogonality between   $\bbeta_2$ and  $\bbeta_1 = \frac{\bu-\bv}{\|\bu-\bv\|_2}$: $$\big\langle \bu-\bv,\langle \bv,\bv-\bu\rangle \bu+\langle \bu,\bu-\bv\rangle \bv\big\rangle = 0.$$
    Therefore, $(\bbeta_1,\bbeta_2)$ is an orthogonal basis for $\text{\rm span}(\bu,\bv)$, which leads to $$\bu = \langle \bu ,\bbeta_1\rangle \bbeta_1+ \langle \bu,\bbeta_2\rangle \bbeta_2:=u_1\bbeta_1+u_2\bbeta_2\quad\text{and}\quad\bv = \langle\bv, \bbeta_1\rangle\bbeta_1 + \langle \bv,\bbeta_2\rangle \bbeta_2:=v_1\bbeta_1 + v_2\bbeta_2$$ for some coordinates $(u_1,u_2,v_1,v_2)$. 
    Note that $\langle\bbeta_2,\bu-\bv\rangle=0$ implies $v_2=u_2$.  Moreover, by 
   $
        \big\langle\bu,\langle\bv,\bv-\bu\rangle\bu+\langle \bu,\bu-\bv\rangle\bv\big\rangle = \|\bu\|_2^2\|\bv\|_2^2-(\bu^\top\bv)^2>0$ we obtain $u_2>0$. It remains to verify $u_1>v_1$, and this follows from $
        u_1-v_1 = \langle \bu-\bv,\bbeta_1\rangle=\|\bu-\bv\|_2>0.$ 
   For this case all that remains is to show  (\ref{eq:explicit_uuv}). We get started from $  u_1 = \langle\bu,\bbeta_1\rangle =\frac{\langle \bu,\bu-\bv\rangle}{\|\bu-\bv\|_2}$ and $v_1= \langle\bv,\bbeta_1\rangle=\frac{\langle\bv,\bu-\bv\rangle}{\|\bu-\bv\|_2}$, which are uniquely determined by $(\bu,\bv)$. 
 We also have $u_2 = \langle\bu,\bbeta_2\rangle$, and we write    $$
     \bbeta_2 = \frac{\langle \bv,\bv-\bu\rangle\bu+\langle\bu,\bu-\bv\rangle\bv}{\|\bu-\bv\|_2\cdot (\|\bu\|_2^2\|\bv\|_2^2-(\bu^\top\bv)^2)^{1/2}}$$ to obtain 
$$
        u_2= \langle\bu,\bbeta_2\rangle = \frac{\|\bu\|_2^2\|\bv\|_2^2-(\bu^\top\bv)^2}{\|\bu-\bv\|_2\sqrt{\|\bu\|_2^2\|\bv\|_2^2-(\bu^\top\bv)^2}}= \frac{\sqrt{\|\bu\|_2^2\|\bv\|_2^2-(\bu^\top\bv)^2}}{\|\bu-\bv\|_2}.$$
  Since $u_1,v_1$ are uniquely determined by $(\bu,\bv)$, the uniqueness of $u_2$ also follows from the equations: $ u_1^2+u_2^2=\|\bu\|_2^2,~u_2\ge 0,~u_1=\frac{\langle\bu,\bu-\bv\rangle}{\|\bu-\bv\|_2}.$

    \paragraph{(Case 2: Parallel $\bu,\bv$)} Since   $\bu,\bv\in \mathbbm{A}_{\alpha,\beta}$ and $\bu\neq \bv$, in this case we have $\bv=\lambda_0\bu$ for some non-zero $\lambda_0\ne 1$. We simply select a specific $\bbeta_2\in \mathbb{R}^n$ (which is not unique when $n>2$) satisfying $\|\bbeta_2\|_2=1$ and $\langle\bbeta_1,\bbeta_2\rangle=0$, and note that we can simply use $
        \bbeta_1=\frac{\bu-\bv}{\|\bu-\bv\|_2}=\frac{\sign(1-\lambda_0)\cdot\bu}{\|\bu\|_2}$ to linearly express $(\bu,\bv)$ as $$
    \bu = \sign(1-\lambda_0)\|\bu\|_2\bbeta_1 \quad\text{and}\quad  \bv=\lambda_0\bu = \lambda_0\sign(1-\lambda_0)\|\bu\|_2\bbeta_1.$$ 
     Thus we have $u_2=0$, and $(u_1,v_1)=(\sign(1-\lambda_0)\|\bu\|_2,\lambda_0\sign(1-\lambda_0)\|\bu\|_2)$ that satisfies $u_1-v_1=|1-\lambda_0|\|\bu\|_2>0$. %By expanding 
     %$(\bbeta_1,\bbeta_2)$ to an orthonormal basis in $\mathbb{R}^n$,
     % (\ref{eq:para_orthomatrix}) immediately follows from (\ref{eq:para_decom}). 
  Recall that $u_1$ and $v_1$ are already uniquely determined by $(\bu,\bv)$ as $u_1=\frac{\langle\bu,\bu-\bv\rangle}{\|\bu-\bv\|_2}$ and $v_1=\frac{\langle\bv,\bu-\bv\rangle}{\|\bu-\bv\|_2}$. Also we have shown $u_2=0=\frac{(\|\bu\|_2^2\|\bv\|_2^2-(\bu^\top\bv)^2)^{1/2}}{\|\bu-\bv\|_2}$ (the second equality is due to $\bv=\lambda_0\bu$), hence the formula for $u_2$ in (\ref{eq:explicit_uuv}) remains valid. Again, $u_2$ is uniquely determined by $(\bu,\bv)$ due to  the arguments used in Case 1.
       The proof is complete.
\end{proof}
\subsubsection{The Proof of Lemma \ref{lem:maximum}}
\begin{proof} 
    Using $\eta = \sqrt{\frac{\pi e}{2}}\tau$ and the definitions of $g_\eta(a,b),h_\eta(a,b)$, we have 
     $$\eta g_\eta (a,b) = \exp\big(\frac{a^2+b^2-\tau^2}{2(a^2+b^2)}\big)\frac{\tau(\tau^2a^2+b^2(a^2+b^2))}{(a^2+b^2)^{5/2}}$$ and $$\eta h_\eta(a,b)=   \exp\big(\frac{a^2+b^2-\tau^2}{2(a^2+b^2)}\big)\frac{\tau ab(a^2+b^2-\tau^2)}{(a^2+b^2)^{5/2}}.$$   
     Substituting them into $F(\eta,a,b)$ and then using the polar coordinates $u_1=\rho\cos(\theta),u_2=\rho\sin(\theta)$ with the feasible domain $\rho\in[\alpha,\beta],\theta\in[0,2\pi]$,  we proceed as 
    \begin{subequations}\label{eq:cal_fixedeta}
        \begin{align}
&F^2\Big(\eta,a,b\Big)\nn\\\nn&= \left(1- \exp\Big(\frac{a^2+b^2-\tau^2}{2(a^2+b^2)}\Big)\frac{\tau[\tau^2a^2+b^2(a^2+b^2)]}{(a^2+b^2)^{5/2}}\right)^2 +\exp \Big(\frac{a^2+b^2-\tau^2}{a^2+b^2}\Big)\frac{\tau^2a^2b^2(a^2+b^2-\tau^2)^2}{(a^2+b^2)^5}\\\nn
&=1-2\exp\Big(\frac{a^2+b^2-\tau^2}{2(a^2+b^2)}\Big)\cdot\frac{\tau[\tau^2a^2+b^2(a^2+b^2)]}{(a^2+b^2)^{5/2}}\\\nn
&\quad+\exp \Big(\frac{a^2+b^2-\tau^2}{a^2+b^2}\Big)\frac{\tau^2[a^2b^2(a^2+b^2-\tau^2)^2+(\tau^2a^2+b^2(a^2+b^2))^2]}{(a^2+b^2)^5}\\\nn
& = 1-2\exp\Big(\frac{1}{2}-\frac{\tau^2}{2\rho^2}\Big)\Big(\Big(\frac{\tau}{\rho}\Big)^3\cos^2(\theta)+\frac{\tau}{\rho}\sin^2(\theta)\Big) +^2\exp\Big(1-\frac{\tau^2}{\rho^2}\Big)\Big(\Big(\frac{\tau}{\rho}\Big)^6\cos^2(\theta)+\Big(\frac{\tau}{\rho}\Big)^2\sin^2(\theta)\Big) \\
& = \Big[1- w^3\exp\Big(\frac{1-w^2}{2}\Big)\Big]^2 +  \exp\Big(\frac{1-w^2}{2}\Big)w(1-w^2)\Big[(w+w^3)\exp\Big(\frac{1-w^2}{2}\Big)-2\Big]\sin^2(\theta)\label{eq:uequaltaurho}\\
&\le \left(\max\Big\{\Big|1- w^3\exp\Big(\frac{1-w^2}{2}\Big)\Big|,1- w\cdot\exp\Big(\frac{1-w^2}{2}\Big)\Big\}\right)^2,  \label{eq:opti_theta_3}
\end{align}
    \end{subequations}
    where in (\ref{eq:uequaltaurho}) we  let $w=\frac{\tau}{\rho}\in [\frac{\tau}{\beta},\frac{\tau}{\alpha}]$, use $\cos^2(\theta)=1-\sin^2(\theta)$ and then perform some rearrangement; also, since (\ref{eq:uequaltaurho}) is linear on $\sin^2(\theta)\in[0,1]$, its maximum is either attained at $\sin^2(\theta)=0$ or $\sin^2(\theta)=1$, and further note that  (\ref{eq:uequaltaurho}) becomes $[1- w^3\exp(\frac{1-w^2}{2})]^2$ when $\sin^2(\theta)=0$, or $(1- w\exp(\frac{1-w^2}{2}))^2$ when $\sin^2(\theta)=1$, hence (\ref{eq:opti_theta_3}) follows.  
    It is evident that (\ref{eq:opti_theta_3}) is attainable. 
    By further optimizing $w\in[\frac{\tau}{\beta},\frac{\tau}{\alpha}]$, (\ref{eq:cal_fixedeta}) yields  the desired claim. 
\end{proof}

\end{appendix}

%\end{localsize}
\end{document}